\documentclass[11pt]{amsart}

\hsize \textwidth
\advance \hsize by -\marginparwidth
\evensidemargin \oddsidemargin

\usepackage{bm}
\usepackage{multirow}
\usepackage{xcolor}
\usepackage{thmtools}
\usepackage{caption}
\usepackage{amsmath,amsfonts, amsthm, amssymb,cancel}
\usepackage{microtype}
\usepackage{enumerate}
\usepackage{stackengine,graphicx}
\graphicspath{{figures/}}
\DeclareFontFamily{U}{rcjhbltx}{}
\DeclareFontShape{U}{rcjhbltx}{m}{n}{<->rcjhbltx}{}
\DeclareSymbolFont{hebrewletters}{U}{rcjhbltx}{m}{n}

\newcommand\mo[2][c]{%
  \bgroup%
  \setstackEOL{ }%
  \setstackgap{L}{0pt}%
  \Longstack[#1]{#2}%
  \egroup%
}

\let\aleph\relax\let\beth\relax
\let\gimel\relax\let\daleth\relax

\DeclareMathSymbol{\aleph}{\mathord}{hebrewletters}{39}
\DeclareMathSymbol{\beth}{\mathord}{hebrewletters}{98}
\DeclareMathSymbol{\gimel}{\mathord}{hebrewletters}{103}
\DeclareMathSymbol{\daleth}{\mathord}{hebrewletters}{100}

\DeclareMathSymbol{\lamed}{\mathord}{hebrewletters}{108}
\DeclareMathSymbol{\mem}{\mathord}{hebrewletters}{109}
\DeclareMathSymbol{\ayin}{\mathord}{hebrewletters}{96}
\DeclareMathSymbol{\tsadi}{\mathord}{hebrewletters}{118}
\DeclareMathSymbol{\qof}{\mathord}{hebrewletters}{113}
\DeclareMathSymbol{\shin}{\mathord}{hebrewletters}{152}

\makeatletter
\@ifpackageloaded{MnSymbol}\@tempswafalse\@tempswatrue
\if@tempswa
  \DeclareFontFamily{U}{MnSymbolC}{}
  \DeclareSymbolFont{MnSyC}{U}{MnSymbolC}{m}{n}
  \SetSymbolFont{MnSyC}{bold}{U}{MnSymbolC}{b}{n}
  \DeclareFontShape{U}{MnSymbolC}{m}{n}{
      <-6>  MnSymbolC5
     <6-7>  MnSymbolC6
     <7-8>  MnSymbolC7
     <8-9>  MnSymbolC8
     <9-10> MnSymbolC9
    <10-12> MnSymbolC10
    <12->   MnSymbolC12}{}
  \DeclareFontShape{U}{MnSymbolC}{b}{n}{
      <-6>  MnSymbolC-Bold5
     <6-7>  MnSymbolC-Bold6
     <7-8>  MnSymbolC-Bold7
     <8-9>  MnSymbolC-Bold8
     <9-10> MnSymbolC-Bold9
    <10-12> MnSymbolC-Bold10
    <12->   MnSymbolC-Bold12}{}
  \DeclareMathSymbol{\aleph}{\mathord}{MnSyC}{"AF}
  \DeclareMathSymbol{\beth}{\mathord}{MnSyC}{"B0}
  \DeclareMathSymbol{\gimel}{\mathord}{MnSyC}{"B1}
  \DeclareMathSymbol{\daleth}{\mathord}{MnSyC}{"B2}
\fi
\makeatother

\DeclareFontFamily{U}{mathx}{\hyphenchar\font45}
\DeclareFontShape{U}{mathx}{m}{n}{
      <5> <6> <7> <8> <9> <10>
      <10.95> <12> <14.4> <17.28> <20.74> <24.88>
      mathx10
      }{}
\DeclareSymbolFont{mathx}{U}{mathx}{m}{n}
\DeclareFontSubstitution{U}{mathx}{m}{n}
\DeclareMathAccent{\widecheck}{0}{mathx}{"71}

\usepackage{scalerel,stackengine}

\def\bra#1{\mathinner{\langle{#1}|}}
\def\ket#1{\mathinner{|{#1}\rangle}}
\def\braket#1{\mathinner{\langle{#1}\rangle}}
\def\Braket#1{\langle{#1}\rangle}

\definecolor{myteal}{RGB}{0,140,100}
\definecolor{myorange}{RGB}{238,119,51}
\definecolor{myred}{RGB}{204,51,17}

\DeclareSymbolFont{extraup}{U}{zavm}{m}{n}
\DeclareMathSymbol{\varheart}{\mathalpha}{extraup}{86}

\def\Real{{\mathbf R}}
\def\Complex{{\mathbf C}}
\def\Sphere{{\mathbf S}}
\def\Schwartz{{\mathcal{S}}}

\def\Hausdorff{{\mathcal{H}}}

\def\Expec{{\mathbf E\,}}

\def\eps{{\varepsilon}}
\def\One{{\mathbf{1}}}

\def\PhaseSpace{{\Real^{2d}}}
\def\Evol{{\mathcal{E}}}

\def\lsim{{\,\lesssim\,}}

\def\bbZ{{\mathbb{Z}}}
\def\bbN{{\mathbb N}}

\newcommand{\Wigner}[1]{{\mathcal{W}_{#1}}}

\newcommand{\Ft}[1]{{\widehat{#1}}}

\newcommand{\bbk}[1]{{]{#1}[}}

\newcommand{\diff}{\mathop{}\!\mathrm{d}}

\numberwithin{equation}{section}

\DeclareMathOperator{\Id}{Id}

\DeclareMathOperator{\Mult}{mult}

\newcommand{\noset}{\varnothing}
\newcommand{\mbf}[1]{{\mathbf{#1}}}
\newcommand{\mcal}[1]{{\mathcal{#1}}}
\newcommand{\ovln}[1]{{\overline{#1}}}
\newcommand{\wtild}[1]{{\widetilde{#1}}}

\newcommand{\back}{\triangleleft}
\newcommand{\next}{\triangleright}

\newcommand{\mstack}[2]{\begin{smallmatrix} #1 \\ #2 \end{smallmatrix}}

\DeclareMathSymbol{\mlq}{\mathord}{operators}{``}
\DeclareMathSymbol{\mrq}{\mathord}{operators}{`'}
\newcommand{\qt}[1]{{\mlq\mlq#1\mrq\mrq}}

\DeclareMathOperator{\Op}{Op}

\DeclareMathOperator{\Rept}{Re}

\DeclareMathOperator{\supp}{supp}

\DeclareMathOperator{\argmax}{argmax}
\DeclareMathOperator{\trace}{tr}

\DeclareMathOperator{\Scaff}{Scaff}

\DeclareMathOperator{\Range}{range}

\DeclareMathOperator{\Typical}{Typ}
\DeclareMathOperator{\Rung}{Rung}

\DeclareMathOperator{\Junk}{junk}
\DeclareMathOperator{\Imrec}{imm}
\DeclareMathOperator{\ext}{ext}

\DeclareMathOperator{\Base}{Base}

\DeclareMathOperator{\Nbd}{Nbd}

\DeclareMathOperator{\Imm}{Imm}
\DeclareMathOperator{\LB}{LB}
\DeclareMathOperator{\sign}{sgn}

\newtheorem{theorem}{Theorem}[section]
\newtheorem{lemma}[theorem]{Lemma}

\newtheorem{definition}[theorem]{Definition}
\newtheorem{corollary}[theorem]{Corollary}
\newtheorem{proposition}[theorem]{Proposition}

\title{Quantum diffusion via an approximate semigroup property}
\author{Felipe Hern{\'a}ndez}
\email{felipehb@stanford.edu}
\date{\today}
\begin{document}

\begin{abstract}
In this paper we introduce a new approach to the diffusive limit of the weakly random Schrodinger equation, first
studied by L. Erdos, M. Salmhofer, and H.T. Yau.  Our approach is based on a wavepacket decomposition of the evolution operator,
which allows us to interpret the Duhamel series as an integral over piecewise linear paths.
 We relate the geometry of these paths to combinatorial features of a diagrammatic expansion which allows us to express the error terms in the
expansion as an integral over paths that are exceptional in some way.  These error terms are bounded using geometric arguments.
The main term is then shown to have a semigroup property, which allows us to iteratively increase the timescale of validity of
an effective diffusion.  This is the first derivation of an effective diffusion equation from the random Schrodinger equation that
is valid in dimensions $d\geq 2$.
\end{abstract}
\maketitle

\tableofcontents
\section{Introduction}

\subsection{The kinetic limit for the Sch{\"o}dinger equation}

In this paper, we study the equation
\begin{equation}
\label{eq:schro}
i\partial_t \psi = -\frac{1}{2}\Delta\psi + \eps V\psi
\end{equation}
with a stationary random potential $V$.  An example of a potential we will consider is a mean-
zero Gaussian random field with a smooth and compactly supported two point correlation function
\begin{equation}
\label{eq:two-point-corr}
\Expec V(x)V(y) = R(x-y).
\end{equation}
with $R\in C_c^\infty(\Real^d)$.  Our approach works for more general
potentials that are stationary, have finite range of dependence, and have bounded moments in $C^k$
(for $k>20d$, say).

The equation~\eqref{eq:schro} is a simple model for wave propagation in a random environment.  It
also has a more direct physical significance, as it models the motion of a cold electron in a disordered
environment~\cite{spohn77}.  We are interested in this paper in the regime where the frequency of the initial condition
$\psi$ is comparable to the correlation length of the potential, which we consider to be of unit scale.
This regime is out of reach of both traditional WKB-type semiclassical approximations,
which are more appropriate for high-frequency $\psi$, and of homogenization techniques, which are appropriate
for low-frequency $\psi$.

This regime was first rigorously studied by H. Spohn in~\cite{spohn77} , who showed that the spectral
density $\mu_t(p) = \Expec |\psi_{t/\eps^2}(p)|^2$ converges, in the semiclassical limit $\eps\to 0$,
to a weak solution of a spatially homogeneous kinetic equation
\begin{equation}
\label{eq:homo-kinetic}
\partial_t \mu(p) = \int \delta(|p|^2-|p'|^2) \Ft{R}(p-p') [\mu(p') - \mu(p)] \diff p',
\end{equation}
where $R(x)$ is the two-point correlation function defined in~\eqref{eq:two-point-corr}.  The term
$\delta(|p|^2-|p'|^2)$ enforces conservation of kinetic energy, which is appropriate in the limit
$\eps\to 0$ since the potential energy becomes negligible.  The time scale between scattering events
is on the order $\eps^{-2}$, which can be heuristically justified by using the Born approximation of
the solution of~\eqref{eq:schro}.  Spohn's technique for demonstrating~\eqref{eq:homo-kinetic} was to
write out the Duhamel expansion for the solution $\psi_t(p)$
to~\eqref{eq:schro} in momentum space, take an expectation of the quantity $|\psi_t(p)|^2$ using the Wick rule
for the expectation of a product of Gaussian random variables, and separate terms
into a main term and an error term.  The error terms are controlled by additional cancellations
and the main terms are compared to a series expansion for the solution of~\eqref{eq:homo-kinetic}.
Spohn's analysis of the Dyson series allowed him to control the solution
up to times $c \eps^{-2}$ for some small constant $c>0$.

This proof technique has been used by many authors since to improve upon our understanding
of~\eqref{eq:schro}.  Most notably, L. Erd{\"o}s and H.T. Yau in a series of
works~\cite{eylinear98,erdosyau2000} were able to improve the time scale to arbitrary kinetic times of the form
of the order $O(\eps^{-2})$
while also demonstrating the weak convergence of the Wigner function
\[
\Wigner{\psi}(x,p) := \int e^{iy\cdot p} \overline{\psi}(x-y/2)\psi(x+y/2)\diff y
\]
to the solution of the linear Boltzmann equation
\begin{equation}
\label{eq:boltz}
\partial_t \rho + p\cdot\nabla_x \rho =
\eps^2 \int \delta(|p|^2-|p'|^2) \Ft{R}(p-p') [\rho(x,p') - \rho(x,p)] \diff p'.
\end{equation}
Introducing the rescaled coordinates $T=\eps^2 t$, $X = \eps^2x$ along with the
rescaled solution
\[
\rho^\eps_T(X,p) := \rho_{\eps^{-2}T}(\eps^{-2}x, p),
\]
the equation~\eqref{eq:boltz} can be written
\[
\partial_T \rho^\eps + p\cdot \nabla_X \rho
= \int \delta(|p|^2-|p'|^2) \Ft{R}(p-p') [\rho^\eps(X,p')-\rho^\eps(X,p)]\diff p'.
\]
In an impressive sequence of refinements to this work, L. Erdos, M. Salmhofer and
H.T. Yau~\cite{esyQuantumBoltzmann04,esyI08,esyII07} were able to improve the
timescale even further to diffusive times $\eps^{-2-\kappa}$ for some positive $\kappa>0$
(in fact, one can take $\kappa=1/370$ when $d=3$).  At this timescale a diffusion equation emerges.
The principle is that the momentum variable is no longer relevant to the evolution because it becomes
uniformly distributed over the sphere within time $O(\eps^{-2})$, and all that remains of the momentum
information is the kinetic energy variable $e =|p|^2/2$.
Moreover, for diffusive times $\eps^{-2-\kappa}$ the particle travels a distance $\eps^{-2-\kappa/2}$
so the diffusive length scale is $\eps^{-2-\kappa/2}$.

For solutions $\rho$ of the linear Boltzmann equation~\eqref{eq:boltz}, the particle
distribution $f$ defined by
\begin{equation}
\label{eq:diffusive-scaling}
f_T(X,e) = \int_{|p|^2/2=e} \rho_{\eps^{-2-\kappa}T}(\eps^{-2-\kappa/2}X,p) \diff \Hausdorff^{n-1}(p)
\end{equation}
converges in the limit $\eps\to 0$ to a solution of the diffusion equation
\begin{equation}
\label{eq:diffusion}
\partial_T f_T = D_e \Delta_X f_T,
\end{equation}
where $D_e$ is a diffusion coefficient depending on the energy $e$.  See~\cite{esyI08} for more details on the limiting diffusion equation.

To reach the diffusive time scale in which the particle
experiences infinitely many scattering events, one must consider terms with $\eps^{-c}$ collisions, which
produces more than $\eps^{-\eps^{-c}}$ diagrams when one applies the Wick expansion.  To deal with the explosion
in the number of terms, Erdos, Salmhofer, and Yau developed a resummation technique to more accurately estimate
the sizes of the terms and additionally had to exploit intricate cancellations coming from the combinatorial
features of the diagrams considered.

\subsection{Statement of the main result}
In this paper, we provide an alternative derivation of the linear Boltzmann equation which is also valid up to
diffusive times but with a fundamentally different approach.  In our proof, we use a wavepacket decomposition
of the evolution operator.  The wavepacket decomposition allows us to keep information about the position and
momentum of the particle simultaneously (up to the limits imposed by the uncertainty principle), and we therefore
express the solution as an integral over piecewise linear paths in phase space.

To make the connection between operators and the linear Boltzmann equation, we use the Weyl quantization
$a\in C^\infty(\PhaseSpace)\mapsto \Op^w(a)\in\mathcal{B}(L^2(\Real^d))$ defined by
\[
\Op^w(a) f(x)
= \int e^{i(x-y)\cdot p} a((x+y)/2,p) f(y)\diff y \diff p.
\]
The relationship between the Weyl quantization and the Wigner transform is given by the identity
\[
\langle \Op^w(a)\psi,\psi\rangle = \int a(x,p) \Wigner{\psi}(x,p) \diff x\diff p.
\]
In particular, applying this identity to the solution $\psi_t = e^{-itH}\psi$ to (\ref{eq:schro})
where $H$ is the random Hamiltonian
\[
H=-\frac{1}{2}\Delta + \eps V,
\]
we have
\[
\int a(x,p) \Expec \Wigner{\psi_t}(x,p) \diff x\diff p
= \Expec \langle \Op^w(a) e^{-itH}\psi, e^{-itH}\psi\rangle =
\langle \Expec e^{itH} \Op^w(a)e^{-itH} \psi,\psi\rangle.
\]
Therefore in order to answer questions about the weak convergence of $\Wigner{\psi_t}$, it suffices
to study the quantum evolution channel
\begin{equation}
\label{eq:evol-def}
\Evol_t[A] := \Expec e^{itH} A e^{-itH}
\end{equation}
applied to operators of the form $A=\Op^w(a)$ with sufficiently regular symbols $a$.   In particular,
we will show that for suitable observables $a_0$ and for times $t\leq \eps^{-2-\kappa}$, we have
\[
\|\Evol_t[\Op^w(a_0)] - \Op^w(a_t)\|_{op} = o(1),
\]
where $a_t$ solves the dual of the linear Boltzmann equation~\eqref{eq:boltz},
\begin{equation}
\label{eq:dual-boltz}
\partial_t a - p\cdot\nabla_x a =
\eps^2 \int \delta(|p|^2-|p'|^2) \Ft{R}(p-p') [a(x,p') - a(x,p)] \diff p'.
\end{equation}

A natural norm that we use on $a$ which also controls the operator norm of $\Op^w(a)$ is the
$C^k$ norm with $k=2d+1$ (see Appendix~\ref{sec:wp-quantization} for a self-contained
proof that the $C^{2d+1}$ norm of $a$ controls the operator norm of $\Op^w(a)$).
We will use a $C^k$ norm which is rescaled to the appropriate length scales of the problem.  Because
the time scale between scattering events is $\eps^{-2}$, a natural spatial length scale is $\eps^{-1}$.
For the rest of the paper we write $r=\eps^{-1}$ for this length scale.
This is the length scale of a wavepacket that remains coherent between scattering events.  Conversely
the natural length scale in momentum is $\eps = r^{-1}$.
This ``microscopic'' scale is the one we use for the wavepacket decomposition of the operator $e^{itH}$.

On the other hand, a natural ``macroscopic'' length scale of the problem is $\eps^{-2}$, which is the distance
that a particle with momentum $O(1)$ travels between scatterings.  A natural ``macroscopic''
length scale in momentum is $O(1)$, which is the impulse applied to a particle in a typical scattering event.

The following norm measures the smoothness of an observable at these length scales:
\[
\|a\|_{C^k_{r,L}}
:= \sum_{|\alpha_x|+|\alpha_p|\leq k}
\sup_{(x,p)} |(rL\partial_x)^{\alpha_x} (r^{-1}L\partial_p)^{\alpha_p} a(x,p)|.
\]
When $L=1$, this norm probes the microscopic smoothness of observables, whereas
when $L=\eps^{-1}$, the norm probes the macroscopic smoothness.

We make one final comment before we state the main result of the paper, which is that we will not treat
the evolution of low-frequency modes.  In dimension $d=2$, the scattering cross section of a low frequency wave
with momentum $|p|\ll 1$ is still on the order $\eps^2$ but the speed of travel is only $|p|$, so the distance
between typical scattering events is only $|p|\eps^{-2}$ rather than $\eps^{-2}$.  Because scattering events
are more closely spaced,
the bounds coming from the geometric arguments we use deteriorate and we make no attempt
to understand what happens in this regime.  In higher dimensions the scattering cross section also shrinks
with momentum so that one could in principle first approximate the evolution of low frequency modes by a
free evolution with no potential and therefore recover the result for all frequencies.
We do not make this argument in this paper.
\begin{theorem}
\label{thm:main-result}
For each $d\geq 2$, there exists $\theta=\theta(d)>0$ and  $\kappa=\kappa(d)>0$ such that the following holds.
Let $V$ be an admissible potential as described in Definition~\eqref{def:admissible-V},
and let $a_0\in C^{2d+1}(\PhaseSpace)$ be a classical observable
supported away from zero momentum;
\[
\supp a_0 \subset \{(x,p)\in\PhaseSpace \mid |p| \geq \eps^{\theta(d)}\}.
\]
Suppose moreover that $a_t$ solves~\eqref{eq:dual-boltz} with initial condition $a_0$.
Then
\begin{equation}
\label{eq:main-op-est}
\|\Evol_t[\Op^w(a_0)] - \Op^w(a_t)\|_{op} \leq
C_d \eps^{2+\kappa} t \|a_0\|_{C^{2d+1}_{\eps^{-1},\eps^{-0.5}}}.
\end{equation}
In particular, for arbitrary $\psi_0\in L^2(\Real^d)$ and $\psi_t$ solving~\eqref{eq:schro} it follows that
\begin{equation}
\label{eq:wigner-boltz}
\int \Wigner{\psi_t}(x,p) a_0(x,p)\diff x\diff p
= \int \Wigner{\psi_0}(x,p) a_t(x,p)\diff x\diff p + O(\eps^{2+\kappa} t \|a_0\|_{C^{2d+1}_{\eps^{-1},\eps^{-0.5}}}\|\psi\|_{L^2}^2).
\end{equation}
\end{theorem}

To see how the diffusion equation~\eqref{eq:diffusion} emerges as a scaling limit, we consider
observables of the form
\[
a_0(x,p) = \bar{a}(\eps^{2+\kappa/2} x, p)
\]
with $\bar{a}\in C^{2d+1}$.  With this rescaling, we have
\[
\|a_0\|_{C^{2d+1}_{\eps^{-1},\eps^{-1}}} \leq \|\bar{a}\|_{C^{2d+1}}.
\]
In particular, $\|a_0\|_{C^{2d+1}_{\eps^{-1},\eps^{-0.5}}}$ is bounded uniformly in $\eps$.  Moreover,
the solution $a_t$ solves~\eqref{eq:diffusion}.

One major difference between Theorem~\ref{thm:main-result} and the main
results of~\cite{esyI08,esyII07}, apart from the very different approaches to the proof,
is that our result holds in dimension $d=2$.  At first this may appear to be
in contradiction with the conjectured phenomenon of Anderson localization in $d=2$, but the contradiction disappears
when one compares the timescale $\eps^{-2-\kappa}$ considered in this paper to the expected length scale
$e^{\eps^{-2}}$ of localization in this dimension.  Indeed, it is expected
that the particle exhibits diffusive behavior for an exponentially long time
before getting trapped by localization.

The exponent $\kappa(d)$ can in principle be extracted from the proof.  However
in this paper we focus on demonstrating the new technique in its simplest form and therefore do
not attempt to
optimize $\kappa(d)$.  Perhaps with some optimization one could obtain $\kappa(3)$ comparable to $1/370$,
but the proof we give yields a bound of the order $\kappa(3) \sim 10^{-6}$.

\subsection{A heuristic sketch of the argument}
\label{sec:heuristics}

\subsubsection{The phase space path integral}
The main idea behind the proof of Theorem~\ref{thm:main-result} is to focus on
justifying an approximate
semigroup property
\begin{equation}
\label{eq:appx-semigroup}
\Evol_{2t}[A] \approx \Evol_t [\Evol_t[A]],
\end{equation}
for suitable operators $A$ including operators of the form $A=\Op^w(a)$.
Observe that the approximation~\eqref{eq:appx-semigroup} has the following physical interpretation.
Let
\begin{align*}
H_1 &= -\frac{1}{2} \Delta + \eps V_1 \\
H_2 &= -\frac{1}{2}\Delta + \eps V_2.
\end{align*}
be Hamiltonians with two independently sampled potentials, and observe that
\[
\Evol_t\circ\Evol_t[A] = \Expec e^{itH_2} e^{itH_1} A e^{-itH_1} e^{-itH_2}.
\]
In other words, $\Evol_t\circ\Evol_t$ represents an evolution with a potential that abruptly changes into an
independently sampled potential at time $t$.  Although such a resampling of the potential drastically changes the evolution of the
wavefunction $\psi_t$ itself,  we will see that the effect on observables is minimal

To prove the approximate semigroup property we approximate the evolution operator
$e^{-itH}$ as an integral over piecewise linear paths in phase space, representing the possible paths of a particle
as it scatters.  To decompose phase space we use a family of wavepackets of the form
\[
\phi_{x,p}(y) := r^{-d/2} e^{iy\cdot p} \chi_{env}((x-y)/r),
\]
where $\chi_{env}\in C_c^\infty(\Real^d)$ is a fixed envelope normalized in $L^2$ and satisfying some additional
conditions described in Appendix~\ref{sec:wp-quantization}.   The functions $\phi_{x,p}$ are localized
in space to scale $r$ and in momentum to scale $r^{-1}$.  We use the notation $\xi=(x,p)$
and write $\ket{\xi}$ as a shorthand for the function $\phi_\xi$.

The use of a phase-space path integral already represents a departure from
previous approaches to the problem.
Indeed, since the paper of Spohn~\cite{spohn77} it has been customary to write the terms of
the Duhamel series expansion for $e^{itH}$ in the Fourier basis.  We will see that by using the spatial
localization of the particle we can more easily compute the expectation appearing in the integrand
without the need for a full Wick expansion (or cumulant decomposition in the case of a non-Gaussian potential).

The free evolution of a wavepacket approximates the motion of a free classical particle in the sense that
\[
e^{-it\Delta/2} \ket{(x,p)} \approx e^{it|p|^2/2} \ket{(x+tp,p)}
\]
for $t\ll r^2$.  For multiplication against the potential we use the identity
\begin{equation}
\label{eq:V-mult-rule}
V\ket{(x,p)} = \int \Ft{V_x}(p'-p) \ket{(x,p')}\diff p',
\end{equation}
where $V_x$ is the potential $V$ multiplied by a cutoff near $y$ which is $1$ in a ball large enough to contain the support of $\chi_{env}$,
\[
V_x(y) = b((y-x)/r) V(y).
\]

We can use these two identities to write an expansion of the evolution of a wavepacket $e^{itH}\ket{(x,p)}$ as an integral over paths in which the phase
space point travels in straight lines with an occasional impulse from the potential causing discontinuities in the momentum variable.
We represent these piecewise linear paths as a tuple $\omega=(\mbf{s},\mbf{p},\mbf{y})$
with $\mbf{s}=(s_0,\cdots,s_k)$ being the sequence of times between the scattering events (satisfying $\sum s_j=t$ and $s_j\geq 0$),
$\mbf{p}=(p_0,p_1\cdots,p_k)$ being the sequence of momentum variables which we require to have the same
magnitude $|p_j|=|p_{j'}|$ and with initial momentum $p_0=p$,
and $\mbf{y}=(y_1,\cdots,y_k)$ being the sequence of scattering locations defined by
\begin{align*}
y_1&=x+s_0p \\
y_{j+1}&=y_j+s_jp_j.
\end{align*}
 An example of a path is depicted in Figure~\ref{fig:simple-path}.

\begin{figure}
\centering
\includegraphics[scale=1.75]{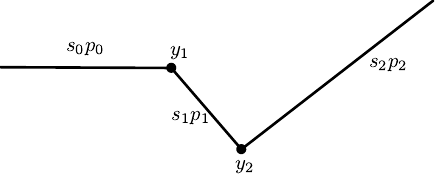}
\caption{A sample scattering path of a particle with $2$ collisions.
The displacement between consecutive collisions is given by $s_jp_j$, where $p_j\in\Real^d$ is a momentum
vector with constrained kinetic energy $|p_j|^2/2$, and $s_j\geq 0$ is the time between collisions.}
\label{fig:simple-path}
\end{figure}

Each such path $\omega$ defines an operator $O_\omega$ which approximately acts on wavepackets by
\begin{equation}
O_\omega\ket{(x,p)} = e^{i\varphi(\omega)} \prod_{j=1}^k \Ft{V_{y_j}}(p_j-p_{j-1})  \ket{y_k+s_kp_k, p_k},
\end{equation}
where $\varphi(\omega)$ is a deterministic phase accumulated from the stretches of free evolution.  These
phases do not matter for this sketch of the proof.  However, in the actual proof we
use stationary phase to ensure that the geometric constraints $y_{j+1}=y_j+s_jp_j$ are approximately
satisfied and to show that kinetic energy is approximately conserved.
Then, at least formally, we can write out the path integral for the evolution of
a wavepacket $\ket{\xi}$ as
\begin{equation}
\label{eq:sketch-path-integral}
e^{itH}\ket{\xi} = \int O_\omega\ket{\xi} \diff \omega.
\end{equation}

We apply this decomposition of the evolution to investigate the
approximate semigroup property for operators of the form
\[
A = \int_{\PhaseSpace} a(\xi) \ket{\xi}\bra{\xi} \diff \xi,
\]
which are local in phase space in the sense that
\[
A\ket{\xi} \approx a(\xi)\ket{\xi}.
\]
Using the path integral~\eqref{eq:sketch-path-integral}
in the definition of $\Evol_t[A]$ we obtain
\[
\Evol_t[A] =
\int_{\PhaseSpace} \iint
\Expec O_{\omega'}^* \ket{\xi}\bra{\xi} O_{\omega} \diff\omega\diff\omega'\diff \xi.
\]
What is important for this sketch of the proof
 is to simply investigate which pairs of paths $\omega$  and $\omega'$ have
\[
\Expec O_{\omega'}^*\ket{\xi}\bra{\xi} O_{\omega} \not= 0.
\]
In particular, we are interested in understanding for which sequence of positions $y_j,y'_j$ and
impulses $q_j, q'_j$ we have
\[
\Expec \prod_{j=1}^k \Ft{V_{y_j}}(q_j) \prod_{j'=1}^{k'} \Ft{V_{y'_{j'}}}^*(q'_{j'}) \not = 0.
\]
Using the fact that $V$ is real and therefore $\Ft{V_y}^*(q) = \Ft{V_y}(-q)$, we can rewrite the expectation above
as
\[
\Expec \prod_{b\in[k]\sqcup[k']} \Ft{V_{y_b}}(q_b),
\]
where $[k]\sqcup[k']$ is shorthand for the doubled index set $[k]\times\{0\}\cup [k']\times\{1\}$
and the $q_b$ impulses are reversed for $b=(j,1)$, so in particular
$q_{j,1} = - (p'_j - p'_{j-1})$.  Because $V_y$ are localized,
the expectation above splits along a partition
$P(\mbf{y},\mbf{y}')\in\mcal{P}([k]\sqcup[k'])$ defined by the clusters of the collision locations $y_b$ (so that
$b\sim_P b'$ when $|y_b-y_{b'}|\lsim r$).  That is, we have an identity of the form
\[
\Expec \prod_{b\in[k]\sqcup[k']} \Ft{V_{y_b}}(q_b)
= \prod_{S\in P(\mbf{y},\mbf{y}')} \Expec \prod_{b\in S} \Ft{V_{y_b}}(q_b).
\]
Within each cluster the expectations is zero unless the sum of the impulses is zero.  This
``conservation of momentum'' condition is a consequence of the stationarity of the potential and is
made rigorous in Lemma~\ref{lem:admissible-V}.  In the case of Gaussian potentials, this is a
consequence of the Wick formula and the identity
\[
\Expec \Ft{V}(p) \Ft{V}(q) = \Ft{R}(p) \delta(p+q).
\]
We are led to the following purely geometric constraints on the pair of paths $\omega,\omega'$.

\begin{definition}
Two paths $\omega,\omega'$ are said to be \emph{compatible} if the partition $P(\mbf{y},\mbf{y'})$
has no singletons and
\[
\sum_{b\in S} q_b = 0
\]
for every $S\in P(\mbf{y},\mbf{y'})$.
\end{definition}

\subsubsection{Geometric characterization of the error term}
This geometric notion of \textit{compatible} paths is perhaps the most significant idea in the proof.
Indeed, the point is that we can usefully manipulate the path integral \textit{before computing an
expectation}.  In other words, we will first decompose the path
integral according to geometry and then allow the geometry to dictate the combinatorics of the diagrams.

Now, we take a step back to appreciate which paths contribute to the error term in the semigroup property.  As the
discussion above indicates, the semigroup property compares the evolution $\Evol_t$ with a fixed potential
$V$ to the evolution $\Evol_{t/2}\circ\Evol_{t/2}$ during which the potential is refreshed from $V_1$ to $V_2$
at time $t/2$.  To keep track of which potential a scattering event sees we introduce the collision
times $t_b$, defined by $t_0=0$ and $t_{b+1} = t_b + s_b$.
The following condition suffices to ensure that the pair of paths $(\omega,\omega')$ has an expected
amplitude that is unaffected by the possibility of a time-dependent potential.
\begin{definition}
Two paths $\omega,\omega'$ are \emph{time-consistent} if $t_b=t_{b'}$ for all pairs of indices
$b,b'\in[k]\sqcup[k']$ such that $y_b=y_{b'}$.  Pairs $(\omega,\omega')$ that are not time-consistent
are said to be \emph{time-inconsistent}.
\end{definition}
Observe that $\omega$ is compatible with itself and is also time-consistent.
We will show that in fact the paths $(\omega,\omega')$ in which
$\omega=\omega'$ or is otherwise a small perturbation of $\omega$ form the bulk of the contribution to
$\Evol_t$.

To understand the error term, we characterize pairs of paths $(\omega,\omega')$ which are compatible
but time-inconsistent. A simple way that a pair $(\omega,\omega')$ could be time-inconsistent is
if either $\omega$ or $\omega'$ have a recollision.  A simple example of a pair of compatible paths which
are time-inconsistent due to a recollision is depicted in Figure~\ref{fig:recollision-example}

\begin{figure}
\centering
\includegraphics{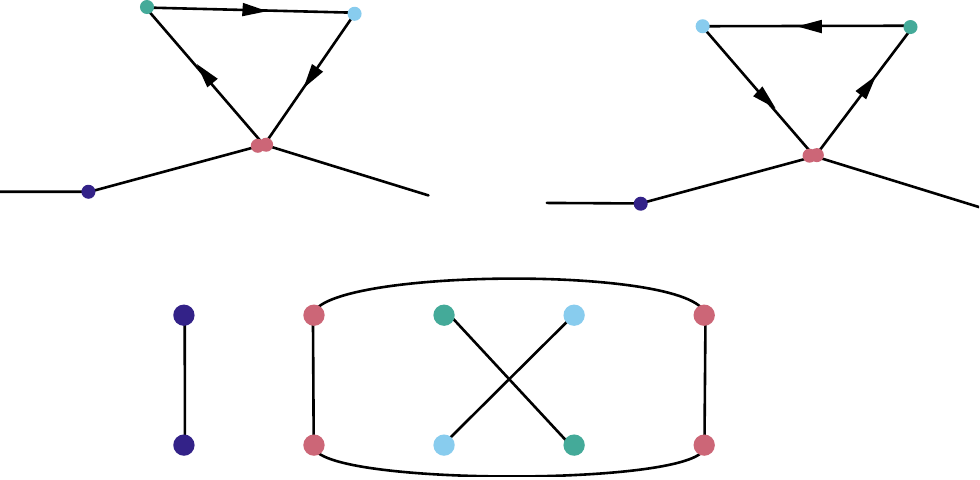}
\caption{An example of a pair of paths with a recollision.  The paths $\omega$ and $\omega'$ are depicted at the top
left and top right, and the collisions are colored according to the cluster in the partition $P(\omega,\omega')$.
On the bottom, an abstract depiction of the partition $P(\omega,\omega')$.  Notice that there is a red cluster with $4$
collisions, two from $\omega$ and two from $\omega'$.}
\label{fig:recollision-example}
\end{figure}

There is another geometric feature which we call a ``tube event'' which occurs when three collisions are collinear
(in general, when they lie on narrow tube).  Tube events can also lead to time inconsistencies, as depicted
in Figure~\ref{fig:parallel-example}

\begin{figure}
\centering
\includegraphics{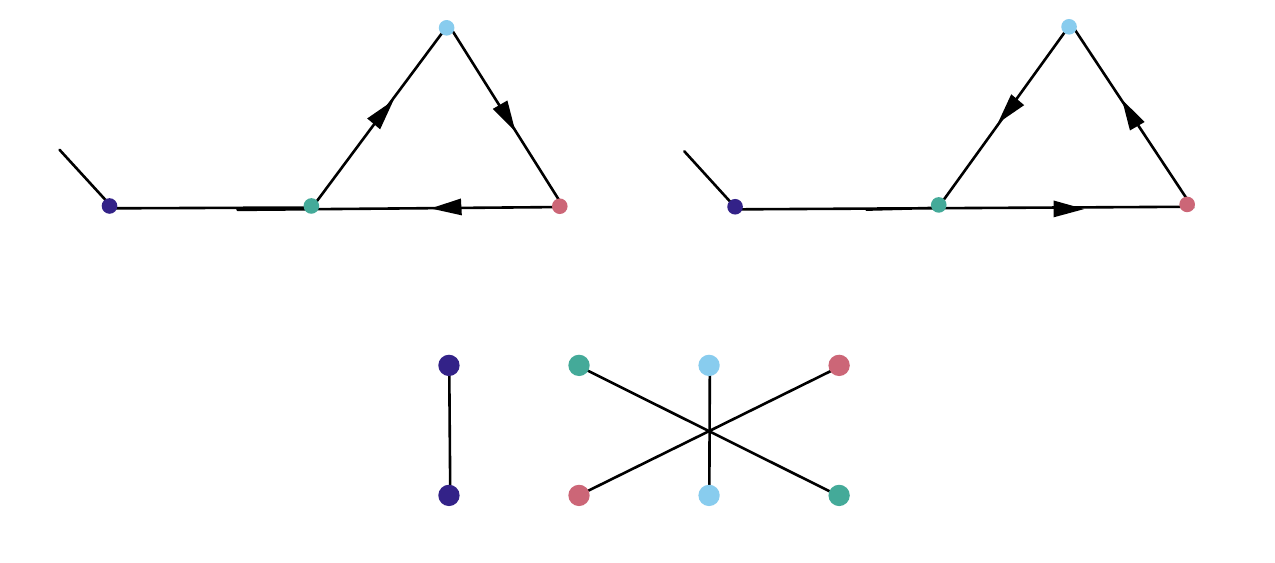}
\caption{An example of a tube event.  Note that there are three collisions lying on a line, but neither $\omega$
nor $\omega'$ forms a recollision.  Nonetheless, the collisions are time-inconsistent because the second
collision of $\omega$ coincides with the fourth collision of $\omega'$, as shown in the diagram of $P(\omega,\omega')$
below.}
\label{fig:parallel-example}
\end{figure}

We note that the estimation of the contribution of these special geometric features replaces the need
for crossing estimates such as the ones studied in~\cite{MR2296805}.
Our diagrams are bounded using relatively simple-minded volumetric
considerations (essentially, after taking care of deterministic cancellations in the integral, we use the
triangle inequality and account for the contribution of each degree of freedom).  This simple-minded
approach works particularly well for subkinetic timescales $\tau\lsim \eps^{-2+\kappa/2}$ in which one
only needs $k_{max}=O(\kappa^{-1})$ collisions in the series expansion to approximate $e^{i\tau H}$
and therefore all combinatorial factors are bounded by a (very large) absolute constant.

The general strategy of the proof therefore is as follows:
\begin{enumerate}
\item Classify the geometric behaviors that can lead to time-inconsistencies.
\item Partition the path integral into paths with bad behaviors and paths without bad behaviors.
\item Use geometric estimates to bound the operator norm of the contribution of the bad paths.
\end{enumerate}
The main new feature of this proof strategy is that the path integral is partitioned
\textit{before the expectation is computed}.  That is, we do not decompose the expectation until
we already have some information about the partition $P(\omega,\omega')$.  This is in contrast to the
traditional approach used in~\cite{spohn77,esyI08,esyII07,erdosyau2000} which is summarized below.
\begin{enumerate}
\item Expand the expectation using the Wick rule or a cumulant expansion.
\item Partition the diagrams according to complexity by a \textit{combinatorial} criterion.
\item Use oscillatory integral estimates to bound the contributions of the bad diagrams.
\end{enumerate}

\subsubsection{Reaching the diffusive timescale}
To reach the diffusive timescale we prove a semigroup property of the form
\[
\Evol_{N\tau} \approx \Evol_\tau^N
\]
where $\tau=\eps^{-2+\kappa/2}$ and $N=\eps^{-\kappa}$.
The challenge we face in trying
 to understand the evolution operator $e^{itH}$ for times $t\sim \eps^{-2-\delta}$ is that one needs to resolve
at least $\eps^{-\delta}$ collisions.  This requires a path integral in a space of dimension $\eps^{-\delta}$.
If we then try to use crude estimates to bound the contribution of the terms in the Duhamel expansion
we may lose a factor of $C^{\eps^{-\delta}}$.  What we need to do is take into account cancellations that occur
between the terms in the Duhamel expansion.  In~\cite{esyI08,esyII07} this is done by renormalizing the
propagator.  This is equivalent to viewing $H=-\Delta/2+\eps V$ not as a perturbation of
the free evolution $-\Delta/2$ but as a perturbation of $-\Delta/2 + \eps^2\Theta$ where $\Theta$ is a multiplier
operator that takes into account the effect of immediate recollisions.  The multiplier $\Theta$ has a nonzero
imaginary part so that $e^{-i\Delta/2 + i\Theta}$ decays exponentially in time.  This exponential
decay exactly matches the exponential growth in the volume of the path inetgral.
The value of $\Theta$ is also chosen so that a precise cancellation occurs in diagrams with immediate recollisions.

We take an alternative approach to resummation.  The idea is that we first write
\[
e^{itH} = e^{iN\tau H} = e^{i\tau H} \cdots e^{i\tau H},
\]
where $N\sim\eps^{-\kappa}$ and $\tau\sim\eps^{-2+\kappa/2}$.  Each of the terms $e^{i\tau H}$ is expanded
as a Duhamel series of $k_{max}$ terms.  We then partition the resulting path integral into pieces depending on
geometric features of the paths and decompose the expectation using this geometric information.  When this is done
we resum the terms in the Duhamel series corresponding to segments that do not have any geometrically
special collisions.  This can be intepreted as a way of writing the evolution
channel $\Evol_{N\tau}$ as a perturbation of the refreshed evolution channel $\Evol_\tau^N$.   This seems to be a more
general strategy for deriving kinetic limits -- the resummation procedure is dictated  by the
desired semigroup structure.

Another important point of comparison concerns the diagrammatic expansion we derive
to reach the diffusive time scale.  In both this paper and in~\cite{esyII07,esyI08} one expands the solution as a
sum over diagrams which are stratified in some way by combinatorial complexity.  In~\cite{esyII07, esyI08} the more
complex diagrams contain more opportunities to find decay via crossing estimates, which
are nontrivial bounds on oscillatory integrals.
In this paper, we first
split the path integral itself according to geometric complexity and then bound the combinatorial complexity
of the diagrams associated to paths of a fixed geometric complexity.  The difference between these approaches is
summarized in Figure~\ref{fig:diagram-table}.
\begin{figure}
\begin{tabular}{l c||c|c|c|c}
\multirow{6}{0.13\textwidth}{Geometric complexity of paths}
&&\multicolumn{4}{|c}{Combinatorial complexity of diagrams} \\
\cline{3-6}
&&  1 & 2 & 3 & 4\\
\cline{2-6}
&1 & $\eps^c N^C$ & - & - & - \\
\cline{2-6}
&2 & $\eps^{2c} N^C$ & $\eps^{2c} N^{2C}$ & - & - \\
\cline{2-6}
&3 & $\eps^{3c} N^C$ & $\eps^{3c}N^{2C}$ & $\eps^{3c}N^{3C}$ & - \\
\cline{2-6}
&4 & $\eps^{4c} N^C$ & $\eps^{4c}N^{2C}$ & $\eps^{4c}N^{3C}$ & $\eps^{4c}N^{4C}$
\end{tabular}
\caption{A cartoon of the contribution of various diagrams.  Diagrams have a combinatorial complexity, and
there are at most $N^{Ck}$ diagrams having complexity exactly $k$.  Moreover paths have a geometric complexity,
and the volume of the paths with geometric complexity $k$ is $\eps^{ck}$.  The approach taken
in~\cite{esyII07,esyI08} is to sort diagrams by combinatorial complexity and then show that the only
contributions to diagrams with high combinatorial complexity also have high geometric complexity.
In this paper, we first sort the paths by geometric complexity and show that paths with low geometric
complexity only contribute to diagrams with low combinatorial complexity.  In summary, we sum along the rows
of this table whereas previous works proceed by summing over the columns.  Note that as long
as $N\ll\eps^{-c'}$ the sum of the contributions is small.}
\label{fig:diagram-table}
\end{figure}

\subsubsection{More explanation of the diagrammatic expansion}
To reach subkinetic times we used the following crude idea to verify the approximate semigroup property:
either the pair of paths $(\omega,\omega')$ has a nontrivial geometric event, or it does not.  If there is a
nontrivial geometric event, we use the triangle inequality inside the path integral and the geometric
information about the event to pick up a factor of $\eps^c$, which is small enough to suppress
the large constant appearing from two inefficiencies in our argument.  The first inefficiency is to fail to
take into account precise cancellations in the path integral, which costs us a factor of $C^k$ where $k$ is the
number of collisions.  The second inefficiency is the failure to take into account the combinatorial
constraints imposed on the collision partition.  The constraints come from ``negative information'' about the
path -- as an oversimplification, if a collision index $b$ is not part of a tube event or a recollision event, then
it must form part of a ladder or anti-ladder.   By failing to take into account this information, we bound
the number of partitions we must sum over by a large combinatorial factor $(C k_{max})^{k_{max}}$ rather than a factor
that depends on the precise geometric constraints on the path.

To reach diffusive times we must make our bounds more efficient on both fronts.  To perform our resummation, we
introduce in Section~\ref{sec:extended-paths} the notion of an ``extended path'', which is a path formed
from $N$ segments each describing the evolution of the particle on an interval of length $\tau$.  An extended
path is a sequence of path segments with phase space points in between consecutive segments,
\[
\Gamma = (\xi_0,\omega_1,\xi_1,\xi_2,\omega_2,\xi_3,
\cdots, \xi_{2\ell-2},\omega_\ell,\xi_{2\ell-1},\cdots,
\xi_{2N-2},\omega_N,\xi_{2N-1}).
\]
An example of an extended path is drawn in Figure~\ref{fig:extended-path}.
\begin{figure}
\centering
\includegraphics{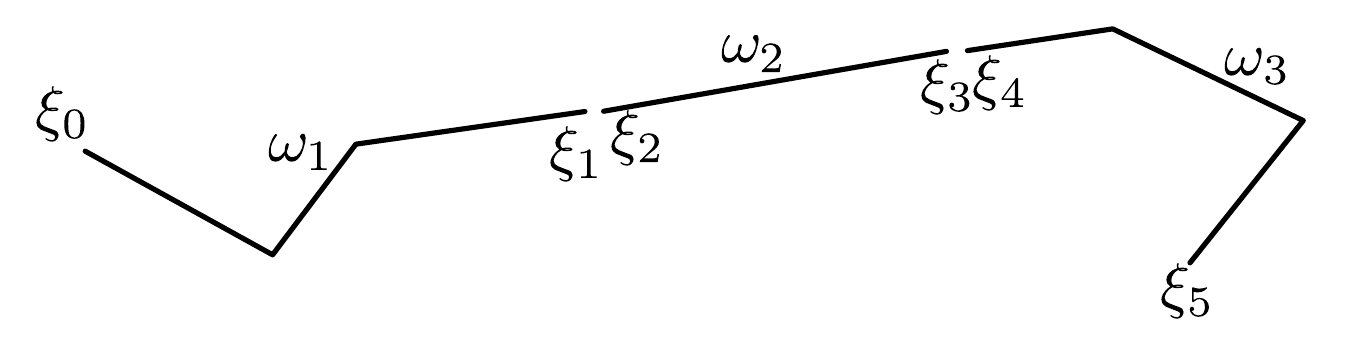}
\caption{A depiction of an extended path $\Gamma$.  The segment $\omega_j$ describes a piecewise linear path
between the endpoints $\xi_{2j-2}$ and $\xi_{2j-1}$.  The $\xi_j$ variables are phase space pairs $(x,p)$
describing the position and momentum of the particle at the boundary of the path segments.}
\label{fig:extended-path}
\end{figure}

Given an extended path we define an operator $O_\Gamma$ by
\[
O_\Gamma = \ket{\xi_0}\bra{\xi_{2N-1}} \prod_{\ell=1}^N \braket{\xi_{2\ell-2}|O_{\omega_\ell}|\xi_{2\ell-1}},
\]
so that including the sum over all possible collision numbers $0\leq k_j\leq k_{max}$ of each segment
in the integral, we have
\[
e^{iN\tau H} \approx \int O_\Gamma \diff \Gamma,
\]
where there is an error term in the approximation that is described in Section~\ref{sec:path-sketch}.

To write down the evolution channel $\Evol_{N\tau}$ we therefore arrive at an integral of the form
\begin{equation}
\Evol_{N\tau}[A] \approx \int \Expec O_{\Gamma^+}^* A O_{\Gamma^-} \diff\Gamma^+\diff\Gamma^-.
\end{equation}
Before we take the expectation, we will split up the pairs of paths $(\Gamma^+,\Gamma^-)$ according to their
geometric properties.
The key result we will need is a description of the structure of the correlation
partition of the paths in terms of the geometric features.  This is done in Section~\ref{sec:path-combo},
 which characterizes the allowed partitions using
ladders and anti-ladders.  Here we simply provide a quick sketch.  An example of a ladder partition
on the disjoint union $[k]\sqcup[k] = [k]\times\{+,-\}$ is
\[
P_{lad} = \{(j,+),(j,-)\}_{j\in [k]}.
\]
An example of an anti-ladder partition on $[k]\sqcup[k]$ is the partition
\[
P_{anti} = \{(j,+), (k+1-j,-)\}_{j\in[k]}.
\]
The ladder and anti-ladder partitions are drawn in Figure~\ref{fig:lads-and-antilads}
\begin{figure}
\includegraphics{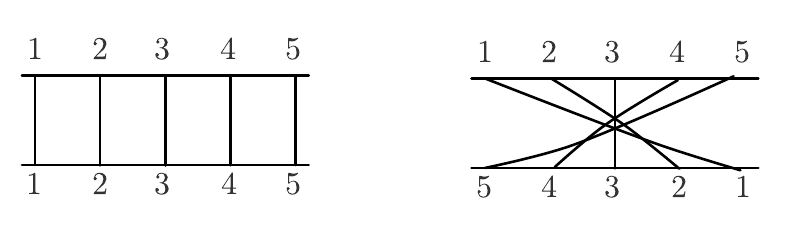}
\caption{An example of a ladder and an anti-ladder partition with five rungs.}
\label{fig:lads-and-antilads}
\end{figure}
The main result of Section~\ref{sec:path-combo} is Lemma~\ref{lem:interval-pair}, which states that
collisions that are not part of a geometric feature (so-called ``typical collisions'') form part of either
a ladder or an anti-ladder structure in the collision partition of $(\Gamma^+,\Gamma^-)$.
Figure~\ref{fig:diagram-figure} illustrates the main result in a special case.
\begin{figure}
\centering
\includegraphics{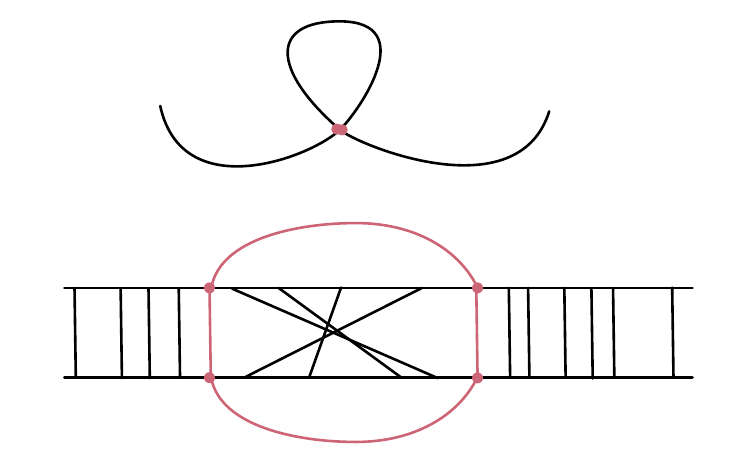}
\caption{An example illustrating Lemma~\ref{lem:interval-pair}.  On the top, a path that has a single recollision
event.  At the bottom, an example of a collision partition compatible with this single recollision event.  Note
that every collision that is not part of the recollision is either part of a ladder or an anti-ladder.}
\label{fig:diagram-figure}
\end{figure}

The next step is to partition the path integral according to the geometric information of the paths, which we
encapsulate in a structure that we call a ``skeleton'' $\mcal{F}(\Gamma^+,\Gamma^-)$.  Given a skeleton
$\mcal{F}$, Lemma~\ref{lem:interval-pair} allows us to construct a set of partitions $\mcal{Q}_\mcal{F}$ such that
for any pair of paths $(\Gamma^+,\Gamma^-)$ with $\mcal{F}(\Gamma^+,\Gamma^-)= \mcal{F}$,
the collision partition $P(\Gamma^+,\Gamma^-)\in\mcal{Q}_\mcal{F}$.  In fact, we have the stronger statement
that for such pairs of paths,
\begin{equation}
\label{eq:skeleton-expec}
\Expec O_{\Gamma^+}^*A O_{\Gamma^-} = \sum_{P\in\mcal{Q}(\mcal{F})} \Expec_P O_{\Gamma^+}^*A O_{\Gamma^-},
\end{equation}
where the sum is over partitions of the collision indices of $(\Gamma^+,\Gamma^-)$, and $\Expec_P$ is shorthand for
a  splitting of the expectation along the partition $P$.

Writing $\One_{\mcal{F}}(\Gamma^+,\Gamma^-)$ for the indicator function that $\mcal{F}(\Gamma^+,\Gamma^-)=\mcal{F}$,
we then attain the following decomposition for the path integral:
\begin{align*}
\Evol_{N\tau}[A]
&= \int \Expec O_{\Gamma^+}^* A O_{\Gamma^-} \diff \Gamma^+\diff\Gamma^- \\
&= \sum_{\mcal{F}} \int \One_{\mcal{F}}(\Gamma^+,\Gamma^-)\Expec O_{\Gamma^+}^* A O_{\Gamma^-}
\diff\Gamma^+\diff\Gamma^-\\
&= \sum_{\mcal{F}} \sum_{P\in\mcal{Q}_\mcal{F}}
\int \One_{\mcal{F}}(\Gamma^+,\Gamma^-)\Expec_P (O_{\Gamma^+})^* A O_{\Gamma^-}
\diff\Gamma^+\diff\Gamma^-.
\end{align*}

The benefit of decomposing the path integral in this way is that the expectation $\Expec_P$ splits in a known
way.  On the other hand there is now the challenge of dealing with the indicator function $\One_\mcal{F}$.
The reason this indicator functions causes a problem is not
the discontinuity (this could be solved by using a smoother partition of unity) but rather the global nature
of the constraints.  In particular, $\One_{\mcal{F}}(\Gamma^+,\Gamma^-)$ includes a product of indicator functions for
each \textit{negative constraint}, that is each pair of collisions that does not form a recollision or a tube event.
The negative constraints are needed to be able to apply Lemma~\ref{lem:interval-pair}, but they make it difficult
to exploit the cancellations needed.
To get around this we use a special form of the inclusion-exclusion principle that is tailored to this purpose.
In particular, in Section~\ref{sec:skeletons} we decompose the indicator function $\One_\mcal{F}$ in the form
\begin{equation}
\label{eq:inc-exc}
\One_{\mcal{F}} = \sum_{\mcal{F}'\geq \mcal{F}} G_{\mcal{F},\mcal{F}'},
\end{equation}
where we impose a partial ordering $\leq$ on skeletons, and
where $G_{\mcal{F},\mcal{F}'}$ is supported on the set of pairs $(\Gamma^+,\Gamma^-)$ such that
$\mcal{F}(\Gamma^+,\Gamma^-)\geq \mcal{F}'$.  In the decomposition~\eqref{eq:inc-exc}, the terms
$G_{\mcal{F},\mcal{F}'}$ depend only on the variables involving collisions that are in the support of the skeleton
$\mcal{F}'$ (that is, collisions involved in a recollision, a cone event, or a tube event).

A challenge is to find a better way to handle the sum over partitions in $\mcal{Q}_\mcal{F}$.
For this we introduce the concept of \textit{colored operators}.  Given a ``coloring'' function $\chi$ which assigns a unique color to
each collision in $\Gamma$, we define the colored operator $O_\Gamma^\chi$ to be an analogue
of $O_\Gamma$ which replaces each instance of the potential $V$ with an appropriately chosen independent copy
of $V$.  Then given a skeleton $\mcal{F}$, we construct two sets of colors $\Psi^+(\mcal{F})$ and $\Psi^-(\mcal{F})$
so that
\[
\sum_{P\in\mcal{Q}(\mcal{F})} \Expec_P O_{\Gamma^+}^*A O_{\Gamma^-}
=
\sum_{\chi^+\in\Psi^+(\mcal{F})}
\sum_{\chi^-\in\Psi^-(\mcal{F})}\Expec (O_{\Gamma^+}^{\chi^+})^*A O_{\Gamma^-}^{\chi^-}
=: \Expec (O_{\Gamma^+}^{\Psi^+})^* A O_{\Gamma^-}^{\Psi^-}.
\]
The precise definition of colored operators is given in Section~\ref{sec:colored-ops}, and the construction
of colorings that reproduce the partition collection $\mcal{Q}_\mcal{F}$ is done in Section~\ref{sec:coloring-machine}.
The benefit of writing the expectation in this way is that we can use the ``operator Cauchy-Schwartz''
inequality
\[
\|\Expec X^*AY\|_{op} \leq \|A\|_{op} \|\Expec X^*X\|_{op}^{1/2} \|\Expec Y^*Y\|_{op}^{1/2}
\]
where $X$ and $Y$ are random operators, to simplify the estimation of the contribution from paths with
skeleton $\mcal{F}$.  The result of this Cauchy-Schwartz procedure is depicted in Figure~\ref{fig:cauchy-schwartz}.
\begin{figure}
\centering
\includegraphics{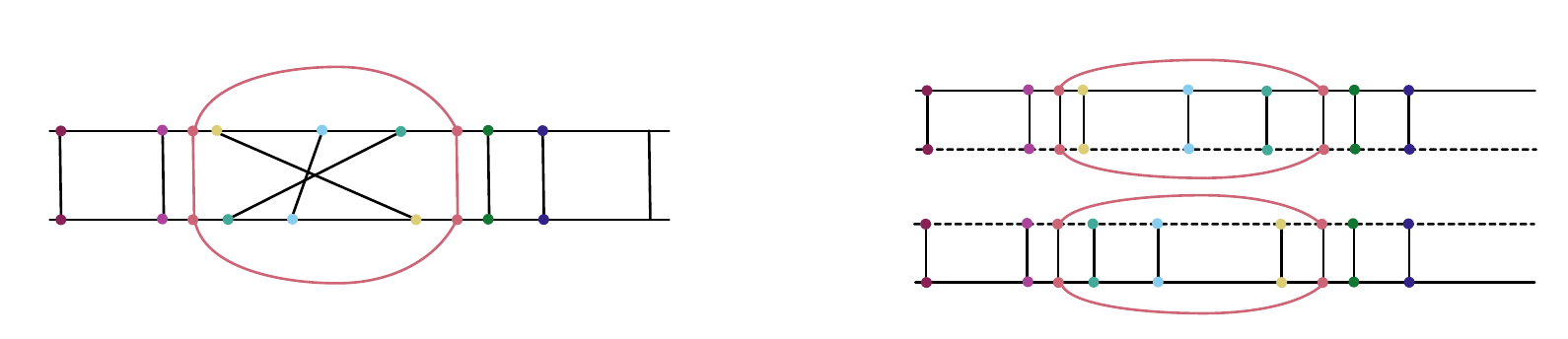}
\caption{On the left, a partition with a collision coloring chosen so that matched collisions have the same color.  When the
operator Cauchy-Schwartz inequality is applied, a copy of the top and bottom rows are produced and matched to each other,
converting the anti-ladder portion of the partition into a ladder (right).}
\label{fig:cauchy-schwartz}
\end{figure}

More precisely, we will apply the operator Cauchy-Schwartz inequality to the operator
\begin{align*}
\Evol_{\mcal{F},\mcal{F}'}[A] &:=
\int G_{\mcal{F},\mcal{F}'}(\Gamma^+,\Gamma^-)
\sum_{P\in\mcal{Q}_\mcal{F}} \Expec_P (O_{\Gamma^+})^* A O_{\Gamma^-}
\diff\Gamma^+\diff\Gamma^- \\
&= \int G_{\mcal{F},\mcal{F}'}(\Gamma^+,\Gamma^-)
\Expec (O_{\Gamma^+}^{\Psi^+})^* A O_{\Gamma^-}^{\Psi^-}.
\diff\Gamma^+\diff\Gamma^-.
\end{align*}
To do this we first split $G_{\mcal{F},\mcal{F}'}$ as a mixture of product functions of the form
\begin{equation}
G_{\mcal{F},\mcal{F}'}(\Gamma^+,\Gamma^-) = \int
H(\theta) \chi_{\mcal{F},\mcal{F}',\theta}^+(\omega)
\chi_{\mcal{F},\mcal{F}',\theta}^-(\omega')\diff \theta.
\end{equation}
We also decompose the coloring sets $\Psi^\pm(\mcal{F},\mcal{F}')$ into carefully chosen components
which are specified by a data structure called a \textit{scaffold}, so we decompose
\[
\Psi^\pm(\mcal{F},\mcal{F}') = \bigcup_{\Scaff\in \mcal{D}^\pm(\mcal{F},\mcal{F}')}
\Psi^\pm(\Scaff).
\]
Then by applying the operator Cauchy-Schwartz inequality we arrive at the estimate
\begin{equation}
\begin{split}
\label{eq:CS-op-sketch}
\|\Evol_{\mcal{F},\mcal{F}'}[A]\|_{op}
&\leq \|A\|_{op}
(\#\mcal{D}^+(\mcal{F},\mcal{F}'))
\sum_{\Scaff\in\mcal{D}^+(\mcal{F},\mcal{F}')}
\big\|
\int |H(\theta)| \Expec (X^+_{\Scaff,\theta})^* (X^+_{\Scaff,\theta})\diff\theta\big\|_{op}^{1/2} \\
&\qquad
(\#\mcal{D}^-(\mcal{F},\mcal{F}')) \sum_{\Scaff\in\mcal{D}^-(\mcal{F},\mcal{F}')}
\big\|
\int |H(\theta)| \Expec (X^-_{\Scaff,\theta})^* (X^-_{\Scaff,\theta})\diff\theta\big\|_{op}^{1/2},
\end{split}
\end{equation}
where
\[
X^\pm_{\Scaff,\theta}
:= \int \chi_{\mcal{F},\mcal{F}',\theta}^\pm(\omega) \sum_{\psi\in\Psi^\pm(\Scaff)} O_\Gamma^\psi\diff \Gamma.
\]
This calculation involving the Cauchy-Schwartz inequality is done more carefully in Section~\ref{sec:diagrams}.

The point of defining the scaffolds is that we can arrange that the operators
$\Expec (X^\pm_{\Scaff,\theta})^*(X^\pm_{\Scaff,\theta})$ involve sums over partitions that are formed
only from ladders, not anti-ladders.
The point is that ladder partitions have a semigroup structure in the sense
that the concatenation of two ladders is a ladder.  We use this structure to more easily use exact cancellations.
More precisely, we use the fact that the ladder partitions form a good approximation to the evolution
$\Evol_\tau$ at subkinetic times and the fact that $\Evol_\tau$ is a contraction in operator norm in order
to obtain the bounds we need.  These bounds are proven in Section~\ref{sec:ladders-superop}.

We also point out that the use of Cauchy-Schwartz in this way is lossy, but it only loses a factor of
$N^{C\|\mcal{F}\|}$.  We can afford to lose this factor because the operator norm appearing in the
right hand side~\eqref{eq:CS-op-sketch} will have order $\eps^{c\|\mcal{F}\|}$.  Roughly this is because
we obtain a factor of $\eps$ for each special geometric event described in $\mcal{F}$.
A careful argument is needed to ensure that one can obtain additional factors of $\eps$ for each recollision
(say) in an integral over paths containing multiple recollisions, and this is done
in Section~\ref{sec:diffusive-bds}.

\subsection{An abbreviated review of related works}
Here we point out some related works, making no attempt at providing a complete review of the rich
field of dynamics in random fields.

The rigorous study of the random Schrodinger equation began with the previously mentioned
work of H. Spohn~\cite{spohn77}.  As mentioned previously, Spohn's analysis was extended to the kinetic
time in~\cite{erdosyau2000} and then to diffusive times in~\cite{esyII07,esyI08}.  Each of these papers
considers the convergence in mean of the Wigner function to the solution of a kinetic equation.  A natural
question is to understand the size of the fluctuations of the Wigner function.  An analysis was carried
out by T. Chen in~\cite{chenhighermean2006} which showed that in fact that the $r$-th moments of the Wigner
function are bounded for any $r<\infty$.  Chen's analysis was later improved by M. Butz in~\cite{butz2013}.
We also point out the work of J. Lukkarinen and H. Spohn in~\cite{lukkarinen2007kinetic}, which shows
that the diagrammatic methods applied to the Schr{\"o}dinger equation can also be used to derive kinetic
equations for a random wave equation.  See the review~\cite{kineticwavereview} for a more complete discussion
of the kinetic regime for waves in a random medium.

Other regimes of interest are the homogenization regime in which the wavelength of the initial condition
is substantially longer than the decorrelation length scale of the potential.  This was studied
by G. Bal and N. Zhang in~\cite{balzhang2014}, where a homogenized equation  with a constant effective potential
is shown to describe the evolution of the average wave function.
This limit was further studied by T. Chen, T. Komorowski, and L. Ryzhik in~\cite{avgwavefx2018}.
An entirely different approach to the study of the average wave function in the kinetic regime
was introduced by M. Duerinckx and C. Shirley in~\cite{duerinckxshirley21}.  There the authors use
ideas from spectral theory to understand the evolution operator, and are able to show with this
method that the average wave function decays exponentially on the kinetic time scale.

The high frequency regime, in which the wavelength of the initial
condition is much shorter than the decorrelation length of the potential, was considered
by G. Bal, T. Komorowski, and L. Ryzhik in~\cite{bkrselfavg2003}. There the authors derive a Fokker-Planck
equation for the evolution of the Wigner function.
%See also~\cite{timesplitting04} for a discussion of
%the high frequency regime in random media. \red{I am not sure \cite{timesplitting04} is a good reference, looks more
%a reference imposed by your advisor. But I have nore read that paper in years.}

The study of the random Schrodinger equation falls into a larger body of work of understanding the emergence
of apparently irreversible phenomena from reversible dynamics~\cite{spohnreview88}.
From this point of view, the
random Schrodinger equation is simply the one-particle quantum mechanical manifestation of a larger phenomenon.

The classical version is the stochastic acceleration problem given by the ordinary differential equation
\[
\ddot{x} = -\eps \nabla V(x)
\]
where again $V$ is a stationary random potential.
Diffusive behavior for the stochastic acceleration problem was first
demonstrated by H. Kesten and G. Papanicolau in~\cite{KestenPapanicolau81} for dimensions $d\geq 3$.  For a
special class of potentials their argument was then applied to the two dimensional case
by D. D\"{u}rr, S. Goldstein, and J. Lebowitz in~\cite{twodmotion87}, and then
T. Komorowski and L. Ryzhik lifted the restriction on the potentials in~\cite{KomorowskiRyzhik06diffusion}.
The argument used by Kesten and Papanicolau
inspired the semigroup approach taken in this paper.  The connection is that Kesten and Papanicolau define a modified
version of the stochastic acceleration problem in which unwanted correlations are ruled out by fiat.  They then
show that this modified evolution is unlikely to have recollisions after all, and therefore is close to the original
evolution.  In a similar way we define an evolution (the refreshed evolution $\Evol_s^m$) which removes
unwanted correlations and use properties of this evolution to study the true evolution $\Evol_{ms}$.  Although
this is where the similarities end, it does seem that a further unification of the proof techniques may
be possible one day.

There are a number of other classical models of particles moving in a random environment.
A popular model
is the Lorentz gas, in which a billiard travels through $\Real^d$ with some obstacles placed according to a Poisson
process.  A pioneering paper in the study of the Lorentz gas is~\cite{bbsgas83} where a linear Boltzmann
equation is derived at the Boltzmann-Grad limit of the model.  A review of this model is provided
in~\cite{diffusiongas14}.  We refer the reader also to some exciting recent
developments in this field~\cite{walksrandom20,LedgerTothValko18,LutskoToth20}.  It seems that the classical
models of particles moving in random environment contain many of the same difficulties of understanding the quantum
evolution.  A deeper understanding of the phase-space path integral may lead us to a better understanding
of the relationship between the classical and quantum problems.

The random Schrodinger equation is also closely related to wave-kinetic theory in which one studies the
evolution of random waves with a nonlinear interaction (see~\cite{MR3971580} for a physically-motivated introduction
to this theory).  A pioneering work in this field is the paper of Lukkarinen and Spohn~\cite{weaklynonlinear11},
in which a wave kinetic equation is derived for the nonlinear Schrodinger equation for initial conditions
that are perturbations of an equilibrium state.  In a series of
works~\cite{faougermainhani16,buckmastershatahgermain20,DengHani21,onsetturbulence21} a wave kinetic
equation was derived for the nonlinear Schrodinger equation on a torus with more general initial
conditions.  Independently,
in~\cite{staffilani2021wave} a wave kinetic equation was derived for the Zakharov-Kuznetsov equation on $\Real^d$
for $d\geq 2$.   Each of these works follows the traditional strategy of writing out a diagrammatic expansion
for the solution and finding sources of cancellation in the error terms and comparing the main terms to a
perturbative expansion of the kinetic equation.  It seems possible that the wavepacket decomposition
used in this paper and the approximate-semigroup argument could be used to make further progress in wave-kinetic
theory.

\subsection{Acknowledgements}
The author is very grateful to Lenya Ryzhik for years of support, advice, and many clarifying discussions. The author also warmly
thanks Minh-Binh Tran for many helpful conversations about the paper.  The author is supported by the Fannie and John Hertz Foundation.

\section{More detailed outline of the proof}
\label{sec:outline}
In this section we lay out the main lemmas used to prove Theorem~\ref{thm:main-result}.

The proof involves analysis of three time scales.  The first time scale is the time $\eps^{-1.5}$ during which
the particle is unlikely to scatter at all and in particular is unlikely to experience more
than one scattering event.  The main result we need from this time scale shows that the linear Boltzmann
equation agrees with the evolution $\Evol_s$ with an error that is very small in operator norm.
This calculation is standard and is reproduced in Appendix~\ref{sec:boltzmann} for the sake of completeness.
The calculation only involves two terms from the Duhamel expansion
 of $e^{-itH}$ so there are no combinatorial difficulties.
\begin{restatable}{proposition}{shorttime}
\label{prp:short-time-compare}
There exists $\theta(d)>0$ such that the following holds:
Let $a_0\in C^{2d+1}$ be an observable
supported on the set $\{(x,p) | |p|>\eps^\theta\}$, and suppose that $a_s$
solves the linear Boltzmann equation~\eqref{eq:dual-boltz}.  Then for $\sigma\leq \eps^{-1.5}$,
\begin{equation}
\|\Op(a_\sigma) - \Evol_\sigma[\Op(a_0)]\|_{op}
\lsim \eps^{2.1}\sigma \|a_0\|_{C^{2d+1}_{\eps^{-1},\eps^{-0.25}}}.
\end{equation}
\end{restatable}

To use Proposition~\ref{prp:short-time-compare} along with the semigroup approximation strategy, we
need the following regularity result for the short time evolution of the linear Boltzmann equation.
\begin{lemma}
\label{lem:short-regularity}
There exists $\theta=\theta(d)>0$ such that the following holds:
Let $a_s$ solve the linear Boltzmann equation~\eqref{eq:dual-boltz}
and $\supp a_0 \subset\{(x,p)\mid |p|\geq \eps^\theta\}$.  Then for $s\leq \eps^{-2.05}$,
\[
\|a_s\|_{C^{2d+1}_{\eps^{-1},\eps^{-0.25}}} \leq C \|a_0\|_{C^{2d+1}_{\eps^{-1},\eps^{-0.5}}}.
\]
\end{lemma}
Lemma~\ref{lem:short-regularity} is proven with a simple and suboptimal argument in
Appendix~\ref{sec:regularity},
where we prove a slightly stronger version in Lemma~\ref{lem:boltz-regularity}.

Using Proposition~\ref{prp:short-time-compare} and Lemma~\ref{lem:short-regularity} we can prove that the
``$\sigma$-refreshed'' evolution $\Evol_\sigma^n$ approximates the linear Boltzmann equation up to a diffusive
timescale.
\begin{corollary}
\label{cor:refreshed-evolution}
For $\sigma=\eps^{-1.5}$ and $m\in\bbN$ such that $m\sigma\leq \eps^{-2.05}$, and $a_s$
solving~\eqref{eq:dual-boltz},
\begin{equation}
\label{eq:refreshed-evol-compare}
   \|\Evol_\sigma^m[\Op(a_0)] - \Op(a_{m\sigma})\|_{op} \leq C \eps^{2.1}m\sigma
\|a_0\|_{C^{2d+1}_{\eps^{-1},\eps^{-0.5}}}.
\end{equation}
\end{corollary}
\begin{proof}
We define the quantity
\[
F_j := \sup_{\|a_0\|_{C^{2d+1}_{\eps^{-1},\eps^{-0.5}}}=1}
\|\Evol_\sigma^j[\Op(a_0)] - \Op(a_{j\sigma})\|_{op}.
\]
By Proposition~\ref{prp:short-time-compare}, $F_1\leq C \eps^{2.1}\sigma$.  To obtain a bound for
$F_{j+1}$ from $F_j$ we write
\[
\|\Evol_\sigma^m[\Op(a)] - \Op(a_{m\sigma})\|_{op}
\leq \|\Evol_\sigma[\Op(a_{(m-1)\sigma})] - \Op(a_{m\sigma})\|_{op}
+ \|\Evol_\sigma^m[\Op(a)] - \Evol_\sigma[\Op(a_{(m-1)\sigma})]\|_{op}.
\]
The first quantity is bounded using Proposition~\ref{prp:short-time-compare} and Lemma~\ref{lem:short-regularity}.
The second term is bounded by $F_j$
using the fact that $\Evol_\sigma$ is linear and is a contraction in the operator norm:
\[
\|\Evol_\sigma[A]\|_{op} =
\|\Expec e^{i\sigma H} A e^{-i\sigma H}\|_{op}
\leq \Expec
\|e^{i\sigma H} A e^{-i\sigma H}\|_{op} = \|A\|_{op}.
\]
Therefore we obtain the bound
\[
F_{j+1} \leq F_j + C \eps^{2.1}\sigma.
\]
In particular,
\[
F_m \leq C\eps^{2.1} m\sigma,
\]
so~\eqref{eq:refreshed-evol-compare} follows.
\end{proof}

The more substantial component of the proof of Theorem~\ref{thm:main-result} is the approximate
semigroup property relating the ``refreshed'' evolution $\Evol_\sigma^M$ to the correct evolution
channel $\Evol_{M\sigma}$.  For the purposes of proving an approximate semigroup property it is more convenient to
work with the ``wavepacket quantization'' defined by.
\[
\Op(a) := \int_{\Real^d} \ket{\xi}\bra{\xi} a(\xi)\diff\xi.
\]
In Appendix~\ref{sec:wp-quantization} we will show that the wavepacket
quantization is close to the Weyl quantization, in the sense that
\[
\|\Op(a) - \Op^w(a)\|_{op}
\leq \eps^{0.05} \|a\|_{C^{2d+1}_{\eps^{-1},\eps^{-0.1}}}.
\]
In general, we will be interested in operators of the form
\[
A = \int \ket{\xi}\bra{\eta} a(\xi,\eta)\diff\xi
\]
with kernel $a(\xi,\eta)$ satisfying $|a(\xi,\eta)|\leq C\|A\|_{op}$ and
supported on near the diagonal.  To quantify this we introduce the distance $d_r$
on $\PhaseSpace$ so that, writing $\xi=(\xi_x,\xi_p)$ and $\eta=(\eta_x,\eta_p)$,
\[
d_r(\xi,\eta) := r^{-1} |\xi_x-\eta_x| + r|\xi_p-\eta_p|.
\]
More precisely, we are interested in families of \emph{good operators}, defined below.
\begin{definition}[Good operators]
An operator $A\in\mathcal{B}(L^2(\Real^d))$ is said to be $(C_1,C_2,\delta)$-good if
there exists a function $a:\PhaseSpace\times\PhaseSpace\to\Complex$ supported in the set
\[
\supp a \subset \{(\xi,\eta)\in \PhaseSpace\times\PhaseSpace \mid
d_r(\xi,\eta) < C_1, |\xi_p| \geq \eps^{\theta} - C_2\eps\}
\]
such that
\[
\|A - \int a(\xi,\eta) \ket{\xi}\bra{\eta} \diff\xi\diff\eta\|_{op} \leq \delta.
\]
\end{definition}
Note that the rank one projection onto a
wavepacket $\ket{\xi}\bra{\xi}$ is (formally) a $(0,0,0)$-good operator if $|\xi_p|\geq \eps^{\theta}$,
but its Wigner transform has smoothness only at the microscopic scale $(r,r^{-1})$.  Similarly if
$a\in C^{2d+1}$ is an observable supported on $\{(x,p)||p|\geq \eps^{\theta}\}$, then
the wavepacket quantization $\Op(a)$ is a $(0,0,0)$-good operator.
Moreover, by~\eqref{eq:quantization-compare} we have that $\Op^w(a)$ is a $(0,0,\eps^{1/2})$-good operator
if $a\in C^{2d+1}_{\eps^{-1},\eps^{-1/2}}$.

The first step in the proof of the approximate semigroup property is to verify a semigroup property up to
times $\eps^{-2+\kappa/2}$.
\begin{proposition}
\label{prp:subkinetic-semigp}
If $A$ is a $(C_1,C_2,\delta)$-good operator with $C_1\leq \eps^{-0.1}$
and $C_2\leq \frac{1}{2} \eps^{\theta}$ and $\eps^{-1.5}<  s < \eps^{-2+\kappa/2}$
then
\[
\|\Evol_{2s}[A] - \Evol_s^2[A]\|_{op}
\leq \eps^{2.1} s \|A\|_{op} + 2\delta.
\]
and moreover $\Evol_s[A]$ is a $(C_1+ |\log\eps|^{10}, C_2-10^3\eps, \delta + \eps^{100})$-good
operator.
\end{proposition}

Proposition~\ref{prp:subkinetic-semigp} is proved by first comparing $\Evol_s$ to the expectation over
ladders, and then observing that the semigroup property holds for ladders.  The first step is done
in Section~\ref{sec:subkinetic-semigroup} where we prove Proposition~\ref{prp:short-ladder-compare}.
The derivation of Proposition~\ref{prp:subkinetic-semigp} from Proposition~\ref{prp:short-ladder-compare}
is explained by Lemma~\ref{lem:ladder-semigroup}.

Using Proposition~\ref{prp:subkinetic-semigp} we can prove a comparison result between the linear Boltzmann equation
and the quantum evolution for times up to $\eps^{-2+\kappa/10}$.
\begin{corollary}
If $A$ is a $(C_1,C_2,\delta)$-good operator with
$C_1 \leq \frac{1}{2}\eps^{-0.1}$ and  $C_2\leq \eps^\theta$ then with $\sigma=\eps^{-1.5}$,
\[
\|\Evol_\sigma^m[A] - \Evol_{m\sigma}[A]\|_{op}
\leq C_\kappa \eps^{2.1}\sigma \|A\|_{op} + 2\delta + \eps^{20}
\]
for $m$ such that $m\sigma \leq \eps^{-2+\kappa/2}$.
\end{corollary}
\begin{proof}
We perform an iteration, defining the error
\[
E_m := \sup_{\substack{\|A\|_{op}=1\\ A \text{ is } (C_1-K_1m,C_2+K_2m,\beta-\delta m){-good}}}
\|\Evol_{2^m\sigma}[A] - \Evol_\sigma^{2^m}[A]\|_{op},
\]
where we choose $C_1 = \frac{1}{2}\eps^{-0.1}$, $K_1=|\log\eps|^{10}$, $C_2 = \frac{1}{4}\eps^{0.1}$,
$K_2=10^3\eps$, $\beta = \eps^{100}$, and $\delta = \eps^{100}$.  Let
$\mcal{A}_m$ be the class of admissible operators in the supremum defining $E_m$.  The significant
point about $\mcal{A}_m$ is that $\Evol_s[A]\in\mcal{A}_m$ when $A\in\mcal{A}_{m+1}$.
To find a recursion for $E_m$, we write
\begin{equation}
\begin{split}
\|\Evol_{2^{m+1}s}[A] - \Evol_s^{2^{m+1}}[A]\|_{op}
&\leq
\|\Evol_{2^{m+1}s}[A] - \Evol_{2^ms}^2[A]\|_{op}
+ \|\Evol_{2^ms}^2[A] - \Evol_s^{2^{m+1}}[A]\|_{op} \\
&\leq
\|\Evol_{2^{m+1}s}[A] - \Evol_{2^ms}^2[A]\|_{op}
+ \|\Evol_{2^ms}[\Evol_{2^ms}[A] - \Evol_s^{2^m}[A]]\|_{op} \\
&\qquad\qquad+ \|(\Evol_{2^ms} - \Evol_s^{2^m}) \Evol_s^{2^m}[A]\|_{op}.
\end{split}
\end{equation}
Since $\Evol_s$ is a contraction in the operator norm, and since $\Evol_s$ maps
$\mcal{A}_{m+1}$ into $\mcal{A}_m$ we have by Proposition~\ref{prp:subkinetic-semigp}
that
\[
\|\Evol_{2^{m+1}s}[A] - \Evol_s^{2^{m+1}}[A]\|_{op}
\leq  ((\eps^{2.1} 2^ms) \|A\|_{op} + 2\delta) + E_m.
\]
Taking a supremum over $A$ we obtain the relation
\[
E_{m+1} \leq \eps^{2.1}2^m\sigma  + 2E_m.
\]
Since $E_0=0$, we obtain
\[
E_m \leq m \eps^{2.1}2^m\sigma.
\]
\end{proof}

The remaining ingredient needed to prove Theorem~\ref{thm:main-result} is a semigroup property that holds up to
diffusive times.  This is substantially more difficult than establishing the semigroup property for subkinetic
times because of the need for resummation in the Duhamel series.  The main result is the following.

\begin{proposition}
\label{prp:diffusive-semigp}
There exists $\kappa=\kappa(d)$ such that the following holds:  If $A$ is a
$(\eps^{-0.1}, \frac{1}{2}\eps^{0.1},\delta)$-good operator, then with
$\tau=\eps^{-2+\kappa/2}$ and $N=\lfloor \eps^{-\kappa}\rfloor$,
\[
\|\Evol_{N\tau}[A] - \Evol_\tau^N[A]\|_{op} \leq \eps^{c\kappa} \|A\|_{op} + 2\delta.
\]
\end{proposition}

Having sketched the argument proving the main result, we now outline the remaining sections in the paper.
In Section~\ref{sec:path-sketch} we explain the phase space path integral approximation we use throughout the
paper.  Then in Section~\ref{sec:subkinetic-semigroup} we introduce the ladder superoperator $\mcal{L}_s$
which is a main character in the derivation of the approximate semigroup property.
Section~\ref{sec:subkinetic-semigroup} contains the bulk of the proof of Proposition~\ref{prp:subkinetic-semigp}
and contains most of the main ideas of the paper.

The remaining sections in the paper are dedicated to the proof of Proposition~\ref{prp:diffusive-semigp}.  In
Section~\ref{sec:extended-paths} we write down the path integral used to represent the solution operator up to this
time.  Then in Section~\ref{sec:path-combo} we clarify the relationship between the geometry of paths and the
combinatorial features of their collision partitions.  Then in Section~\ref{sec:skeletons} we split up the path integral
according to the geometry of the paths.  To exploit the combinatorial structure of the correlation partition,
we introduce the formalism of colored operators in Section~\ref{sec:colored-ops} and
Section~\ref{sec:coloring-machine}.  Then in Section~\ref{sec:diagrams} we finally write out our version of
a ``diagrammatic expansion'' (which is different than previous expansions in that the first term of the expansion
for $\Evol_{N\tau}[A]$ is the refreshed evolution $\Evol_\tau^N[A]$).  The diagrams are bounded in
Section~\ref{sec:diffusive-bds}.

The remaining sections contain proofs of more technical results needed throughout the argument, and are referenced
as needed.

\section{A sketch of the derivation of the path integral}
\label{sec:path-sketch}
In this section we state the precise version of the phase-space path integral alluded to
in Section~\ref{sec:heuristics}.  The proofs of the assertions made in this section
are given in Sections~\ref{sec:duhamel},~\ref{sec:first-operator-bd},
and~\ref{sec:free-intersperse}.

The first step is to write out an expansion for $e^{-is H}$ that is valid for times
$s\leq \tau :=\eps^{-2+\kappa/2}$ in terms of paths.
More precisely, a path
$\omega = (\mbf{s},\mbf{p},\mbf{y})$ having $k$ collisions is a tuple
containing a list of collision invervals
$\mbf{s}\in\Real_+^{k+1}$ satisfying $\sum_{j=0}^k s_j=s$, momenta $\mbf{p} \in (\Real^d)^{k+1}$, and
collision locations $\mbf{y} \in (\Real^d)^k$.  Each path $\omega$ is
associated to the operator $O_\omega$ defined by
\begin{equation*}
\Ft{O_\omega \psi} = \delta_{p_k} \Ft{\psi}(p_0)
e^{i\varphi(\omega)} \prod_{j=1}^k \Ft{V_{y_j}}(p_j-p_{j-1}),
\end{equation*}
where $\varphi(\omega)$ is the phase function
\[
\varphi(\omega) = \sum_{j=0}^k s_j |p_j|^2/2 + \sum_{j=1}^k y_j\cdot (p_j-p_{j-1}).
\]
In Dirac notation, we express $O_\omega$ as
\begin{equation}
\label{eq:path-op-def}
O_\omega  = \ket{p_k}\bra{p_0}
e^{i\varphi(\omega)} \prod_{j=1}^k \Ft{V_{y_j}}(p_j-p_{j-1}).
\end{equation}
Here $V_y(x) = V(x-y)\chi^V(x)$ is a localized and shifted version
of the potential, with localization $\chi^V$ having width
$r$ and satisfying $\int \chi^V = 1$.

Let $\Omega_k(s)$ denote the space of paths with $k$ collisions
and duration $s$,
\begin{equation*}
\Omega_k(s) = \triangle_k(s)
\times(\Real^d)^{k+1} \times (\Real^d)^k,
\end{equation*}
where $\triangle_k(s)\subset \Real_+^{k+1}$ is the set of tuples
of time intervals summing to $s$,
\[
\triangle_k(s) = \{\mbf{s}=(s_0,\cdots,s_k)\in\Real_+^{k+1}\mid
\sum_{j=0}^k s_j = s\}.
\]
We will see in Section~\ref{sec:duhamel} that the Duhamel expansion can be formally written
\begin{equation*}
e^{isH} =
\sum_{k=0}^\infty T_k :=
\sum_{k=0}^\infty \int_{\Omega_k(s)} O_\omega \diff \omega.
\end{equation*}
In this integral there is no need for the collision locations
$\mbf{y}$ to have any relationship with the variables $\mbf{s}$
and $\mbf{p}$.  There is however significant cancellation due to the
presence of the phase $e^{i\varphi(\omega)}$.  For example, by integrating
by parts in the $p_j$ variables and using the identity,
\begin{align*}
\partial_{p_j} \varphi(\omega) = y_j + s_jp_j - y_{j+1},
\end{align*}
we can reduce the path integral to paths which satisfy
\[
|y_{j+1} - (y_j + s_jp_j)| \lessapprox r.
\]
Integration by parts in the $s_j$ is somewhat more delicate because
of the hard constraint $\sum s_j = \tau$.  By decomposing
this hard constraint as a sum of softer constraints, we can
impose a cutoff on the weaker conservation of kinetic energy condition
\[
||p_j|^2/2 - |p_{j'}|^2/2| \lessapprox \max\{|p_j|r^{-1}, s_j^{-1},
s_{j'}^{-1}\}.
\]
The integration by parts argument will allow us to
construct a function $\chi^{path}$ supported on the
set of such ``good'' paths and for which
\[
T_k \approx
\int_{\Omega_k(s)} \chi^{path}(\omega) O_\omega \diff\omega,
\]
with an error that is negligible in operator norm.  To be more precise,
given a tolerance $\alpha$ (which we set to be $|\log\eps|^{10}$), we define
\begin{equation}
\label{eq:Omega-alpha-def}
\begin{split}
\Omega_{\alpha,k}(s) =
\{(\mbf{s},\mbf{p},\mbf{y})\in\Omega_k(s) \mid
&|y_{j+1} - (y_j+s_jp_j)| \leq \alpha r \text{ for all } j\in[1,k-1]
\text{ and } \\
&
||p_j|^2/2 - |p_{j'}|^2/2| \leq \alpha \max\{ s_{j'}^{-1}, s_j^{-1},
|p_j|r^{-1}, \alpha r^{-2}\}
\text{ for all } j,j'\in [0,k] \}.
\end{split}
\end{equation}
Within $\Omega_{\alpha,k}(s)$ we also define the subset
\begin{equation}
\label{eq:omega-xi-eta-def}
\begin{split}
    \Omega_{\alpha,k}(t;\xi,\eta)
    :=
    \{(\mbf{s},\mbf{p},\mbf{y})\in\Omega_{\alpha,k}(t) \mid
        &|y_1 - (\xi_x + s_0p_0)| \leq \alpha r,
        |\eta_x - (y_k+s_kp_k)| \leq \alpha r, \\
        &|p_0-\xi_p| \leq \alpha r^{-1}, \text{ and }
    |p_k-\eta_p| \leq \alpha r^{-1} \},
\end{split}
\end{equation}
where $\xi,\eta\in\PhaseSpace$ with $\xi=(\xi_x,\xi_p)$ and
$\eta=(\eta_x,\eta_p)$.

The following lemma is simply a careful application of integration by parts, and is done
in Section~\ref{sec:duhamel}.
\begin{lemma}
    \label{lem:int-by-pts}
    There exists a cutoff function
    $\chi\in C^\infty(\Omega_k(\tau)\times\PhaseSpace\times\PhaseSpace)$
    supported in the set
    \[
        \supp \chi \subset
        \{(\omega,\xi,\eta) \mid
        \omega\in\Omega_{\alpha,k}(\tau; \xi,\eta)\}
    \]
    such that, with
    \begin{equation}
        T_k^\chi(\tau) :=
        (i\eps)^k \int_{\Omega_k(\tau)\times\PhaseSpace\times\PhaseSpace}
        \ket{\eta}\bra{\xi} \chi(\omega,\xi,\eta)
        \braket{\eta| O_\omega |\xi} \diff \omega\diff\xi\diff\eta,
    \end{equation}
    we have the approximation
    \begin{equation}
        \|T_k^\chi - T_k\|_{op} \leq
        \eps^{-C_d k}\|V\|_{C^{10d}}^{Ck} \exp(-c\alpha^{0.99}).
    \end{equation}
\end{lemma}

The point of Lemma~\ref{lem:int-by-pts} is that it allows us to neglect the contribution of
``physically unreasonable paths'' -- those that either badly violate conservation of kinetic energy
or the transport constraints $y_{k+1}\approx y_k+s_kp_k$.

We remark that Lemma~\ref{lem:int-by-pts} is deterministic in the sense
that that the conclusion holds for all potentials, and when we apply it
we will simply need moment bounds for the $C^{10d}$ norm of the potential
(after being cutoff to a ball of large radius).  With the choice
$\alpha=|\log\eps|^{10}$, and assuming that
$\|V\|_{C^{10d}}\leq \eps^{-100}$ (say), the right hand side is still
$O(\eps^K)$ for any $K>0$.

Having given a description to the collision operators $T_k$, it remains
to estimate moments of the form $\Expec \|T_k^\chi\|_{op}^M$, for which we
use the moment method:
\[
\Expec \|T_k^\chi\|_{op}^{2M} \leq \Expec \trace ((T_k^\chi)^*T_k^\chi)^M.
\]
Note that this step is where the cutoff on the potential is crucial --
without the cutoff $\chi_R$ the trace above would be infinite.
In Section~\ref{sec:first-operator-bd} we prove
Lemma~\ref{lem:collision-strength}, which states that
\[
\Big(\Expec \|T_k^\chi(s)\|_{op}^{2M}\Big)^{1/M}
        \leq R^{C/m} \eps^2 s \,\,(C(kM)^C |\log\eps|^{10}).
\]
The presence of the factor $(km)^C$ makes this bound unsuitable
for reaching diffusive time scales.  However this bound is good enough
to approximate $e^{is H}$ by $e^{-is \Delta/2}$ for times
$s\leq \eps^{-1.1}$ (say).  We use this result in Section~\ref{sec:free-intersperse}
to define a modified operator $T_{k,\rho}^\chi$
which involves a first short period of free evolution.  More precisely, given a time
$\sigma>0$ we construct a function $\rho_\sigma\geq 0$
that is supported on the interval $[\sigma2\sigma]$,
is Gevrey regular, and satisfies $\int \rho_\sigma =1$.  Then we define
\[
\wtild{U}_{\tau,\sigma} := \int_\sigma^{2\sigma} \int_\sigma^{2\sigma}
e^{is\Delta/2} e^{-i(\tau-s-s')H} e^{is'\Delta/2} \rho_\sigma(s)\rho_\sigma(s')
\diff s\diff s'.
\]
We will fix for the remainder of the paper $\sigma = \eps^{-1.5}$.  In Section~\ref{sec:free-intersperse}
we use Lemma~\ref{lem:int-by-pts} to prove Lemma~\ref{lem:free-replace}, which justifies the approximation
\[
\|\Evol_{N\tau}[A] - \Expec (\wtild{U}_{\tau,\sigma}^*)^N A \wtild{U}_{\tau,\sigma}^N\|_{op}
\leq \eps^{0.2}.
\]
This will allow us to restrict the path integral to a
space of paths which do not have a collision too close to either endpoint,
\[
\Omega_{k,\alpha}(\tau,\sigma;\xi,\eta)
:= \{(\omega,\xi,\eta) \in \Omega_{k,\alpha}(\tau;\xi,\eta) \mid
s_0 \geq \sigma \text{ and } s_k \geq \sigma\}.
\]
The operator $\wtild{U}_{\tau,\sigma}$ also has a path integral expansion in terms of collision
operators $T_k^{\chi,\sigma}$, which in addition to having the smooth cutoff $\chi$ have a cutoff
enforcing that $\omega\in \Omega_{k,\alpha}(\tau,\sigma;\xi,\eta)$.

Combining the above arguments we obtain the following approximation result for the evolution
operator $\wtild{U}_{\tau,\sigma}$.
\begin{proposition}
\label{prp:main-approximation}
There exists a bounded smooth function
$\chi_\sigma\in C^\infty(\Omega_k(\tau)\times\Real^{2d}\times\Real^{2d})$ supported in the set
\[
\supp (\chi_\sigma) \subset
\Omega_{\alpha,k}(\tau,S;\xi,\eta) :=
\{(\omega,\xi,\eta) \mid
\omega \in\Omega_{k,\alpha}(\tau;\xi,\eta) \text{ and }s_0 \geq \eps^{-1.6},
\text{ and } s_k \geq \eps^{-1.6}\},
\]
and such that, with
\begin{equation}
\label{eq:Tkrho-def}
T_{k}^{\chi,\sigma}
:= (i\eps)^k \int_{\Omega_k(\tau)\times\Real^{2d}\times\Real^{2d}}
\ket{\eta}\bra{\xi}
\chi_\sigma(\omega,\xi,\eta) \braket{\eta|O_\omega|\xi} \diff\omega\diff\xi\diff\eta,
\end{equation}
we have the Duhamel expansion
\[
\wtild{U}_{\tau,\sigma} := \sum_{k=0}^{k_{max}} T_k^{\chi,\sigma}
+ R_{k_{max}}^{\chi,\sigma},
\]
where the remainder $R_{k_{max}}^{\chi,\rho}$ has the expression
\begin{equation}
\label{eq:Rkmax-def}
R_{k_{max}}^{\chi,\sigma}
:= (i\eps)^k
\int_0^\tau\diff s
\int_{\Omega_{k_{max}}(s)\times\Real^{2d}\times\Real^{2d}}
e^{i(\tau-s) H} \ket{\eta}\bra{\xi}
\chi_\sigma(\omega,\xi,\eta) \braket{\eta|O_\omega|\xi} \diff\omega\diff\xi\diff\eta,
\end{equation}
and
\begin{equation}
\label{eq:operator-error-bd}
\left(\Expec \|e^{i\tau H} - \wtild{U}_{\tau,\sigma}\|_{op}^{\eps^{-0.1}}\right)^{\eps^{0.1}}
\leq \eps^{0.3}
\end{equation}
\end{proposition}

We use the operator $\wtild{U}_{\tau,\sigma} =: U_{\tau,\sigma}  + R_{k_{max}}^{\chi,\sigma}$
to decompose $e^{iN\tau H}$ as follows:
\[
e^{iN\tau H} \approx U_{\tau,\sigma}^N
+ \sum_{j=1}^N e^{i(N-j)\tau H} R_{k_{max}}^{\chi,\sigma} U_{\tau,\sigma}^{j-1}.
\]
Let $\wtild{\Evol}_{N\tau}$ be the superoperator formed from the main term,
\[
\wtild{\Evol}_{N\tau}[A] := \Expec (U_{\tau,\sigma}^N)^* A U_{\tau,\sigma}^N.
\]
Let $R_j$ be the operator $R_{k_{max}}^{\chi,\sigma} U_{\tau,\sigma}^{j-1}$,
which also can be written $R_j=\int_0^\tau e^{i(\tau-s)H}R_{j,s}\diff s$ with
\begin{equation}
\label{eq:remainder-def}
R_{j,s} := \int_{\Omega_{k_{max}}(s)\times\Real^{2d}\times\Real^{2d}}
\ket{\eta}\bra{\xi}U_{\tau,\sigma}^{j-1}
\chi_\sigma(\omega,\xi,\eta) \braket{\eta|O_\omega|\xi} \diff\omega\diff\xi\diff\eta.
\end{equation}
Then by an application of the operator Cauchy-Schwarz inequality and the triangle inequality we have the estimate
\begin{equation}
\label{eq:tilde-approx}
\|\Evol_{N\tau}[A] - \wtild{\Evol}_{N\tau}[A]\|_{op}
\leq \|A\|_{op}(\eps^{0.3} + N\tau \max_{j\in[N]}\sup_{s\in[0,\tau]}\|\Expec R_{j,s}^*R_{j,s}\|_{op}).
\end{equation}
In the course of understanding the evolution $\wtild{\Evol}_{N\tau}$ we will derive estimates that as a byproduct
prove
\begin{equation}
\label{eq:remainder-bound}
\max_{j\in[N]}\sup_{0\leq s\leq \tau} \|\Expec R_{j,s}^*R_{j,s}\|_{op} \leq \eps^{50}.
\end{equation}
In Section~\ref{sec:remainders} we explain how this bound is obtained as a modification of the argument
used to control the diffusive diagrams.

\section{The ladder approximation for $\Evol_\tau$}
\label{sec:subkinetic-semigroup}
In this section we sketch the proof that the evolution channel $\Evol_s$
is well approximated by a sum over ladder diagrams when $s\leq \tau = \eps^{-2}N^{-\kappa/10}$.
This is closely related to the semigroup property which we will explore in a later section.
The statement of the main result of the section, Proposition~\ref{prp:short-ladder-compare}, is given
in Section~\ref{sec:ladder-statement} after some preliminary calculations which motivate
the definition of the ladder superoperator.

\subsection{An introduction to the channel $\Evol_s$}
For  times $s\leq\tau=\eps^{-2}N^{-\kappa/10}$, we may use Lemma~\ref{lem:int-by-pts} to write
\begin{equation}
\label{eq:forward-path-integral}
e^{-isH} = \sum_{k=0}^{k_{max}}
\int_{\Omega_k(s)} \chi_S(\omega,\xi,\eta)\ket{\eta}\bra{\xi}
\braket{\eta|O_\omega|\xi}
\diff \omega\diff\xi\diff\eta + E
\end{equation}
where $\chi_S(\omega,\xi,\eta)$ is a smooth function supported on the set $\Omega_{\alpha,k}(s,\sigma;\xi,\eta)$
and the approximation is up to an error $E$ that satisfies $\Expec \|E\|_{op}^2 \leq \eps^{0.2}$.

The operator $e^{isH}$ can similarly be expressed as an integral over paths,
\begin{equation}
\label{eq:backward-path-integral}
e^{isH} = (e^{-isH})^* =
\sum_{k=0}^{k_{max}}
\int_{\Omega_k(s)} \chi_S(\omega,\xi,\eta)\ket{\xi}\bra{\eta}
\braket{\eta|O_\omega|\xi}^*
\diff \omega\diff\xi\diff\eta + E^*.
\end{equation}

We now use~\eqref{eq:forward-path-integral} and~\eqref{eq:backward-path-integral} to
write an expansion for $\Evol_s[A]$.  We will drop the summation over $k$ and handle
the sum implicitly in the integral over $\Omega(s) = \bigcup_{j=0}^{k_{max}}\Omega_k(s)$.
\begin{equation*}
\Evol_s[A] \approx
\int
\ket{\xi_0^-}
\braket{\xi_1^-|O_{\omega^-}|\xi_0^-}^*
\braket{\xi_1^-|A|\xi_1^+}
\braket{\xi_1^+|O_{\omega^+}|\xi_0^+}\bra{\xi_1^+}
\chi_S(\omega^+,\xi_0^+,\xi_1^+)
\chi_S(\omega^-,\xi_0^-,\xi_1^-)\diff\bm{\omega}\diff\bm{\xi}.
\end{equation*}
up to a remainder that is bounded by $O(\eps^{0.2}\|A\|_{op})$ in operator norm.

To express the operator more compactly we introduce some notation.  We write
$\bm{\Gamma} = (\xi_0^+,\omega^+,\xi_1^+; \xi_0^-,\omega^-,\xi_0^-) = (\Gamma^+;\Gamma^-)$
for the full path, and then define the path cutoff function
\[
\Xi(\bm{\Gamma}) :=
\chi_S(\omega^+,\xi_0^+,\xi_1^+)
\chi_S(\omega^-,\xi_0^-,\xi_1^-).
\]
Moreover, we will stack like terms to keep the integrand more organized.
That is we will write $\mstack{A}{B}$ to mean the product $AB$.  With this notation,
\[
\Evol_s[\Op(a)] \approx
\int \ket{\xi_0^-}\bra{\xi_0^+}
\braket{\xi_1^-|A|\xi_1^+}
\Xi(\Gamma)
\Expec \mstack{\braket{\xi_1^+|O_{\omega^+}|\xi_0^+}}{\braket{\xi_1^-|O_{\omega^-}|\xi_0^-}^*}
\diff\Gamma
\]

Given a path $\omega=(\mbf{s},\mbf{p},\mbf{y}) \in \Omega_k$,
the random amplitude $\braket{\xi|O_\omega|\eta}$ is given by
\[
\braket{\xi| O_\omega|\eta} =
\braket{\xi|p_0}\braket{p_k|\eta}
e^{i\varphi(\omega,\xi,\eta)}
\prod_{j=1}^k \Ft{V_{y_j}}(p_j-p_{j-1}).
\]
Since $V$ is real and therefore $\Ft{V}(q)^* = \Ft{V}(-q)$, the
term in the expectation can be written
\begin{equation}
\Expec \mstack{\braket{\xi_1^+|O_{\omega^+}|\xi_0^+}}{\braket{\xi_1^-|O_{\omega^-}|\xi_0^-}^*}
=
\mstack
{\braket{\xi_1^+|p_{+,k_+}}\braket{p_{+,0}|\xi_0^+}}
{\braket{\xi_1^-|p_{-,k_-}}\braket{p_{-,0}|\xi_0^-}}
e^{i(\varphi(\omega^+)-\varphi(\omega^-))}
\Expec \prod_{a\in K} \Ft{V_{y_a}}(q_a),
\end{equation}
where
\[
K_{k_+,k_-} = \{(\ell,j) \mid \ell\in\{+,-\} \text{ and } j\in [1,k_\ell]\}
\]
and
\[
q_{(\ell,j)} =
\begin{cases}
p_{\ell,j}-p_{\ell,j-1}, &\ell = +\\
p_{\ell,j-1}-p_{\ell,j}, &\ell = -.
\end{cases}
\]
We write $X=X(\Gamma)$ for the collision set $X = \{(y_a,q_a)\}_{a\in K_{(k_1,k_2)}}$.

To split up the expectation, let
$P(X)=P(\mbf{y})\in\mcal{P}(K)$  be the finest partition such
that $|y_a-y_{a'}|\leq 2\alpha r$ implies that $a$ and $a'$ belong to the
same set.  Then
\[
\Expec \prod_{a\in K} \Ft{V_{y_a}}(q_a),
= \prod_{S\in P(X)} \Expec \prod_{a\in S}\Ft{V_{y_a}}(q_a).
\]
For admissible potentials, Lemma~\ref{lem:admissible-V} implies that
\begin{equation}
    \label{V-moment-est}
\Big|\Expec \prod_{a\in K} \Ft{V_{y_a}}(q_a)\Big|
\leq
r^{-|K|d}
\prod_{a\in K} (1+|q_a|)^{-10d}
\sum_{P'\leq P(\mbf{y})}
\prod_{S\in P'}
(C_V |S|)^{2|S|}
r^{d} b_{|S|r^{-1}}(\sum_{j\in S} q_j),
\end{equation}
with $b_t(x) := \exp(-c|t^{-1}x|^{0.99})$.

This quantity is only nonnegligible when $|\sum_{j\in S}q_j|\lessapprox r^{-1}$
for each $S\in P(\mbf{y})$.
This leads us to define the notion of a $\beta$-complete collision set.
\begin{definition}[Complete collision sets]
    A collision set $X=\{(y_j,q_j)\}_{j=1}^k$ is $\beta$-complete if
\begin{equation}
\Big|\sum_{j\in S} q_j\Big|\leq \beta|S| r^{-1}
\end{equation}
holds for every $S\in Q(\mbf{y})$.
\end{definition}

An early approximation we can make is to reduce the integration over
paths $\omega^+$ and $\omega^-$ to only paths which are $\beta$-complete
for $\beta=|\log\eps|^{20}$.
\begin{lemma}
\label{lem:complete-collisions}
Let $s\leq \tau$ and let $A$ be a band-limited operator with bandwidth at most $\eps^{-0.5}$.
Then for $\beta>|\log\eps|^{20}$,
\begin{equation}
\begin{split}
\big\|
\int \ket{\xi_0^-}\bra{\xi_0^+}
\braket{\xi_1^-|A|\xi_1^+}
\Xi(\Gamma)
\Expec \mstack{\braket{\xi_1^+|O_{\omega^+}|\xi_0^+}}{\braket{\xi_1^-|O_{\omega^-}|\xi_0^-}^*}
(1 - \One_{X(\Gamma)\text{ is }\beta\text{-complete}})
\diff\Gamma\big\|_{op}
\leq \eps^{100}\|A\|_{op}.
\end{split}
\end{equation}
\end{lemma}
\begin{proof}
The difference is an integral over paths which form $\beta$-incomplete
collision sets, and the norm of the integrand is at most
$\eps^{|\log\eps|^5}$ for such paths.  The volume of integration for fixed
$\xi_0^+$ or for fixed $\xi_0^-$ is only $\eps^{-C}$ , so the result follows
upon applying the Schur test.
\end{proof}

\subsection{The structure of partitions from generic paths}
The idea is that the main contribution to the channel $\Evol_s$ should come from paths
$\bm{\Gamma}$ that are \textit{generic}.  We define generic paths as those that do not have
\textit{incidences}.

In general incidences are any geometric feature of a path that can change its correlation
structure.  The simplest type of incidence is a recollision.
\begin{definition}[Recollisions]
\label{def:segment-recollision}
A \emph{recollision} in a path $\omega\in\Omega_{k}$ is a pair
$a,a'\in [1,k]$ such that $|y_a-y_{a'}| \leq 2r$.

A special kind of recollision is an \emph{immediate recollision}, which
satisfies $a'=a+1$ and $|p_{a+1}-p_{a-1}|\leq 10\alpha r^{-1}$.
Let $I^{imm}(\omega)\subset [1,k]$
be the set of all indices belonging to a immediate recollision,
and let $I^{rec}(\omega)$ be the set of indices belonging to a recollision that
is not an immediate recollision.
\end{definition}

At this point we stop to observe that immediate recollisions do not substantially alter
the trajectory of a path.
\begin{figure}
\centering
\includegraphics{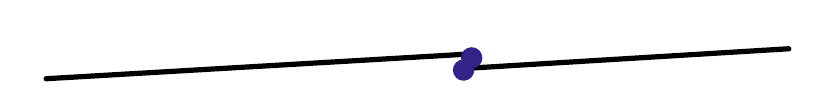}
\caption*{An example of an immediate recollision in a path.  Note that such a recollision can only
cause a small displacement in position on the order $O(r)$ and a small perturbation in momentum on
the order $O(r^{-1})$.}
\end{figure}

The first fact we prove is that the time betweeen recollisions cannot be too large.
\begin{lemma}
\label{lem:short-rec-time}
Let $\omega\in\Omega_{k,\alpha}(\tau,\sigma;\xi,\eta)$ be a path with
$|\xi_p|\geq r^{-1}\sigma^{-1}$.  If $|y_j-y_{j+1}|\leq 2r$, then $s_j\leq 10|p_0|^{-1}\alpha r$.
\end{lemma}
\begin{proof}
First, the condition $|\xi_0|\geq r^{-1}\sigma^{-1}$ and the constraint
$|p_0-(\xi_0)_p| \leq \alpha r^{-1}$ imply $|p_0|\leq 2\alpha r^{-1}\sigma^{-1}$.
Then, since $s_0\geq \sigma$
\[
||p_j|^2/2-|p_0|^2/2| \leq \alpha \max\{s_j^{-1}, |p_0| r^{-1}\}.
\]
Assuming that $s_j\geq 10|p_0|^{-1}\alpha r$, it follows that
\[
||p_j|-|p_0|| \leq \alpha 2\alpha r^{-1}.
\]
But the condition $|y_j-y_{j+1}|\leq 2r$ implies
\[
|s_jp_j|\leq 4\alpha r,
\]
so that $s_j \leq 4\alpha|p_j|^{-1}r \leq 10 |p_0|^{-1}r$.
\end{proof}

To state the second fact, we introduce the notion of the collision time $t_a$, simply
defined by
\[
t_{\pm,j} := \sum_{0\leq j'<j} s_{\pm,j'}.
\]
\begin{lemma}
\label{lem:path-simplification}
Let $\bm{\Gamma}$ be a $|\log\eps|^{20}$-complete path, and suppose that $\{a,a+1\}\in P(\bm{\Gamma})$
for every immediate recollision $a$.  If $a<a'\in K(\bm{\Gamma})\setminus I^{imm}(\bm{\Gamma})$
are two consecutive collisions when ignoring immediate recollisions, then
\begin{equation}
\label{eq:simplified-path-straight}
|y_{a'} - (t_{a'}-t_a) p_a| \leq 4m^2 \alpha^{20} r.
\end{equation}
\end{lemma}
\begin{proof}
To prove this, first observe that $a'=a+2m+1$ for some number $m$ of immediate recollisions between $a$ and
$a'$, and
\[
y_{a+2m+1} - y_a = \sum_{j=0}^m y_{a+2j+1} - y_{a+2j} + \sum_{j=1}^m y_{a+2j} - y_{a+2j-1}.
\]
The latter terms are each bounded by $2r$ because $(a+2j-1,a+2j)$ are all immediate recollisions.
The former terms are well approximated by $s_{a+2j}p_{a+2j}$, so we have
\[
|y_{a+2m+1}-y_a - \sum_{j=0}^m s_{a+2j}p_{a+2j}| \leq 2m\alpha r.
\]
Next we observe that, since $(a+2j-1,a+2j)$ forms a pair in $P(\bm{\Gamma})$ and $\bm{\Gamma})$
is $|\log\eps|^{20}$-complete, $|q_{a+2j-1}+q_{a+2j}|\leq 2|\log\eps|^{20}r^{-1}$.  Expanding the
definition of $q_{a+2j-1}$ and $q_{a+2j}$ it follows that
\[
|q_{a+2j-1} + q_{a+2j}|
= |p_{a+2j-1} - p_{a+2j-2} + p_{a+2j} - p_{a+2j-1}| =
|p_{a+2j} - p_{a+2(j-1)}| \leq 2|\log\eps|^{20}r^{-1}
\]
for every $1\leq j\leq m$.  In particular, $|p_{a+2j}-p_a| \leq 2|\log\eps|^{20} m r^{-1}$ for each $a$,
and therefore
\[
|y_{a+2m+1}-y_a - (\sum_{j=0}^m s_{a+2j})p_{a}| \leq 2m^2\alpha^{20} r.
\]
Finally, we observe that
\[
t_{a+2m+1} - t_a = \sum_{j=0}^m s_{a+2j} + \sum_{j=1}^m s_{a+2j-1}.
\]
The latter sum is bounded by $10m|p_0|^{-1}\alpha r$ by Lemma~\ref{lem:short-rec-time}.
\end{proof}

Recollisions form just one type of incidence.  It is possible that paths $\omega^+$ and
$\omega^-$ have a nontrivial collision structure even if neither path has a recollision.
Consider for example the paths
\[
    \omega^+ = (\mbf{s}^+,\mbf{p}^+,\mbf{y}^+) = (
        ((4,2,1,2,4), (v, -v, v, -v, v), (4v, 2v, 3v, v))
    \in \Omega_0(13,1)
\]
and
\[
    \omega^- := (\mbf{s}^-,\mbf{p}^-,\mbf{y}^-) = ((4,3,2,1,3), (v, -v, v, -v, v), (4v, v, 3v, 2v))
    \in \Omega_0(13,1)
\]
where $v\in\Sphere^{d-1}$ is any unit vector.  This example is depicted in Figure~\ref{fig:one-dim-example}.  Then
the collision partition associated to $\omega^+$ and $\omega^-$ is given by
\begin{align*}
    P = \{ \{(+,1), (-,1)\},
        \{(+,2), (-,4)\},
        \{(+,3), (-,3)\},
        \{(+,4), (-,2)\} \}.
\end{align*}
\begin{figure}
\centering
\includegraphics{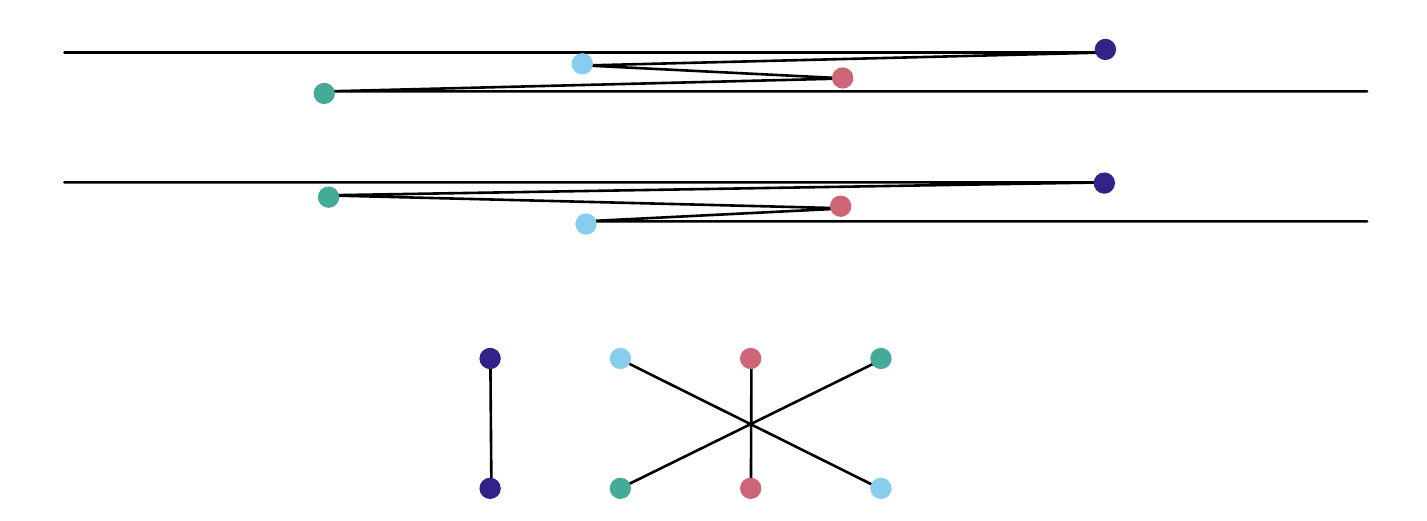}
\caption{Another example of a pair of paths with a nontrivial collision partition and such that neither $\omega$
nor $\omega'$ has a recollision event. The paths are depicted with a slight downward drift to clarify the order of the
collisions.  Note that this behavior is typical in one dimension.}
\label{fig:one-dim-example}
\end{figure}

This partition has a nontrivial structure because the second collision of
$\omega^-$ correlates with the fourth collision of $\omega^+$ and vice-versa.
This behavior is not uncommon for paths that are constrained to one dimension.
The problem with the above example is that there are non-consecutive
collisions in $\Gamma$ which can be visited by a single path with two
collisions.

We need to introduce another type of incidence to prevent this behavior
which we call a \emph{tube incidence}.

\begin{definition}[Tube incidences]
A \emph{tube incidence} for a path $\omega \in \Omega_{\alpha,k}(s,S)$
is a pair $(a,a') \in [0,k]$ of collisions such that there exists a collision
$b\in I^{imm}(\omega)$ with $a<b<a'$ and there exists a time $s\in\Real$ such that
\begin{equation}
\label{eq:tube-inc-def}
|y_a + sp_a - y_{a'}| \leq 100k^4\alpha^{20}r.
\end{equation}
Above we use the convention $y_0 = y_1 - s_0p_0$.
We set $I^{tube}(\omega)$ to be the set of pairs $(a,a')$ that form a tube incidence.
\end{definition}
A tube incidence occurs when a particle scattering out of site $a$ can choose to
``skip'' its next collision (possibly at $b$) and instead scatter at site $a'$.

The key idea is that the partition $P(\mbf{y})$ of a doubled path $\bm{\Gamma}$ is
severely constrained when neither $\omega^+$ nor $\omega^-$ have an incidence.
In particular, the partition must be a generalized ladder.  To define a generalized
ladder we first define a ladder partition.
\begin{definition}[Ladder partitions]
\label{def:ladders}
Let $A$ and $B$ be two finite ordered sets with $|A|=|B|$.  The \emph{ladder matching}
$P_{lad}\in\mcal{P}(A\sqcup B)$ is the unique matching of the form
$P_{lad}=\{\{a,\varphi(a)\}\}_{a\in A}$ where $\varphi:A\to B$ is the unique order-preserving
bijection between $A$ and $B$.
\end{definition}

We can now state the main result.
\begin{lemma}
\label{lem:is-ladder}
Let $\bm{\Gamma} = (\xi_0^+,\omega^+,\xi_1^+;\xi_0^-,\omega^-,\xi_1^-)$ be a doubled
path such that $X(\bm{\Gamma})$ is $|\log\eps|^{20}$-complete
and $d_r(\xi_0^+,\xi_0^-)\leq |\log\eps|^{20}$.  Then at least one of the following
holds:
\begin{itemize}
\item One of $\omega^+$ or $\omega^-$ has an incidence.  That is,
\[
I^{rec}(\omega^+)\cup I^{tube}(\omega^+)\cup I^{rec}(\omega^-)\cup I^{tube}(\omega^-)\not=\noset.
\]
\item The partition $P(\bm{\Gamma})$ has a cell with more than two elements.
\item The partition $P(\bm{\Gamma})$ is a generalized ladder in the sense that
(1) $P$ saturates the set $I^{imm}(\omega^+)\cup I^{imm}(\omega^-)$ and
(2) The restriction $P|_{K(\bm{\Gamma})\setminus I^{imm}(\omega^+)\setminus I^{imm}(\omega^-)}$
is a ladder partition on the set
$K(\Gamma^+)\setminus I^{imm}(\omega^+)\sqcup K(\Gamma^-)\setminus I^{imm}(\omega^-)$.
\end{itemize}
\end{lemma}
\begin{proof}
We will assume that we are not in either of the first two cases, so that $P(\bm{\Gamma})$ is a perfect
matching of $K(\bm{\Gamma})$ and neither $\omega^+$ nor $\omega^-$ has an incidence.

First we observe that every immediate recollision $(a,a+1)$ must also be a cell in $P(\bm{\Gamma})$.
Since $P(\bm{\Gamma})$ is assumed to be a perfect matching, it follows that $P(\bm{\Gamma})$ saturates
the set of immediate recollisions.

It remains to show that $P|_{K(\bm{\Gamma})\setminus I^{imm}(\bm{\Gamma})}$ forms a ladder partition.
Order the collisions in $K(\Gamma^+)\setminus I^{imm}(\Gamma^+)$ as $(a_1,a_2,\cdots, a_m)$, and the
collisions of $K(\Gamma^-)\setminus I^{imm}$ as $(b_1,b_2,\cdots, b_{m'})$ (with $b_{j+1} > b_j$ and $a_{j+1}>a_j$).We first observe that there are no pairs $\{a_j,a_{j'}\}\in P(\bm{\Gamma})$ because then $(j,j')$ would
form a recollision.  Likewise $\{b_j,b_{j'}\}\not\in P(\bm{\Gamma})$ for any $j,j'$.  This shows that
$m=m'$.  To show that $P(\bm{\Gamma})$ is a generalized ladder, it suffices to check that
$\{a_j,b_j\}\in P(\bm{\Gamma})$ for every $j\in[m]$.  We prove this by induction on $j$.

The base case is that $j=1$.  Choose $j_1$ such that $\{a_1,b_{j_1}\}\in P(\bm{\Gamma})$.
then $|y_{a_1} - y_{b_{j_1}}|\leq 2r$.  Using~\eqref{eq:simplified-path-straight}, we have
\begin{align*}
|y_{a_1} - (y_{+,0} + t_{a_1}p_{+,0})| &\leq 4k^2 \alpha^{20}r \\
|y_{b_1} - (y_{-,0} + t_{b_1}p_{-,0})| &\leq 4k^2 \alpha^{20}r.
\end{align*}
It follows that for $s= t_{a_1}-t_{b_1}$,
\[
|y_{b_{j_1}} - (y_{+,0} + sp_{-,0}| \leq 10 k^2 \alpha^{20}r,
\]
so that either $j_1=1$ or else $j_1>1$ and therefore $(0,b_{j_1})$ is a tube incidence for
$\Gamma^-$.

The proof of the inductive step follows the same argument, with an additional calculation
to show that $|p_{a_j}-p_{b_j}|\leq m|\log\eps|^{20}r^{-1}$ if $\{a_{j'},b_{j'}\}\in P(\bm{\Gamma})$
for all $j'\leq j$.
\end{proof}

\subsection{Ladders and the semigroup property}
\label{sec:ladder-statement}
Now we state the main result of the section, which is that the expectation appearing in
$\Evol_s[A]$ is well approximated by simply summing over the ladder matchings.

\begin{definition}[Generalized ladders]
\label{def:generalized-ladders}
A generalized ladder partition on the set $[k_+]\sqcup[k_-]\{(\pm, j) \mid j\leq k_{\pm}\}$
is a partition $P$ of $[k_+]\sqcup[k_-]$ such that there exists a set
$I_+^{imm}\subset [k_+]$ and $I_-^{imm}\subset [k_-]$ such that
$\{(+,j),(+,j+1)\}\in P$ for every $j\in I_+^{\Imm}$ and
$\{(-,j),(-,j+1)\}\in P$ for every $j\in I_-^{\Imm}$, and such that
$P|_{[k_+]\sqcup[k_-]\setminus \{+\}\times I_+^{\Imm}\setminus \{-\}\times I_-^{\Imm}}$
is a ladder partition on the set
$([k_+]\setminus I_+^{imm})\sqcup ([k_-]\setminus I_-^{imm})$.
\end{definition}
We set $\mcal{Q}_{gl}\subset \mcal{P}([k_+]\sqcup[k_-])$ to be the set of generalized ladder partitions.

To relate partitions to the superoperator $\Evol_s$ we first define the notion of a $P$-expectation.
Given a partition $P\in\mcal{P}([k_+]\sqcup[k_-])$,
$\omega^+\in \Omega_{k_+}$ and $\omega^-\in\Omega_{k_-}$,
we define
\begin{equation}
\label{eq:p-expec}
\Expec_P \mstack{\braket{\xi_1^+|O_{\omega^+}|\xi_0^+}}{\braket{\xi_1^-|O_{\omega^-}|\xi_0^-}^*}
=
\mstack
{\braket{\xi_1^+|p_{+,k_+}}\braket{p_{+,0}|\xi_0^+}}
{\braket{\xi_1^-|p_{-,k_-}}\braket{p_{-,0}|\xi_0^-}}
e^{i(\varphi(\omega^+)-\varphi(\omega^-))}
\eps^{k_++k_-}
\prod_{S\in P} \Expec \prod_{a\in S} \Ft{V_{y_a}}(q_a).
\end{equation}
We then define the ladder expectation to be the sum over all ladder partitions,
\[
\Expec_{lad}
\mstack{\braket{\xi_1^+|O_{\omega^+}|\xi_0^+}}{\braket{\xi_1^-|O_{\omega^-}|\xi_0^-}^*}
:=
\sum_{P\in\mcal{Q}_{rl}([k_+]\sqcup[k_-])}
\Expec_P \mstack{\braket{\xi_1^+|O_{\omega^+}|\xi_0^+}}{\braket{\xi_1^-|O_{\omega^-}|\xi_0^-}^*}
\]

We are now ready to define the ladder superoperator $\mcal{L}_s$.
\begin{equation}
\label{eq:ladder-op-def}
\mcal{L}_s[A] :=
\int \ket{\xi_0^-}\bra{\xi_0^+}
\braket{\xi_1^-|A|\xi_1^+}
\Xi(\Gamma)
\Expec_{lad} \mstack{\braket{\xi_1^+|O_{\omega^+}|\xi_0^+}}{\braket{\xi_1^-|O_{\omega^-}|\xi_0^-}^*}
\diff\Gamma.
\end{equation}
For convenience we define

The main result of this section is that the ladder superoperator $\mcal{L}_s$ is a good approximation
to the evolution $\Evol_s$.
\begin{proposition}
\label{prp:short-ladder-compare}
Let $A$ be an operator with good support and let $s\leq \eps^{-2-\kappa/2}$.  Then
\[
\|\Evol_s[A] - \mcal{L}_s[A]\|_{op}
\leq C\eps^{2.1}s \|A\|_{op}.
\]
\end{proposition}

Before we are ready to prove Proposition~\ref{prp:short-ladder-compare} we must first establish that
the sum over \emph{all} generalized ladders is the same as the sum over ther \emph{correct}
generalized ladder in the case that $\bm{\Gamma}$ is a path with no incidences.
\begin{lemma}
Let $\bm{\Gamma}$ be a $|\log\eps|^{20}$-complete path and suppose that
$d_r(\xi_0^+,\xi_0^-)\leq |\log\eps^{-20}$, and $|\xi_p|\geq r^{-1}\sigma^{-1}$.  Suppose moreover
that $\bm{\Gamma}$ has no incidences and that $P(\bm{\Gamma})$ is a matching.  Then
\[
\Expec \mstack{\braket{\xi_1^+|O_{\omega^+}|\xi_0^+}}{\braket{\xi_1^-|O_{\omega^-}|\xi_0^-}^*}
=
\Expec_{lad} \mstack{\braket{\xi_1^+|O_{\omega^+}|\xi_0^+}}{\braket{\xi_1^-|O_{\omega^-}|\xi_0^-}^*}.
\]
\end{lemma}
\begin{proof}
By Lemma~\ref{lem:is-ladder}, it follows that for \emph{some}
generalized ladder $P\in\mcal{Q}_{gl}([k_+]\sqcup[k_-])$,
\[
\Expec
\mstack{\braket{\xi_1^+|O_{\omega^+}|\xi_0^+}}{\braket{\xi_1^-|O_{\omega^-}|\xi_0^-}^*}
=
\Expec_P \mstack{\braket{\xi_1^+|O_{\omega^+}|\xi_0^+}}{\braket{\xi_1^-|O_{\omega^-}|\xi_0^-}^*}.
\]
We will show that if $P'\in\mcal{Q}_{gl}([k_+]\sqcup[k_-])$ is another generalized ladder
with $P'\not= P$, then
\[
\Expec_{P'} \mstack{\braket{\xi_1^+|O_{\omega^+}|\xi_0^+}}{\braket{\xi_1^-|O_{\omega^-}|\xi_0^-}^*}
= \prod_{S\in P} \Expec \prod_{a\in S} \Ft{V_{y_a}}(q_a)
= 0.
\]
Indeed, since $P'\not=P$ there exists some $a,b,b'\in[k_+]\sqcup[k_-]$ such that
$\{a,b\}\in P=P(\bm{\Gamma})$ and $\{a,b'\}\in P'$ with $b\not=b$.  But since
$\{a,b'\}\not\in P(\bm{\Gamma})$, $|y_a-y_{b'}| > 2r$ so the expectation
\[
\Expec \Ft{V_{y_a}}(q_a) \Ft{V_{y_b}}(q_{b'}) = 0
\]
vanishes.
\end{proof}

\subsection{The proof of Proposition~\ref{prp:short-ladder-compare}}
The error $\Evol_s[A]-\mcal{L}_s[A]$ can be written as a path integral
\[
\Evol_s[A] - \mcal{L}_s[A]
=
\int \ket{\xi_0^-}\bra{\xi_0^+} \braket{\xi_1^-|A|\xi_1^+}
\Xi(\Gamma)
\big(
\Expec
\mstack{\braket{\xi_1^+|O_{\omega^+}|\xi_0^+}}{\braket{\xi_1^-|O_{\omega^-}|\xi_0^-}^*}
-
\Expec_{lad}
\mstack{\braket{\xi_1^+|O_{\omega^+}|\xi_0^+}}{\braket{\xi_1^-|O_{\omega^-}|\xi_0^-}^*}\big)
\diff\bm{\Gamma}.
\]
The argument of Lemma~\ref{lem:complete-collisions} still works to show that we can restrict
the integral to paths that are $|\log\eps|^{20}$-complete.  Moreover, using the support
condition on $A$ we can also restrict to the case that $d_r(\xi_1^+,\xi_1^-)\leq |\log\eps|^{20}$
and that $|(\xi_1^+)_p|\geq (r\sigma)^{-1}$.

Under these constraints, the only paths that contribute to the path integral above are those for which
either $P(\bm{\Gamma})$ has a cell of more than two elements and paths which have some kind of incidence.
Let $\One^{bad}(\bm{\Gamma})$ be the indicator function for such paths.   We decompose this function
according to the exact partition $P(\bm{\Gamma})$ and the exact incidence set $I^{rec}(\bm{\Gamma})$,
$I^{tube}(\bm{\Gamma})$,
\begin{align*}
\One^{bad}(\bm{\Gamma})
&= \sum_{P\in\mcal{P}([k_+]\sqcup [k_-])} \sum_{I^{rec},I^{tube}}
\One(P(\bm{\Gamma})=P)
\One(I^{rec}(\bm{\Gamma})=I^{rec})
\One(I^{tube}(\bm{\Gamma})=I^{tube}) \\
&=: \sum_{P\in\mcal{P}([k_+]\sqcup [k_-])} \sum_{I^{rec},I^{tube}}
\One_{P,I^{rec},I^{tube}}(\bm{\Gamma}).
\end{align*}
where the sum includes the constraint that either $I^{rec}\cup I^{rec}\not=\noset$ or else
$P$ has a cell with at least three elements.
Then we have the estimate
\begin{align*}
\|\Evol_s[A] - \mcal{L}_s[A]\|_{op}
\leq
\sum_{P,I^{rec},I^{tube}}
\Big\|
\int \ket{\xi_0^-}\bra{\xi_0^+} &\braket{\xi_1^-|A|\xi_1^+}
\Xi(\Gamma)
\One_{P,I^{rec},I^{tube}}(\bm{\Gamma}) \\
&\big(
\Expec
\mstack{\braket{\xi_1^+|O_{\omega^+}|\xi_0^+}}{\braket{\xi_1^-|O_{\omega^-}|\xi_0^-}^*}
-
\Expec_{lad}
\mstack{\braket{\xi_1^+|O_{\omega^+}|\xi_0^+}}{\braket{\xi_1^-|O_{\omega^-}|\xi_0^-}^*}\big)
\diff\bm{\Gamma}\Big\|_{op}.
\end{align*}
Using the triangle
inequality we bound
\[
\big|\Expec
\mstack{\braket{\xi_1^+|O_{\omega^+}|\xi_0^+}}{\braket{\xi_1^-|O_{\omega^-}|\xi_0^-}^*}
-
\Expec_{lad}
\mstack{\braket{\xi_1^+|O_{\omega^+}|\xi_0^+}}{\braket{\xi_1^-|O_{\omega^-}|\xi_0^-}^*}\big|
\leq
\sum_{Q\leq P(\bm{\Gamma})}
\big|\Expec_Q
\mstack{\braket{\xi_1^+|O_{\omega^+}|\xi_0^+}}{\braket{\xi_1^-|O_{\omega^-}|\xi_0^-}^*}\big|,
\]
and now applying the Schur test we estimate
\begin{equation}
\label{eq:schur-err-bd}
\begin{split}
\|\Evol_s[A] - \mcal{L}_s[A]\|_{op}
\leq
\sum_{k_+,k_-}
&\sum_{P,I^{rec},I^{tube}} \\
&\sup_{\xi_0^-}
\Big|
\int \braket{\xi_1^-|A|\xi_1^+} \Xi(\Gamma)
\One_{P,I^{rec},I^{tube}}(\bm{\Gamma})
\sum_{Q\leq P(\bm{\Gamma})}
|\Expec_Q \mstack{\braket{\xi_1^+|O_{\omega^+}|\xi_0^+}}{\braket{\xi_1^-|O_{\omega^-}|\xi_0^-}^*}|
\diff\omega^\pm \diff\xi_1^\pm\diff\xi_0^- \Big|^{1/2} \\
&\times \sup_{\xi_0^+}
\Big|
\int \braket{\xi_1^-|A|\xi_1^+} \Xi(\Gamma)
\One_{P,I^{rec},I^{tube}}(\bm{\Gamma})
\sum_{Q\leq P(\bm{\Gamma})}
|\Expec_Q \mstack{\braket{\xi_1^+|O_{\omega^+}|\xi_0^+}}{\braket{\xi_1^-|O_{\omega^-}|\xi_0^-}^*}|
\diff\omega^\pm \diff\xi_1^\pm\diff\xi_0^+ \Big|^{1/2}.
\end{split}
\end{equation}

The first step to bound the term
$\Expec_Q
\mstack{\braket{\xi_1^+|O_{\omega^+}|\xi_0^+}}{\braket{\xi_1^-|O_{\omega^-}|\xi_0^-}^*}$
appearing in the integrand.
Then, expanding out the formula~\eqref{eq:p-expec} we have
\[
\big|
\Expec_Q
\mstack{\braket{\xi_1^+|O_{\omega^+}|\xi_0^+}}{\braket{\xi_1^-|O_{\omega^-}|\xi_0^-}^*}\big|
\leq
\mstack
{|\braket{\xi_1^+|p_{+,k_+}}\braket{p_{+,0}|\xi_0^+}|}
{|\braket{\xi_1^-|p_{-,k_-}}\braket{p_{-,0}|\xi_0^-}|}
e^{i(\varphi(\omega^+)-\varphi(\omega^-))}
\prod_{S\in P(\bm{\Gamma})\vee Q} \Expec \prod_{a\in S} \Ft{V_{y_a}}(q_a).
\]
Now we use
\[
|\Braket{p|\xi}| \leq C r^{d/2} \exp(-c (r|(\xi)_p-p|)^{0.5})
\]
as well as Lemma~\ref{lem:admissible-V} to estimate
\begin{align*}
\big|
\Expec_Q
\mstack{\braket{\xi_1^+|O_{\omega^+}|\xi_0^+}}{\braket{\xi_1^-|O_{\omega^-}|\xi_0^-}^*}\big|
\leq C r^{2d}\eps^{k_++k_-}
&\mstack{ \exp(-c(r|(\xi_1^+)_p - p_{+,k_+}|)^{0.5})
\exp(-c(r|(\xi_0^+)_p - p_{+,0}|)^{0.5})}
{\exp(-c(r|(\xi_1^-)_p - p_{-,k_-}|)^{0.5})
\exp(-c(r|(\xi_0^-)_p - p_{-,0}|)^{0.5})} \\
&r^{-d(k_++k_-)} (C(k_++k_-))^{k_++k_-}
\prod_{Q'\leq P(\bm{\Gamma})\vee Q} \prod_{S'\in Q'}
r^d \exp\Big(-c \big|r|S|^{-1} \sum_{a\in S} q_a\big|^{0.99}\Big) \\
&\qquad\qquad\qquad\qquad \times \prod_{a\in [k_+]\sqcup[k_-]} (1+|q_a|)^{-20d}.
\end{align*}
We collect the important terms above in the function $E_Q(\bm{\Gamma})$,
\begin{equation}
\label{eq:EQ-def}
E_Q(\bm{\Gamma}) :=
r^{2d} r^{-d(k_++k_-)}\eps^{k_++k_-}
\prod_{S\in Q} r^d\exp(-c|\frac{r}{|S|}\sum_{a\in S} q_a|^{0.99})
\times \prod_a (1+|q_a|)^{-20d}.
\end{equation}
Because $k_+,k_-\leq k_{max}$, the combinatorial factors contribute at most an absolute constant.
Therefore Proposition~\ref{prp:short-ladder-compare} reduces to the following integral bound.
\begin{lemma}
\label{lem:bad-path-bd}
Let $A$ be an admissible operator, let $k_+,k_-\in[k_{max}]$,
let $Q\in\mcal{P}([k_+]\sqcup[k_-])$, and let $(P,I^{rec},I^{imm})$ be a triple such that if
$I^{rec}\cup I^{tube}=\noset$, then $P$ has a cell of more than two elements.  Then
\begin{equation}
\label{eq:bad-path-int}
\sup_{\xi_0^-} \Big|
\int \braket{\xi_1^-|A|\xi_1^+} \Xi(\Gamma) E_{Q\vee P}(\bm{\Gamma})
\One_{P,I^{rec},I^{tube}}(\bm{\Gamma})
\diff\omega^\pm \diff\xi_1^\pm\diff\xi_0^- \Big|
\leq C \eps^{2.25}\tau.
\end{equation}
\end{lemma}

To estimate~\eqref{eq:bad-path-int} we use the
following simple lemma, which is just an iterated application of Fubini's theorem and the triangle
inequality.
\begin{lemma}
\label{lem:abstract-holder}
Let $\mcal{X}=X_1\times \cdots \times X_N$ be the product of
$N$ measure spaces $X_j$, and let $\mcal{X}_j = X_1\times\cdots \times X_j$ be
the product of the first $j$ factors.
 Then for any positive functions $F_j:\mcal{X}_j\to\Real^+$,
\begin{equation}
\int_{\mcal{X}} \prod_{j=1}^N F_j(x_1,\cdots,x_j) \diff X
\leq
\prod_{j=1}^N \sup_{x'_1,\dots,x'_{j-1}}
\int_{X_j} F_j(x'_1,\dots,x'_{j-1}, x_j) \diff x_j.
\end{equation}
\end{lemma}

To apply Lemma~\ref{lem:abstract-holder} we need to order the variables in $\bm{\Gamma}$ and
bound the integrand as a product of functions constraining each variable in $\bm{\Gamma}$
as a function only of variables that come earlier in the ordering.
The reader may find it useful to refer to Figure~\ref{fig:ladder-components} for a quick overview
of the constraints on the variables.

\begin{figure}
\includegraphics{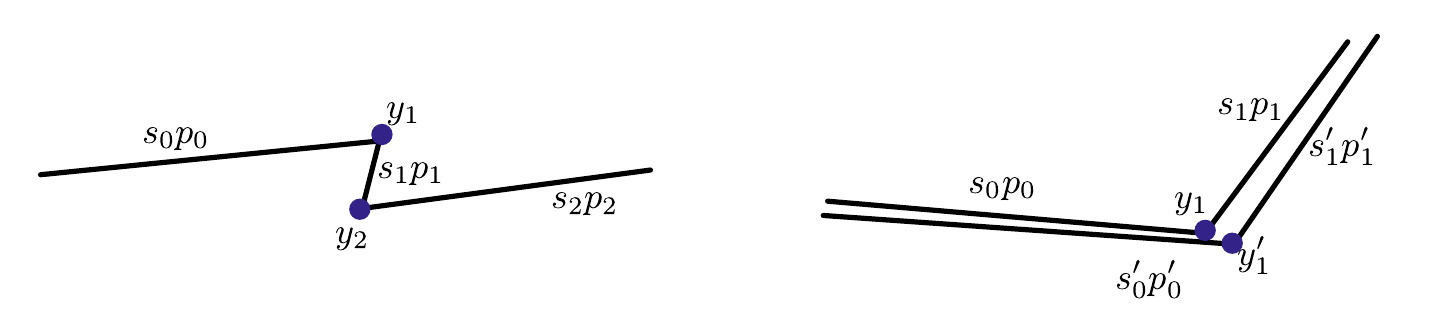}
\caption{Two typical collision pairs that can appear.  On the left, a recollision.  The contribution from
recollisions is heuristically counted as follows:  The $s_0$ variable has only the constraint $0\leq s_0\leq \tau$.
The $s_1$ variable is then constrained to $s_1\lsim r|p_1|^{-1}$ so that $|y_1-y_2|\leq r$.  Then $p_1$ is chosen
from an annulus of width $r^{-1}$ and radius $|p_0|$ and $p_2$ is effectively constrained by a delta function
in the momentum variables.  The total contribution is at most $\eps^2 \tau\lsim 1$.
On the right, a typical ``rung'' collision forming a ladder.  The time variable $s_0$ satisfies $0\leq s_0\leq \tau$,
but then $s_0'$ is constrained by $|s_0-s'_0|\lsim r|p_0|^{-1}$ so that $|y_1-y_1'|\leq r$.  The momentum variable
$p_1$ is again chosen from an annulus of thickness $r^{-1}$ and width $|p_0|$, and $p_1'$ is constrained by a delta
function to match $p_1$.  Again the contribution is bounded by $\eps^2 \tau$.}
\label{fig:ladder-components}
\end{figure}

We use the following ordering of the variables:
\begin{equation}
\label{eq:Gamma-ordering}
\bm{\Gamma} =
(\xi_0^+, p_{+,0},s_{+,0},y_{+,1},\cdots,y_{+,k_+},p_{+,k_+},s_{+,k_+},
\xi_1^+, \xi_1^-, p_{-,k_-}, s_{-,k_-},y_{-,k_-}, \cdots,y_{-,1} p_{-,0},s_{-,0},\xi_0^-).
\end{equation}
Given a variable label $\lambda\in \{\qt{y}_a,\qt{p}_a,\qt{s}_a\}_{a\in K(\bm{\Gamma})}
\cup \{\qt{\xi}_\ell^\pm\}_{\ell\in\{0,1\}}$, we define the partial
paths $\bm{\Gamma}_{<\lambda}$ to be the sequence of variables preceding $\lambda$.  Thus for example
$\bm{\Gamma}_{<\qt{y}_{+,1}} = (\xi_0^+,p_{+,0},s_{+,0})$.
We also define a total ordering $\leq'$ on $[k_+]\sqcup[k_-]$ implied by the ordering
of the variables~\eqref{eq:Gamma-ordering}, in which $(+,j)\leq'(-,j')$ for any $j,j'$ and
$(\pm,j) \leq (\pm,j')$ when $\pm j\leq \pm j'$ (that is, the ordering is reversed for the negative
indices, as indicated in~\eqref{eq:Gamma-ordering}).

The next step is to bound the integrand
as a product of constraints assigned to each variable as a function of the prior variables.
$\One_{P,I^{rec},I^{tube}}(\bm{\Gamma})$.
We will write out the ``standard term''
 in the integrand that does not use the indicator function as a product of constraints of
the $p$, $y$, and $\xi$ variables.
\[
|\braket{\xi_1^-|A|\xi_1^+}|
|\Xi(\bm{\Gamma})|
E_{Q\vee P} (\bm{\Gamma})
\leq \eps^{k_++k_-}
F_p(\bm{\Gamma}) F_y(\bm{\Gamma}) F_\xi(\bm{\Gamma}).
\]
The constraints on momentum come from several sources. First there is a term $\prod (1+|q_a|)^{-20d}$
ensuring that no impulse is too large.  Second, there is a term enforcing conservation of kinetic energy.
Third, there are terms ensuring that $p_{+,0}$ and $p_{-,k_-}$ match with
$(\xi_0^+)_p$ and $(\xi_1^-)_p$, respectively.  Finally, there are the constraints (approximate
delta functions) coming from the expectation.  We also take the factor of $r^{2d}$
from~\eqref{eq:EQ-def} and distribute one factor of $r^d$ to each $p_{+,0}$ and $p_{-,k_-}$:
\begin{align*}
F_p(\bm{\Gamma}) :=
&\prod (1+|p_a-p_{a-1}|)^{-20d}
\prod \One(||p_{\pm,j}|-|p_{\pm,0}|| \leq \alpha \max\{s_{\pm,j}^{-1}, r^{-1}\})\\
&\times
(r^d \One(|p_{+,0} - (\xi_0^+)_p|\leq \alpha r^{-1})) (r^d \One(|p_{-,{k_-}} - (\xi_1^-)_p|\leq \alpha r^{-1}))
\prod_{S\in {Q\vee P}} r^d \exp(-ck_{max}^{-1}|r\sum_{a\in S}q_a|^{0.99}) \\
\end{align*}
The constraints in the position variables are determined by the compatibility conditions
$|y_{a+1}-(y_a+s_ap_a)|\leq \alpha r$ and the compatibility of the first and last collisions of
$\omega^+$ and $\omega^-$ against the boundaries $\xi_0^+$ and $\xi_1^-$.  We take the factor
of $r^{-d(k_++k_-)}$ from~\eqref{eq:EQ-def} and distribute one $r^{-d}$ to each $y_a$ variable:
\begin{align*}
F_y(\bm{\Gamma}) :=
&\prod (r^{-d}\One(|y_{a+1} - (y_a+s_{a}p_{a})|\leq \alpha r)) \\
&\qquad \times (r^{-d}\One(|y_{+,1} - ((\xi^+_0)_x + s_0p_0)| \leq \alpha r))
(r^{-d}\One(|y_{-,k_-} - ((\xi^-_1)_x - s_{-,k_-}p_{-,k_-})|)).
\end{align*}
The constraint on the $\xi$ variables comes from the compatibility with the path, along with
the support condition on $A$
\begin{align*}
F_\xi(\bm{\Gamma}) :=
\One(d_r(\xi_1^+, (y_{+,k_+}+s_{+,k_+}p_{+,k_+},p_{+,k_+})) \leq \alpha)
\One(d_r(\xi_0^-, (y_{-,1}-s_{-,0}p_{-,0},p_{-,0})) \leq \alpha)
(1 + d_r(\xi_1^+,\xi_1^-))^{-20d}.
\end{align*}
There are also indirect constraints on the $s$ variables coming indirectly from the
combination of the compatibility conditions
$|y_{a+1}-(y_a+s_ap_a)|\leq \alpha r$ and constraints of the form
$|y_a - y_b|\leq 2kr$ for collisions $a\sim_P b$ of the same cell of the partition $P(\bm{\Gamma})$.
We also note that by Lemma~\ref{lem:path-simplification}, we have
$|y_{a'}-(y_a + (t_{a'}-t_a)p_a)| \leq C \alpha^{20} r$ when $a$ and $a'$ are collisions separated only
by immediate recollisions.

We have no worry of ``double-dipping'' on the basic compatibility constraints
such as $\One(|y_{a+1}-(y_a+s_ap_a)|\leq \alpha r)$ because for example
\[
\supp ( F_p(\bm{\Gamma}) F_y(\bm{\Gamma}) F_\xi(\bm{\Gamma}))
\subset \{|y_{a+1}-(y_a+s_ap_a)|\leq \alpha r\},
\]
so we can freely apply extra such indicator functions where they are useful.

\subsubsection{The ``standard constraint'' bounds}
In this section we use the partitions $P$ and $Q$ to ``assign constraints'' to each variable.

In particular, we write define functions $f_\lambda(\bm{\Gamma}_{\leq \lambda})$ such that
\begin{align*}
F_p(\bm{\Gamma}) &\leq
\prod f_{\qt{p},a}(\bm{\Gamma}_{\leq\qt{p},a})\\
F_y(\bm{\Gamma}) &\leq
\prod
f_{\qt{y},a}(\bm{\Gamma}_{\leq\qt{y},a}) \\
F_\xi(\bm{\Gamma}) &\leq
f_{\qt{\xi},1,+}(\bm{\Gamma}_{\leq\qt{\xi},1,+})
f_{\qt{\xi},1,-}(\bm{\Gamma}_{\leq\qt{\xi},1,-})
f_{\qt{\xi},0,-}(\bm{\Gamma}_{\leq\qt{\xi},0,-}).
\end{align*}
These ``standard constraint'' functions can simply be read off of the definitions of $F_p$, $F_y$,
and $F_\xi$.
\begin{equation}
f_{\qt{p},a}(\bm{\Gamma}_{\leq\qt{p},a})
:=
\begin{cases}
r^d \exp(-ck_{max}^{-1} | r\sum_{a\in S}q_a|^{0.99}),
&a = \max_{\leq'} S \text{ for some } S\in P\vee Q\\
r^d \One(|p_{-,k_-} - (\xi_1^-)_p|\leq \alpha r^{-1}),
&a = (-,k_-) \\
r^d \One(|p_{+,0} - (\xi_0^+)_p|\leq \alpha r^{-1}),
&a = (+,0) \\
(1+|p_a-p_{a-1}|)^{-20d} \One(||p_a| - |p_{+,0}||\leq \alpha \max\{|p_0|^{-1}s_a^{-1},r^{-1}\}),
&\text{ else}.
\end{cases}
\end{equation}
The standard $y$ constraints are given by
\begin{equation}
f_{\qt{y},a}(\bm{\Gamma}_{\leq\qt{y},a}) :=
\begin{cases}
r^{-d}\One(|y_{+,1} - ((\xi^+_0)_x + s_0p_0)| \leq \alpha r),
&a = (+,1) \\
r^{-d}\One(|y_{-,k_-} - ((\xi^-_1)_x - s_{-,k_-}p_{-,k_-})|),
&a = (-,k_-) \\
r^{-d}\One(|y_{+,j} - (y_{+,j-1}+s_{+,j-1}p_{+,j-1})|\leq \alpha r), &a = (+,j)
\text{ for } j>1\\
r^{-d}\One(|y_{-,j+1} - (y_{-,j}+s_{-,j}p_{-,j})|\leq \alpha r), &a = (-,j)
\text{ for } j<k_-.
\end{cases}
\end{equation}
Finally, the standard $\xi$ constraints are
\begin{align*}
f_{\qt{\xi},1,+}(\bm{\Gamma}_{\leq\qt{\xi},1,+})
&:= \One(d_r(\xi_1^+, (y_{+,k_+}+s_{+,k_+}p_{+,k_+},p_{+,k_+})) \leq \alpha) \\
f_{\qt{\xi},1,-}(\bm{\Gamma}_{\leq\qt{\xi},1,-})
&:= (1 + d_r(\xi_1^+,\xi_1^-))^{-20d} \\
f_{\qt{\xi},0,-}(\bm{\Gamma}_{\leq\qt{\xi},0,-})
&:= \One(d_r(\xi_0^-, (y_{-,1}-s_{-,0}p_{-,0},p_{-,0})) \leq \alpha).
\end{align*}

The contributions from the $\qt{\xi}$ and $\qt{y}$ variables are easy to account for
\begin{lemma}[Standard position and phase space bounds]
For any $\lambda\in\{(\qt{y},a), (\qt{\xi},\ell,\pm)\}$
\[
\sup_{\bm{\Gamma}_{<\lambda}}
\int f_{\lambda}(\bm{\Gamma}_{\leq\lambda})\diff \Gamma_\lambda
\leq C.
\]
\end{lemma}

The momentum constraints are slightly more complicated.
\begin{lemma}[Standard momentum bounds]
If $a=\max_{\leq'}S$ for some $S\in P\vee Q$ or $a\in\{(-,k_-),(+,0)\}$, then
\[
\sup_{\bm{\Gamma}_{<\qt{p},a}}
\int f_{\qt{p},a}(\bm{\Gamma}_{\leq\qt{p},a})\diff p_a
\leq C.
\]
Otherwise, if $a\not\sim_P a+1$ is not the first collision of an immediate recollision, then
\begin{equation}
\label{eq:std-momentum-bd}
\sup_{\bm{\Gamma}_{<\qt{p},a}}
\int f_{\qt{p},a}(\bm{\Gamma}_{\leq\qt{p},a})\diff p_a
\leq C r^{-1} \min\{|p_0|^{d-1}, 1\}.
\end{equation}
\end{lemma}
\begin{proof}
Only the second bound needs proof.  Since $a+1\not\sim_P a$, it follows from
Lemma~\ref{lem:short-rec-time} that $s_a \gtrsim \alpha |p_0|^{-1} r$.  Therefore
$|p_0|^{-1}s_a^{-1}\lsim r^{-1}$, so $p_a$ is constrained to an annulus of thickness $r^{-1}$
and radius $|p_0|$.  This annulus has volume $r^{-1}|p_0|^{d-1}$.  If $|p_0|\gtrsim 1$ the additional
factor $(1+|p_a-p_{a-1}|)^{-20d}$ ensures that $p_a$ is essentially also confined to a ball of unit radius.
\end{proof}

The immediate recollisions require some more detailed attention.
If $a\sim a+1$ is an immediate recollision, then we group the variables $(s_a,p_a)$ and use the following
estimate.
\begin{lemma}
For any $|p_0|>r^{-1}\sigma^{-1} \geq \eps^{0.5}$ and $q\in\Real^d$,
\begin{equation}
\label{lem:sp-integral}
\begin{split}
\int_0^{10\alpha |p_0|^{-1}r}\int_{\Real^d}
&\One(||p|^2/2-|p_0|^2/2| \leq \alpha s^{-1})
(1 + |p-q|)^{-20d} \diff q\diff s \\
&\leq C \min\{|p_0|^{-1}, |p_0|^{d-2}\} (1 + \log(\eps^{-1})).
\end{split}
\end{equation}
\end{lemma}
\begin{proof}
We split the integral over $s$ into dyadic intervals $[2^k,2^{k+1}]$ for $k\in\bbZ$.  On this interval,
the variable $p$ is retricted to an annulus of radius $|p_0|$ and width $|p_0|^{-1}2^{-k}$.  Moreover,

Now consider the case $|p_0|\gtrsim 1$.  In this case the factor $(1+|p-q|)^{-20d}$ additionally
restricts the integration over $p$ to a unit ball, and the integration over $p_0$ produces a factor
on the order $\min\{1,|p_0|^{-1} 2^{-k}\}$.  Integrating over $s$ to produce a factor $2^k$
and summing over $k$ such that $2^k\leq 20\alpha |p_0|^{-1}r$, we obtain the bound
\begin{align*}
\int_0^{10\alpha |p_0|^{-1}r}\int_{\Real^d}
&\One(||p|^2/2-|p_0|^2/2| \leq \alpha s^{-1})
(1 + |p-q|)^{-20d} \diff q\diff s \\
&\leq \sum_{\substack{k\in\bbZ \\ 2^k < |p_0|^{-1}}} 2^k
+ \sum_{\substack{|p_0|^{-1} < 2^k < 20\alpha |p_0|^{-1}r}} |p_0|^{-1} \\
&\leq |p_0|^{-1}(1 + \log(20\alpha r)).
\end{align*}

The second case we consider is that $|p_0|\lesssim 1$.  In this case, the annulus of radius $|p_0|$
and width $|p_0|^{-1}2^{-k}$ has volume on the order $|p_0|^{d-2} 2^{-k}$.  The bound then follows
from integrating in $s$ and summming over $k$, as above.
\end{proof}

Conspicuously missing from the discussion above is the integration over the time variables.  For many
of the time variables we simply use the constraint $s_a\leq \tau$ to pick up a factor of $\tau$.
Additional constraints come from the partition $P$.  Suppose that $a\leq' b$ and $a\sim_P b$.  Then there
is a constraint $|y_b - y_a|\leq 2 r$, which coupled with the constraint
$|y_b - (y_{b-1}+s_{b-1}p_{b-1})|\leq \alpha r$ imposes a constraint on $s_{b-1}$ in terms of the variables
$y_a$,$y_{b-1}$, and $p_{b-1}$, which are all in $\Gamma_{<\qt{s},b}$.  This constraint picks up a factor
of $|p_0|^{-1}r$ instead of $\tau$:
\begin{equation}
\label{eq:s-cluster-bd}
\sup_{\Gamma_{<\qt{s},b-1}}
\int \One(|y_a - (y_{b-1}+s_{b-1}p_{b-1})|\leq \alpha r) \diff s_{b-1}
\leq |p_0|^{-1}r.
\end{equation}

\subsubsection{The case that $P$ has a cluster}
With just the bounds we have already proven, it is possible to obtain a good estimate in the case
that $P$ has a set of more than $2$ elements.  This is the simplest case, as suggested by
Figure~\ref{fig:cluster}.
\begin{figure}
\includegraphics{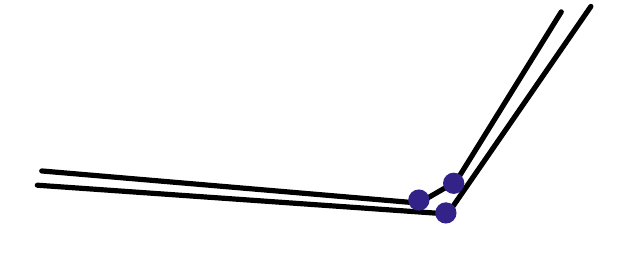}
\caption{The case that the collision partition has a cluster.  There is only one full time degree of freedom
contributing a factor of $\tau$, but three factors of $\eps$ from the potential.  Such clusters therefore contribute
$\eps^3\tau$ rather than $\eps^2\tau$.}
\label{fig:cluster}
\end{figure}

\begin{lemma}[The cluster bound]
\label{lem:cluster-bd}
Let $A$ be an admissible operator, let $k_+,k_-\in[k_{max}]$,
let $Q\in\mcal{P}([k_+]\sqcup[k_-])$, and let $P\in\mcal{P}([k_+]\sqcup[k_-])$ be a partition having a
cell of more than two elements.  Then
\begin{equation}
\label{eq:cluster-path-int}
\sup_{\xi_0^-} \Big|
\int \braket{\xi_1^-|A|\xi_1^+} \Xi(\Gamma) E_{Q\vee P}(\bm{\Gamma})
\One_{P}(\bm{\Gamma})
\diff\omega^\pm \diff\xi_1^\pm\diff\xi_0^- \Big|
\leq C^{k_++k_-} (\eps^2\tau)^{(k_++k_-)/2}\tau^{-1}.
\end{equation}
\end{lemma}
\begin{proof}
We bound the standard part of the integrand by the product $F_pF_yF_s$ as described in the previous
section and apply Lemma~\ref{lem:abstract-holder}.  The $\qt{y}$ and $\qt{\xi}$ variables contribute
a factor of $C^{k_++k_-}$.  The product of the contributions from the $(s_a,p_a)$ pairs coming from
immediate recollisions produces a factor of $(C\log\eps^{-1})^{n_r}$, where $n_r$ is the number of immediate
recollision clusters of $P$.
To account for the $\qt{s}$ and $\qt{p}$ variables, let $P'\subset P$
The $\qt{p}$ variables contribute a total of
$(Cr^{-1}\min\{|p_0|^{d-1},1\})^{|P|-n_r}$ by taking the product of the integral over all $p_a$ with
$a=\max_{\leq'}S$ for each $S\in P$.
Then for each $s_a$ variable that is the first in its cluster (of which there are $|P|$),
we get a trivial factor of $\tau$.
Each of the rest of the $s_a$ variables contribute $|p_0|^{-1}r$ by~\eqref{eq:s-cluster-bd}.
The product of all of these factors gives
\begin{equation}
\label{eq:normal-partition-bd}
\begin{split}
\sup_{\xi_0^-} \Big|
&\int \braket{\xi_1^-|A|\xi_1^+} \Xi(\Gamma) E_{Q\vee P}(\bm{\Gamma})
\One_{P}(\bm{\Gamma})
\diff\omega^\pm \diff\xi_1^\pm\diff\xi_0^- \Big| \\
&\qquad\leq C^{k_++k_-}
\eps^{k_++k_-}\tau^{|P|}
(C \min\{|p_0|^{d-2},|p_0|^{-1}\})^{|P|-n_r}
\end{split}
\end{equation}
The first and last factors are bounded by $C^{k_++k_-}$.  Then the fact that $P$ is not a perfect
matching implies $|P| < (k_++k_-)/2$ and since $|P|$ is an integer in particular it follows that
$|P|\leq (k_++k_-)/2 - 1$.  This proves~\eqref{eq:cluster-path-int}.
\end{proof}

\subsubsection{The recollision case}
To complete the proof of Lemma~\ref{lem:bad-path-bd} we need to find a way to use the additional
constraints coming from a recollision or tube incidence.  A simplified version of the argument is
presented in Figure~\ref{fig:recollision-explanation}

\begin{figure}
\includegraphics{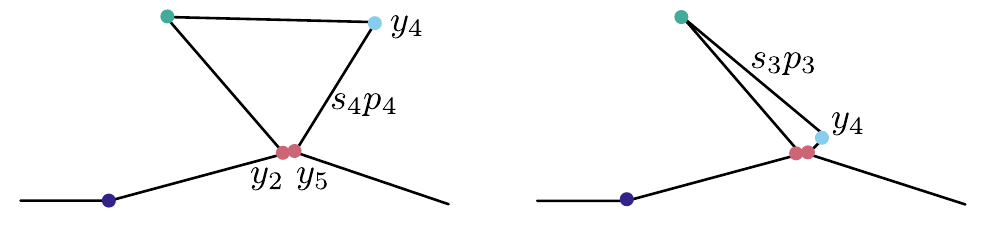}
\caption{A path with a recollision at indices $(2,5)$.  On the left, the case that $y_4$ is far from $y_5$.  In this
case the momentum variable $p_4$ is constrained to be approximately parallel to $y_2-y_4$.  On the right, the
case that $y_4$ is close to $y_5$.  In this case there is an additional constraint on the time variable $s_3$.}
\label{fig:recollision-explanation}
\end{figure}

Suppose that $(a,b)$ is a recollision occuring in $\omega^+$,
$\sign(a)=\sign(b)=+$.
The idea is that a recollision at $(a,b)$ typically enforces
a strong constraint on the momentum \emph{before} the collision at $b$.  Indeed, if $|y_a-y_b|\leq 2r$
then $|s_{b-1}p_{b-1} - (y_a-y_{b-1})| \leq 2r$, so in particular
\[
|p_{b-1}-z| \leq 2r s_{b-1}^{-1},
\]
with $z = (y_a-y_{b-1}) s_{b-1}^{-1}$.  If $|y_a-y_{b-1}| > 10\alpha r$, then
 $s_{b-1} \geq |y_a-y_{b-1}||p_0|^{-1}$.
This is where the constraint on $p_{b-1}$ comes from.
On the other hand, if $|y_a-y_{b-1}|$ is small, then there is a constraint on $s_{b-2}$ of
exactly the same kind as~\eqref{eq:s-cluster-bd}.  The only additional subtlety to deal with is the
possibility that $b-1$ is itself a recollision or immediate recollision, in which case
$p_{b-1}$ is localized by the momentum constraint on a collision cluster and further localization is
not helpful.  This is only a minor difficulty, and so we now prove
\begin{lemma}[The recollision bound]
Let $A$ be an admissible operator, let $k_+,k_-\in[k_{max}]$,
let $P,Q\in\mcal{P}([k_+]\sqcup[k_-])$, and let $I^{rec}\not=\noset$ be a nonempty set of
recollisions in $[k_+]$.  Then
\begin{equation}
\label{eq:rec-path-int}
\sup_{\xi_0^-} \Big|
\int \braket{\xi_1^-|A|\xi_1^+} \Xi(\Gamma) E_{Q}(\bm{\Gamma})
\One_P(\bm{\Gamma})
\One_{I^{rec}}(\bm{\Gamma})
\diff\omega^\pm \diff\xi_1^\pm\diff\xi_0^- \Big|
\leq C^{k_++k_-} (\eps^2\tau)^{(k_++k_-)/2} (\eps^{-1.5}\tau^{-1}).
\end{equation}
\end{lemma}
\begin{proof}
Let $(a,a')\in I^{rec}$ be the recollision with minimal $a'$.  Then let $b<a'$ be the first
collision before $a'$ that is not an immediate recollision.  Because $(a,a')$ is
not an immediate recollision, it is clear that $b\not=a'$.  Moreover, because $(a,a')$ is minimal,
the index $b$ is not of the form $\max_{\leq'} S$ for any $S\in Q$.

We bound the indicator function for a recollision as a sum of indicator functions
depending on the distance $|y_{b-1}-y_b|$
\[
\One_{I^{rec}}(\bm{\Gamma})
\leq \One(|y_a-y_b|\leq 2r)\One(|y_{b-1}-y_b|\geq Kr)
+ \One(|y_{b-1}-y_b| \leq Kr),
\]
and use this to split~\eqref{eq:rec-path-int} as a sum of two integrals each corresponding to
a different term.

For the first term we follow the proof of Lemma~\ref{lem:cluster-bd}, with the modification that we set
\[
f'_{\qt{p},b-1} (\bm{\Gamma}_{\leq\qt{p}_{b-1}})
=
\One(|p_{b-1} - |p_0|\frac{y_{b-1}-y_a}{|y_{b-1}-y_a|}| \leq |p_0|K^{-1}).
\]
In this case $p_{b-1}$ is sampled from the intersection of an annulus with thickness $r^{-1}$ and
radius $|p_0|$ and a ball or radius $|p_0|K^{-1}$,
so that
\[
\sup_{\bm{\Gamma}_{<\qt{p}_{b-1}}}
\int f'_{\qt{p},b-1}(\bm{\Gamma}_{\leq\qt{p}_{b-1}})\diff p_{b-1}
\leq C r^{-1} \min\{(K^{-1}|p_0|)^{d-1}, 1\}.
\]
Applying this estimate in the place of the bound~\eqref{eq:std-momentum-bd} along with the rest
of the argument that leads to~\eqref{eq:normal-partition-bd} yields
\begin{equation}
\label{eq:far-rec-bd}
\begin{split}
\sup_{\xi_0^-} \Big|
&\int \braket{\xi_1^-|A|\xi_1^+} \Xi(\Gamma) E_{Q\vee P}(\bm{\Gamma})
\One_{P}(\bm{\Gamma})
\One(|y_a-y_b|\leq 2r)\One(|y_{b-1}-y_b|\geq Kr)
\diff\omega^\pm \diff\xi_1^\pm\diff\xi_0^- \Big| \\
&\qquad\leq C^{k_++k_-}
\eps^{k_++k_-}\tau^{|P|}
(C \min\{|p_0|^{d-2},|p_0|^{-1}\})^{|P|-n_r-1}
(C \min\{K^{1-d}|p_0|)^{d-2},|p_0|^{-1}\})
\end{split}
\end{equation}
The last factor is maximized when $|K^{1-d}||p_0|^{d-1} = |p_0|^{-1}$, which occurs
when $|p_0|= K$.  In this case we obtain a savings of $K^{-1}$ over the
bound~\eqref{eq:normal-partition-bd}, and therefore conclude
\begin{equation}
\label{eq:far-rec-bd}
\begin{split}
\sup_{\xi_0^-} \Big|
&\int \braket{\xi_1^-|A|\xi_1^+} \Xi(\Gamma) E_{Q\vee P}(\bm{\Gamma})
\One_{P}(\bm{\Gamma})
\One(|y_a-y_b|\leq 2r)\One(|y_{b-1}-y_b|\geq Kr)
\diff\omega^\pm \diff\xi_1^\pm\diff\xi_0^- \Big| \\
&\qquad\leq C^{k_++k_-} \eps^{k_++k_-}\tau^{|P|} K^{-1}.
\end{split}
\end{equation}

The second term to deal with is the integral involving $\One(|y_{b-1}-y_b|\leq Kr)$.  In this case,
we apply the bound~\eqref{eq:s-cluster-bd} to get a factor of $Kr$ instead of $\tau$ for the integration
over the variable $s_{b-2}$.  Thus
\begin{equation}
\label{eq:close-rec-bd}
\begin{split}
\sup_{\xi_0^-} \Big|
&\int \braket{\xi_1^-|A|\xi_1^+} \Xi(\Gamma) E_{Q\vee P}(\bm{\Gamma})
\One_{P}(\bm{\Gamma})
\One(|y_{b-1}-y_b|\leq Kr)
\diff\omega^\pm \diff\xi_1^\pm\diff\xi_0^- \Big| \\
&\qquad\leq C^{k_++k_-} \eps^{k_++k_-}\tau^{|P|} (Kr \tau^{-1}).
\end{split}
\end{equation}
Choosing $K = \eps^{-0.5}$ yields the desired result.
\end{proof}

\subsubsection{The tube incidence case}
The final remaining case is that $\bm{\Gamma}$ does not have a recollision but does have a tube incidence.
Suppose that the first tube incidence occurs at $(a,b)$.  Then combining the tube incidence
constraint~\eqref{eq:tube-inc-def} with the compatibility condition
$|y_b-y_{b-1}-s_{b-1}p_{b-1}|\leq 2r$, we conclude that there exists $s\in\Real$ such that
\[
|s_{b-1}p_{b-1} + (y_{b-1} - y_a) + sp_a| \leq C\alpha^{20}r.
\]
In other words, the vector $s_{b-1}p_{b-1}$ lies on the tube with thickness $\alpha^{20}r$, axis
$p_a$, and passing through $y_{b-1}-y_a$.  If $p_a$ is transverse to $p_{b-1}$ then this imposes
a strong constraint on the time variable $s_{b-1}$.  On the other hand if $p_a$ is parallel with
$p_{b-1}$ then this imposes a constraint on the momentum variable $p_{b-1}$.
Both cases are depicted in Figure~\ref{fig:tube-explanation}

\begin{figure}
\includegraphics{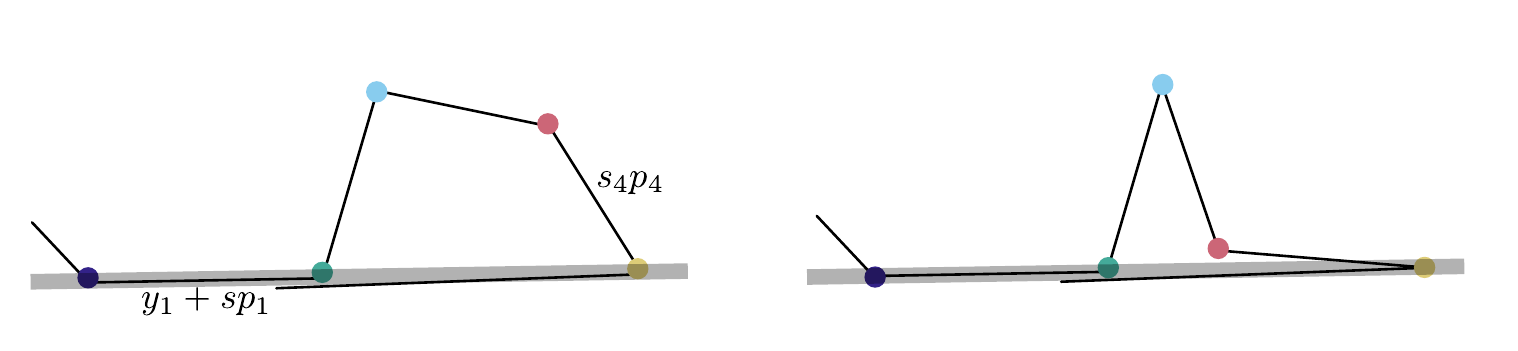}
\caption{Paths with a tube incidence at $(1,5)$, so that $|y_5-(y_1+sp_1)|\lessapprox r$ for some $s\in\Real$.
On the left, an example in which $p_4$ is transverse to $p_1$.  In this case the time variable
$s_4$ is constrained so that $y_5$ can lie on the gray tube.
On the right, an example in which $p_4$ is approximately parallel to $p_1$.  In this case $s_4$ is much less
constrained, but $p_4$ is much more constrained.  In either case there is a gain of a factor of (at least)
$\eps^{1/2}$.}
\label{fig:tube-explanation}
\end{figure}

\begin{lemma}[The tube incidence bound]
Let $A$ be an admissible operator, let $k_+,k_-\in[k_{max}]$,
let $P,Q\in\mcal{P}([k_+]\sqcup[k_-])$, and let $I^{tube}\not=\noset$ be a nonempty set of
tube incidences in $[k_+]$, and suppose moreover that $I^{rec} = \noset$ so that there are
no recollisions.
\begin{equation}
\label{eq:rec-path-int}
\sup_{\xi_0^-} \Big|
\int \braket{\xi_1^-|A|\xi_1^+} \Xi(\Gamma) E_{Q\vee P}(\bm{\Gamma})
\One_P(\bm{\Gamma})
\One_{I^{rec}}(\bm{\Gamma})
\diff\omega^\pm \diff\xi_1^\pm\diff\xi_0^- \Big|
\leq C^{k_++k_-} (\eps^2\tau)^{(k_++k_-)/2} (\eps^{-1.5}\tau^{-1}).
\end{equation}
\end{lemma}
\begin{proof}
Let $(a,a') \in I^{tube}$ be a tube incidence in $\omega^+$.
Let $b<a'$ be the last collision before $a'$ that does not belong to an immediate recollision,
so that $|y_{a'} - (y_b + (t_{a'}-t_b)p_b)| \leq C\alpha^{10} r$
We decompose the indicator function for the tube incidence according to the angle between $p_b$ and
$p_a$,
\begin{align*}
\One_{I^{tube}}(\bm{\Gamma})
\leq
&\One(|(y_a + s p_a) - (y_b + (t_{a'}-t_b) p_{b'})|\leq 100k^4\alpha^{20}r \text{ for some } s\in\Real)
\One(\theta(p_a,p_b) \geq \delta) \\
&+ \One(\theta(p_a,p_b) \leq \delta),
\end{align*}
where $\theta(u,v) := \min\{|u/|u| - v/|v||, |u/|u|+v/|v||\}$, and $\delta$ is a small parameter that we
will optimize shortly.

This splits the integral as a sum of two terms.  In the first term, we use the fact that the
set of times for which the bound
\[
|(y_a + s p_a) - (y_b + (t_{a'}-t_b) p_{b'})|\leq 100k^4\alpha^{20}r
\]
holds is an interval with width at most $C\alpha^{20}r (|p_0| \theta(p_a,p_b))^{-1}$,
and so we obtain this factor instead of $\tau$ in the integration over the $s_{a'-1}$ variable.

In the second term, we gain a factor of $\delta^{d-1}$ as the integration over the annulus
contributes $r^{-1}\delta^{d-1} |p_0|^{d-1}$ rather than $r^{-1}|p_0|^{d-1}$.
Therefore
\begin{equation}
\label{eq:transverse-bd}
\begin{split}
\sup_{\xi_0^-} \Big|
&\int \braket{\xi_1^-|A|\xi_1^+} \Xi(\Gamma) E_{Q\vee P}(\bm{\Gamma})
\One_{P}(\bm{\Gamma})
\One(|y_a-y_b|\leq 2r)\One(|y_{b-1}-y_b|\geq Kr)
\diff\omega^\pm \diff\xi_1^\pm\diff\xi_0^- \Big| \\
&\qquad\leq C^{k_++k_-}
\eps^{k_++k_-}\tau^{|P|}
\min\{
[C\alpha^{20}r\tau^{-1} (|p_0| \delta)^{-1} + \delta^{d-1})].
\end{split}
\end{equation}
Choosing $\delta = \eps^{1/d}$ produces an extra factor of $\eps^{d-1/d} \leq \eps^{1/2}$, as desired.
\end{proof}

\section{Iterating the path integral}
\label{sec:extended-paths}
In this section we set up the path integral we use to reach diffusive
times $N\tau = \eps^{-2-9\kappa/10}$.  In this section we focus on the main
term
\[
\wtild{\Evol}_{N\tau}[A] = \Expec (U_{\tau,\sigma}^N)^* A U_{\tau,\sigma}^N.
\]
Each of the operators $U_\tau$ is expanded as an integral over path segments.
Together these path segments form an elongated path that we call an \textit{extended path}.
An extended path is a tuple of path segments interspered with phase space points $\xi_\ell$,
as follows
\[
\Gamma = (\xi_0,\omega_1,\xi_1,\xi_2,\omega_2,\xi_3,
\cdots, \xi_{2\ell-2},\omega_\ell,\xi_{2\ell-1},\cdots,
\xi_{2N-2},\omega_N,\xi_{2N-1}).
\]
We sometimes also reorder $\Gamma$ as a tuple of phase space points $\bm{\xi}$ and
path segments $\bm{\omega}$, so
$\Gamma = (\bm{\xi},\bm{\omega})\in (\PhaseSpace)^{2N} \times \Omega(\tau)^N$.  We write
$\Omega_{ext}^N$ to denote this space of extended paths.

The extended paths that will actually contribute to our integral satisfy additional compaibility
conditions that ensure that the phase space endpoints of the segments match up nicely.
We define the set $\Omega_{\alpha,S}^N\subset\Omega_{ext}^N$ to be
\[
\Omega_{\alpha,S}^N :=
\{\Gamma = (\bm{\xi},\bm{\omega})\in (\PhaseSpace)^{2N}\times \Omega^{N} \mid
\omega_\ell \in \Omega_{\alpha}(\tau,S;\xi_{2\ell-2},\xi_{2\ell-1}) \text{ and }
d_r(\xi_{2\ell-1},\xi_{2\ell})<N^2\}.
\]
Note that the superscript $N$ in $\Omega_{\alpha,S}^N$ is just notation -- the
set is not itself a Cartesian product of $N$ copies of $\Omega_{\alpha, S}$.

We will be integrating over paths, which means we need to specify a measure.  Give $\Omega_k$ the measure
$d_k\omega$ which is the Lebesgue measure on $\triangle_{k+1}(\tau) \times (\Real^d)^{k+1} \times (\Real^d)^k$.
Then we assign $\Omega$ the measure formed on the union, so that given a continuous function
$f\in C^0(\Omega)$,
\[
\int_{\Omega} f(\omega)\diff\omega :=
\sum_{k=0}^{k_{max}} \int_{\Omega_k} f(\omega) \diff_k\omega.
\]
The measure $\diff\omega$ allows us to define the measure $\diff\Gamma$, informally written
\[
\diff\Gamma = \prod_{\ell=0}^{2N-1}\diff\xi_\ell \prod_{\ell=1}^N\diff\omega_\ell.
\]
Note that each path $\Gamma\in\Omega_{\alpha,S}^N$ specifies a collision number sequence
$\mbf{k}(\Gamma) = (k_1,\cdots,k_N)$, where $k_\ell$ is the number of collisions in the segment
$\omega_\ell$.  Thus integration over $\Omega_{\alpha,S}^N$ implicitly includes a sum over $\mbf{k}$.

We define a smooth cutoff
$\Xi$ on the space of extended paths that has support in $\Omega_{\alpha,S}$, defined by
\[
\Xi(\Gamma) = \prod_{\ell=1}^N \chi_{\alpha,S}(\omega_\ell,\xi_{2\ell-2},\xi_{2\ell-1}).
\]
The collisions in a path $\Gamma\in\Omega_{ext}^N$ are indexed by pairs $(\ell,j)$ indicating the segment
index $\ell$ and the collision index $j$ within the segment.
So that we can index all of the collisions at once,
we use the collision index set
\[
K(\Gamma) = \{(\ell,j) \mid \ell\in[1,N] \text{ and }j\in[1,k_j]\}.
\]
Sometimes we also need $K_0(\Gamma)$ which is
\[
K_0(\Gamma) = \{(\ell,j) \mid \ell\in[1,N] \text{ and }j\in[0,k_j]\}.
\]
We will use letters like $a,b,c$ to refer to the indices in $K(\Gamma)$.
Also, we impose a
ordering on $K$ using the lexicographic order, and write $a+_K1$ to mean the successor of the index
$a$ in the collision set $K(\Gamma)$.
The index set $K_{all}$ refers to the set of all possible collision indices
\[
K_{all} = [1,N]\times[1,k_{max}] = \{(\ell,j) \mid \ell\in[1,N], j\in[1,k_{max}]\}.
\]
We define the function $\ell:K_{all}\to[1,N]$ which sends $a=(\ell,j)$ to $\ell(a) = \ell$,
which we also call the \textit{segment index} of the collision index $a$.

The path $\Gamma$ defines an operator $O_\Gamma$
\begin{equation}
\label{eq:OGamma-def}
O_\Gamma =
\ket{\xi_0}\bra{\xi_{2N-1}}
\prod_{\ell=1}^N \braket{\xi_{2\ell-2}|O_{\omega_\ell}|\xi_{2\ell-1}}
\prod_{\ell=1}^{N-1} \braket{\xi_{2\ell-1}|\xi_{2\ell}}.
\end{equation}

To write down an expression of $\Evol_{N\tau}[A]$
we will integrate over pairs of paths $(\Gamma^+,\Gamma^-)$, and will therefore
need to index into the path variables of both paths simultaneously.  We therefore define
the signed index sets $K^+(\Gamma)$ and $K^-(\Gamma)$, with
\[
K^\pm(\Gamma) = \{(\ell,j,\pm) \mid (\ell,j)\in K(\Gamma)\}.
\]
This allows us to define the doubled index set
\[
\mbf{K}(\Gamma^+,\Gamma^-)
= K^+(\Gamma^+)\sqcup K^-(\Gamma^-)
\]
Analagously, we write $\mbf{K}_{all} = K_{all}^+\sqcup K_{all}^-$.
The set $\mbf{K}_{all}$ has a partial ordering induced by the ordering on $K_{all}$, with
$(a,+)$ being incomparable to $(a',-)$ for any $a,a'\in K_{all}$.
Moreover, we write $\sign(a)$ to be the
``sign'' of an index in $\mbf{K}$ (so $\sign(a)=+1$ if $a\in K^+_{all}$ and $\sign(a)=-1$
if $a\in K^-_{all}$).  Finally, given an index $a = (\ell,\pm,j)$, we write
$a^*$ for the index of opposite sign, so $a^* = (\ell,\mp,j)$.

These signed collision indices will be used to index into doubled paths,
$\bm{\Gamma} = (\Gamma^+,\Gamma^-)$.  Now there is an abuse of notation we will sometimes use, which is
that we will sometimes go back and forth between signed and unsigned indices when the path $\Gamma$ is
fixed (either being $\Gamma^+$ or $\Gamma^-$).

We can now write down an integral expression for $\wtild{\Evol}_{N\tau}[A]$,
\[
\wtild{\Evol}_{N\tau}[A] =
\int_{\Omega_{\alpha,S}^N\times\Omega_{\alpha,S}}
\Xi(\Gamma^+) \Xi(\Gamma^-)
\Expec O_{\Gamma^-}^* A O_{\Gamma^+}\diff\bm{\Gamma},
\]
with an error on the order $O(\eps^{0.1}\|A\|_{op})$ in operator norm.  We further restrict to pairs
of paths $(\Gamma^+,\Gamma^-)$ that produce a nonnegligible contribution to
$\Expec O_{\Gamma^-}^*AO_{\Gamma^+}$.  The naive bound on the operator norm coming from applying
the triangle inequality and estimating volumes is on the order $O(\eps^{-CN})$, and therefore we
will have to consider any pairs of paths $(\Gamma^+,\Gamma^-)$ satisfying
\[
\|\Expec O_{\Gamma^-}^*AO_{\Gamma^+}\|_{op} \geq \eps^{CN}.
\]
Expanding the operator $O_{\Gamma}$ using the definition~\eqref{eq:OGamma-def}
and~\eqref{eq:path-op-def}, and simply bounding all deterministic quantities with an appropriate
factor of $\eps^{-C}$, we obtain a simple bound
\begin{equation}
\begin{split}
\|\Expec O_{\Gamma^-}^*AO_{\Gamma^+}\|_{op}
&\leq
|\braket{\xi_0^+|A|\xi_0^-}|
 \Expec \prod_{\ell=1}^N
\braket{\xi_{2\ell-2}^+|O_{\omega_\ell^+}|\xi_{2\ell-1}^+}^*
\braket{\xi_{2\ell-2}^-|O_{\omega_\ell^-}|\xi_{2\ell-1}^-} \\
&\leq
\eps^{-CN} |\braket{\xi_0^+|A|\xi_0^-}|
\Expec \prod_{a\in K(\Gamma^+)} \Ft{V_{y_a^+}}(p_a^+-p_{a-1}^+)^*
\prod_{a\in K(\Gamma^-)} \Ft{V_{y_a^-}}(p_a^--p_{a-1}^-).
\end{split}
\end{equation}
First we rewrite the expression above using signed indices $a$, defining the
impulse variables $q_a = p_a-p_{a-1}$ when $a\in K^+(\Gamma^+)$ and
$q_a = p_{a-1}-p_a$ for $a\in K^-(\Gamma^-)$.  Then, because the potential $V$ is real
and therefore $\Ft{V}(q)^* = \Ft{V}(-q)$, we can write more succinctly
\[
\Expec \prod_{a\in K(\Gamma^+)} \Ft{V_{y_a^+}}(p_a^+-p_{a-1}^+)^*
\prod_{a\in K(\Gamma^-)} \Ft{V_{y_a^-}}(p_a^--p_{a-1}^-).
= \Expec \prod_{a\in K(\bm{\Gamma})} \Ft{V_{y_a}}(q_a).
\]
We can split the expectation using the geometry of the points $y_a$.
More precisely, given $\bm{\Gamma}$ we define a graph $G(\bm{\Gamma})$
on the vertex $K(\bm{\Gamma})$ in which $(a,b)$ is an edge of $G(\bm{\Gamma})$
when $|y_a-y_b|\leq 2r$.  The graph $G(\bm{\Gamma})$ then defines a partition
$P(\bm{\Gamma})\in\mcal{P}(K(\bm{\Gamma}))$ which is the partition of collisions
into the connected components of the graph.  Since $V_{y_a}$ is independent of $V_{y_b}$
when $|y_a-y_b|\geq 2r$, we can split the expectation using the partition $P(\bm{\Gamma})$,
\begin{equation}
\label{eq:partition-potential}
\Expec \prod_{a\in K(\bm{\Gamma})} \Ft{V_{y_a}}(q_a).
= \prod_{S\in P(\bm{\Gamma})} \Expec \prod_{a\in S} \Ft{V_{y_a}}(q_a).
\end{equation}
The moment bound on the Fourier coefficients $\Ft{V_{y_a}}$ in particular implies
that for any sequence $(y_j,q_j)\in\PhaseSpace$,
\[
\big| \Expec \prod_{j=1}^k \Ft{V_{y_j}}(q_j)\big|
\leq \eps^{-Ck} k^{Ck} \exp(-c (k^{-1} r |\sum_j q_j|)^{0.99}).
\]
This crude estimate places already a significant constraint on the geometry of
the paths $\bm{\Gamma}$ that contribute to the evolution operator $\Evol_{N\tau}$.
We say that a pair of paths $\bm{\Gamma})$ is $\beta$-complete if for every $S\in P(\bm{\Gamma})$,
$|\sum_{j\in S} q_j|\leq \beta r$.

The calculation above shows the following
\begin{lemma}
\label{lem:complete-extended-path}
For a band-limited operator $A$ and $\beta=N^4$,
\[
\|
\int_{\Omega_{\alpha,S}^N\times\Omega_{\alpha,S}^N}
\Xi(\Gamma^+) \Xi(\Gamma^-)
(1 - \One_{\bm{\Gamma} \textrm{ is }\beta-\text{complete}})
\Expec O_{\Gamma^-}^* A O_{\Gamma^+}\diff\bm{\Gamma}\|_{op}
\leq \eps^{100}\|A\|_{op}.
\]
\end{lemma}
\begin{proof}
The bound follows from an application of the Schur test in the wavepacket basis, along with the
estimate
\[
|\Braket{\xi_{2N-1}^+ | \Expec O_{\Gamma^-}^*AO_{\Gamma^+} |\xi_{2N-1}^-}|
\leq
\eps^{-CN} |\Braket{\xi_0^+|A|\xi_0^-}| \exp(-c N^2)
\]
for paths $(\Gamma^+,\Gamma^-)$ which are not $N^4$-complete.
\end{proof}
For convenience we simply write $\Gamma^+\sim\Gamma^-$ when $\bm{\Gamma}$ is $N^4$-complete.

In the next section we explore in more detail the geometric features of paths $\bm{\Gamma}$
which are $N^4$-complete.

\section{Geometry and combinatorics of extended paths}
\label{sec:path-combo}
The point of this section is to more precisely specify the relationship between
the geometry of paths (as described by their ``events'') and the combinatorial properties
of the collision partition that the paths induce.  Our end goal is to
derive a diagrammatic expansion for $\Evol_{N\tau}$ which is stratified by a measure of
combinatorial complexity that matches geometric complexity.  The geometric complexity will allow
us to bound complicated diagrams by additional factors of $\eps^c$, which is needed to reach diffusive
times.

\subsection{From extended paths to concatenated paths}
For the arguments in this section it will be easier to work with a single ``concatenated path''
rather than an extended path.  We will therefore define a map
\begin{equation*}
\bm{\omega} = (\omega_1,\omega_2,\cdots,\omega_N) \mapsto (\mbf{S},\mbf{P},\mbf{Y})
\end{equation*}
that sends sequences of path segments to a single path.  Given such a sequence
$\bm{\omega}$, let $K_0(\bm{\omega})$ be the collision indexing set
\[
K_0(\bm{\omega}) = \{(\ell,j) \mid \ell\in [1,N], j\in [0,k_\ell]\}
\]
and $K(\bm{\omega}) = K_0(\bm{\omega}) \setminus ([1,N]\times\{0\})$.
Then, given $\bm{\omega}$ we define $\mbf{S}$ by
\begin{equation*}
S_j := \sum_{\substack{b\in K_0(\bm{\omega}) \\a_j\leq b< a_{j+1}}} s_b,
\end{equation*}
where $a_j$ is the $j$-th element of $K(\bm{\omega})$, with the convention that $a_0=0$ and
$a_{(\sum k_j)+1}= \infty$  are minimal and maximal elements, respectively.  An alternative
definition of $S_j$ starts by defining the collision times $T_{\ell,j}$ for each collision
in $\omega_\ell$, by
\[
T_{\ell,j} = \ell\tau + \sum_{j'<j} s_{\ell,j'}.
\]
Then $S_a$ satisfies $S_a = T_{a+1}-T_a$ (with the convention that $T_{1,0}=0$ and $T_{N,k_N} =N\tau$).
We can also assign the momentum variables for each
$a\in K_{\mbf{k}}$ by setting $P_a = p_a$, and also $P_{1,0} = p_{1,0}$.  Finally, $Y_a = y_a$.

The following lemma shows that the concatenated path $(\mbf{S},\mbf{P},\mbf{Y})$
approximately satisfies the compatibility conditions and approximate conservation of energy
expected of a single path segment.

\begin{lemma}[Path concatenation lemma]
\label{lem:path-concatenation}
If $\Gamma =(\bm{\xi},\bm{\omega}) \in \Omega_{\alpha,S}^N$ is an extended
path then the
concatenated path $(\mbf{S},\mbf{P},\mbf{Y})$ formed from the tuple $\bm{\omega}$
satisfies the compatibility condition
\begin{equation}
|Y_{a+1} - (Y_a + S_aP_a)| \leq C\alpha (r + S_a^2\tau^{-1}r^{-1} + S_a\tau^{-1} r).
\end{equation}
for each $a\in K(\Gamma)$, as well as the conservation of kinetic energy
conditions
\begin{equation}
||P_a|^2/2 -|P_{a'}|^2/2|| \leq \alpha
\max\{S_a^{-1}, S_{a'}^{-1}, |P_a| r^{-1}\}
+ \alpha |T_a - T_{a'}| \tau^{-1}r^{-1}.
\end{equation}
for $a,a'\in K(\Gamma)$.
\end{lemma}

Combining Lemma~\ref{lem:path-concatenation} with the argument
of Lemma~\ref{lem:path-simplification}, we obtain the following corollary
\begin{corollary}
\label{cor:simplified-path}
Let $\bm{\Gamma}=(\Gamma^+,\Gamma^-) \in \Omega_{\alpha,S}^N\times\Omega_{\alpha,S}^N$
be a pair of paths that is $N^4$ complete, and let $(\mbf{S}^+,\mbf{P}^+,\mbf{Y}^+)$
be the concatenated path formed from $\Gamma^+$.  Suppose that $a,a'\in K(\Gamma^+)$
are two collision indices such that for all $a<b<a'$, either $\{b,b-1\}\in P(\bm{\Gamma})$
or $\{b,b+1\}\in P(\bm{\Gamma})$.  Then
\[
|Y_{a'} - (Y_a + (T_{a'}-T_a)P_a)|\leq 100 k_{max}^2 N^6 r.
\]
\end{corollary}

To prove Lemma~\ref{lem:path-concatenation} we first need an elementary result showing that
the kinetic energy cannot drift too much between the endpoints of a path segment.
\begin{lemma}
\label{lem:ke-drift}
    If $\omega=(\mbf{s},\mbf{p},\mbf{y})\in\Omega_{\alpha,k}(\tau,\sigma;\eta,\xi)$
and $|\xi_p| \geq r\sigma^{-1}$, then
$||\xi_p|^2/2 - |\eta_p|^2/2|\leq \alpha |\xi_p| r^{-1}$.  In particular,
$||\xi_p| - |\eta_p||\leq \alpha r^{-1}$.
\end{lemma}
\begin{proof}
For $\omega$ to be $\alpha$-compatible with $\eta$ and $\xi$ it follows
that $|p_0-\xi_p|\leq \alpha r^{-1}$ and $|p_k - \eta_p|\leq \alpha r^{-1}$.
Moreover, since $s_0, s_k\geq S$, it follows from the kinetic energy condition
that
\[
||p_0|^2/2 - |p_k|^2/2|\leq \alpha \min\{\sigma^{-1}, |\xi_p| r^{-1}\} \leq \alpha |\xi_p| r^{-1}.
\]
The conclusion that $||\xi_p|-|\eta_p||\leq \alpha r^{-1}$ follows from an elementary
argument.
\end{proof}
\begin{proof}
For each $\ell\in[1,N]$, set
\[
z_\ell := \begin{cases}
(\xi_\ell)_x - (y_{\ell,k_\ell} +s_{\ell,k_\ell}p_{\ell,k_\ell}), & k_\ell > 0 \\
(\xi_\ell)_x - (\xi_{\ell-1})_x+\tau p_{\ell,0}, & k_\ell = 0.
\end{cases}
\]
Moreover, for $a\in K(\Gamma)$, set
\[
z_a = \begin{cases}
y_a - (y_{a-1} + s_{a-1}p_{a-1}), & a = (\ell,j) \text{ for some } j>1,\ell\in[1,N] \\
y_{\ell,1} - ((\xi_{\ell})_x + s_{\ell,0}p_{\ell,0}), & a = (\ell,0)\text{ for some }\ell\in[1,N].
\end{cases}
\]
The idea is that these `$z$' variables represent the offset in the collision positions and the ``checkpoints''
$\xi_\ell$.   We also need to keep track of the offset in the momenta variables, so similarly we define
\[
u_\ell^- := (\xi_\ell)_p - p_{\ell,k_\ell}
\]
and
\[
u_\ell^+ := p_{\ell+1,0} - (\xi_\ell)_p.
\]
Because $\Gamma\in\Omega_{\alpha,\mbf{k}}(\tau,S)$, each of the $z$ variables satisfies
$|z_b| \leq \alpha r$, and each of the $u$ variables satisfies $|u_\ell^{\pm}| \leq \alpha r^{-1}$.
We can use these error variables to write
\[
Y_{j+1} = Y_j + S_j P_j + \sum u_\ell^+ s_\ell^+ + \sum u_{\ell}^- s_\ell^- + \sum z_b,
\]
where the sum over $u_\ell^{\pm}$ includes all $\ell$ in between the collisions $a_j$ and $a_{j+1}$
(of which there are at most $k_{max} S_j\tau^{-1}$) and the times $s_\ell^+$ are all bounded by
$S_j$.  Moreover the sum over position displacements $z_b$ includes at most $S_j\tau^{-1}$ elements
(one for each displacement at a checkpoint, and a final one for the collision $z_{a_j}$ itself.
The bound then follows immediately.

The bound on the conservation of kinetic energy follows from applying
Lemma~\ref{lem:ke-drift} on each segment between the segments containing the collisions $a$ and $a'$,
of which there are $|T_a-T_{a'}|\tau^{-1}$.
\end{proof}

\subsection{Geometric features of paths}

We have already seen in~\eqref{eq:partition-potential} that the partition $P(\bm{\Gamma})$,
which is determined by the geometric configuration of the positions $y_a$ for $a\in K(\Gamma)$,
determines how the expectation $\Expec O_\Gamma^* A O_\Gamma$ splits.  The fundamental idea of this paper
is that the extra compatibility conditions and conservation of kinetic energy conditions that an extended
path satisfy significantly constrain the geometry of paths $\bm{\Gamma}$ that are $N^4$-complete.

The big idea of this paper is that the partition $P(\bm{\Gamma})$ can be significantly constrained
using much less information than the full dependency graph $G(\bm{\Gamma})$ of the scattering locations.
Specifically, we need only record pairs of collisions involved in certain rare events.   The
first kind of event is a recollision.

\begin{definition}[Recollisions of extended paths]
The recollision graph $R(\bm{\Gamma})$ of an extended path $\bm{\Gamma}$ consists of all
pairs $(a,a')\subset K(\bm{\Gamma})$ such that $\sign(a)=\sign(a')$,
$a'\not= a \pm 1$, and $|y_a-y_{a'}|\leq 2r$.

We also define the set $E^{imm}(\bm{\Gamma})$ of immediate recollisions to be pairs $(a,a')$
with $a'=a\pm 1$ and $|y_a-y_{a'}\leq 2r$.
\end{definition}

The second kind of event is what we call a ladder-breaking event.  A ladder-breaking event
is a pair $(a,a')$ of non-consecutive collisions
for which it is possible to enter $a$ at \textit{some}
momentum $v$ and travel to the collision $a'$.  This can be encoded as a geometric constraint on the
collision data $(y_a,q_a)$ and $(y_{a'},q_{a'})$.

\begin{definition}[Ladder-breaking events]
The ladder-breaking event set $\LB(\bm{\Gamma})$ is
the set of all pairs $(a,a')\in K(\bm{\Gamma})^2$ such that there exists $b\in K(\bm{\Gamma})$
with $a<b<a'$ and $b$ not belonging to an immediate recollision,
and such that there exists $v\in\Real^d$ and $t\in\Real$ such that
\begin{align*}
    ||v|^2/2 - |p_{1,0}|^2/2| &\leq \alpha^2  N^{10} |p_{1,0}|r^{-1} \\
    ||v-q_a|^2/2-|p_{1,0}|^2/2| &\leq \alpha^2 N^{10} |p_{1,0}| r^{-1} \\
    ||v+q_{a'}|^2/2-|p_{1,0}|^2/2| &\leq \alpha^2 N^{10} |p_{1,0}| r^{-1} \\
    |y_a + tv - y_{a'}| &\leq \alpha^2 N^{10} r.
\end{align*}
In addition, we say that $(0,a)$ forms a ladder-breaking event if there exists
$b\in K(\bm{\Gamma})\setminus I^{imm}(\bm{\Gamma})$ such that $b<a$ and
there exists $t\in\Real$ such that
\begin{align*}
    ||p_{1,0}+q_a|^2/2-|p_{1,0}|^2/2| &\leq \alpha^2 N^{10} |p_{1,0}| r^{-1} \\
    |y_0 + tp_{1,0} - y_{a}| &\leq \alpha^2 N^{10} r.
\end{align*}
\end{definition}

Ladder breaking events themselves are difficult to work with because they cannot be determined
solely by the collision data $(y_a,q_a)$ and $(y_b,q_b)$, but instead have an additional
nonlocal constraint that every collision $a<a'<b$ is an immediate recollision.  We can simplify this
by working with two other sets of \emph{local} events that together cover the set of ladder-breaking events.

The first of these is a tube event, which is similar to the tube event we have already seen.
\begin{definition}[Tube events]
A tube event in a path $\Gamma$ is a pair $(a,b)$ with $a<b$ and $a,b\not\in E^{imm}(\Gamma)$
 satisfying the condition that there exists $s\in\Real$ such that
\[
|y_b - (y_a + s p_{a-1})| \leq N^{12} r.
\]
We write $E^{tube}(\Gamma)$ for the set of all tube events.
\end{definition}

The other type of event is called a cone event.  Cone events are more complicated, and an example is visualized
in Figure~\ref{fig:cone-event}.
\begin{figure}
\includegraphics{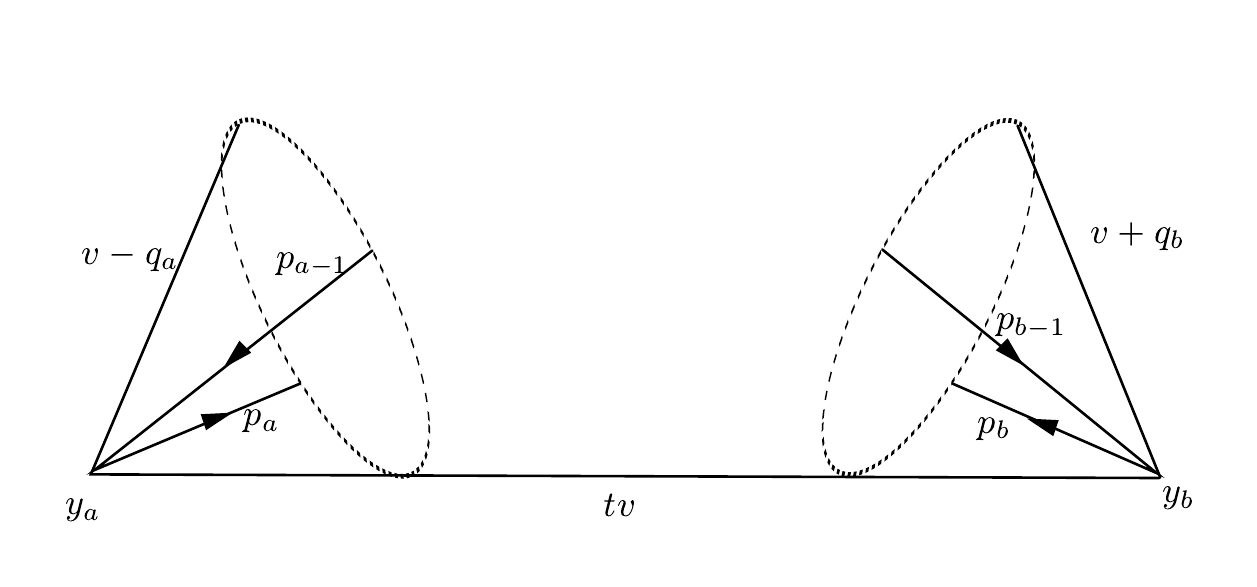}
\caption{A depiction of a cone event at $(a,b)$.  The collision site $y_b$ does \emph{not} lie near the ray
$y_a+sp_a$ (which points out of the page), but there \emph{does} exist some momentum $v$ for which
$y_b \approx y_a+tv$, and so that both $v-q_a$ and $v+q_a$ have approximately the correct kinetic energy.
This occurs when the cone drawn at $y_a$ which contains the rays $y_a+sp_a$ and $y_a+sp_{a-1}$ also contains the
point $y_b$ and vice-versa, hence the name.}
\label{fig:cone-event}
\end{figure}

\begin{definition}[Cone events]
A cone event in a path $\Gamma$ is a pair $(a,b)$ with $a<b$
such that there exists $v\in\Real^d$ and $s\in\Real$ such that
\begin{align*}
    ||v|^2/2 - |p_{1,0}|^2/2| &\leq \alpha^2  N^{10} |p_{1,0}|r^{-1} \\
    ||v-q_a|^2/2-|p_{1,0}|^2/2| &\leq \alpha^2 N^{10} |p_{1,0}| r^{-1} \\
    ||v+q_{a'}|^2/2-|p_{1,0}|^2/2| &\leq \alpha^2 N^{10} |p_{1,0}| r^{-1} \\
    |y_a + tv - y_{a'}| &\leq \alpha^2 N^{10} r,
\end{align*}
and moreoever $|v - p_a| \geq N^8 r^{-1}$.
We write $E^{cone}(\Gamma)$ for the set of all cone events.
\end{definition}

We observe that the tube events and cone events ``cover'' the ladder-breaking events.
\begin{lemma}
Let $(a,b)\in \LB(\bm{\Gamma})$.  Then either $(a,b)\in E^{cone}(\Gamma)$ or
there exists $c$ such that $a<c<b$ and $(c,b)\in E^{tube}(\Gamma)$.
\end{lemma}
\begin{proof}
Suppose that $(a,b)\in \LB(\bm{\Gamma})\setminus E^{cone}(\Gamma)$.  Then in particular
$p_a$ itself satisfies $|y_a + sp_a - y_b| \leq N^{10} r$.  Let $a'$ be the collision index
satisfying $a<a'<b$ and such that $a'$ is not an immediate recollision.  Then $(a',b)$
is a tube event.
\end{proof}

The final piece of information we need is a trimmed form of the recollision graph $G(\bm{\Gamma})$ that
removes all isolated edges.  Recall that the dependency graph $G(\bm{\Gamma})$ is defined to be
\[
G(\bm{\Gamma}) = \{(a,b)\in\binom{K(\bm{\Gamma})}{2} \mid |y_a-y_b|\leq 2r\}.
\]
We use the graph $G(\bm{\Gamma})$ to define the cluster graph $C(\bm{\Gamma})$ below.

\begin{definition}[Cluster graph]
The cluster graph $C(\bm{\Gamma})$ of a doubled path $\bm{\Gamma}$ is defined to be
\[
G_C(\bm{\Gamma}) :=
\{ (a,b)\in G(\bm{\Gamma}) \mid
\text{ there exists } c\in K(\bm{\Gamma}) \setminus \{a,b\} \text{ such that }
(a,c)\in G(\bm{\Gamma}) \text{ or } (b,c)\in G(\bm{\Gamma})\}.
\]
We write $\supp G_C(\bm{\Gamma})$ to mean the set of indices contained in an edge of $C(\bm{\Gamma})$.
\end{definition}

We say that a collision index $a\in K(\bm{\Gamma})$ is \textit{typical} if
there does not exist $b\sim_{P(\bm{\Gamma})} a$ such that $b$
belongs to a recollision, tube, or cone event, or a cluster in $C(\bm{\Gamma})$.
We write $\Typical(\bm{\Gamma})$ for the set of typical indices.

\begin{figure}
\includegraphics[scale=0.5]{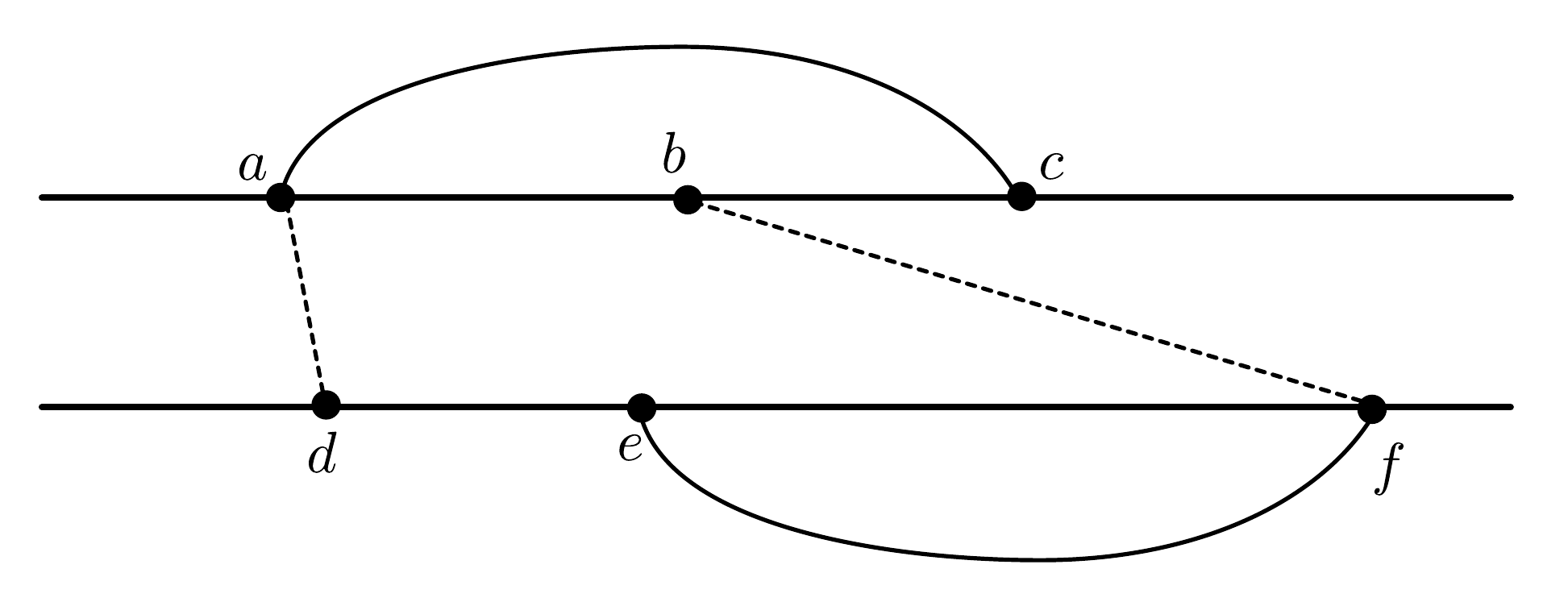}
\caption{An abstract representation of the geometric data of a path.  In this depiction, $(a,c)$ forms a
recollision as well as $(e,f)$, and $a\sim d$, $b\sim f$.  Note that $b$ and $d$ are atypical collisions even though
they do not belong to recollisions, tube events, or cone events.}
\label{fig:diagram-example}
\end{figure}

\subsection{Combinatorial consequences of events}
In this section we see that the partition $P(\bm{\Gamma})$ has a rigid structure on the typical
collision indices.

The main lemma is the following simple on the geometry of consecutive collisions.
Given a path $\bm{\Gamma}$ and a collision index $a\in K(\bm{\Gamma})$, we define
$\next(a)\in K(\bm{\Gamma})$ to be the next collision index after $a$ that is not
an immediate recollision, so
\[
\next(a) := \min\{a'\in K(\bm{\Gamma}) \mid a'>a \text{ and }
a'\not\in E^{imm}(\bm{\Gamma})\}.
\]
Likewise we set $\back(a)$ to be the collision prior to $a$ that is not an
immediate recollision,
\[
\back(a) := \min\{a'\in K(\bm{\Gamma}) \mid a'<a \text{ and }
a'\not\in E^{imm}(\bm{\Gamma})\}.
\]
\begin{lemma}
\label{lem:ladder-baby}
Let $\bm{\Gamma}$ be an $N^4$-complete extended path, and let
$a\in K(\Gamma^+)$, $b,b'\in K(\Gamma^-)$ satisfy $\{a,b\},\{\next(a),b'\}\in P(\bm{\Gamma})$.
Then the following inequalities hold:
\begin{align*}
||p_a|^2/2 - |p_{1,0}|^2/2| &\leq C \alpha^2 N^8 |p_{1,0}| r^{-1} \\
||p_a - q_b|^2/2 - |p_{1,0}|^2/2| &\leq C \alpha^2 N^8 |p_{1,0}| r^{-1} \\
||p_a+ q_{b'}|^2/2 - |p_{1,0}|^2/2| &\leq C \alpha^2 N^8 |p_{1,0}| r^{-1} \\
|y_b + (T_{\next(a)} - T_a)p_a - y_{b'}| &\leq C\alpha^2 N^8 r.
\end{align*}
\end{lemma}
\begin{proof}
The last inequality follows from Corollary~\ref{cor:simplified-path}
along with the bounds $|y_a - y_b|\leq 2r$ and $|y_{\next(a)} - y_{b'}|\leq 2r$
that follow from $\{a,b\},\{\next(a),b'\}\in P(\bm{\Gamma})$.

The first inequality follow from the kinetic energy conservation condition.

The second and third inequalities follow from the same kinetic energy conservation condition coupled with
the bounds
\begin{align*}
|p_{\next(a)} - (p_a + q_{b'})| &\leq N^4 r^{-1} \\
|p_{a-1} - (p_a - q_b)| &\leq N^4 r^{-1}
\end{align*}
which follow from the condition that $\bm{\Gamma}$ is $N^4$-complete.
\end{proof}

We use Lemma~\ref{lem:ladder-baby} to obtain a local rigidity statement for the partition
$P(\bm{\Gamma})$.
\begin{lemma}[Ladder-forming lemma]
\label{lem:ladder}
Let $\bm{\Gamma}$ be an $N^4$-complete extended path, and
let $a\in K(\Gamma^+)$, $b\in K(\Gamma^-)$ be collisions
satisfying $\{a,b\}\in P(\bm{\Gamma})$.
Suppose further that
$\{a,\next(a)\}\subset \Typical(\bm{\Gamma})$
and $b$ does not belong to a ladder-breaking event of $\bm{\Gamma}$.
Then either
$\next(a)\sim_{P(\bm{\Gamma})}\next(b)$
or $\next(a)\sim_{P(\bm{\Gamma})} \back(b)$.

Moreover, if $\{\back(a),a\} \subset \Typical(\bm{\Gamma})$,
then either $\back(a) \sim_{P(\bm{\Gamma})}\next(b)$ or
$\back(a) \sim_{P(\bm{\Gamma})} \back(b)$.
\end{lemma}
\begin{proof}
We will treat only the case $\{a,\next(a)\} \subset \Typical(\bm{\Gamma})$,
as the case $\{a,\back(a)\}\subset \Typical(\bm{\Gamma})$ is similar.

First, we observe that it must be the case that $\next(a)\sim_{P(\bm{\Gamma})} b'$ for some index
$b'\in K(\Gamma^-)$, or else $\next(a)$ would belong to a recollision and therefore not
be typical.
If $\{\next(a),b'\}\not\in P(\bm{\Gamma})$,  then there is an additional collision belonging to
the cluster containing $\{\next(a),b\}$, so $\next(a)\in\supp C(\bm{\Gamma})$
and therefore $\next(a)\not\in\Typical(\bm{\Gamma})$.

We have therefore reduced to the case that $\{\next(a),b'\}\in P(\bm{\Gamma})$
for some $b'\in K(\Gamma^-)$.  Now if $b'$ is not $\next(b)$ or $\back(b)$
then by Lemma~\ref{lem:ladder-baby} the pair $(b,b')$ forms a ladder-breaking event,
which is a contradiction.
\end{proof}

The final ingredient we need is a sufficient condition for the first collision of $\wtild{K}^+$
to pair with the first collision of $\wtild{K}^-$.
\begin{lemma}
\label{lem:first-rung}
Suppose that $\bm{\Gamma}$ is an $N^4$-complete path with no ladder-breaking event of the form $(0,a)$,
and suppose that $|p_{1,0,+} - p_{1,0,-}| \leq N\alpha$.
Let $\next^+(0)$ and $\next^-(0)$ be the first collisions in $\Gamma^+$ and $\Gamma^-$,
respectively, that are not immediate recollisions.
If $\{\next^+(0),\next^-(0)\}\subset \Typical(\bm{\Gamma})$, then
in fact $\{\next^+(0),\next^-(0)\}\in P(\bm{\Gamma})$.
\end{lemma}
\begin{proof}
Since $\next^+(0)\in\Typical(\bm{\Gamma})$, in particular it follows that
$\{\next^+(0),b\}\in P(\bm{\Gamma})$ for some $b\in K(\Gamma^-)$.  If $b\not=\next^-(0)$, then one of $(0,b)$
or $(0,\next^-(0))$ forms a ladder-breaking event.
\end{proof}

\subsection{Global structure of the partition $P(\bm{\Gamma})$}
Lemmas~\ref{lem:ladder} and~\ref{lem:first-rung} provide only local information about the
partition $P(\bm{\Gamma})$.  We will turn this local information into a global rigidity result.
The rigidity result we will prove is that the pairs formed by typical collisions can be partitioned
into a small number of special partitions called ladder and anti-ladder matchings.

\begin{figure}
\includegraphics[scale=0.5]{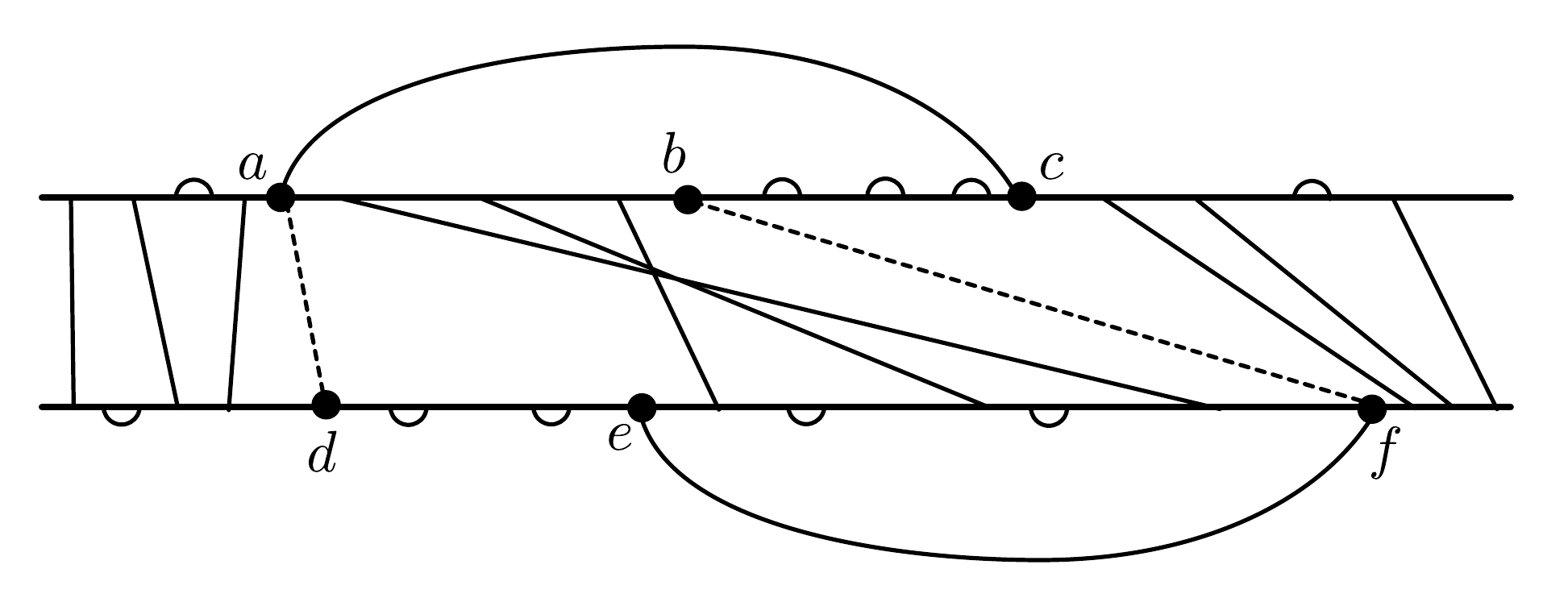}
\caption{An example of a partition compatible with the geometric data of Figure~\ref{fig:diagram-example}.
In this case the interval set is $\{(0,a),(a,b),(b,c),(c,\infty),(0,d),(d,e),(e,f),(f,\infty)\}$.
The pair of intervals $\{(a,b),(e,f)\}$ forms a renormalized anti-ladder, and $(b,c)$ shows an
example of a simple partition.}
\label{fig:filled-diagram}
\end{figure}

First we define ladders and anti-ladders in some generality.
\begin{definition}[Ladders and anti-ladders]
\label{def:ladders}
Let $A$ and $B$ be two finite disjoint ordered sets with $|A|=|B|$.  The \emph{ladder matching}
$P_{lad}\in\mcal{P}(A\sqcup B)$ is the unique matching of the form
$P_{lad}=\{\{a,\varphi(a)\}\}_{a\in A}$ where $\varphi:A\to B$ is the unique order-preserving
bijection between $A$ and $B$.
The \emph{anti-ladder} matching is given by
$\{\{a,\varphi^-(a)\}\}_{a\in A}$ with $\varphi^-$ being the unique order-reversing
bijection.
\end{definition}
Note that a partition $P\in\mcal{P}(A\sqcup B)$ can be both a ladder and an anti-ladder if $|A|=|B|=1$.

The partition $P(\bm{\Gamma})$ can also have many instances of immediate recollisions, so
we introduce renormalized ladders and anti-ladders to allow for immediate recollisions to exist
between the ``rungs'' of the ladder.
\begin{definition}[Renormalized ladders and anti-ladders]
\label{def:renormalized-ladders}
Let $A$ and $B$ be two finite disjoint ordered sets.  A renormalized ladder is any partition
$P\in\mcal{P}(A\sqcup B)$ for which there exist sets $I_A\subset A$ and $I_B\subset B$
satisfying the following properties:
\begin{itemize}
\item Every set $S\in P$ has exactly two elements, $|S|=2$.  That is, $P$ is a matching.
\item If $a,a'\in A$ and $a\sim_P a'$, then $a=a'+_A1$ or $a'=a+_A1$ and $a,a'\in I_A$.
\item Likewise, if $b,b'\in B$ and $b\sim_P b'$, then $b=b'+_B1$ or $b'=b+_B1$ and $b,b'\in I_B$.
\item The partition $P$ restricted to the set $(A\setminus I_A)\sqcup (B\setminus I_B)$
is a ladder partition.
\end{itemize}
Let $\mcal{Q}_{rl}\subset \mcal{P}(A\sqcup B)$ be the set of renormalized ladder
partitions.  We define renormalized anti-ladders in a similar way (that is, satisfying the first
three conditions but so that $P$ restricted to $(A\setminus I_A)\sqcup (B\setminus I_B)$ is an
anti-ladder).  Let $\mcal{Q}^-_{rl}$ denote the set of renormlized anti-ladder partitions.
\end{definition}

The other important type of partition is a simple partition, which can be thought of as a ladder
with no rungs.

\begin{definition}[Simple partitions]
Let $A$ be a finite ordered set.  The simple partition on $A$ is the perfect matching $M$
such that every pair $\{a,a'\}\in M$ satisfies $a'=a+_A1$ or $a=a'+_A1$.
\end{definition}

We will partition the index set $K(\bm{\Gamma})$ into a collection of intervals that cover
the typical collision indices and an exception set of atypical collisions.
Given an extended path $\bm{\Gamma}$
on a collision index set $K(\bm{\Gamma})$, and $a,b\in K(\bm{\Gamma})$ with $\sign(a)=\sign(b)$,
we define the collision interval $(a,b)_{K(\bm{\Gamma})}$ to be the set of collisions
\[
(a,b)_{K(\bm{\Gamma})} := \{c\in K(\bm{\Gamma}) \mid a < c < b\}.
\]
We also define the ``one-sided'' intervals
\[
(0,a)_{K(\bm{\Gamma})} := \{c\in K(\bm{\Gamma}) \mid c < a\}.
\]
and
\[
(a,\infty)_{K(\bm{\Gamma})} := \{c\in K(\bm{\Gamma}) \mid a<c\}.
\]
In a similar way we can define closed intervals such as $[a,b]_{K(\bm{\Gamma})}$, which
include the endpoints $a$ and $b$.
We note that intervals are totally ordered sets, with an ordering inherited from the
partial order on $K(\bm{\Gamma})$.  We write $\sign(I)$ to mean the sign of $I$, so
that $\sign((a,b)_{K(\bm{\Gamma})}) = \sign(a)=\sign(b)$.
An interval $(a,b)_{K(\bm{\Gamma})}$ is said to be \textit{typical}
with respect to $\bm{\Gamma}$ if $(a,b)_{K(\bm{\Gamma})}\subset \Typical(\bm{\Gamma})$.
A maximal typical interval is a typical interval not strictly contained in any other typical interval.
The maximal typical intervals are those of the form $(a,b)_{K(\bm{\Gamma})}$ for
$a,b\not\in\Typical(\bm{\Gamma})$ or $(0,a)_{K(\bm{\Gamma})}$ or $(a,\infty)_{K(\bm{\Gamma})}$ for
$a\not\in\Typical(\bm{\Gamma})$.
We define $\mcal{I}(\bm{\Gamma})$ to be the collection of maximal typical intervals defined on
$\bm{\Gamma}$.

The main global result is that maximal typical intervals either form ladders, anti-ladders, or simple
partitions.
\begin{lemma}
\label{lem:interval-pair}
Let $\bm{\Gamma}$ be an $N^4$-complete extended path, and let $\mcal{I}(\bm{\Gamma})$ be
its collection of maximal typical intervals.
Then for each interval $I=(a,b)_{K(\bm{\Gamma})}\in\mcal{I}$, at least one of the following holds:
\begin{itemize}
\item The interval $I$ forms a simple partition, in the sense that $P(\bm{\Gamma})$ is $I$-local
and the restriction $P(\bm{\Gamma})|_I$ is a simple partition.
\item There exists $I'\in \mcal{I}(\bm{\Gamma})$
with $\sign(I)=-\sign(I')$ such that $P(\bm{\Gamma})$ is local
to $I\cup I'$ and $P|_{I\cup I'}$ is either a renormalized ladder partition or a
renormalized anti-ladder partition.
\end{itemize}
\end{lemma}
\begin{proof}
Suppose that there exists $a\in I$ such that $a\sim_{P(\bm{\Gamma})} b$ for
some $b\in K(\bm{\Gamma})$ with $\sign(b)=-\sign(a)$.  Note that for such an index,
$a\in \wtild{K}(\bm{\Gamma})$. The collision index $b$ belongs
to a unique interval $I'$ with $\sign(I')=-\sign(I)$.  Then the fact that $P(\bm{\Gamma})|_{I\cup I'}$
forms a ladder or an anti-ladder follows from repeated application of Lemma~\ref{lem:ladder}.

On the other hand, if $I$ does not contain any external collisions then we are in the first case,
and $I$ must form a simple partition.
\end{proof}

\subsection{Geometric features of ladders}
We conclude this section by recording geometric features of paths that correlate along a ladder.

\begin{lemma}
\label{lem:rung-lemma}
Let $\bm{\Gamma}$ be a $N^4$-complete path and suppose that
$a\in K(\Gamma^+)$, $b\in K(\Gamma^-)$ are collision indices
satisfying $\{a,b\}\in P(\bm{\Gamma})$ and $\{\next(a),\next(b)\}\in P(\bm{\Gamma})$.
Then the following hold:
\begin{align}
\label{eq:time-shift-bd}
|(t_{\next(a)} - t_a) - (t_{\next(b)} - t_b)| &\leq C N^6 |p_{1,0}|^{-1}r,\\
\label{eq:momentum-shift-bd}
| (p_{\next(b)}-p_{\next(a)}) - (p_b-p_a)| &\leq C N^4 r^{-1}, \\
\label{eq:momentum-close-bd}
|p_a - p_b| &\leq  N^6 r (t_{\next(a)}-t_a)^{-1}.
\end{align}
\end{lemma}
\begin{proof}
First observe that~\eqref{eq:momentum-shift-bd} is a direct consequence of
the fact that $\bm{\Gamma}$ is $N^4$-complete.

The bounds~\eqref{eq:time-shift-bd} and~\eqref{eq:momentum-close-bd} are less trivial.
Now we prove~\eqref{eq:time-shift-bd}.  By Corollary~\ref{cor:simplified-path},
\[
|(y_{\next(a)} - y_a) - (t_{\next(a)}-t_a)p_a| \leq \alpha N^6r.
\]
On the other hand, $|y_{\next(a)}-y_{\next(b)}|\leq 4r$ and $|y_{a}-y_{b}|\leq 4r$, so also
\[
|(y_{\next(a)} - y_a) - (t_{\next(b)}-t_b)p_b| \leq 2\alpha N^6r.
\]
It follows that
\begin{equation}
\label{eq:tp-tp}
|(t_{\next(a)}-t_a) p_a - (t_{\next(b)}-t_b)p_b| \leq 3\alpha N^6 r.
\end{equation}
In particular, using $|u-v| \geq ||u|-|v||$ for any vectors $u,v\in\Real^d$,
we have
\[
|(t_{\next(a)}-t_a) |p_a| - (t_{\next(b)}-t_b)|p_b|| \leq 3\alpha N^6 r.
\]
By the conservation of kinetic energy condition, $||p_a| -|p_b||\leq C\alpha N r^{-1}$.
Therefore
\[
(t_{\next(b)}-t_b)||p_a|-|p_b||\leq C\alpha N^2 r^{-1} \tau \leq C\alpha N^2 r,
\]
so we have
\[
|p_a||(t_{\next(a)}-t_a)  - (t_{\next(b)}-t_b)| \leq 4\alpha N^6 r.
\]
Dividing through by $|p_a| \geq \frac{1}{2} |p_{1,0}|$ finishes the proof of~\eqref{eq:time-shift-bd}.

The remaining nontrivial bound is~\eqref{eq:momentum-close-bd}.  We obtain this
from~\eqref{eq:tp-tp} after first replacing $t_{\next(b)}-t_b$ by $t_{\next(a)}-t_a$
using~\eqref{eq:time-shift-bd} to obtain
\[
(t_{\next(a)}-t_a) |p_a - p_b| \leq C\alpha N^6 r.
\]
Dividing through by $t_{\next(a)}-t_a$ finishes the result.
\end{proof}

Recall that for a phase space point $\xi=(x,p)\in\PhaseSpace$, we define the free evolution $U_s(\xi)$,
also written simply $\xi_s$, by
\[
U_s((x,p)) := (x+sp,p).
\]
\begin{definition}[Twinning in phase space]
We say that two points $\xi,\eta\in\PhaseSpace$ are $(\beta,s)$-twinned if
$d_r(\xi,U_s(\eta))\leq \beta$.
\end{definition}

A corollary of Lemma~\ref{lem:rung-lemma} is that the phase space endpoints of a long ladder must be twinned.
\begin{lemma}
\label{lem:twinning-sync}
Suppose that $\bm{\Gamma}$ is an $N^4$-complete path and the restriction of the
partition to the pair of intervals
$[a,a']_{K^+(\Gamma^+)} \cup [b,b']_{K^-(\Gamma^-)}$ forms a renormalized ladder partition
so that in particular $\{a,b\}\in P(\bm{\Gamma})$ and $\{a',b'\}\in P(\bm{\Gamma})$.
Suppose moreover that $\ell(a')-\ell(a) > 1$.

Then for some $s\in\Real$, $\xi_{2\ell(a)-1}$ is $s$-twinned to $\xi_{2\ell(b)-1}$ and
$\xi_{2\ell(a')-2}$ is $(N^{10},s)$-twinned to $\xi_{2\ell(b')-2}$.
\end{lemma}
\begin{proof}
Because there are at most $k_{max}$ collisions in a segment, there must exist a pair of
consecutive external collisions $c,c'\in[a,a']$ that are paired with consecutive
collisions $d,d'\in[b,b']$ and which satisfy $|t_c-t_{c'}| \geq k_{max}^{-1}\tau$.
Applying Lemma~\ref{lem:rung-lemma} we therefore conclude that $|p_c-p_d|\leq N^8 r^{-1}$.
Then iteratively applying the bound~\eqref{eq:momentum-shift-bd} we conclude
that $|p_a - p_b| \leq 2N^8 r^{-1}$ and $|p_{a'}-p_{b'}|\leq 2N^8 r^{-1}$.

%TODO: Write this out on a park.
\end{proof}

\section{Skeletons and diagrams}
\label{sec:skeletons}
One way to apply Lemma~\ref{lem:interval-pair} is to decompose the path integral
according to the data
\[(E^{rec}(\bm{\Gamma}),E^{tube}(\bm{\Gamma}),E^{cone}(\bm{\Gamma}),
G_C(\bm{\Gamma}), G_A(\bm{\Gamma})),
\]
where $G_A$ is the graph of atypical collisions,
\[
G_A(\bm{\Gamma}) :=
\{(a,b)\in G(\bm{\Gamma}) \mid a\not\in \Typical(\bm{\Gamma}) \text{ or } b\not\in \Typical(\bm{\Gamma})\}.
\]
Within each piece, the information about the events and the graphs $G_C$ and $G_A$ about places a strong
constraint on $P(\bm{\Gamma})$ by Lemma~\ref{lem:interval-pair} and therefore one can hope to split
up the expectation
$\Expec O_{\Gamma^-}^* A O_{\Gamma^+}$ in a nice way.  The problem with this approach is that there are simply
too many graphs to sum over.  As an extreme example, suppose that the graph $C(\bm{\Gamma})$ is connected.
This places a strong constraint on every single collision in $\Gamma$, so one can hope that the contribution
of such paths is on the order $\eps^{cN}$.  On the other hand, there are $2^{\Theta(N^2)}$ connected paths
on $CN$ collisions, and the combinatorial factor $2^{N^2}$ swamps the factor of $\eps^{cN}$.

We therefore need a more parsimonious way of recording enough information about the geometry of $\bm{\Gamma}$
that does not contain too much information about $\bm{\Gamma}$ and which lends itself to a decomposition of
the sum.

The approach we take in this paper is
to record minimal spanning forests of the graphs $E^{rec}(\bm{\Gamma})$,$E^{tube}(\bm{\Gamma})$,
$E^{cone}(\bm{\Gamma})$, $G_C(\bm{\Gamma})$, and $G_A(\bm{\Gamma})$.  To be more precise, let
$\mbf{G}_{all} := \binom{\mbf{K}_{all}}{2}$ be the set of all pairs of distinct collision indices
in $\mbf{K}_{all}$.  Let $\mbf{G}^\pm\subset E(\mbf{K}_{all})$ be the set of pairs $(a,b)$ with
$\sign(a)\not=\sign(b)$ (so $\mbf{G}_{all}$ is the complete graph on $\mbf{K}_{all}$, and
$\mbf{G}^\pm$ is a complete bipartite graph).
We fix an arbitrary assignment of weights $w:\mbf{G}_{all}\to\Real^+$
to each edge so that each edge has a distinct weight, and so that $w(e)>w(e')$ whenever
$e\in \mbf{G}^\pm$ and $e'\not\in \mbf{G}^\pm$.

Then, given a graph $G\subset \mbf{G}_{all}$, we define the minimal spanning forest $F_G$ to be the graph
of minimal total weight which has the same connected components as $G$.

Before we define the skeleton of a graph, we record a slight modification to the graph $G(\bm{\Gamma})$ that we
will use.  We redefine
\begin{align*}
G(\bm{\Gamma}) = \{(a,b) \in E(\mbf{K}(\Gamma))
\mid
|y_a-y_b|\leq 2r \text{ and }\sign(a)\not=\sign(b)
\text{ or } |y_a-y_b|\leq 4r\}.
\end{align*}
The utility of this new definition is that, if $(a,b),(a,c)\in G(\bm{\Gamma})\cap E^\pm$,
then $(b,c)\in G(\bm{\Gamma})$.  This, along with the dominance of the weights
of the edge set $E^\pm$, ensure that the minimal forest $F_G(\bm{\Gamma})$ has the property
that no two edges in $F_G(\bm{\Gamma})\cap E^\pm$ are incident to a vertex.  We also note that this
modification of the definition of $G(\bm{\Gamma})$, though it may change the cluster graph and therefore
the set of typical indices, otherwise does not affect the conclusions of Lemma~\ref{lem:interval-pair}.

\begin{definition}[Geometric skeleton]
The  skeleton   $\mcal{F}(\bm{\Gamma})$  of a path $\bm{\Gamma}$ consists of the
tuple of forests $(F_{rec}(\bm{\Gamma}),F_{tube}(\bm{\Gamma}), F_{cone}(\bm{\Gamma}),F_C(\bm{\Gamma}),
F_A(\bm{\Gamma})$
The \emph{complexity} of the skeleton
$\mcal{F}$, written $\|\mcal{F}\|$
is given by $\|\mcal{F}\| := |F_{rec}| + |F_{tube}| + |F_{cone}| + |F_C| + |F_A|$,
and the \emph{support} $\supp \mcal{F}(\bm{\Gamma})\subset\mbf{K}_{all}$
of the skeleton is the set of collision indices contained in an edge in
$F_{rec}(\bm{\Gamma})\cup F_{tube}(\bm{\Gamma})\cup F_{cone}(\bm{\Gamma})\cup F_C(\bm{\Gamma})\cup F_A(\bm{\Gamma})$.
\end{definition}

An important combinatorial feature of the cluster forest is that it cannot have too many
edges in $E^\pm$.  More precisely, the number of edges in $F_C\cap E^\pm$ is
controlled by $|F_C|$.
\begin{lemma}
Let $F_C$ and $F_A$ be the cluster forest and atypical forest (respectively)
of some path $\bm{\Gamma}$, and let
$S_{AC}\subset \mbf{K}_{all}$ be the set of collision indices that are contained in an
edge of $(F_C\cup F_A)\cap E(\mbf{K}_{all}) \setminus E^\pm$.  Then
\[
|(F_C\cup F_A)\cap E^\pm| \leq 2|S_R|.
\]
\end{lemma}
\begin{proof}
This follows from the observation above that no two edges in $F_C\cap E^\pm$
or $F_A\cap E^{\pm}$ are incident to any vertex.
\end{proof}

Lemma~\ref{lem:interval-pair} determines the structure of the partition $P(\bm{\Gamma})$
for the typical indices, which are those not contained in the support $\supp\mcal{F}(\bm{\Gamma})$.
Fortunately the partition $P(\bm{\Gamma})|_{\supp\mcal{F}(\bm{\Gamma})}$ is exactly the connected
components of the graph $F_C\cup F_R$.  In particular, we have the following result.
\begin{lemma}
Let $\bm{\Gamma}$ be an extended path with skeleton $\mcal{F}$.  Then
$\supp\mcal{F} = K(\bm{\Gamma}) \setminus \Typical(\bm{\Gamma})$, and $P(\bm{\Gamma})$
is a refinement of the partition $\{\supp\mcal{F}, \Typical(\bm{\Gamma})\}$
(that is, $P(\bm{\Gamma})$ saturates $\supp\mcal{F}$).

Moreover, if $\bm{\Gamma}'$ is another extended path with $\mcal{F}(\bm{\Gamma}')=\mcal{F}$,
then $P(\bm{\Gamma})|_{\supp\mcal{F}} = P(\bm{\Gamma}')|_{\supp\mcal{F}}$.
\end{lemma}
We write $P_{\mcal{F}}$ for the partition $P(\bm{\Gamma})|_{\supp\mcal{F}}$.

\subsection{Decomposing paths using skeletons}
We will decompose
the integral over $\Omega_{ext}^N\times\Omega_{ext}^N$ according to the skeleton $\mcal{F}(\bm{\Gamma})$.
Let $\One_{\mcal{F}}:\Omega_{\alpha,S}^N\times\Omega_{\alpha,S}^N\to\{0,1\}$
be the indicator function of the set of paths $\bm{\Gamma}$ with skeleton exactly $\mcal{F}$.
We can use these indicator functions to partition the domain of integration of the path integral
as follows:
\begin{equation}
\label{eq:sum-over-skeletons}
\Xi(\bm{\Gamma})
\One(\Gamma^+\sim\Gamma^-) \Expec O_{\Gamma^+}^*A O_{\Gamma^-} \diff\Gamma^+\diff\Gamma^- \\
= \sum_{\mcal{F}}
\iint
\Xi(\bm{\Gamma})
\One(\Gamma^+\sim\Gamma^-)
\One_{\mcal{F}}(\bm{\Gamma})
\Expec O_{\Gamma^+}^*A O_{\Gamma^-} \diff\bm{\Gamma}.
\end{equation}

In the next section we will see how to use the structure of the $P(\bm{\Gamma})$
to write down a more useful expression for $\Expec O_{\Gamma^+}^*AO_{\Gamma^-}$, which is the point
of decomposing the integral as a sum over skeletons.  The cost of this manipulation is that we are
left with the indicator function $\One_{\mcal{F}}(\bm{\Gamma})$, which is difficult to work with because it
depends on all of the path variables at the same time.  For example, the indicator function
$\One_{\mcal{F}}$ includes many ``negative constraints'' that ensure that collisions belong to different
components of $F_R$ \textit{ do not } form a recollision.  Fortunately the following lemma, which is simply
an abstraction of the identity~\eqref{eq:forest-indicator-decomposition}
proven in Appendix~\ref{sec:forest-appendix}, provides a more convenient expression for the function
$\One_\mcal{F}$.
\begin{lemma}
\label{lem:abstract-forest-decomp}
For any graph $G\subset \binom{[m]}{2}$, the indicator function for
the minimal forest $F_G$ of $G$ being equal to $F$ has the representation
\[
\One(F_G=F) = \sum_{F'\supset F} \chi_{F,F'}(G),
\]
where the sum is over forests $F'$ containing $F$.  In
addition the functions  $\chi_{F,F'}$ have a
structure in terms of the connected components $P_{F'}$ of $F'$,
in the sense that
\[
\chi_{F,F'}(G) = \prod_{S\in P_{F'}} \chi_{F,F',S}(G_S).
\]
Here $G_S$ is the induced subgraph on the vertex set $S$.
Moreover, $|\chi_{F,F'}(G)|\leq 1$ for any $F,F',G$ and
$\supp \chi_{F,F'}(G) \subset \{G \mid G\supset F'\}$.
Finally, the function $\chi_{\noset,\noset}=1$ is the constant function.
\end{lemma}

We will apply Lemma~\ref{lem:abstract-forest-decomp} to write
$\One_\mcal{F} = \sum_{\mcal{F}'\supset \mcal{F}} \chi_{\mcal{F},\mcal{F}'}(\bm{\Gamma})$.
The only addiitonal observation we need is that the forests
$F_{E^{rec}(\bm{\Gamma})}$,
$F_{E^{tube}(\bm{\Gamma})}$,
$F_{E^{cone}(\bm{\Gamma})}$, and
$F_{C(\bm{\Gamma})}$ are determined only using edge indicator functions
$\One((a,b)\in E^{rec}(\bm{\Gamma}))$,
$\One((a,b)\in E^{tube}(\bm{\Gamma}))$,
$\One((a,b)\in E^{cone}(\bm{\Gamma}))$,
$\One((a,b)\in G(\bm{\Gamma}))$, which are each functions only of the
path variables $((p_{a-1},p_a,y_a),(p_{b-1},p_b,y_b))$.

\begin{lemma}
\label{lem:indicator-decomp}
The indicator function $\One_{\mcal{F}}(\bm{\Gamma})$ can be written as a sum
\[
\One_{\mcal{F}}(\bm{\Gamma}) = \sum_{\mcal{F}'\geq \mcal{F}} \chi_{\mcal{F},\mcal{F}'}(\bm{\Gamma}),
\]
where each function $\chi_{\mcal{F},\mcal{F}'}$ depends only on the variables in the
support $\supp\mcal{F}'$ and is bounded by the indicator function $\One_{\mcal{F}'}(\bm{\Gamma})$,
\[
\sup_{\bm{\Gamma}}|\chi_{\mcal{F},\mcal{F}'}(\bm{\Gamma})| \leq \One_{\mcal{F}'}(\bm{\Gamma}).
\]
Moreover, $\chi_{\mcal{F},\mcal{F}'}$ has the partial product structure
\begin{equation}
\label{eq:chi-split}
\begin{split}
\chi_{\mcal{F},\mcal{F}'}(\bm{\Gamma})
=
(-1)^{c(\mcal{F},\mcal{F}')}
&\chi^+_{\mcal{F},\mcal{F}'}(\Gamma^+)
\chi^-_{\mcal{F},\mcal{F}'}(\Gamma^-)
\chi^{clust}_{\mcal{F},\mcal{F}'}(\mbf{y}_{\supp (F_{C'}\cup F_{A'})}),
\end{split}
\end{equation}
where $\mbf{y}_{\supp (F_{C'}\cup F_{A'})}$ is the tuple of the values of $y_a$ for
$a\in \supp F_{C'}\cup F_{A'}$, and $\chi^{clust}_{\mcal{F}}$ obeys the bound
\[
\chi^{clust}_{\mcal{F},\mcal{F}'}(\mbf{y}_{\supp F_{C'}\cup F_{A'}})
\leq \prod_{a\sim_{F_{C'}} b} \One(|y_a-y_b|\leq 4Nr)
\prod_{a\sim_{F_{A'}} b} \One(|y_a-y_b|\leq 4Nr).
\]
\end{lemma}

To estimate the contribution of the diagrams we will need to express
$\chi_{\mcal{F},\mcal{F}'}$ as a mixture of
functions with a product structure.
The term that doesn't split is the function $\chi^{clust}_{\mcal{F},\mcal{F}'}$, which depends
on the variables $\mbf{y}_{\supp F_{C'}\cup \supp F_{A'}}$.  To keep track of the variables
affected by the cutoff $\chi^{clust}_{\mcal{F},\mcal{F}'}$ we introduce the sets
$J_{clust}^\pm(\mcal{F}')$
\[
J_{clust}^\pm(\mcal{F}') := (\supp F_{C'}\cup\supp F_{A'}) \cap \mbf{K}^\pm_{all}.
\]
Then, for a tuple of variables $Y^\pm = \{y_a\}_{a\in J_{clust}^\pm}$, which we consider
as a function $Y:J_{clust}^{pm}\to\Real^d$ and we write $\supp Y$ to mean the domain of this
function.  Then given any such labelled tuple of positions $Y$ we define
 we write
\begin{equation}
\label{eq:GY-def}
\chi^{\pm,Y}_{\mcal{F},\mcal{F}'}(\Gamma)
:= \chi^\pm_{\mcal{F},\mcal{F}'}(\Gamma)
\delta(\mbf{y}_{\supp Y}(\Gamma) - Y).
\end{equation}
This allows us to decompose $\chi_{\mcal{F},\mcal{F}'}(\bm{\Gamma})$ as an integral
of functions that do split as a product of functions of $\Gamma^+$ and $\Gamma^-$ alone:
\begin{equation}
\label{eq:split-G}
\chi_{\mcal{F},\mcal{F}'}(\bm{\Gamma})
= \int \chi_{\mcal{F},\mcal{F}'}^{+,Y^+}(\Gamma^+)
\chi_{\mcal{F},\mcal{F}'}^{-,Y^-}(\Gamma^-)
\chi^{clust}_{\mcal{F},\mcal{F}'}(Y^+,Y^-)\diff Y^+\diff Y^-,
\end{equation}
where the integral is over tuples $Y^+:J^+_{clust}\to\Real^d$
and $Y^-:J^-_{clust}\to\Real^d$.

\subsection{The partitions associated to a skeleton}
Using a skeleton $\mcal{F}$ we define two collections of partitions.  The first is the collection
of partitions that \emph{actually} arise as $P(\bm{\Gamma})$ for some $N^4$-complete extended
path $\bm{\Gamma}$.
\[
\mcal{Q}'_{\mcal{F}} = \{P(\bm{\Gamma}) \mid \mcal{F}(\bm{\Gamma})=\mcal{F}
\text{ and } \bm{\Gamma} \text{ is } N^4-\text{complete}\}.
\]
The second collection is of partitions that match the criteria of
Lemma~\ref{lem:interval-pair} and the constraint $P|_{\supp \mcal{F}} = P_{\mcal{F}}$.

To define this collection of partitions it is useful to define an abstraction of the maximal atypical
interval set $\mcal{I}(\bm{\Gamma})$ that does not depend on the path $\bm{\Gamma}$ directly.
Recall that $\mcal{I}(\bm{\Gamma})$ is the set of maximal atypical intervals of the form
$(a,b)_{K(\bm{\Gamma})}$, which are subsets of $K(\bm{\Gamma})$.  It is useful
to define an abstract interval collection which only contains the
information about the endpoints of the intervals, and is determined only by the set of
atypical collision indices.
\begin{definition}[Abstract interval collections]
Let $S\subset\mbf{K}_{all}$ be a set of collision indices.  We define the abstract interval
collection $\mcal{I}^{abs}(S)$ to be the set of all pairs
\[
\mcal{I}^{abs}(S) :=
\{(a,b) \in (S\cup \{0,\infty\})^2 \mid a<b\text{ and there does not exist }
c\in S\text{ such that } a<c<b\}.
\]
Given a collision index set $K\subset\mbf{K}_{all}$ containing $S$
 and an abstract interval $I=(a,b)\in\mcal{I}^{abs}(S)$, we define the realization
\[
I_K = (a,b)_K := \{c\in K \mid a<c<b\}.
\]
\end{definition}

We can now define the set of partitions canonically associated to a skeleton.
\begin{definition}[Canonical partition collection]
Given a skeleton $\mcal{F}$, we define the canonical collection of partitions
$\mcal{Q}_\mcal{F}$ to be any partition $P\in\mcal{P}(K)$ on a collision index set
$K$ containing $\supp\mcal{F}$ such that
\begin{itemize}
\item $P|_{\supp\mcal{F}} = P_{\mcal{F}}$.
\item For every $I\in\mcal{I}^{abs}(\mcal{F})$, either $P$ saturates $I_K$ and
$P|_{I_K}$ is a simple partition, or there exists $I'\in\mcal{I}^{abs}(\mcal{F})$
such that $\sign(I)=-\sign(I')$, $P$ saturates $I_K\cup I'_K$, and
the restriction $P|_{I_K\cup I'_K}$ is a generalized ladder.
\item If $(a,a+1)\in P$ is an immediate recollision, then $\ell(a)=\ell(a+1)$.
\end{itemize}
\end{definition}
The last condition comes from the fact that recollisions must occur within a time of
$N|p|^{-1}r\ll \eps^{-1.5}$, which is the guaranteed time of free evolution between segments.

A restatement of Lemma~\ref{lem:interval-pair} in terms of the collection $\mcal{Q}'_{\mcal{F}}$
and $\mcal{Q}_\mcal{F}$ is the statement
\begin{equation}
\label{eq:real-vs-abstract-Q}
\mcal{Q}_{\mcal{F}}' \subset \mcal{Q}_\mcal{F}.
\end{equation}
That is, every partition of the form $P(\bm{\Gamma})$ for an $N^4$-complete path
with skeleton $\mcal{F}(\bm{\Gamma})$ belongs to $\mcal{Q}_{\mcal{F}}$.

We conclude this section by showing that the expectation
\[
\Expec \prod_{a\in K(\bm{\Gamma})} \Ft{V_{y_a}}(q_a).
\]
splits as a sum over all partitions in $\mcal{Q}_{\mcal{F}}$.
\begin{lemma}
\label{lem:sum-over-partitions}
Let $\bm{\Gamma}$ be an $N^4$-complete path with skeleton $\mcal{F}$.  Then
\begin{equation}
\label{eq:sum-over-partitions}
\begin{split}
\Expec \prod_{a\in K(\bm{\Gamma})} \Ft{V_{y_a}}(q_a)
&= \sum_{\substack{P\in \mcal{Q'}_{\mcal{F}} \\ \supp P = K(\bm{\Gamma})}}
\prod_{S\in P} \Expec \prod_{a\in S}\Ft{V_{y_a}}(q_a) \\
&= \sum_{\substack{P\in \mcal{Q}_{\mcal{F}} \\ \supp P = K(\bm{\Gamma})}}
\prod_{S\in P} \Expec \prod_{a\in S}\Ft{V_{y_a}}(q_a).
\end{split}
\end{equation}
Note that in this sum, we only consider partitions of the index set $K(\bm{\Gamma})$, which
is indicated by the support constraint $\supp P = K(\bm{\Gamma})$.
\end{lemma}

The key observation is the following incompatibility of diagrams in $\mcal{Q}_{\mcal{F}}$.
\begin{lemma}
\label{lem:partition-incompatibility}
Let $P,P'\in\mcal{Q}_\mcal{F}$ be two partitions in $\mcal{Q}_{\mcal{F}}$
with $\supp P = \supp P'$.  Then either there exists sets $S\in P$ and $S'\in P'$ such that
$S\cap S'$ is a singleton set, or $P=P'$.
\end{lemma}
\begin{proof}
Suppose $P\not= P'$, so that there exist sets $S\in P$ and $S'\in P'$
such that $S\not= S'$ and $S\cap S'\not=\noset$.
If $|S|>2$, then $S$ is a cluster, and is therefore specified by $\mcal{F}(\bm{\Gamma})$.  Thus
it follows that $|S|=|S'|=2$.  Since $0<|S\cap S'| < |S|$, it follows that $|S\cap S'|=1$.
\end{proof}

\begin{proof}[Proof of Lemma~\ref{lem:sum-over-partitions}]
Since $\bm{\Gamma}$ is $N^4$-complete and
$\mcal{F}(\bm{\Gamma})=\mcal{F}$ it follows by~\eqref{eq:real-vs-abstract-Q}
that $P(\bm{\Gamma})\in\mcal{Q}_{\mcal{F}}$.
Thus for $P=P(\bm{\Gamma})$, we have
\[
\Expec \prod_{a\in K(\bm{\Gamma})} \Ft{V_{y_a}}(q_a)
= \prod_{S\in P} \Expec \prod_{a\in S}\Ft{V_{y_a}}(q_a).
\]
Now suppose that $P\not= P'\in\mcal{Q}_\mcal{F}$ and $\supp P = \supp P'$.  We will show that
\[
\prod_{S'\in P'} \Expec \prod_{a\in S'}\Ft{V_{y_a}}(q_a) = 0
\]
Indeed, by Lemma~\ref{lem:partition-incompatibility}, there is some $S'\in P'$ and $S\in P$
such that $S\cap S'=\{a_0\}$.  This collision site $a_0$ is separated from any other
collision in $S'$, so
\[
\Expec \prod_{a\in S'}\Ft{V_{y_a}}(q_a) =
\Expec \Ft{V_{y_{a_0}}}(q_{a_0})
\Expec \prod_{a\in S'\setminus\{a_0\}}\Ft{V_{y_a}}(q_a) = 0,
\]
since $\Expec V = 0$.
\end{proof}

\section{Colored operators}
\label{sec:colored-ops}
In this section we use combinatorial features of the partition set $\mcal{Q}_\mcal{F}$ to write
different expressions for~\eqref{eq:sum-over-partitions}.  The key idea is the notion of a
\textit{colored operator}, which is a path operator with a correlation structure specified by a coloring
function $\psi$.

To define a colored operator, let $\mcal{C}$ be a fixed countable collection of colors, and
let $\{V^{(c)}\}_{c\in\mcal{C}}$ be a collection of independently sampled potentials $V^{(c)}$,
and let $c_{imm}\in\mcal{C}$ be a special color reserved for immediate recollisions.
We write $\mcal{C}_{ext} = \mcal{C}\setminus \{c_{imm}\}$ for the remaining colors in $\mcal{C}$.

We say that a coloring $\psi:K\to\mcal{C}$
on a collision index set $K$ is \emph{valid} if $\psi^{-1}(c_{imm})$ can be partitioned into
consecutive pairs,
\[
\psi^{-1}(c_{imm}) = \{(a,a+_K 1)\}_{a\in A_{imm}(\psi)},
\]
where $A_{imm}(\psi)$ is the set of first elements of each such pair and $\ell(a)=\ell(a+1)$ for
each such pair (so that recollisions cannot cross segment boundaries).

Then, given a coloring function $\psi:[1,k]\to\mcal{C}$ and a path $\omega\in\Omega_k$,
we define the colored operator
\begin{equation}
\label{eq:colored-op-def}
O_\omega^\psi := \ket{p_0}\bra{p_k} e^{i\varphi(\omega)}
\prod_{a\in A_{imm}(\psi)}
\Expec \Ft{V_{y_a}}(p_a-p_{a-1}) \Ft{V_{y_{a+1}}}(p_{a+1}-p_a)
\Expec \prod_{\psi(j)\not= c_{imm}} \Ft{V_{y_j}^{(\psi(j))}}(p_j-p_{j-1}).
\end{equation}

We then extend this definition to extended paths $\Gamma$  and colorings $\psi:K(\Gamma)\to\mcal{C}$
by taking a product.  Given $\psi:K(\Gamma)\to\mcal{C}$,
let $\psi_\ell$ be the restriction of the coloring $\psi$ to the collisions with index
$(\ell,j)$ for some $j$.
\begin{equation}
\label{eq:extended-colored-op}
O_{\Gamma}^{\psi} =
\ket{\xi_0}\bra{\xi_{2N-1}}
\prod_{\ell=1}^N \braket{\xi_{2\ell-2}|O_{\omega_\ell}^{\psi_\ell}|\xi_{2\ell-1}}
\prod_{\ell=1}^{N-1} \braket{\xi_{2\ell-1}|\xi_{2\ell}}
\end{equation}

\subsection{Partitions from colorings}
In this section we relate expectations involving colored operators to expectations that are split
according to a partition.
Let $\psi^+:K^+\to\mcal{C}$ and $\psi^-:K^-\to\mcal{C}$ be colorings on collision index sets $K^+$
and $K^-$.  We then define the coloring $\psi^+\oplus \psi^-:K^+\sqcup K^-\to\mcal{C}$ to have
restriction $\psi^+$ on $K^+$ and $\psi^-$ on $K^-$.

\begin{definition}
\label{def:ptition-from-coloring}
Let $\psi^+$ and $\psi^-$ as above.  The partition $P(\psi^+,\psi^-)\in\mcal{P}(K^+\sqcup K^-)$
is defined so that $a\sim_{P(\psi^+,\psi^-)} b$ if either
$\psi(a)=\psi(b)\in\mcal{C}_{ext}$ or
$b=a+1$ and $a\in A_{imm}(\psi)$.

Given sets $\Psi^+$ and $\Psi^-$ of colorings, we write $\mcal{Q}(\Psi^+,\Psi^-)$ for the
set of all partitions $P(\psi^+,\psi^-)$ with $\psi^+\in\Psi^+$ and $\psi^-\in\Psi^-$,
excluding partitions containing singleton cells,
\[
\mcal{Q}(\Psi^+,\Psi^-) :=
\{P(\psi^+,\psi^-)\mid \psi^+\in\Psi^+,\psi^-\in\Psi^-,
\text{ and for every } S\in P(\psi^+,\psi^-), |S|\geq 2\}.
\]
For convenience we also write $\mcal{Q}(\Psi)=\mcal{Q}(\Psi,\Psi)$.
\end{definition}

The following identity is then clear from the definition that for deterministic operators $A$,
\[
\Expec (O^{\psi^+}_{\Gamma^+})^*
A
O^{\psi^-}_{\Gamma^-}
= \Expec_{P(\psi^+,\psi^-)} O_{\Gamma^+}^* AO_{\Gamma^-},
\]
where the latter expectation is split according to the partition $P(\psi^+,\psi^-)$.
The identity written above is only formal in the sense that $\Expec_P$ is not a well-defined
operation, but what is meant is that
\begin{align*}
\prod_{A_{imm}(\psi^+)\cup A_{imm}(\psi^-)}
&\Expec \Ft{V_{y_a}}(q_a)\Ft{V_{y_{a+1}}}(q_{a+1}) \\
&\times \Expec\Big[
\prod_{\substack{a\in K^+\\ \psi^+(a)\not=c_{imm}}}
\Ft{V_{y_a}^{(\psi^+(a))}}(q_a))
\prod_{\substack{a\in K^-\\ \psi^-(a)\not=c_{imm}}}
\Ft{V_{y_a}^{(\psi^-(a))}}(q_a)) \Big]
\\
&= \prod_{S\in P(\psi^+,\psi^-)}
\Expec \prod_{a\in S} \Ft{V_{y_a}}(q_a).
\end{align*}
Then, defining the colored operator $O^\Psi_\Gamma$ associated to a set
of colorings
\[
O^\Psi_\Gamma = \sum_{\psi\in\Psi} O^\psi_{\Gamma},
\]
we have
\[
\Expec (O^{\Psi^+}_{\Gamma^+})^* A O^{\Psi^-}_{\Gamma^-}
= \sum_{P\in \mcal{Q}(\Psi^+,\Psi^-)} \Expec_P (O_{\Gamma^+})^*A O_{\Gamma^-}.
\]

In Section~\ref{sec:coloring-machine}
we will construct sets of
of colorings $\Psi^+(\mcal{F})$ and $\Psi^-(\mcal{F})$
associated to a skeleton $\mcal{F}$ so that
$\mcal{Q}(\Psi^+(\mcal{F}),\Psi^-(\mcal{F})) = \mcal{Q}_\mcal{F}$.
In particular, this implies that for $N^4$-complete extended paths $\bm{\Gamma}$
with skeleton $\mcal{F}$, we have the formula
\begin{equation}
\label{eq:op-coloring-expec}
\Expec O_{\Gamma^+}^*A O_{\Gamma^-}
=
\Expec (O_{\Gamma^+}^{\Psi^+(\mcal{F})})^* A O_{\Gamma^-}^{\Psi^-(\mcal{F})}.
\end{equation}

\subsection{The product structure of $O_\Gamma$}
We observe that the path operator $O_\Gamma$ has the product structure
\[
O_\Gamma = O_{\Gamma,1} O_{\Gamma,2} \cdots O_{\Gamma,N}
\]
where
\[
O_{\Gamma,\ell} = \ket{\xi_{2\ell-2}}
\braket{\xi_{2\ell-2}|O_{\omega_\ell}|\xi_{2\ell-1}} \bra{\xi_{2\ell-1}}.
\]
Given an index $\ell$, we write $\Gamma[\ell]$ for the
tuple of variables $(\xi_{2\ell-2},O_{\omega_\ell},\xi_{2\ell-1})$.

The colored path operators have a similar structure
\[
O_{\Gamma}^\psi =  O_{\Gamma,1}^{\psi_1}\cdots O_{\Gamma,N}^{\psi_N}
\]
where
\[
O_{\Gamma[\ell]}^{\psi_\ell} = \ket{\xi_{2\ell-2}}
\braket{\xi_{2\ell-2}|O_{\omega_\ell}^{\psi_\ell}|\xi_{2\ell-1}} \bra{\xi_{2\ell-1}}.
\]

We introduce another notation to write these products using the tensor product space
\[
\mcal{H}_N := \bigotimes_{j=1}^N L^2(\Real^d).
\]
For each $\ell\in[N]$ we define the map
\[
\iota_\ell: \mcal{B}(L^2(\Real^d)) \to \mcal{B}(\mcal{H}_N)
\]
which sends an operator $A$ to an operator $\iota_\ell(A)$ acting locally at the $\ell$-th slot.
That is, we write
\[
\iota_\ell(A) := \Id\otimes\cdots\otimes\Id\otimes A\otimes \Id\otimes\cdots\otimes\Id
\]
where the $A$ appears in the $\ell$-th tensor product slot.  We also define the map
\[
\Mult: \mcal{B}(\mcal{H}_N) \to \mcal{B}(L^2(\Real^d))
\]
which is defined on simple tensors by the formula
\[
\Mult(A_1\otimes A_2\otimes\cdots\otimes A_N) = A_1A_2\cdots A_N,
\]
and which extends by linearity to the space $\mcal{B}(\mcal{H}_N)$.

Using this notation we can express the product structure
of $O_\Gamma$ for example by the identity
\[
O_\Gamma = \Mult (\prod_{\ell=1}^N \iota_\ell(O_{\Gamma,\ell})),
\]
with no need to worry about the ordering of the product since $\iota_\ell(A)$ commutes with
$\iota_{\ell'}(B)$ when $\ell\not=\ell'$.

%We can also use tensor products to state a stronger form of the identity~\eqref{eq:skeleton-ops-expec},
%which is that for extended paths $\bm{\Gamma}$ that are $N^4$-complete and have skeleton $\mcal{F}$,
%\[
%\Expec O_{\Gamma^+}^*\otimes O_{\Gamma^-}
%= \Expec O_{\mcal{F},+,\Gamma^+}^* \otimes O_{\mcal{F},-,\Gamma^-}.
%\]
%This identity in particular implies~\eqref{eq:skeleton-ops-expec} by applying the linear map
%$O_1\otimes O_2 \mapsto O_1 AO_2$ for a given deterministic operator $A$.

\subsection{Product structure for coloring sets}
\label{sec:product-colorings}
In this section we describe two kinds of product structure that can be used to simplify the operator
$O^\Psi_{\Gamma}$.

Let $Q\in\mcal{P}([N])$ be a partition of the segment index set $[N]$.  Let $(\psi_q)_{q\in Q}$ be a
tuple of local coloring functions $\psi_q:K_q\to \mcal{C}$ with domains $K_q\subset \ell^{-1}(q)$.  Such
a tuple defines a global coloring function $\bigoplus_q \psi_q :K\to\mcal{C}$ with $K =\bigcup K_q$, defined so that
$(\bigoplus_q \psi_q)|_{K_{q'}} = \psi_{q'}$.   Given sets of local colorings $\Psi_q$, we define the product
\[
\prod_{q\in Q} \Psi_q := \{ \bigoplus_{q\in Q} \psi_q \mid \psi_q\in \Psi_q\}.
\]
We say moreover that the tuple of coloring function sets $\{\Psi_q\}_{q\in Q}$ \emph{splits}
if there is a partition of the external collision colors $\mcal{C}_{ext} = \bigcup_{q\in Q} \mcal{C}_q$
into disjoint sets $\mcal{C}_q$ such that, for every $\psi_q\in\Psi_q$, $\Range(\psi_q)\subset \mcal{C}_q\cup\mcal{C}_{imm}$.

Given a local coloring function $\psi_q\in \Psi_q$, we define the local path operator $O^{\psi_q}_{\Gamma}$ by
\[
O^{\psi_q}_{\Gamma} = \prod_{\ell\in q} \iota_q(O^{\psi_\ell}_{\Gamma_\ell}).
\]
Then we set $O^{\Psi_q}_\Gamma = \sum_{\psi_q\in\Psi_q} O^{\psi_q}_{\Gamma}$.

The following observation connects the product structure of a coloring set $\Psi$ to the tensor product
structure of the operators $O^\Psi_{\Gamma}$.
\begin{lemma}
\label{lem:operator-product}
Let $Q\in\mcal{P}([N])$ be a partition of $[N]$, and let
Let $\Psi = \prod_{q\in Q}\Psi_q$ be a set  of coloring functions that is a product of sets of local coloring
functions $\Psi_q$.  Then the operator $O^\Psi_\Gamma$ has the product structure
\[
O^\Psi_\Gamma = \Mult\big( \prod_{q\in Q} O^{\Psi_q}_{\Gamma}\big).
\]
If moreover $\Psi$ splits with respect to $Q$, so that $\Range(\Psi_q)\cap\Range(\Psi_{q'})\subset \mcal{C}_{imm}$
for $q\not=q'$, then the operators $O^{\Psi_q}_{\Gamma}$ are independent.
\end{lemma}

The independence of the operators $O^{\Psi_q}_{\Gamma}$ is useful because it allows us to split up the expectation of
expressions involving $O^\Psi_\Gamma$.  In particular we will be computing $\Expec (O^\Psi_{\Gamma^+})^*O^\Psi_{\Gamma^-}$.
To write an expression for this operator we work on the squared Hilbert space $\mcal{H}_N\otimes\mcal{H}_N$
and introduce the squared multiplication map
\[
\Mult^\pm(A_1^+\otimes \cdots\otimes A_N^+ \otimes A_1^-\otimes \cdots A_N^-)
:= (A_N^+)^* (A_{N-1}^+)^*\cdots (A_1^+)^* A_1^- A_2^-\cdots A_N^-.
\]
With this definition, we have for any paths $\Gamma^+,\Gamma^-$ and any coloring functions $\psi^+,\psi^-$
the identity
\[
(O^{\psi^+}_{\Gamma^+})^* O^{\psi^-}_{\Gamma^-}
= \Mult^\pm \big( \prod_{\ell\in [N]}
\iota_\ell(O^{\psi^+_\ell}_{\Gamma^+_\ell})\otimes \iota_\ell(O^{\psi^-_\ell}_{\Gamma^-_\ell}) \big).
\]
Then, if $\Psi=\prod_{q\in Q}\Psi_q$ and $\Psi_q$ split $\mcal{C}_{ext}$, then
\[
\Expec (O^{\Psi}_{\Gamma^+})^* O^{\Psi}_{\Gamma^-}
= \Mult^\pm \big( \prod_{q\in Q} \Expec O^{\Psi_q}_{\Gamma^+}\otimes O^{\Psi_q}_{\Gamma^-} \big).
\]

\section{Constructing partitions from colorings}
\label{sec:coloring-machine}
In this section we provide the construction for the coloring set $\Psi^+(\mcal{F})$ and $\Psi^-(\mcal{F})$
such that $\mcal{Q}(\Psi^+(\mcal{F}),\Psi^-(\mcal{F})) = \mcal{Q}_{\mcal{F}}$ for a skeleton $\mcal{F}$.
We then show how to $\Psi^\pm(\mcal{F})$ into components that have a product structure and which split.
This construction will be used to simplify the calculation of the diagrammatic expansion.

\subsection{Coloring a ladder}
The building blocks of the partitions in $\mcal{Q}_F$
 are renormalized ladders and anti-ladders.  In this section
we demonstrate how to construct these partitions from colorings.  We will construct the set of colorings
by first assigning to each renormalized ladder a canonical coloring.

To define the colorings, we first investigate in some more detail the structure of a ladder partition.
Let $A$ and $B$ be disjoint ordered sets, and let $P\in\mcal{Q}_{rl}(A\sqcup B)$ be a renormalized
ladder partition.  The cells in $P$ are all pairs, and come two distinct types.  The first type is a
``rung''.  A rung is a pair $(a,b)\in A\times B$.  Let $\Rung(P)\subset P$ be the set of rungs in $P$.
The second type is an immediate recollision of the form $(a,a+_A1)$ or $(b,b+_B1)$.  Let
$\Imm(P)$ be the set of immediate recollisions.

We place an ordering on the set of rungs
in which $(a,b)\leq (a',b')$ if $a\leq a'$.
There is a unique order-preserving bijection
$\varphi_{\Rung(P)}:\Rung(P)\to [|\Rung(P)|]$ which enumerates the runs.

We are ready to define the canonical coloring of $A\sqcup B$ associated to the renormalized ladder $P$.
The color palette $\mcal{C}$ is given by
\[
\mcal{C} = \mcal{C}_{ext}\cup\{c_{imm}\} = \{\qt{\ext}\}\times \bbN \cup \{\qt{\Imrec}\}.
\]
For $x\in A\cup B$, we define the canonical coloring
\begin{equation}
\label{eq:psiP-def}
\psi_{P}(x)
=
\begin{cases}
(\qt{\ext},\varphi_{\Rung(P)}(a,b)), & x\in (a,b)\in\Rung(P) \\
\qt{\Imrec}
& x\in \supp \Imm(P).
\end{cases}
\end{equation}
We note that this construction works without modification if $P$ is a renormalized anti-ladder partition.
By construction, the coloring function $\psi_P$ recovers the ladder partition $P$
in the following way.
\begin{lemma}
\label{lem:get-ladder-back}
Let $P\in\mcal{Q}_{rl}(A\sqcup B)$ be a renormalized (anti-)ladder on $A\sqcup B$, and let $\psi_P$ be
the coloring defined in~\eqref{eq:psiP-def}.  Let $\psi_A = (\psi_P)|_A$ and $\psi_B = (\psi_P)|_B$.  Then
$P=P(\psi_A,\psi_B)$, where $P(\psi_A,\psi_B)$ is defined by Definition~\ref{def:ptition-from-coloring}.
\end{lemma}

We define $\Psi^+_{rl}(A)$ to be the set of all coloring functions of the form $(\psi_P)|_A$
for some $P\in \mcal{Q}_{rl}(A,B)$ where $B$ is any finite set.  That is,
\[
\Psi^+_{rl}(A) = \{(\psi_P)|_A \mid P \in \mcal{Q}_{rl}(A,B)\text{ for some finite set } B\}.
\]
We likewise define $\Psi^-_{rl}(A)$ to be the half-coloring functions generated from \emph{anti-ladders}.
If $P$ is an antiladder on $A\sqcup B$, then in fact $(\psi_P)|_A\in\Psi^+_{rl}(A)$ (since the numbering
of the rungs is still increasing on $A$), whereas $(\psi_P)|_B$ has decreasing rung indices.  Therefore we
define
\[
\Psi^-_{rl}(A) = \{(\psi_P)|_A \mid P\in\mcal{Q}^-_{rl}(B,A) \text{ for some finite set }B\}.
\]
%Finally, we set of ``simple'' colorings that come from ladders with no rungs.
%\[
%\Psi^0(A) = \{(\psi_P)_A \mid P\in \mcal{Q}_{rl}(A,\noset)\}.
%\]

We will need to slightly modify the definitions of $\Psi^{\pm}_{rl}(A)$ to allow for labels other
than $\qt{\ext}$ in the color palette $\mcal{C}$.  Let $\lambda_{ext}$ be such a label,
and define the relabelled color palette
\[
\mcal{C}(\lambda_{ext}) := \{\lambda_{ext}\}\times\bbN \cup\{\qt{\Imrec}\}.
\]
Then we set $\Psi^\pm_{rl}(A;\lambda_{ext})$ to be the analogues of $\Psi^\pm_{rl}(A)$
on the color palette $\mcal{C}(\lambda_{ext})$.

As a consequence of Lemma~\ref{lem:get-ladder-back} we see that the coloring
sets $\Psi^\pm_{rl}(A)$ can be used to reconstruct the generalized ladder partitions.
\begin{lemma}
Let $A$ and $B$ be finite sets.  Then the following hold:
\begin{align*}
\mcal{Q}_{rl}(A,B) &= \mcal{Q}(\Psi^+_{rl}(A),\Psi^+_{rl}(B))
= \mcal{Q}(\Psi^-_{rl}(A),\Psi^-_{rl}(B)) \\
\mcal{Q}^-_{rl}(A,B) &= \mcal{Q}(\Psi^+_{rl}(A),\Psi^-_{rl}(B)).
\end{align*}
\end{lemma}

Now we define a partition of the ladder colorings according to the number of rungs.
Set $\Psi^{\pm,h}_{rl}(A)\subset \Psi^{\pm}_{rl}(A)$ to be the colorings with exactly $h$ rungs,
that is,
\[
\Psi^{\pm,h}_{rl}(A) =
\{\psi\in \Psi^\pm_{rl}(A) \mid \psi^{-1}((\qt{ext},h)) \not=\noset
\text{ and } \psi^{-1}((\qt{ext},h+1)) = \noset\}.
\]
Note that if $h\not=h'$ and $\psi\in \Psi^{\pm,h}_{rl}(A)$ and
$\psi'\in \Psi^{\pm,h}_{rl}(B)$, then $P(\psi,\psi')$ has a singleton set.

\subsection{Coloring a skeleton}
Now we use the construction from the previous section to define  colorings $\Psi^+(\mcal{F};K^+)$
and $\Psi^-(\mcal{F};K^-)$  that assemble the partitions of $\mcal{Q}_{\mcal{F}}$.
Given a skeleton $\mcal{F}$, define the color palette
\[
\mcal{C}(\mcal{F}) = P_\mcal{F}\,
\,\cup\, \mcal{I}^{abs}(\mcal{F})^2 \times \bbN\,\cup\,
\{\qt{\Imrec}\}
\]
Roughly speaking, the first set of colors $P_{\mcal{F}}$ will be used to color the collisions
in $\supp\mcal{F}$, while the second set of colors $\mcal{I}^{abs}(\mcal{F})^2\times\bbN$ are used to color
the rungs of ladders, and $\qt{\Imrec}$ is used for immediate recollisions.

\begin{definition}[Colorings of $K^+$]
Let $\mcal{F}$ be a skeleton and $K^+$ be a collision set containing $\supp \mcal{F} \cap K^+_{all}$.
The coloring collection $\Psi^+(\mcal{F};K^+)$ is the collection of functions
$\psi:K^+\to\mcal{C}(\mcal{F})$ such that the following holds for some
partial matching $M$ of $\mcal{I}^{abs}(\mcal{F})$:
\begin{itemize}
\item For $a\in\supp\mcal{F}$, $\psi(a) = S\in P_{\mcal{F}}$ is the cell $S$ containing the
index $a$.
\item If $I\in\mcal{I}^{abs}(\mcal{F})$ with $\sign(I)=1$ is an abstract interval
that is not matched by $M$, $I\not\in \supp M$, then $\psi^+(a)=\qt{\Imrec}$ for every
$a\in I_K$.
\item If $(I,I')\in M$ is a pair in $M$ and $\sign(I)=+1$, then
$\psi^+|_{I_{K^+}} \in \Psi^+(I_{K^+};\{I,I'\})$.
\item For every $\ell\in [N]$, the number of collisions $a$ satisfying $\ell(a)=\ell$
and $\psi(a)=\qt{\Imrec}$ is even.
\end{itemize}
We then write $\Psi^+(\mcal{F}) = \bigcup_{K^+} \Psi^+(\mcal{F};K^+)$, where the union is
over all collision sets $K^+$.
\end{definition}

The colorings $\Psi^-(\mcal{F};K^-)$ are defined in a very similar way, but with the possibility
that an interval is colored as an anti-ladder.

\begin{definition}[Colorings of $K^-$]
Let $\mcal{F}$ be a skeleton and $K^-$ be a collision set containing $\supp \mcal{F} \cap K^-_{all}$.
The coloring collection $\Psi^-(\mcal{F};K^-)$ is the collection of functions
$\psi:K^-\to\mcal{C}(\mcal{F})$ such that the following holds for some
partial matching $M$ of $\mcal{I}^{abs}(\mcal{F})$ and some
orientation function $\aleph:M\to \{\pm\}$
\begin{itemize}
\item For $a\in\supp\mcal{F}$, $\psi(a) = S\in P_{\mcal{F}}$ is the cell $S$ containing the
index $a$.
\item If $I\in\mcal{I}^{abs}(\mcal{F})$ with $\sign(I)=1$ is an abstract interval
that is not matched by $M$, $I\not\in \supp M$, then $\psi^+(a)=\qt{\Imrec}$ for every
$a\in I_K$.
\item If $(I,I')\in M$ is a pair in $M$ and $\sign(I')=-1$, then
$\psi^+|_{I_{K^+}} \in \Psi^{\aleph(I,I')}(I_{K^-};I,(I,I'))$.
\item For every $\ell\in [N]$, the number of collisions $a$ satisfying $\ell(a)=\ell$
and $\psi(a)=\qt{\Imrec}$ is even.
\end{itemize}
We then write $\Psi^-(\mcal{F}) = \bigcup_{K^-} \Psi^-(\mcal{F};K^-)$, where the union is
over all collision sets $K^-$.
\end{definition}

The following lemma shows that the colorings $\Psi^+(\mcal{F})$ and $\Psi^-(\mcal{F})$ reconstruct
the family of partitions $\mcal{Q}_{\mcal{F}}$.
%Given $\psi^+:K^+\to\mcal{C}(\mcal{F})$
%and $\psi^-:K^-\to\mcal{C}(\mcal{F})$ we define the partition $P(\psi^+,\psi^-)\in\mcal{P}(K^+\sqcup K^-)$
%to be the partition induced by the coloring $\psi^+\oplus\psi^-:K^+\sqcup K^-\to\mcal{C}(\mcal{F})$.

\begin{lemma}
\label{lem:colors-to-partitions}
The coloring sets $\Psi^{\pm}(\mcal{F})$ assemble $\mcal{Q}_{\mcal{F}}$ in the sense that
\[
\mcal{Q}(\Psi^+(\mcal{F}),\Psi^-(\mcal{F})) = \mcal{Q}_\mcal{F}.
\]
\end{lemma}

In the remainder of the section we will partition $\Psi^+(\mcal{F})$ and $\Psi^-(\mcal{F})$ into
pieces that have a product structure and which split in the way described in Section~\ref{sec:product-colorings}.

\subsection{Partitioning $\Psi^{\pm}(\mcal{F})$ into sets with product structure}
There are two reasons that $\Psi^\pm(\mcal{F})$ cannot not have a product structure with respect
to a partition $Q\in\mcal{P}([N])$.  The first is that there are global correlations across segment
intervals coming from the matching $M$.  The second is that the rung counts within a ladder communicate
across segments.  We will partition $\Psi^\pm(\mcal{F})$ into sets with fixed matchings $M$ and fixed rung
counts so that within certain intervals of $[N]$ we recover a product structure.  The partition we
choose also has some extra information that will be useful in the diarammatic estimates.

We start by partitioning the $\Psi^{\pm}(\mcal{F})$ according to the matching $M$.
We define the matching $M$ associated to a coloring $\psi$ to be the set of all
pairs of abstract intervals $(I,I')\in\mcal{I}^{abs}(\mcal{F})$ which share at least one rung:
\[
M_\psi := \{(I,I')\in\mcal{I}^{abs}(\mcal{F}) \mid \psi^{-1}((\{I,I'\},1)) \not=\noset\}.
\]
Given $(I,I')\in M_\psi$ we then define the orientation
\[
\aleph_\psi(I,I') :=
\begin{cases}
-,
\psi^{-1}((\{I,I'\}),2)\cap K^+ < \psi^{-1}((\{I,I'\}),1)\cap K^+ \\
+, \text{ else}.
\end{cases}
\]
That is, $\aleph_\psi(I,I')=-1$ if $(I,I')$ has at least two rungs and they are in reverse order, and is
$+1$ otherwise.

We partition $\Psi^-(\mcal{F})$ according to the values of $M$ and $\psi$,
\[
\Psi^-(\mcal{F};M,\aleph) := \{\psi \in \Psi^-(\mcal{F}) \mid M_\psi = M, \aleph_\psi = \aleph\}.
\]
This yields a partition of $\Psi^-(\mcal{F})$ into sets $\Psi^-(\mcal{F};M,\aleph)$ into sets which
share the same global properties.

We further subdivide $\Psi^-(\mcal{F};M,\aleph)$  according to the rung counts at specified collisions.
Given a coloring $\psi$ we define the rung count function $\ovln{\psi}$ by
\[
\ovln{\psi}(a)
= \begin{cases}
j, & \psi(a) = (\{I,I'\},j)\text{ for some }(I,I')\in M \\
\noset, & \text{ else.}
\end{cases}
\]
We also define the segment rung count function $m_\psi(a)$ to be the number of rungs within the segment
(and interval ) containing $a$,
\[
m_\psi(a) := \#\{ b \in I_K\cup I'_K \mid \psi(b) = (\{I,I'\},j) \text{ and }\ell(b)=\ell(a)\}.
\]
Finally we define the neighbors $\next_\psi(a)$ and $\back_\psi(a)$
of an index $a$.  These are the next rung after the index $a$ and the previous rung before index $a$,
respectively,
\begin{align*}
\next_\psi(a) &= \min\{a' \mid a'>a\text{ and } \psi(a') \in M\times \bbN\} \\
\back_\psi(a) &= \max\{a' \mid a'<a \text{ and }\psi(a') \in M\times \bbN\}.
\end{align*}
In the above definitions we use the convention $\min\noset = \infty$ and $\max\noset=0$,
so that $\next_\psi(a) = \infty$ if there are no rungs after $a$, and $\back_\psi=0$ if there are
no rungs preceding $a$.

Given an index $a$, we then define the neighborhood tuple
\[
\Nbd_\psi(a) :=
(\ovln{\psi}(a), \next_\psi(a), \ovln{\psi}(\next_\psi(a)), m_\psi(\next_\psi(a)),
\back_\psi(a), \ovln{\psi}(\back_\psi(a)), m_\psi(\back_\psi(a)))
\]
Then given a set $S$ we define $\Nbd_\psi(S) = (\Nbd_\psi(a))_{a\in S}$.
Given a set $S$ and a neighborhood tuple $\Nbd(S)$ of a set we define the coloring sets
\[
\Psi^-(\mcal{F};M,\aleph,\Nbd(S))
:= \{\psi\in\Psi^-(\mcal{F};M,\aleph) \mid \Nbd_\psi(S) = \Nbd(S)\}.
\]
These coloring sets clearly partition $\Psi^-(\mcal{F})$, and we will next show that they
have a product structure.   We call the tuples $(\mcal{F},M,\aleph,\Nbd(S))$
\emph{scaffolds}.\footnote{A scaffold is a structure that can support many ladders and sticks.}
The specific choice we make for the set $S$ is to take $S=\supp^-\mcal{F}'$ for
some skeleton $\mcal{F}'\supset \mcal{F}$.  We write
$\mcal{S}^\pm(\mcal{F},\mcal{F}')$ for the set of scaffolds of the form
$\Scaff = (\mcal{F},M,\aleph,\Nbd(\supp^\pm \mcal{F}'))$.  Then given any
skeletons $\mcal{F}$ and $\mcal{F}'\supset \mcal{F}$, we have the partition
\[
\Psi^\pm(\mcal{F})
= \bigcup_{\mcal{S}^\pm(\mcal{F},\mcal{F}')} \Psi^\pm(\Scaff).
\]
We observe that there are not \emph{too many} scaffolds associated to a pair of skeletons $\mcal{F}\subset \mcal{F}'$,
\begin{equation}
\label{eq:scaffold-count}
\#\mcal{S}^\pm(\mcal{F},\mcal{F}') \leq N^{C\|\mcal{F}'\|}.
\end{equation}

We will show in the next section that $\Psi^\pm(\Scaff)$ has a product structure as described in
Section~\ref{sec:product-colorings}.

\subsubsection{Product structure of $\Psi^-(\Scaff)$}
\label{sec:scaff-product}
For this section we fix a scaffold $\Scaff=(\mcal{F};M,\aleph,\Nbd(S))$ and demonstrate
that $\Psi^\pm(\Scaff)$ has a product structure.  Again we focus for concreteness on the set
$\Psi^-(\Scaff)$ which is slightly more complicated than $\Psi^+(\Scaff)$.
We define the \emph{base} of the scaffold $\Scaff$ to be the set $S = \supp^-\mcal{F}'$ along with the neighboring rungs specified
by $\Nbd(S)$,
\[
\Base(\Scaff) = S \cup \{\next(a)\}_{a\in S} \cup \{\back(a)\}_{a\in S}.
\]
We will write simply $\Base$ instead of $\Base(\Scaff)$ in this section, as we have fixed for now the scaffold
$\Scaff$.
Then define
$\mcal{I}^{abs}(\Base)$ to be the abstract interval collection with endpoints
$\Base$.   Since $S\subset\supp\mcal{F}$, for each $J=\bbk{a,b}\in\mcal{I}^{abs}(\Base)$
there exists a unique $I=\bbk{a',b'}\in\mcal{I}^{abs}(\mcal{F})$ such that
$a'\leq a$ and $b\leq b'$.  In particular, for any collision set $K$ and
$J\in\mcal{I}^{abs}(\Base)$, $J_K \subset I_K$ for some $I\in\mcal{I}^{abs}(\mcal{F})$

We say that $J\in\mcal{I}^{abs}(\Base)$ is
\textit{ladder-like} if $J=\bbk{a,b}$ with $a,b\not\in S$.  Otherwise we say that
$J\in\mcal{I}^{abs}(\Base)$ is a \textit{stick} if $J=\bbk{a,b}$ and $a\in S$ or $b\in S$.
Note that if $J$ is ladder-like then $J = \bbk{\next(s),\back(s')}$ for some $s<s'\in S$
(since every element in $\Base\setminus S$ is of the form $\next(s)$ or $\back(s)$ for some $s\in S$).

We use the intervals $\mcal{I}^{abs}(\Base)$ to partition the set $[N]$ as follows.
To each $J = \bbk{a,b}\in\mcal{I}^{abs}(\Base)$ we define the set
\[
\ell(J) := \bbk{\ell(a),\ell(b)} = [\ell(a)+1,\ell(b)-1]\subset [N].
\]
Note that $\ell(J)$ may be empty if $\ell(b)\leq \ell(a)+1$.

The intervals $\ell(J)$ can be used to form the partition $Q\in\mcal{P}([N])$
\[
Q(\Base) = \{\ell(\Base)\} \cup \{\ell(J)\}_{J\in\mcal{I}^{abs}(\Base)}.
\]
Given a realization of a collision index set $K$, we define the partition
$\ell^{-1} Q(\Base)\in\mcal{P}(K)$ by
\[
\ell^{-1}Q(\Base) := \{K\cap\ell^{-1}(\Base) ,\{K \cap \ell^{-1}(\ell(J))\}_{J\in\pi(\mcal{I}^{abs}(\Base))}\}
\]
The first set in the partition is called the ``junk'' set $\Junk(\Base,K)$, and the latter
sets are written
$J^\circ_K$, so
\[
J^\circ_K = \ell^{-1}\ell(J)\cap K = \{a\in K \mid \ell(a)\in \ell(K), a\in J_K\}.
\]

Now we are ready to show that the sets $\Psi^-(\mcal{F};M,\aleph,\Nbd(S))$ have a product structure
with respect to the partition $Q$.  To write the product structure of
$\Psi^\pm(\Scaff)$, define for each $q\in Q(\Base)$ the set of coloring functions
\begin{equation}
\label{eq:Psiq-def}
\Psi^\pm_q(\Scaff) :=
\{\psi|_{\ell^{-1}(q)} \mid \psi\in \Psi^\pm(\Scaff)\}.
\end{equation}
That is, $\Psi^\pm_q$ consists of the restrictions of the colorings $\psi$ in $\Psi^\pm$ to collisions
belonging to the interval $q$.

The main result is the following.

\begin{lemma}
\label{lem:coloring-product}
Let $\Scaff\in\mcal{S}(\mcal{F},\mcal{F}')$ for a pair of skeletons
$\mcal{F}\subset\mcal{F}'$.
Then, with $\Psi_q^-(\Scaff)$ defined by~\eqref{eq:Psiq-def},
the coloring set
$\Psi^-(\Scaff)$ has the product structure
\begin{equation}
\label{eq:Psi-product}
\Psi^-(\Scaff)
= \prod_{q\in Q(\Base)} \Psi_q^-(\Scaff).
\end{equation}
Moreover, letting $\mcal{C}_q = \bigcup_{\psi_q\in\Psi^-_q} \Range(\psi_q)$, we have
for $q\not= q'$ the inclusion $\mcal{C}_q\cap \mcal{C}_{q'} \subset \mcal{C}_{imm}$.
\end{lemma}

To prove Lemma~\ref{lem:coloring-product} we first
make a simple observation concerning the structure of colorings in
$\Psi^-(\Scaff)$.
\begin{lemma}
\label{lem:interval-color-structure}
Let $J\in\mcal{I}^{abs}(\Base)$ be an abstract interval, and let
$\psi\in\Psi^-(\mcal{F};M,\aleph,\Nbd(S))$ be a coloring on collision set $K$.
Let $I\in\mcal{I}^{abs}(\mcal{F})$ be the unique interval satisfying $J_K \subset I_K$.
Then one of the following holds:
\begin{itemize}
\item The interval $J$ is a stick, and in this case $\psi(J^\circ_K) \subset \mcal{C}_{imm}$,
or in other words $\psi|_{J^\circ_K}\in \Psi^{+,0}_{rl}(J^\circ_K)$.
\item The interval $J$ is ladder-like, so that $\{I,I'\}\in M$ for some $I'\in\mcal{I}^{abs}(\mcal{F})$
and $J=(\next(s),\back(s'))$ for some $(s,s')\in S$.  Now if $\aleph(\{I,I'\}) = +1$ then
\begin{equation}
\label{eq:psi-range}
\psi(J^\circ_K) \setminus \mcal{C}_{imm}
= \{\{I,I'\}\}\times \{\ovln{\psi}(\next(s)) + m(\next(s)) \leq j \leq  \ovln{\psi}(\back(s')) - m(\back(s')) \}.
\end{equation}
Otherwise if $\aleph(\{I,I'\})=-1$ then
\[
\psi(J^\circ_K) \setminus \mcal{C}_{imm}
= \{\{I,I'\}\} \times
\{ \ovln{\psi}(\back(s')) + m(\back(s')) \leq j \leq  \ovln{\psi}(\next(s)) - m(\next(s)) \}.
\]
\end{itemize}
\end{lemma}
\begin{proof}
If $J$ is a stick, then $J = (s,\next(s))$ or $J=(\back(s),s)$ for some $s\in S$.  In particular,
by the definition of the functions $\next$ and $\back$ we have $\psi(J_K) \subset \mcal{C}_{imm}$.
Since $J^\circ_K\subset J_K$ this proves the first claim.

We prove the second claim when $\aleph=+1$, as the case $\aleph=-1$ is similar.  In this
case, $J=(\next(s),\back(s'))$, so we have the following identity for the rung count function
$\ovln{\psi}$ on the interval $J_K$:
\[
\ovln{\psi}(J_K) = ]\ovln{\psi}(\next(s)), \ovln{\psi}(\back(s))[.
\]
The claim then follows from observing for example that
$\ovln{\psi}(a) - \ovln{\psi}\next(s)\in [0,m(\next(s)]$ implies $\ell(a)=\ell(\next(s))\not\in\ell(J)$,
whereas $\ovln{\psi(a)}-\ovln{\psi}(\next(s)) > m(\next(s)$ implies $\ell(a)>\ell(\next(s))$.
\end{proof}

We are now ready to prove Lemma~\ref{lem:coloring-product}.
\begin{proof}[Proof of Lemma~\ref{lem:coloring-product}]
To prove the product structure of $\Psi^-(\Scaff)$ it suffices to show
the following: for any $\psi,\psi'\in\Psi^-(\Scaff)$ and any $J\in \mcal{I}^{abs}(\Base)$,
the coloring $\varphi$ defined by $\varphi(a)=\psi(a)$ for $a\in J^\circ_K$ and $\varphi(a)=\psi'(a)$
otherwise is also a valid coloring in $\Psi^-(\Scaff)$.  This is straightforward, and the
key idea is that the number of rungs in any $J^\circ_K$ is fixed, so swapping the values of $\psi$ for $\psi'$
on $J^\circ_K$ does not affect $\Nbd_\psi(S)$.

On the other hand, the fact that $\Psi^-(\mcal{F};M,\aleph,\Nbd(S))$ splits according to the partition $Q(\Base)$
follows from Lemma~\ref{lem:interval-color-structure} and observing that the sets described
in~\eqref{eq:psi-range} are disjoint.
\end{proof}

The final result we need in this section is a characterization of the partition
sets $\mcal{Q}(\Psi^-_q(\Scaff))$, which is a simple consequence
of Lemma~\ref{lem:interval-color-structure}.
\begin{lemma}
Let $(M,\aleph,\Nbd)$ as above, and let $J\in\mcal{I}^{abs}(\Base(\Nbd))$ be an abstract
interval, and let $q=\ell(J)\in Q(\Base)$.  If $J$ is a stick interval,
then the partition collection $\mcal{Q}(\Psi^-_q(\Scaff))$
consists of simple partitions.  On the other hand if $J=]\next(s),\back(s')[$ is a ladder
interval, then $\mcal{Q}(\Psi^-_q(\Scaff))$ consists of generalized ladder partitions
with exactly
\[
\ovln{\psi}(\back(s')) - \ovln{\psi}(\next(s)) + 1 - m(\back(s'))-m(\next(s))
\]
rungs.
\end{lemma}

\section{The diagrammatic expansion}
\label{sec:diagrams}
In this section we start to derive our diagrammatic decomposition of
the evolution channel $\wtild{\Evol}_{N\tau}[A]$.
We start with the expression
\begin{align*}
\wtild{\Evol}_{N\tau}[A] = \int_{\Omega_{ext}^N\times\Omega_{ext}^N}
\Xi(\bm{\Gamma})
\One(\Gamma^+\sim\Gamma^-)
 \Expec O_{\Gamma^+}^*AO_{\Gamma^-}\diff\Gamma^+\diff\Gamma^-,
\end{align*}
First we decompose the integral according to the skeleton of the path $\bm{\Gamma}$.
We first rewrite the expectation using the skeleton operators
using~\eqref{eq:op-coloring-expec}, then we apply Lemma~\ref{lem:indicator-decomp}
to rewrite the indicator function for a skeleton:
\begin{align*}
\Evol_{N\tau}[A]
&= \sum_{\mcal{F}} \iint \Xi(\bm{\Gamma})
\One(\Gamma^+\sim\Gamma^-)
\One_{\mcal{F}}(\bm{\Gamma})
\Expec (O_{\Gamma^+}^{\Psi^+(\mcal{F})})^*A O_{\Gamma^-}^{\Psi^-(\mcal{F})} \diff\bm{\Gamma} \\
&= \sum_{\mcal{F}\leq \mcal{F}'}
\iint \Xi(\bm{\Gamma}) \chi_{\mcal{F},\mcal{F}'}(\bm{\Gamma})
\Expec (O_{\Gamma^+}^{\Psi^+(\mcal{F})})^*A O_{\Gamma^-}^{\Psi^-(\mcal{F})} \diff\bm{\Gamma}.
\end{align*}
We define $\Evol_{\mcal{F},\mcal{F}'}$ to be the superoperator appearing on the right hand side,
\begin{equation*}
\Evol_{\mcal{F},\mcal{F}'}[A]
:= \iint
\Xi(\bm{\Gamma}) \One_{\mcal{F}}(\bm{\Gamma})
\Expec (O_{\Gamma^+}^{\Psi^+(\mcal{F})})^*A O_{\Gamma^-}^{\Psi^-(\mcal{F})} \diff\bm{\Gamma}.
\end{equation*}
Then we rewrite our decomposition of the evolution channel as
\begin{equation}
\label{eq:skeleton-expansion}
\wtild{\Evol}_{N\tau}[A]
= \sum_{\mcal{F}\leq\mcal{F}'}\Evol_{\mcal{F},\mcal{F'}}[A].
\end{equation}
The main term comes from the pair of empty skeletons, $\mcal{F}=\mcal{F}'=\bm{0}$.
In this case the expectation only includes ladder partitions, and moreover
$\chi_{\bm{0},\bm{0}}=1$, so the main term is in fact the ladder superoperator
\begin{equation}
\begin{split}
\Evol_{\bm{0},\bm{0}}[A]
&= \iint \Xi(\bm{\Gamma}) \Expec (O_{\Gamma^+}^{lad})^*A (O_{\Gamma^-}^{lad})\diff\bm{\Gamma} \\
&= \mcal{L}_{N\tau}[A].
\end{split}
\end{equation}

Now we deal with the terms with $\mcal{F}'\not=\bm{0}$.  For any tuple
$Y:J^\pm_{clust}\to \Real^d$ and any scaffold $\Scaff \in\mcal{S}^\pm(\mcal{F},\mcal{F}')$ we
define the operator
\begin{equation}
\label{eq:localized-diagram-operator}
O_{\pm,\Scaff,Y} :=
\int \chi_{\mcal{F},\mcal{F}'}^{\pm,Y}(\Gamma) \Xi(\Gamma) O_{\Gamma}^{\Psi^\pm(\Scaff)} \diff \Gamma
\end{equation}
usin the function $G_{\mcal{F},\mcal{F}'}^{\pm,Y}$ defined in~\eqref{eq:GY-def}.
Then using~\eqref{eq:split-G} and the decomposition
\[
\Psi^\pm(\mcal{F}) = \bigcup_{\Scaff\in\mcal{S}^\pm(\mcal{F},\mcal{F}')} \Psi(\Scaff),
\]
we have the identity
\begin{equation}
\Evol_{\mcal{F},\mcal{F}'}[A]
=
\sum_{\Scaff^+\in\mcal{S}^+(\mcal{F},\mcal{F}')}
\sum_{\Scaff^-\in\mcal{S}^-(\mcal{F},\mcal{F}')}
 \iint \chi^{clust}_{\mcal{F},\mcal{F}'}(Y^+,Y^-)
\Expec O_{+,\Scaff^+,Y^+}^* A O_{-,\Scaff^-,Y^-}
\diff Y^+ \diff Y^-,
\end{equation}
where the integral is over tuples $Y^+:J^+_{clust}\to\Real^d$
and $Y^-:J^-_{clust}\to\Real^d$.

Now we use the following version of the Cauchy-Schwarz inequality, which
is that for any measure $\mu$ on a measure space $\Omega$, if
$X,Y:\Omega\to\mcal{B}(\mcal{H})$
are operator-valued functions on a Hilbert space $\mcal{H}$ then
\[
\big\|\int X(\theta)^*A Y(\theta)\diff\mu(\theta)\big\|_{op}
\leq \|A\|_{op}
\big\|\int X(\theta)^*X(\theta)\diff\mu(\theta)\big\|_{op}^{1/2}
\big\|\int Y(\theta)^*Y(\theta)\diff\mu(\theta)\big\|_{op}^{1/2}.
\]
We use this inequality with the measure being the product of the counting measure on
$\mcal{S}^+(\mcal{F},\mcal{F}')$, the counting measure on $\mcal{S}^-(\mcal{F},\mcal{F}')$,
the Lebesgue measure on the anchor points $Y$, and the random measure on the potentials
$V^{(c)}$ to obtain
\begin{equation}
\label{eq:CS-bd}
\begin{split}
\Big\|\Evol_{\mcal{F},\mcal{F}'}[A]\Big\|_{op}
\leq
CN^{C\|\mcal{F}'\|}
\|A\|_{op}
&\max_{\Scaff^+\in\mcal{S}^+(\mcal{F},\mcal{F}')}
\Big\|
\iint
\chi^{clust}_{\mcal{F},\mcal{F}'}(Y)
\Expec
%TODO: Here is where the momentum localization (lower bound) would go
O_{+,\Scaff^+,Y^+}^*
O_{+,\Scaff^+,Y^+}
\diff Y^+\diff Y^-
\Big\|_{op}^{1/2} \\
&\times
\max_{\Scaff^-\in\mcal{S}^-(\mcal{F},\mcal{F}')}
\Big\|
\iint \chi^{clust}_{\mcal{F},\mcal{F}'}(Y)
\Expec O_{-,\Scaff^-,Y^-}^*
O_{-,\Scaff^-,Y^-}
\diff Y^+\diff Y^-
\Big\|_{op}^{1/2}.
\end{split}
\end{equation}

One simplification we can make in the operators above is to integrate out half of the pair $(Y^+,Y^-)$
using the localization afforded by $\chi_{\mcal{F},\mcal{F}'}^{clust}$,
\begin{equation}
\label{eq:integrate-half-Y}
\Big\|
\iint
\chi^{clust}_{\mcal{F},\mcal{F}'}(Y)
\Expec O_{\pm,\Scaff,Y^\pm}^* O_{\pm,\Scaff,Y^\pm}
\diff Y^+\diff Y^-
\Big\|_{op}
\leq (CNr)^{|J^\mp_{clust}|} \Big\|
\int \Expec O_{\pm,\Scaff,Y^\pm}^* O_{\pm,\Scaff,Y^\pm}
\diff Y^\pm \Big\|_{op}.
\end{equation}
We define $X_{\Scaff,\pm}$ to be the operator appearing in the right hand side above,
which simplifies to
\begin{equation}
\label{eq:X-def}
\begin{split}
X_{\Scaff,\pm}
&= \int
\Xi(\bm{\Gamma})
\chi^\pm_{\mcal{F},\mcal{F}'}(\Gamma^+)
\chi^\pm_{\mcal{F},\mcal{F}'}(\Gamma^-)
\prod_{a\in J^\pm_{clust}(\mcal{F}')} \delta(y_a^+ - y_a^-)
\Expec (O_{\Gamma^+}^{\Psi^\pm(\Scaff)})^* O_{\Gamma^-}^{\Psi^\pm(\Scaff)}
\diff\bm{\Gamma}
\end{split}
\end{equation}
after interating out the delta functions.

In Section~\ref{sec:diffusive-bds} we will prove the following estimate on $X_{\Scaff,\pm}$.
\begin{proposition}
\label{prp:main-diffusive-bd}
There exists constants $C,c>0$ such that the following holds:
Let $\mcal{F}\subset\mcal{F}'$ be a pair of skeletons with $\mcal{F}'\not=\bm{0}$, and let
$\Scaff\in\mcal{S}^\pm(\mcal{F},\mcal{F}')$ be a scaffold.  Then the operator $X_{\Scaff,+}$ defined above
satisfies the operator norm bound
\[
\|X_{\Scaff,\pm}\|_{op}
\leq N^{C\|\mcal{F}'\|}
\eps^{c\|\mcal{F}'\|_\pm} r^{-|J^\pm_{clust}|d},
\]
where
\[
\|\mcal{F}'\|_{\pm}
:= |\supp^\pm (F'_{rec}\cup F'_{tube}\cup F'_{cone})|
+ \sum_{\substack{S\in F_{clust} \\ |S\cap K^\pm_{all}|>1}} (|S|-1).
\]
\end{proposition}

The proof of Proposition~\ref{prp:main-diffusive-bd} uses geometric estimates on the volume of
paths that have the recollision, tube, and cone events described by $\mcal{F}'$ to obtain the
factor of $\eps^{c\|\mcal{F}'\|}$.  Using Proposition~\ref{prp:main-diffusive-bd} we can show
that $\mcal{\Evol}_{N\tau}$ is close to the ladder superoperator.

\begin{corollary}
\label{cor:close-to-ladder}
There exists $\kappa>0$ such that if $N\leq \eps^{-\kappa}$ and for
any operator $A$ with good support, the followin approximation holds:
\[
\|\wtild{\Evol}_{N\tau}[A] - \mcal{L}_{N\tau}[A]\|_{op}
\leq N^C\eps^c.
\]
\end{corollary}
\begin{proof}
The conclusion follows from combining Proposition~\ref{prp:main-diffusive-bd}
with~\eqref{eq:integrate-half-Y} and~\eqref{eq:CS-bd}, and using the counting bound
\[
\{\mcal{F}' \mid \|\mcal{F}'\| \leq k\} \leq N^{Ck}
\]
to complete the sum over $\mcal{F}'\not=\noset$.
\end{proof}

\subsection{On the remainders $\mcal{R}_j$}.
\label{sec:remainders}
To approximate the evolution $\Evol_{N\tau}$ by the cutoff Duhamel series
$\wtild{\Evol}_{N\tau}$ as in~\eqref{eq:tilde-approx}, we need
to prove the bound~\eqref{eq:remainder-bound}
\[
\max_{j\in[N]} \sup_{0\leq s\leq \tau} \|\Expec R_{j,s}^* R_{j,s}\|_{op} \leq \eps^{50}.
\]
We will work with the case $j=N$ and $s=\tau$, which is on the one hand the ``most difficult'' case
in terms of the volume of integration, but is at the same time the case that is notationally easiest
to work with.  Specializing to this case we recall the definition~\eqref{eq:remainder-def} of the operator
$R_{N,\tau}$ below
\[
R_{N,\tau} = \int_{\Omega_{k_{max}}(\tau)\times\Real^{2d}\times\Real^{2d}}
\ket{\eta}\bra{\xi}U_{\tau,\sigma}^{N-1}
\chi_\sigma(\omega,\xi,\eta) \braket{\eta|O_\omega|\xi} \diff\omega\diff\xi\diff\eta.
\]
Expanding the operator $U_{\tau,\sigma}^{N-1}$ as an integral over extended paths,
\begin{align*}
R_{N,\tau} &=
\int_{\Omega_{k_{max}}(\tau)\times\Real^{2d}\times\Real^{2d}}
\int_{\Omega_{\alpha,S}^{N-1}}
\ket{\eta}\bra{\xi} O_{\Gamma}\Xi(\Gamma)
\chi_\sigma(\omega,\xi,\eta) \braket{\eta|O_{\omega}|\xi}
\diff\omega\diff\xi\diff\eta \diff\Gamma \\
&=
\int_{\Omega_{\alpha,S,err}^{N}}
O_{\Gamma}\Xi(\Gamma) \diff\Gamma,
\end{align*}
where $\Omega_{\alpha,S,err}^N\subset\Omega_{\alpha,S}^N$ is the subset of extended paths
having $k_{max}$ collisions in the final segment.

Now we sketch the proof of the bound~\eqref{eq:remainder-bound}.  By performing the same decomposition
as in the previous subsection
(in particular, following the arguments leading to~\eqref{eq:integrate-half-Y} and~\eqref{eq:CS-bd})
with $A=\Id$, we reach a bound of the form
\[
\|\Expec R_{N,\tau}^*R_{N\tau}\|_{op}
\leq \sum_{\mcal{F}\leq\mcal{F}'} N^{C\|\mcal{F}'\|}
\sum_{\Scaff\in\mcal{S}^+(\mcal{F},\mcal{F}')}
\|X_{\Scaff,+,err}\|_{op},
\]
where the only modification in the argument is to keep the condition that the final segment has $k_{max}$
collisions, as in
\begin{equation}
\label{eq:X-err-def}
\begin{split}
X_{\Scaff,\pm,err}
:= \int_{\Omega_{\alpha,S,err}^N\times\Omega_{\alpha,S,err}^N}
\Xi(\bm{\Gamma})
\chi^\pm_{\mcal{F},\mcal{F}'}(\Gamma^+)
\chi^\pm_{\mcal{F},\mcal{F}'}(\Gamma^-)
\prod_{a\in J^\pm_{clust}(\mcal{F}')} \delta(y_a^+ - y_a^-)
\Expec (O_{\Gamma^+}^{\Psi^\pm(\Scaff)})^* O_{\Gamma^-}^{\Psi^\pm(\Scaff)}
\diff\bm{\Gamma}.
\end{split}
\end{equation}
The improvement that yields the factor of $\eps^{50}$ is that, in the proof of Proposition~\ref{prp:main-diffusive-bd}
we will frequently use $\eps^2\tau\leq 1$ and discard such factors.  By simply keeping the
factor of $(\eps^2\tau)^{k_{max}} = \eps^{1000}$ and not bounding it naively by $1$ we easily obtain the
desired factor of $\eps^{50}$.

\section{Bounding the diffusive diagram contributions}
\label{sec:diffusive-bds}
In this section we prove Proposition~\ref{prp:main-diffusive-bd},
which is a bound on the operator norm of $X_{\Scaff,\pm}$.
For concreteness we will work with $X_{\Scaff,+}$.

\subsection{Rewriting the operator $X_{\Scaff,+}$}
First we recall the definition of the operator $X_{\Scaff,+}$,
\begin{equation}
\label{eq:X-def}
\begin{split}
X_{\Scaff,+}
= \int
\Xi(\bm{\Gamma})
\chi^+_{\mcal{F},\mcal{F}'}(\Gamma^+)
\chi^+_{\mcal{F},\mcal{F}'}(\Gamma^-)
\prod_{a\in J^+_{clust}(\mcal{F'})} \delta(y_a^+ - y_a^-)
\Expec (O_{\Gamma^+}^{\Psi^+(\Scaff)})^* O_{\Gamma^-}^{\Psi^+(\Scaff)}\diff\bm{\Gamma}.
\end{split}
\end{equation}
We observe that the functions $\chi^+_{\mcal{F},\mcal{F}'}(\Gamma^\pm)$
depend only on $\Gamma^\pm[\Base]$, where $\Base=\Base(\Scaff)$ is as defined
in Section~\ref{sec:scaff-product}.
%Then
%we have
%\[
%X_{\Scaff,+} = \int \Xi(\Gamma)
%\chi^+_{\mcal{F},\mcal{F}'}(\Gamma^+[\Base])
%\chi^+_{\mcal{F},\mcal{F}'}(\Gamma^-[\Base])
%\prod_{a\in \supp F_{clust}} \delta(y_a^+-y_a^-)
%\Expec
%\wtild{O}_{D,+,\Gamma^+}^* \wtild{O}_{D,+,\Gamma^-}\diff\bm{\Gamma}.
%\]
%
One last observation we make before we exploit the tensor product structure of the operator
$O_{\Gamma^\pm}^{\Psi^+(\Scaff)}$ is that there are compatibility conditions on the phase
space indpoints of intervals in $Q(\Base)$ coming from Lemma~\ref{lem:twinning-sync}.
Recall that $\mcal{I}^{abs}(\Base)$ consists of intervals with endpoints in $\Base$.  Write
\[
\mcal{I}^{abs}(\Base) = \mcal{I}^{stick}(\Base)\cup \mcal{I}^{ladder}(\Base),
\]
where $\mcal{I}^{stick}(\Base)$ consists of those intervals which have as an endpoint a collision in
$\supp^+(\mcal{F}')$ and $\mcal{I}^{ladder}(\Base)$ are those intervals of the form $]\next(s),\back(s')[$
for $s,s'\in\supp^+(\mcal{F}')$.  Then let $Q^{stick}(\Base) = \ell(J)$ for $J\in\mcal{I}^{stick}(\Base)$
and $Q^{ladder}(\Base) = \ell(J)$ for $J\in\mcal{I}^{ladder}(\Base)$ be the intervals in $[N]$
corresponding to these stick and ladder intervals.

These intervals impose some compatibility conditions on the phase space points $\xi_\ell$.
In particular, for sticks $]\ell_1,\ell_2[\in Q^{stick}(\Base)$ the compatibility condition is simply
\[
d_r(\xi^\pm_{2\ell_2-2}, U_{k\tau}(\xi^\pm_{2\ell_1-1})) \leq N^8.
\]
For ladders, the compatibility condition is that there exists $s\in[-\tau,\tau]$ such that
\[
d_r(\xi^+_{2\ell_1-2}, U_s(\xi^-_{2\ell_1-2})) \leq N^8
\]
and that
\[
d_r(\xi^+_{2\ell_2-1}, U_s(\xi^-_{2\ell_2-1})) \leq N^8.
\]
We write $T\big(\mstack{\xi,&\eta}{\xi',&\eta'}\big)$ for the indicator function on the
twinnin condition on ladders and $S_k(\xi,\eta)$ for the ``stick'' condition.

Now we use the  tensor product structure of the operator $O_{\Gamma}^{\Psi^+(\Scaff)}$
given by applying Lemma~\ref{lem:operator-product} to the coloring set $\Psi^+(\Scaff)$
(which has the required product structure owing to Lemma~\ref{lem:coloring-product}).
This yields
\begin{align*}
X_{\Scaff,+}
&= \Mult^\pm\Big[
\big(
\int
\Xi(\bm{\Gamma}[\Base])])
\chi^+_{\mcal{F},\mcal{F}'}(\Gamma^+)
\chi^+_{\mcal{F},\mcal{F}'}(\Gamma^-)
\prod_{a\in J^+_{clust}(\mcal{F}')} \delta(y_a^+-y_a^-) \\
&
\qquad\qquad \prod_{]\ell_1,\ell_2[\in \ell(\mcal{I}^{ladder})(\Base)}
T\big(
\mstack {\xi_{2\ell_1-1}^+,&\xi_{2\ell_2-2}^+}
{\xi_{2\ell_1-1}^-,&\xi_{2\ell_2-2}^-} \big)
\prod_{]\ell_1,\ell_2[\in \ell(\mcal{I}^{stick})(\Base)}
\mstack
{S_k(\xi_{2\ell_1-1}^+,\xi_{2\ell_2-2}^+)}
{S_k(\xi_{2\ell_1-1}^-,\xi_{2\ell_2-2}^-)} \\
&\qquad\qquad\qquad\qquad \Expec
(O_{\Gamma^+[\Base]}^{\Psi^+_{\Base}(\Scaff)})^*
\otimes
O_{\Gamma^-[\Base]}^{\Psi^+_{\Base}(\Scaff)}
\diff\bm{\Gamma}[\Base]\big) \\
&\qquad\qquad\qquad \times
\prod_{J\in \ell(\mcal{I}^{abs})(\Base)}
\big(
\int \Xi(\bm{\Gamma}[J])
\Expec
O_{\Gamma^+[J]}^{\Psi^+_J(\Scaff)}\otimes
O_{\Gamma^-[J]}^{\Psi^+_J(\Scaff)}
\diff\bm{\Gamma}[I] \big) \Big].
\end{align*}

At this point we use the fact that for intervals $J\in\mcal{I}^{ladder}(\Base)$,
the coloring sets $\Psi^+_J(\Scaff)$ are (up to a relabelling of the colors) equal
to $\Psi_{rl}^{+,h}$ for some $h$ specified by the scaffold.  This is a consequence of
Lemma~\ref{lem:interval-color-structure}.  We therefore have the identity
that (for some $h=h(J)$ determined by the scaffold),
\[
L_{k,=h}\big(
\mstack{\xi, & \eta}{\xi', &\eta'}\big) =
\int_{\Omega_{\alpha,\sigma}^k}
\Xi(\bm{\Gamma})
\Expec
\braket{\xi|O^{\Psi^{+}_J(\Scaff)}_{\Gamma^+}|\eta}
\braket{\eta'|O^{\Psi^{+}_J(\Scaff)}_{\Gamma^-}|\xi'}\diff\bm{\Gamma}.
\]
This allows us to collapse the multiplication map $\Mult^\pm$ and obtain
\begin{equation}
\begin{split}
X_{\Scaff,+}
&= \sum_{\psi^+,\psi^-\in \Psi^+(\Scaff)}
\int \ket{\xi_{2N-1}^+}\bra{\xi_{2N-1}^-}
\Xi(\bm{\Gamma}[\Base])
\chi^+_{\mcal{F},\mcal{F}'}(\Gamma^+[\Base])
\chi^+_{\mcal{F},\mcal{F}'}(\Gamma^-[\Base]) \\
&\qquad \prod_{a\in J^+_{clust}(\mcal{F}')} \delta(y_a^+-y_a^-)
\prod_{]\ell_1,\ell_2[\in\ell(\mcal{I}^{ladder})(\Scaff)}
T\big(\mstack
{\xi^+_{2\ell_1-1}, & \xi^+_{2\ell_2-2}}
{\xi^-_{2\ell_1-1}, & \xi^-_{2\ell_2-2}}\big)
L_{\ell_2-\ell_1-1,=h}\big(\mstack
{\xi^+_{2\ell_1-1}, & \xi^+_{2\ell_2-2}}
{\xi^-_{2\ell_1-1}, & \xi^-_{2\ell_2-2}}\big)\\
&\qquad \prod_{]\ell_1,\ell_2[\in\ell(\mcal{I}^{stick})(\Scaff)}
S_k(\xi^+_{2\ell_1-1}, \xi^+_{2\ell_2-2})
S_k(\xi^-_{2\ell_1-1}, \xi^-_{2\ell_2-2}) \\
&\qquad
\prod_{\ell\in \ell(\Base)}
\Braket{\xi^+_{2\ell-2}|p_{\ell,0,+}}
\Braket{p_{\ell,k_\ell,+}|\xi^+_{2\ell-1}}
\Braket{\xi^-_{2\ell-2}|p_{\ell,0,-}}
\Braket{p_{\ell,k_\ell,-}|\xi^-_{2\ell-1}}
e^{i\varphi(\omega_\ell^+)-i\varphi(\omega_\ell^-)} \\
&\qquad\qquad
\prod_{A_{imm}(\psi^+,\psi^-)}
\Expec \Ft{V_{y_a}}(q_a) \Ft{V_{y_{a+1}}}(q_{a+1})
 \Expec \prod_{\substack{a\in K(\bm{\Gamma}_D) \\ \psi(a)\not=\qt{\Imrec}}}
\Ft{V_{y_a}^{\psi(a)}}(q_a)
\diff \bm{\Gamma}
\end{split}
\end{equation}

Now we apply the Schur test to obtain a bound on the operator norm.
Because $X_{\Scaff,+}$ is a self-adjoint operator, we only need to bound the
``row sum'' norm.
We also use the bounds
\[
|\Braket{\xi|p}|
\leq Cr^{d/2} \exp(- c(r |p - \xi_p|)^{0.999})
\]
and Lemma~\ref{lem:admissible-V} to bound the expectation.  Then to deal with the terms
involving the ladder functions $L_k$,
we use the maximal ladder function
\[
M_k(\xi,\eta) := \sup_{\xi',\eta'}T\big(\mstack{\xi,&\eta}{\xi',&\eta'}\big)
|L_k\big(\mstack{\xi,&\eta}{\xi',&\eta'}\big)|.
\]
finally we arrive at the estimate
\begin{equation}
\label{eq:monstrous-bd}
\begin{split}
\|X_{\Scaff,+}\|_{op}
\leq
&(CN)^{C\|\mcal{F}'\|}
\sup_{\xi_{2N-1}^+}
\sum_{\psi^+,\psi^-\in\Psi^+_{\Base}(\Scaff)}
\sum_{P'\leq P(\psi^+,\psi^-)}
\int
\chi^+_{\mcal{F}\mcal{F}'}(\Gamma^+[\Base])
\chi^+_{\mcal{F}\mcal{F}'}(\Gamma^-[\Base]) \\
&
\prod_{a\in J^+_{clust}(\mcal{F}')} \delta(y_a^+-y_a^-)
\prod_{]\ell_1,\ell_2[\in\ell(\mcal{I}^{ladder}(D))}
M_{\ell_2-\ell_1-1}(\xi_{2\ell_2-1}^+, \xi_{2\ell_2-2}^+)
T\big(\mstack
{\xi^+_{2\ell_1-1}, & \xi^+_{2\ell_2-2}}
{\xi^-_{2\ell_1-1}, & \xi^-_{2\ell_2-2}}\big)
\\
&\prod_{]\ell_1,\ell_2[\in\ell(\mcal{I}^{stick})(\Scaff)}
S_{\ell_2-\ell_1-1}(\xi^+_{2\ell_1-1}, \xi^+_{2\ell_2-2})
S_{\ell_2-\ell_1-1}(\xi^-_{2\ell_1-1}, \xi^-_{2\ell_2-2}) \\
&
\prod_{\ell\in\ell(\Base)} r^d \eps^{k_{\ell}} r^{-dk_\ell}
\chi_\alpha(\xi_{2\ell-2},\omega_\ell,\xi_{2\ell-1}) \\
&\prod_{A\in P'} (C|A|)^{2|A|}
r^d\exp(-c|\frac{r}{|A|}\sum_{a\in A}q_a|^{0.99})
\One(y_A \text{ is }r\text{-connected})
\diff\bm{\Gamma}[\Base].
\end{split}
\end{equation}

\subsection{The standard variable constraints}
As in the proof of Proposition~\ref{prp:short-ladder-compare}, the integrand is easier to digest as a product
of constraints on each variable.  The path variables are labelled by the set
\begin{equation}
\begin{split}
\Lambda_D := \{\qt{p_a},\qt{s_a}\}_{a\in K_0(\bm{\Gamma}[\Base])}\cup
\{\qt{y_a}\}_{a\in K(\bm{\Gamma}[\Base])}
&\cup \{ \qt{\xi^+_{2\ell-2}}, \qt{\xi^-_{2\ell-2}},
\qt{\xi^+_{2\ell-1}}, \qt{\xi^-_{2\ell-1}} \}_{\ell\in \ell(\Base)} \\
&\cup\{\qt{\xi^+_0}, \qt{\xi^+_{2N-1}}
\qt{\xi^-_{0}} \qt{\xi^-_{2N-1}}\}
\end{split}
\end{equation}

We impose a total ordering $\leq_\Lambda$ on $\Lambda$ generated by the following rules:
\begin{itemize}
\item For any $a,b\in K_0^+(\bm{\Gamma})$ with $\sign(a)=+$ and $\sign(b)=-$,
$\qt{p_a},\qt{s_a},\qt{y_a}\leq
\qt{p_b}, \qt{s_b}, \qt{y_b}$.
\item If $a,b\in K_0^+(\bm{\Gamma})$ satisfy $\sign(a)=\sign(b)$ and either
$\sign(a)=+$ and $a\leq b$ or $\sign(a)=-$ and $a\geq b$, then
$\qt{y_a}\leq \qt{y_b}$, $\qt{s_a}\leq \qt{s_b}$, $\qt{p_a}\leq \qt{p_b}$.
\item Similarly, if $0\leq m\leq m'\leq 2N-1$, then provided $\qt{\xi^+_m},\qt{\xi^+_{m'}}\in\Lambda_D$,
$\qt{\xi^+_{m}}\leq \qt{\xi^-_m}$,
$\qt{\xi^+_m}\leq \qt{\xi^+_{m'}}$,
$\qt{\xi^-_{m'}}\leq \qt{\xi^-_{m}}$.
\item For any $a\in K_0^+(\bm{\Gamma})$, $\qt{p_a}\leq \qt{s_a} \leq \qt{y_{a+1}}$,
and for any $\ell\in\Base(D)$,
$\qt{\xi^+_{2\ell-2}}\leq \qt{p_{\ell,+,0}}\leq \qt{s_{\ell,+,k_\ell}} \leq \qt{\xi^+_{2\ell-1}}$,
$\qt{\xi^-_{2\ell-1}}\leq \qt{p_{\ell,-,k_\ell}}\leq \qt{s_{\ell,-,0}} \leq \qt{\xi^-_{2\ell-2}}$.
\end{itemize}

We are now ready to assign constraint functions $f_\lambda$ for each label $\lambda\in \Lambda$,
much the same as we did in the calculation for Proposition~\ref{prp:short-ladder-compare}.
First we define the familiar momentum constraints
\begin{equation}
f_{\qt{p}_a}(\bm{\Gamma}_{\leq \qt{p},a}) =
\begin{cases}
r^d \One(|p_{\ell,+,0} - (\xi_{2\ell-2}^+)_p|\leq \alpha r^{-1}), & a = (\ell,+,0) \\
r^d \One(|p_{\ell,-,k_\ell} - (\xi_{2\ell-1}^-)_p|\leq \alpha r^{-1}),
& a = (\ell,-,k_\ell) \\
r^d \exp(-c (Nk_{max})^{-1} | r\sum_{a\in S}q_a|^{0.99}),
&a = \max_{\leq_\Lambda } S \text{ for some } S\in P' \\
(1+|q_a|)^{-20d} \One(||p_a| - |p_{+,0}||\leq \alpha \max\{|p_0|^{-1}s_a^{-1},r^{-1}\}),
&\text{ else}.
\end{cases}
\end{equation}
The position constraints are slightly modified by the presence of the delta function
$\delta(y_a^+-y_a^-)$ for each $a\in J^+_{clust}(\mcal{F}')$.   For
$a\in J^+_{clust}(\mcal{F}')$ we define
\[
f_{\qt{y},a,-}(\bm{\Gamma}_{\leq \qt{y},a,-}) := r^{-d}\delta(y_a^+-y_a^-).
\]
For the rest of the indices $a$ we set
\begin{equation}
f_{\qt{y},a}(\bm{\Gamma}_{\leq\qt{y},a}) :=
\begin{cases}
r^{-d}\One(|y_{\ell,+,1} - ((\xi^+_{2\ell-2})_x + s_{\ell,+,0}p_{\ell,+,0})| \leq \alpha r),
&a = (\ell,+,1) \\
r^{-d}\One(|y_{\ell,-,k_-} - ((\xi^-_{2\ell-1})_x - s_{\ell,-,k_-}p_{\ell,-,k_-})|),
&a = (\ell,-,k_\ell) \\
r^{-d}\One(|y_{\ell,+,j} - (y_{\ell,+,j-1}+s_{\ell,+,j-1}p_{\ell,+,j-1})|\leq \alpha r), &a = (\ell,+,j)
\text{ for } j>1\\
r^{-d}\One(|y_{\ell,-,j+1} - (y_{\ell,-,j}+s_{\ell,-,j}p_{\ell,-,j})|\leq \alpha r), &a = (\ell,-,j)
\text{ for } j<k_\ell.
\end{cases}
\end{equation}
The phase space position $\qt{\xi}$ constraints are now more complicated than the analogous calculation
in the proof of Proposition~\ref{prp:short-ladder-compare} due to the presence of the twinning functions
and ladder functions.
\begin{equation}
f_{\qt{\xi},b,m} :=
\begin{cases}
d_r(\xi^+_{2\ell-1}, (y_{\ell,+,k_\ell}+s_{\ell,+,k_\ell}p_{\ell,+,k_\ell},p_{\ell,+,k_\ell}))
\leq N^2\alpha, &b = +, m = 2\ell-1, \ell\in \Base \\
d_r(\xi^-_{2\ell-2}, (y_{\ell,-,1}-s_{\ell,-,0}p_{\ell,-,0},p_{\ell,-,0}))
\leq N^2\alpha, &b=-, m = 2\ell-1, \ell\in\Base \\
M_{\ell_2-\ell_1-1}
(\xi_{2\ell_1-1}^+, \xi_{2\ell_2-2}^+),
&b=+, \,\,\,]\ell_1,\ell_2[\in\ell(\mcal{I}^{ladder}) \\
d_r(U_{(\ell_2-\ell_1)\tau}(\xi_{2(\ell_1-1)-1}^+), \xi_{2(\ell_2+1)-2}^+)
\leq N^6, &b=+, m= 2\ell_2-2,
]\ell_1,\ell_2[\in \ell(\mcal{I}^{stick}) \\
T\big(\mstack
{\xi^+_{2\ell_1-1}, & \xi^+_{2\ell_2-2}}
{\xi^-_{2\ell_1-1}, & \xi^-_{2\ell_2-2}}\big),
&b=-, m=2\ell_2-2, ]\ell_1,\ell_2[ \in \ell(\mcal{I}^{ladder}) \\
d_r(U_{(\ell_2-\ell_1)\tau}\xi^-_{2(\ell_1-1)-1},
\xi^-_{2\ell_2-2}) \leq N^6
&b=-, m=2\ell_2-2, ]\ell_1,\ell_2[ \in \ell(\mcal{I}^{stick}).
\end{cases}
\end{equation}

The relevant fact we need about the maximal ladder functions is that
\[
\sup_\xi \int M_k(\xi,\eta) \diff \eta \leq CN^{20},
\]
which follows from Lemma~\ref{lem:maximal-ladder-bd}, which is proven in the next section.

The definitions of the $f_\lambda$ functions are defined precisely so that the
integrand in~\eqref{eq:monstrous-bd} is bounded by the product of individual constraints.
That is,
\begin{align*}
&
\prod_{a\in J^+_{clust}(\mcal{F}')} \delta(y_a^+-y_a^-)
\prod_{]\ell_1,\ell_2[\in\ell(\mcal{I}^{ladder}(D))}
M_{\ell_2-\ell_1-1}(\xi_{2\ell_2-1}^+, \xi_{2\ell_2-2}^+)
T\big(\mstack
{\xi^+_{2\ell_1-1}, & \xi^+_{2\ell_2-2}}
{\xi^-_{2\ell_1-1}, & \xi^-_{2\ell_2-2}}\big)
\\
&\prod_{]\ell_1,\ell_2[\in\ell(\mcal{I}^{stick})(\Scaff)}
S_{\ell_2-\ell_1-1}(\xi^+_{2\ell_1-1}, \xi^+_{2\ell_2-2})
S_{\ell_2-\ell_1-1}(\xi^-_{2\ell_1-1}, \xi^-_{2\ell_2-2}) \\
& r^{-d|K(\bm{\Gamma}_D)|} \prod_{\ell\in\Base(D)} r^d
\chi_\alpha(\xi_{2\ell-2},\omega_\ell,\xi_{2\ell-1})
\prod_{A\in P'} r^d\exp(-c|\frac{r}{|A|}\sum_{a\in A}q_a|^{0.99}) \\
&\qquad\qquad\qquad\leq \prod
f_{\qt{p},a}(\bm{\Gamma}_{\leq \qt{p},a})
f_{\qt{y},a}(\bm{\Gamma}_{\leq \qt{y},a})
\prod
f_{\qt{\xi},b,m}(\bm{\Gamma}_{\leq \qt{\xi},b}).
\end{align*}

Thus, using the fact that the maximum cluster in the partition $P(\psi)$ has
has cardinality bounded by $\|\mcal{F}'\|$ to sum over the partitions $P'$, we arrive at the bound
\begin{equation}
\begin{split}
\|\wtild{X}_{\Scaff,+}\|_{op}
&\leq (CN)^{C\|\mcal{F}'\|}
\sup_{\xi^+_{2N-1}}
\max_{\psi^+,\psi^-\in \Psi^+(\Scaff)}
\max_{P'\leq P(\psi^+,\psi^-)} \\
&\qquad \int
\eps^{|K(\bm{\Gamma}[\Base])|}
\chi_{\mcal{F},\mcal{F}'}^+(\Gamma^+[\Base])
\chi_{\mcal{F},\mcal{F}'}^+(\Gamma^-[\Base]) \\
&\qquad\qquad\qquad \prod
f_{\qt{p},a}(\bm{\Gamma}_{\leq \qt{p},a})
f_{\qt{y},a}(\bm{\Gamma}_{\leq \qt{y},a})
f_{\qt{\xi},b,m}(\bm{\Gamma}_{\leq \qt{\xi},b,m}) \\
&\qquad\qquad\qquad \prod_{A\in P'}\One(y_A\text{ is }r\text{-connected})\diff\bm{\Gamma}_D \\
&=: (CN)^{C\|\mcal{F}\|}
\max_{\psi^+,\psi^-\in\Psi^+(\Scaff)}
\max_{P'\leq P(\psi^+,\psi^-)}
Y(\Scaff,\psi^+,\psi^-,P'),
\end{split}
\end{equation}
where
\begin{equation}
\label{eq:YD-def}
\begin{split}
Y(\Scaff,\psi^+,\psi^-,P') := \sup_{\xi^+_{2N-1}}
&\int
\eps^{|K(\bm{\Gamma}[\Base])|}
\chi_{\mcal{F},\mcal{F}'}^+(\Gamma^+[\Base])
\chi_{\mcal{F},\mcal{F}'}^+(\Gamma^-[\Base]) \\
&\qquad\qquad\prod
f_{\qt{p},a}(\bm{\Gamma}_{\leq \qt{p},a})
f_{\qt{y},a}(\bm{\Gamma}_{\leq \qt{y},a})
f_{\qt{\xi},b,m}(\bm{\Gamma}_{\leq \qt{\xi},b,m}) \\
&\qquad\qquad \prod_{A\in P'}\One(y_A\text{ is }r\text{-connected})\diff\bm{\Gamma}_D.
\end{split}
\end{equation}

Before we proceed we clarify the nature of the partitions involved in the above expression.
The first partition clearly present above is the partition $P'$, which is some refinement
of $P(\psi^+,\psi^-)$.  Due to the constraint that $y_A$ is $r$-connected, we have
\[
|y_a - y_b| \leq 4|A|r
\]
for $(a,b)\in A\in P'$.  We can extend these constraints using the localizations from
$\chi_{\mcal{F},\mcal{F}'}^+$.  In particular, $\chi^+_{\mcal{F},\mcal{F}'}$ enforces
\[
|y_a-y_b| \leq 4r
\]
when either $a,b\in K^+_{all}$ and $(a,b)\in F'_{rec}$ or $a,b\in K^-_{all}$ and
$(a^*,b^*)\in F'_{rec}$.  In either case we say $(a,b) \in F'_{rec,+}$, meaning
that $(a,b)$ is a recollision originally present in $\Gamma^+$ and is now mirrored to
$\Gamma^-$ by our application of the Cauchy-Schwarz inequality.  We also have
$y_a = y_{a^*}$ for $a\in J^+_{clust}(\mcal{F}')$.
Thus if $Q$ is the coarsening of $P'$ which has $a\sim_Q b$ when $(a,b)\in F'_{rec,+}$
and $a\sim_Q a^*$ for $a\in J^+_{clust}(\mcal{F}')$, we can get a slightly better
bound for $Y(\Scaff,\psi^+,\psi^-,P)$:
\begin{equation}
\label{eq:YD-def}
\begin{split}
Y(\Scaff,\psi^+,\psi^-,P') := \sup_{\xi^+_{2N-1}}
&\int
\eps^{|K(\bm{\Gamma}[\Base])|}
\chi_{\mcal{F},\mcal{F}'}^+(\Gamma^+[\Base])
\chi_{\mcal{F},\mcal{F}'}^+(\Gamma^-[\Base]) \\
&\qquad\qquad\prod
f_{\qt{p},a}(\bm{\Gamma}_{\leq \qt{p},a})
f_{\qt{y},a}(\bm{\Gamma}_{\leq \qt{y},a})
f_{\qt{\xi},b,m}(\bm{\Gamma}_{\leq \qt{\xi},b,m}) \\
&\qquad\qquad \prod_{A\in Q}\One(y_A\text{ is }r\text{-connected})\diff\bm{\Gamma}_D.
\end{split}
\end{equation}
Implicitly the partition $Q$ depends on $P'$ and $\Scaff$ (through $\mcal{F}$), but as these will be fixed
for the remainder of the section we simply write `Q'.

\begin{lemma}
\label{lem:Y-bd}
The quantity $Y(\Scaff,\psi^+,\psi^-,P')$ defined above satisfies
\[
Y(\Scaff,\psi^+,\psi^-,P') \leq
N^{C\|\mcal{F}'\|} r^{-d|J^+_{clust}|} \eps^{c\|\mcal{F}\|_+}.
\]
\end{lemma}

\subsection{Integrating out $\Gamma^-[\Base]$}
The first step we take towards the proof of Lemma~\ref{lem:Y-bd}
is to integrate out the $\Gamma^-$ variables.

To do so we first write out more explicitly the constraints on the time variables $s_a$.
The constraints on $y_A$ imply constraints on the time variables $s_a$ when $a$ is not of the form $\min(S)-1$
for some $S\in Q$.  In particular, we define the time constraints
\[
f_{\qt{s},a}(\bm{\Gamma}_{\leq\qt{s},a})
:=
\begin{cases}
\One_{[0,\tau]}(s_a), &a = \min(S)-1 \text{ for some } S\in Q \\
\One(|y_{a'} + (y_{a}+s_ap_a)| \leq 2|S|r), &a+1\in S\in Q, a'=\min(S).
\end{cases}
\]
To interate out the variables of $\Gamma^-$ we ignore the term $\chi_{\mcal{F},\mcal{F}'}^+(\Gamma^-[\Base])$
(which is bounded by $1$) and use the product of time constraints $\prod f_{\qt{s},a}$ to
enforce $\One(y_A\text{ is } r\text{-connected})$ for every $A\in Q$.  We therefore bound
\begin{equation}
\label{eq:Ysplit-bd}
\begin{split}
Y(\Scaff,P')
&\leq \sup_{\xi^+_{2N-1}}
\int
\eps^{|K(\bm{\Gamma}[\Base])|}
\chi_{\mcal{F},\mcal{F}'}(\Gamma^+[\Base])
\prod f_{\qt{p},a}(\bm{\Gamma}_{\leq \qt{p},a})
f_{\qt{y},a}(\bm{\Gamma}_{\leq \qt{y},a}) \\
&
\qquad \qquad \qquad \qquad
\prod f_{\qt{\xi},b,m}(\bm{\Gamma}_{\leq \qt{\xi},b,m})
f_{\qt{s},a}(\bm{\Gamma}_{\leq \qt{s},a}) \diff\bm{\Gamma}[\Base] \\
&\leq
\sup_{\Gamma^+[\Base]}
\int
\eps^{|K(\Gamma^-[\Base])|}
\prod_{\Lambda[\Base]^-} f_\lambda(\bm{\Gamma}_{\leq \lambda})
\diff\Gamma^-[\Base] \\
&\qquad \qquad \times \sup_{\xi^+_{2N-1}}
\int
\eps^{|K(\Gamma^+[\Base])|}
\chi_{\mcal{F},\mcal{F}'}^+(\Gamma^+[\Base])
\prod_{\Lambda[\Base]^+} f_\lambda(\bm{\Gamma}_{\leq\lambda}\diff\Gamma^+[\Base].
\end{split}
\end{equation}

In this subsection we count up the constraints on the $\Gamma^-$ variables to bound the first
integral.  To state the result we need to first define a few quantities.
Let $k^+$ be the number of collisions in $\Gamma^+[\Base]$ and $k^-$ be the number of collisions
in $\Gamma^-[\Base]$.  Given a partition $P\in\mcal{P}(K(\bm{\Gamma}_D))$, we let
$m(P) \subset K(\bm{\Gamma}_D)$ be the set
\[
m(P) = \{\min S\mid S\in P\}
\]
of ``first'' collisions in each set of $P$, while
\[
M(P) = \{\max S \mid S\in P\}
\]
is the set of ``last'' collisions.  We write $m^+(P)$, $m^-(P)$ and $M^+(P)$, $M^-(P)$ for the number of
first and last collisions in $K^+$ and $K^-$, respectively.  Then we have the following bound
\begin{lemma}
\label{lem:minus-gamma}
Given a scaffold $\Scaff$ and a partition $P'$,
the following integral over $\Gamma^-$ variables satisfies the bound
\begin{equation}
\label{eq:minus-gamma-bd}
\begin{split}
\sup_{\Gamma^+[\Base]}
\int
\prod_{\Lambda[\Base]^-} f_\lambda(\bm{\Gamma}_{\leq \lambda})
\diff\Gamma^-[\Base]
&\leq N^{C\|\mcal{F}'\|} r^{-d|J^+_{clust}(\mcal{F}')}
\tau^{|m^-(Q)|} \\
&\qquad
\times (E^{-1/2}r)^{k^- - |m^-(Q)|} (r^{-1}\min\{E^{(d-1)/2},1\})^{k^--|M^-(Q)|}.
\end{split}
\end{equation}
\end{lemma}

\begin{proof}
We  integrate each variable independently using Lemma~\ref{lem:abstract-holder}.  The resulting
bound follows from the following rules:
\begin{itemize}
\item Each recollision pair $(s_a,p_a)$ contributes $O(\log(\eps))$ by Lemma~\ref{lem:sp-integral}.
\item Each phase space variable labelled $\qt{\xi^\pm_m}$ contributes $O(N^C)$, and there are
$O(\|\mcal{F}'\|)$ such variables.
\item Each time variable $\qt{s_{a-1}}$ that is the first collision in a cluster $S\in Q(\psi;P')$
 contributes at most $\tau$.  There
are $|m^-(Q)|$ such variables.
\item Each time variable $s_{a-1}$ with $a\not\in m^-(Q)$
cluster contributes $E^{-1/2}r$, where $E = |p_{\ell,0}|^2$
is the kinetic energy of the path.
\item Each momentum variable $p_a$ with $a\in M^-(P')$
contributes $C|S|^d$ because of the
localization $r^d\exp(-c(r|S|^{-1}|q_S|)^{0.9})$.  There are $M^-(P')$ such variables.
Of these, $|S|\leq 2$ for all but at most $\|\mcal{F}'\|$ of the collisions, so the combinatorial
factor is at most $N^{C\|\mcal{F}'\|}$.
\item Each momentum variable $p_{\ell,j}$ that is not the last in its cluster and does not belong to a
recollision contributes $C r^{-1} \min\{|p_{\ell,0}|^{d-1},1\}$ by~\eqref{eq:std-momentum-bd}.
There are $N^- - M^-(P')$ such collisions.
\item Each first momentum variable $p_{\ell,0}$ contributes $O(1)$.
\item Each $y_{\ell,j}$ variables each contribute $O(1)$, except for the variables determined by delta
functions, which contribute $O(r^{-d})$.
\end{itemize}
\end{proof}

We split the analysis of the remaining integral into cases, depending on which of
$F_{clust}$, $F_{rec}$, $F_{tube}$, or $F_{cone}$ is most significant.

\subsection{The cluster dominated case}
The first case we deal with is that
\[
\frac{1}{2}(k^++k^-) - |Q| \geq \frac{1}{10} \|\mcal{F}'\|_+.
\]
This occurs in particular if
\[
\sum_{S\in F'_{clust}} |S\cap K^+_{all}|-1 \geq \frac{1}{4} \|\mcal{F}'\|_+
\]
since each set $S\in F'_{clust}$ belongs to a cell $T\in Q$ containing at least
$2|S\cap K^+_{all}|$ collisions.
In this case, the simpleminded analysis of the previous section works to obtain
a suitable bound for $Y(\Scaff, P')$.

To integrate out the phase space variables we use the following $L^1$-type bound on the maximal
ladder function
\[
\sup_\xi \int M_k(\xi,\eta)\diff \eta \leq N^C.
\]
This bounds the contribution of the phase space variables by an overall factor of $N^{C\|\mcal{F}'\|}$.
Then using the same accounting as in the proof of Lemma~\ref{lem:minus-gamma} we obtain
\begin{equation}
\begin{split}
\sup_{\xi^+_{2N-1}}
\int & \chi_{\mcal{F},\mcal{F}'}^+(\Gamma^+[\Base])
\prod_{\Lambda[\Base]^+} f_\lambda(\bm{\Gamma}_{\leq\lambda}\diff\Gamma^+[\Base] \\
&\leq
(CN)^{C\|\mcal{F}'\|}
\tau^{m^+(Q)} (E^{-1/2}r)^{k^+-m^+(Q)}
(r^{-1}\min\{E^{(d-1)/2},1\})^{k^+-M^+(Q)}.
\end{split}
\end{equation}
Now we combing this bound with
~\eqref{eq:Ysplit-bd} and Lemma~\ref{lem:minus-gamma} and
simplify using
\[
\tau^{|Q|} \leq \tau^{(k^++k^-)/2} \tau^{|Q|-(k^++k^-)/2}
\leq \tau^{(k^++k^-)/2} \eps^{(k^++k^-)/2 - |Q|}
\]
and
\[
\sup_E \min\{ E^{(d-2)/2}, E^{-1/2}\} \leq 1
\]
to conclude
\begin{equation}
\label{eq:Y-cluster-bd}
\begin{split}
Y(\Scaff,P')
&\leq (CN)^{C\|\mcal{F}'\|}
(\eps^2\tau)^{(k^++k^-)/2}
\eps^{(k^++k^-)/2 - |Q|} r^{-d|\supp J^+_{clust}(\mcal{F}')|},
\end{split}
\end{equation}
as desired.

\subsection{The recollision argument}
The second case we deal with is that $|F_{rec}'| \geq \frac{1}{4}\|\mcal{F}\|_+$.
If $(a,b)\in F_{rec}$ is a recollision event, then either $\{a,b\}\in P'$ or
$\{a,b\}\subsetneq S\in Q$ for some $S\in Q$.  If the latter
situation occurs for at least half of the recollisions in $F_{rec}$, then we
are in the situation of the previous section.  In this section we deal with the
former scenario.

To see where the decay comes from in this case, we observe that for each such
pair $\{a,b\}$ the momentum variables essentially satisfy
\[
|p_b - p_{b-1} + p_a - p_{a-1}| \leq 4 \alpha r^{-1}
\]
while also being constrained to an annulus of thickness $r^{-1}$.  If $p_a$, $p_{a-1}$,
and $p_{b-1}$ are chosen randomly, then there is only a small probability that $p_b$ can
be chosen to both be on the annulus and satisfy the linear constraint above.  Thus the key
estimate is the following bound, which we note is closely related to the ``crossing estimate''
that is traditionally used to control diagrams.
\begin{lemma}
\label{lem:crossing-est}
For any kinetic energy level $E \geq \eps^{0.2}$, we have the estimate
\begin{align*}
\sup_{q_0,q_1,q_2\in\Real^d}
\int
\big(\prod_{j=1}^2
\One(||p_j|-E|\leq r^{-1})\big)
&(1 + |p_j-q_j|)^{-20d}
\One(||q_0 - p_2 + p_1|-E|\leq r^{-1}) \diff p_1\diff p_2 \\
&\leq C\eps^{0.1}
\big(\min\{E^{(d-1)/2}, 1\} r^{-1}\big)^2.
\end{align*}
In particular,
\begin{align*}
\sup_{q_0,q_1,q_2,q_3\in\Real^d}
\int
\big(\prod_{j=1}^3
\One(||p_j|-E|\leq r^{-1})\big)
&(1 + |p_j-q_j|)^{-20d}
(r^d \exp(-c(r|q_0-p_2+p_1-p_3|)^{0.5}))
\diff p_1\diff p_2 \diff p_3\\
&\leq C\eps^{0.1}
\big(\min\{E^{(d-1)/2}, 1\} r^{-1}\big)^2.
\end{align*}
\end{lemma}
\begin{proof}
Let $A_{E,r}(z)$ denote the shifted annulus
\[
A_{E,r}(z) := \{p\in\Real^d \mid ||p-z|-E|\leq r^{-1}\},
\]
We note that for $(p_1,p_2)$ that contributes to the integrand,
$p_2 \in A_{E,r}(0)\cap A_{E,r}(q_0-p_1)$.  We split the integral
into two parts.  In the first part, $|q_0-p_1|\geq \eps^{0.5}$.  In this
case, the intersection $A_{E,r}(0)\cap A_{E,r}(q-p_1)$ has volume
at most $(r^{-1}\eps^{0.5}) E^{(d-1)/2} r^{-1}$, from which we gain
the extra factor of $\eps^{0.1}$.  The second integral involves the
constraint $|q_0-p_1|\leq \eps^{0.5}$, which also contributes an
extra factor of $\eps^{0.1}$ in the integration over $p_1$.

For the proof of the second bound, we simply observe that
the constraint $p_3\in A_{E,r}(0)$ and the exponential
localization of $p_3$ enforces that
$p_2\in A_{E,r}(q_0-p_1)$, and we use the argument for the first bound.
The integral in $p_3$ produces only a constant factor.
\end{proof}

To apply Lemma~\ref{lem:crossing-est}, we need to make some observations about the
momentum variables $p_a$.  If $[\ell_1,\ell_2]\in\pi(\mcal{I}^{stick}(D))$ is a stick,
then $|p_{\ell_1-1,k_{\ell_1-1}} - p_{\ell_2+1,0}| \leq CN^C r^{-1}$.  For each
$\{a,b\}\in F_{rec}$, if $a-1$ is of the form $(\ell,0)$ then $\ell-1$ must be the endpoint
of a ``stick'' interval (otherwise it would be a ladder interval, but then the first collision
in the segment must be a rung and not a recollision).  Thus we can isolate the momentum
variables of the form $p_a,p_{a-1},p_b,p_{b-1}$ for $\{a,b\}\in F_{rec}\cap P'$ and integrate
over these variables \emph{first}, leaving the ordering on $\Lambda_D^+$ otherwise unchanged.

We then use the following integral bound, which follows from Lemma~\ref{lem:crossing-est}
above.
\begin{lemma}
\label{lem:multiple-crossings}
Let $\{(a_i,b_i)\}_{i=1}^m\subset \binom{[K]}{2}$ be a collection of disjoint pairs
of indices in $[K]$ satisfying $a_i+1<b_i$ for some $K\in\bbN$.
(meaning that $\{a_i,b_i\}\cap \{a_j,b_j\}=\noset$ for $i\not=j$).
Then, for any $E\geq r/\sigma$,
\begin{equation}
\begin{split}
\int
\prod_{i=1}^K \One(||p_i| - E|\leq r^{-1})
\prod_{j=1}^m r^d &\exp(-c (r|p_{a_i}-p_{a_i-1} + p_{b_i} - p_{b_i-1}|)^{0.5}) \diff \mbf{p} \\
&\leq  C^K r^{-m/100}
\big(\min\{E^{(d-1)/2}, 1\} r^{-1}\big)^{K-m}.
\end{split}
\end{equation}
\end{lemma}

The additional required ingredient is the following elementary combinatorial
fact.
\begin{lemma}
\label{lem:disjoint-subset-finder}
Let $\mcal{S}\subset[K]$ be a collection of sets $S$ containing four elements each,
such that any index $j\in[K]$ is contained in at most two sets.
Then there exists subsets $T(S)\subset S$ for each $S\in\mcal{S}$
such that $|T(S)|\geq 1$ for every $S\in\mcal{S}$, the sets $T(S)$ are
disjoint, and $|T(S)|=3$ for at least $|\mcal{S}|/4$ sets.
\end{lemma}

\begin{proof}[Proof of Lemma~\ref{lem:multiple-crossings} using Lemma~\ref{lem:disjoint-subset-finder}]
We apply Lemma~\ref{lem:disjoint-subset-finder}
to the collection of sets $\mcal{S} = \{\{a_i-1,a_i,b_i-1,b_i\}\}_{i=1}^m$.
Then we reorder $[K]$ so that the elements of $T(S)$ are last
 and the rest of the elements are first.
The triple sets contribute $\eps^{0.1}\big(\min\{E^{(d-1)/2},1\}r^{-1}\big)^2$
and the singletons constribute $O(1)$ from the stretched exponential localization.
Every other index  that is not contained in some $T(S)$ contributes
$\min\{E^{(d-1)/2}, 1\}r^{-1}$.
\end{proof}

From Lemma~\ref{lem:multiple-crossings} we gain an additional improvement
$\eps^{c|F_{rec}|}$ over the bound~\eqref{eq:Y-cluster-bd}.  We therefore obtain
\begin{equation}
\label{eq:Y-rec-bd}
\begin{split}
Y(D,\psi,P')
&\leq (CN)^{C\|\mcal{F}'\|}
\max_{P'\leq P(\psi)}
(\eps^2\tau)^{(k^++k^-)/2}
\eps^{c|\supp^+ F'_{rec}|}
r^{-d|\supp F_{sing}|} \\
&\leq (CN\eps^\beta)^{C\|\mcal{F}\|} r^{-d|\supp J^+_{clust}(\mcal{F}')|}.
\end{split}
\end{equation}

We conclude this subsection with the proof of the combinatorial lemma.
\begin{proof}[Proof of Lemma~\ref{lem:disjoint-subset-finder}]
We build the sets $T(S)$ by a greedy algorithm.  At the $j$-th step we keep track
of the collection of ``completed'' sets $\mcal{C}_j\subset\mcal{S}$, a set
of ``remaining'' sets $\mcal{R}_j\subset\mcal{S}$, a subset choice function
$T_j:\mcal{C}_j\to 2^{[K]}$ and a ``token count'' function
$f_j:\mcal{S}\to \bbZ$.  We also define the ``inaccessible elements''
$I_j\subset[K]$ to be $\bigcup_{S\in\mcal{C}_j}T(S)$.  These are the elements that have
been chosen to be part of a subset already by the $j$th step.

At the initial step we set $\mcal{C}_0 = \noset, \mcal{R}_0=\mcal{S}$,
and $f_j(S)=0$ for all $S$.  At each step we maintain the invariants that
$\sum_{\mcal{S}}f_j(S)\leq 0$, that for each $S\in\mcal{R}_j$,
$|S\setminus I_j| \geq 4 - f_j(S)$, and $f_j(S) = 0$ if $S\cap I_j=\noset$.

At the $j$-th step we arbitrarily pick a set $S\subset \mcal{R}_j$ minimizing
$|S\setminus I_j|$.  If such $S$ has $|S|=4$, then we choose
$T(S)\subset S$ arbitrariliy, set $\mcal{C}_{j+1} = \mcal{C}_j\cup\{S\}$
and $\mcal{R}_{j+1} = \mcal{C}_j\setminus\{S\}$.  We update the token
count function by setting $f_{j+1}(S) = -3$ and
$f_{j+1}(S') = f_j(S') + |S'\cap T(S)|$ for every $S'\in\mcal{R}_j$.  It
is clear that the token count updated in this way maintains the invariants above.

If the minimal $|S\setminus I_j|$ has $|S\setminus I_j|<4$ there are three cases to
consider.  The first is that $|S\setminus I_j| = 3$ and there exists
$S'\in\mcal{R_j}$ such that $S'\setminus I_j = S\setminus I_j$.  In this case
we choose two elements $a$ and $b$ arbitrarily from $S\setminus I_j$
and set $T_{j+1}(S) = \{a\}$ and $T_{j+1}(S')=\{b\}$.  Because both $a$ and $b$
belong to $S$ and $S'$, and no element is contained in more than one set,
there are no other sets $S''$ containing $a$ or $b$.  We then simply
update $\mcal{R}_{j+1} = \mcal{R}_j\setminus\{S,S'\}$
and $\mcal{C}_{j+1} =\mcal{C}_j \cup\{S,S'\}$, and $f_{j+1} = f_j$.

The second case is that $|S\setminus I_j|=3$ and there does not exist another
$S'\in\mcal{R}_j$ such that $S'\setminus I_j = S\setminus I_j$.
In this case we set $T(S) = S\setminus I_j$ and proceed with the token update
rule $f_{j+1}(S) = -3$, $f_{j+1}(S') = f_j(S) + |S'\cap T(S)|$.

The final case is that $|S\setminus I_j|<2$.  In this case we choose an element
$a\in S\setminus I_j$ arbitrarily and assign $T(S)=\{a\}$.  We update the
token function by setting $f_{j+1}(S) = f_j(S)-1$ and $f_{j+1}(S') = f_j(S')+1$
if there exists $S'\in\mcal{R}_j$ such that $a\in S'$.

Repeating this process until $\mcal{R}_j$ is empty, we obtain a collection of
disjoint sets $T(S)\subset S$ for each $S\in\mcal{S}$.  It remains to check
that $|T(S)|=3$ for at least $|\mcal{S}|/4$ sets.  To see this we observe
that $f(S) = -3$ for such sets, whereas $f(S') \geq 1$ if $|T(S')| = 1$
and $\sum f(S)\leq 0$.
\end{proof}

\subsection{The tube event case}
The third case we consider is that $|\supp^+ F'_{tube}|\geq \|\mcal{F}\|_+/4$.  In
this case at least some constant fraction of the tube events
$(a,b)\in F_{tube}$ satisfy $b\not\in\supp F_{rec}$.
and $b-1\not\in M(P')$.  If not, we can reduce to one of the previous cases.
In this case, we will introduce an extra constraint on either the time variable
$s_{b-1}$ or the momentum variable $p_{b-1}$.
In this section we argue that the bound
\begin{equation}
\label{eq:Y-tube-bd}
\begin{split}
Y(\Scaff, P')
&\leq (CN)^{C\|\mcal{F}'\|}
\eps^{\beta |\supp^+ F'_{tube}|}
r^{-d|\supp J^+_{clust}(\mcal{F}')|},
\end{split}
\end{equation}
holds for some $\beta>0$.

Recall that a tube event occurs when $y_b$ lies in the tube
\[
T(y_a,p_a) := \{z \in\Real^d \mid \text{ there exists }
s\in \Real\text{ such that } |y_a + sp_a - z|\leq N^8r\}.
\]
Indeed, if
$|p_{b-1} - p_a|\geq \eps^{0.6}$, then we have
\[
\sup_{y_a,p_a,p_{b-1}} \int \One(y_{b-1}+s_{b-1}p_{b-1}\in T(y_a,p_a))
\diff s_{b-1}
\leq \eps^{-0.6} N^8 r.
\]
And this is an improvement over $\tau$ by a factor of $N^8 \eps^{0.4}$.
On the other hand, if $|p_{b-1}-p_a|\leq \eps^{0.6}$, then we gain an improvement
in the $p_{b-1}$ variable.

The only challenge with implementing this idea is that there may be many
pairs of the form $(a_i,b)\in F_{tube}$ with $a_i<b$.  In this case the constraints
``collide'' and naively there is only a gain of one factor of $\eps^\beta$.

A very simplistic and extraordinarily suboptimal solution to this problem is as follows:
If for some $b$ there are at most $m$ tuples of the form $(a,b)$ (with $m$ to be chosen
momentarily, in fact $m=60d$ works), then we still
gain a factor of $\eps^{\beta/m}$ for each pair in $F'_{tube}$.  If on the other hand
for some $b$ there are at least $m$ pairs $(a_i,b)$ then we integrate over all possible
values of $y_b$ at a cost of $\eps^{-C}$ (with $C<3d$).  Then, having conditioned on the value of $y_b$,
we gain a factor of $\eps^\beta$ on the variables determining $(y_a,p_a)$ by splitting
into cases as follows:
\begin{itemize}
\item In the case $|y_a-y_b|\leq \eps^{-0.5}r$, we gain a factor of $\eps^{0.3}$
in the integration over $s_{a-1}$.
\item In the case $|y_a-y_b|> \eps^{-0.5}r$ we gain a factor of $\eps^{0.1}$ in the
integration over the $p_a$ variable.
\end{itemize}
Either way, each such collision contributes an additional factor of $\eps^{0.1}$.  If
$m>60d$, then $\eps^{0.1(m/2)} < \eps^{3d}$, so that $\eps^{0.1m}\eps^{-C} < \eps^{\beta m}$,
and again it can be seen that each tube incidence in $F'_{tube}$ contributes a factor of $\eps^\beta$.

\subsection{Cone events}
The final case we consider is that $|\supp^+ F'_{cone}|\geq \|\mcal{F}\|_+/4$.
Recall that in a cone event $(a,b)\in F_{cone}$, there exists
some $v\in\Real^d$ such that
\begin{align*}
    ||v|^2/2 - E| &\leq \alpha^2  N^{10} E^{1/2}r^{-1} \\
    ||v-q_a|^2/2-E| &\leq \alpha^2 N^{10} E^{1/2} r^{-1} \\
    ||v+q_{b}|^2/2-E| &\leq \alpha^2 N^{10} E^{1/2} r^{-1} \\
    |y_a + tv - y_{b}| &\leq \alpha^2 N^{10} r,
\end{align*}
and the constraint that this does not hold for $v=p_a$ implies that
there must exist some exterior collision between $a$ and $b$.

For each of these cone events, if neither of the following constraints holds:
\begin{align*}
|y_a - y_b| &\leq \eps^{-0.8} r \\
|\frac{y_a-y_{b-1}}{|y_a-y_{b-1}|} - p_{b-1}| &\leq \eps^{0.2},
\end{align*}
then the momentum variable $p_b$ is constrained to be on the intersection
of the annulus $A_{E,r}(0)\cap A_{E,\delta}(E^{1/2} \frac{y_a-y_b}{y_a-y_b})$.
In any such case there is a gain of a factor of $\eps^\beta$ in one of the variables
$s_{b-1}$, $p_{b-1}$, or $p_b$.  Now if there are too many pairs of the form $(a_i,b)$
with some fixed $b$ we can spend a factor of $\eps^{-C}$ to deduce the values of
$(y_b,p_{b-1},p_b)$ and obtain factors of $\eps^\beta$ instead on
$(s_{a-1},p_{a-1},p_a)$ (the cone event is symmetric so the same geometric argument
works in this reverse case).
Either way we have
\begin{equation}
\label{eq:Y-tube-bd}
\begin{split}
Y(\Scaff, P')
&\leq (CN)^{C\|\mcal{F}'\|}
\eps^{\beta |\supp^+ F'_{tube}|}
r^{-d|\supp J^+_{clust}(\mcal{F}')|},
\end{split}
\end{equation}
and this concludes the final case of the proof of Proposition~\ref{prp:main-diffusive-bd}.

\section{Analysis of the ladder superoperator}
\label{sec:ladders-superop}
In this section we prove some estimates on the ladder superoperator and its variants that appear
in the diagrammatic decomposition of $\wtild{\Evol}_{N\tau}$.

Recall that the single segment ladder superoperator $\mcal{L}$ is defined by
\[
\mcal{L}[A] :=
\int \ket{\xi_0^-}\bra{\xi_0^+}
\braket{\xi_1^-|A|\xi_1^+}
\Xi(\Gamma)
\Expec_{lad}
\mstack{\braket{\xi_1^+|O_{\omega^+}|\xi_0^+}}
{\braket{\xi_1^-|O_{\omega^-}|\xi_0^-}^*}\diff\Gamma.
\]
The ladder superoperator can also be expressed in terms of colored path operators.
Let $O_{lad}$ be the operator
\begin{equation}
\label{eq:Olad-def}
O_{lad} := \sum_{k=0}^{k_{max}}\int_{\Omega_k(\tau)}\int_{(\PhaseSpace)^2} \ket{\xi_0}\bra{\xi_1}
\chi_{\alpha,\sigma}(\omega;\xi_0,\xi_1)
\braket{\xi_0|O_{\omega}^{\Psi^+_{rl}([k])}|\xi_1} \diff\omega \diff \xi_0\diff\xi_1.
\end{equation}
Then
\[
\mcal{L}[A] = \Expec O_{lad}^* A O_{lad}.
\]
We define the $k$-extended ladder superoperator $\mcal{L}_k$ to be the analogue of
$\mcal{L}$ formed from extended paths.  We define the extended ladder operator
$O_{lad,k}$ by integrating over paths in $\Omega_{\alpha,\sigma}^k$ (that is, paths
having $k$ segments):
\begin{equation}
\label{eq:Oladk-def}
O_{lad,k} = \int_{\Omega_{\alpha,\sigma}^k} \Xi(\Gamma) O_\Gamma^{\Psi_{rl}^+(K(\Gamma))} \diff\Gamma.
\end{equation}
Note that $O_{lad,1}=O_{lad}$ (in the case $k=1$ the path $\Gamma = (\xi_0,\omega,\xi_1)$
and the expression expands to~\eqref{eq:Olad-def}.
Then we define
\begin{equation}
\label{eq:kladder-def}
\begin{split}
\mcal{L}_k[A] :=
\Expec O_{lad,k}^* A O_{lad,k} \\
&=
\int_{\Omega_{\alpha,\sigma}^k\times\Omega_{\alpha,\sigma}^k}
\Xi(\bm{\Gamma}) \Expec_{lad} O_{\Gamma^-}^* A O_{\Gamma^+}\diff\bm{\Gamma}.
\end{split}
\end{equation}

The ladder superoperator is an essential
building block that we use in our expansion of the superoperator $\wtild{\Evol}_{N\tau}$.
Analytically it will appear in the form of the \emph{ladder function}
\begin{equation}
\label{eq:ladder-fx}
L_k\big( \mstack{\xi, &\eta}{\xi', &\eta'} \big)
:=
\braket{\eta | \mcal{L}_k[ \ket{\xi}\bra{\xi'}] | \eta'}.
\end{equation}

In our analysis of $\wtild{\Evol}_{N\tau}$ we need a bound on the \textit{maximal ladder function},
\begin{equation}
\label{eq:Mk-def}
M_k (\xi,\eta)
:=
\sup_{\xi',\eta'\in T_{N^6,\tau}(\xi,\eta)}
|L_k\big( \mstack{\xi, &\eta}{\xi', &\eta'} \big)|,
\end{equation}
where
\[
T_{\delta,\tau}(\xi,\eta) =
\bigcup_{|s|\leq \tau}
\{(\xi',\eta')\in(\PhaseSpace)^2\mid \xi\text{ is } (s,\delta)\text{-twinned to } \xi'
\text{ and }\eta\text{ is } (s,\delta)\text{-twinned to } \eta'\}.
\]
The main result of this section is the following bound on the maximal ladder function.
\begin{lemma}
\label{lem:maximal-ladder-bd}
The maximal ladder function defined in~\eqref{eq:Mk-def} satisfies, for $k\leq N$,
\[
\sup_\xi \int M_k(\xi,\eta)\diff \eta
\leq CN^{20d}.
\]
\end{lemma}

We will prove Lemma~\ref{lem:maximal-ladder-bd} by understanding the mapping properties
of the superoperator $\mcal{L}_k$ itself.

\subsection{The semigroup property for $\mcal{L}_k$}
The first step to understandin the superoperator $\mcal{L}_k$ is to compare it to the iterated
ladder superoperator $\mcal{L}^k$ defined by the recursion
$\mcal{L}^{k+1}[A] = \mcal{L}[\mcal{L}^k[A]]$ and $\mcal{L}^1=\mcal{L}$.
The following lemma shows that these superoperators are very close when applied to
operators with good support.
\begin{lemma}
\label{lem:ladder-semigroup}
If $A$ is an operator with good support and $k\leq N$ then
\[
\|\mcal{L}_k[A] -  \mcal{L}^k[A] \|_{op} \leq  C \eps^{100}\|A\|_{op}.
\]
\end{lemma}
\begin{proof}
By Lemma~\eqref{lem:complete-extended-path} we can restrict the domain
of integration in~\eqref{eq:kladder-def} to paths which are $N^4$-complete,
up to a loss of $\eps^{100}$ in the operator norm.
Applying the same reasonin to the iterated ladder $\mcal{L}^k[A]$, we see that it suffices
to bound
\[
\int_{\Omega_{\alpha,\sigma}^k\times\Omega_{\alpha,\sigma}^k}
\Xi(\bm{\Gamma})
\One(\bm{\Gamma} \textrm{ is }N^4-\text{complete})
[ \Expec_{lad} O_{\Gamma^-}^* A O_{\Gamma^+}
- \Expec_{local-lad} O_{\Gamma^-}^* A O_{\Gamma^+}]\diff\bm{\Gamma},
\]
where $\Expec_{local-lad}$ only includes ladder partitions with rungs $(a,b)$
satisfing $\ell(a)=\ell(b)$.  These ladder partitions are precisely what are involved
in the iterated operator $\mcal{L}^k$.
By the good support of $A$, we can further restrict the integral to paths satisfying
$d_r(\xi_0^+,\xi_0^-)\leq N^4$ as well as $|(\xi_0^\pm)_p|\geq \eps^{0.2}$
(up to an error on the order $\eps^{100}\|A\|$).  For such paths
we can use~\eqref{eq:time-shift-bd} to see that
$|t_a-t_b| \leq N^{10}\eps^{-0.2} r$ for all ladder rungs $(a,b)$.
Since $N^{10}\eps^{-0.2}r\leq \sigma$ and no collisions occur within $\sigma$
of a segment boundary $jN\tau$, it follows that $\ell(a)=\ell(b)$
for each such ladder rung, and the integrand above is identically zero.
\end{proof}

As a consequence of Lemma~\ref{lem:ladder-semigroup} and Proposition~\ref{prp:short-ladder-compare}, we
can conclude that for operators $A$ with good support and $k\leq c\eps^{-0.1}$,
\begin{equation}
\label{eq:lk-to-evolk}
\| \mcal{L}_k[A] - \Evol_{\tau}^k[A]\|_{op} \leq C \eps^{2.1}\tau k\|A\|_{op}.
\end{equation}

\subsection{Trace norm bounds of the ladder superoperator}
\label{sec:operator-duality}
The quantity we would like to control involves an integral of the maximal ladder function.  Thus
far we have only proven bounds in the operator norm, which is useful for controlling
pointwise values of the ladder function.  To control such integrals we will need bounds for
the trace-norm of $\mcal{L}_k[\rho]$.

Recall that the trace-norm of an operator $\rho$ is defined by duality against the operator norm,
\[
\|\rho\|_{tr} := \sup_{\|A\|_{op}=1} \trace[A^*\rho].
\]
Operators for which this norm is finite are called trace-class operators.
Using this trace duality pairing,
the channel $\Evol_s$ is dual to the channel $\Evol_{-s}$:
\[
\trace[(\Evol_s[A])^*\rho] = \trace[A^*\Evol_{-s}[\rho]].
\]
We write $\mcal{L}_k^*$ for the dual of the ladder superoperator $\mcal{L}_k$.
This superoperator is obtained from the definition of $\mcal{L}_k$ by simply changing $i$ to $-i$,
so it satisfies the same bounds as $\mcal{L}_k$.

As a corollary we obtain the trace-norm bounds for the superoperator $\mcal{L}_k$.
\begin{lemma}
\label{lem:trace-growth-bd}
Let $\Pi_\delta$ be the multiplier operator
\[
\Ft{\Pi_\delta f}(p) = \One(|p|>\delta) \Ft{f}(p),
\]
Then for any bounded operator $A$, and $\delta > \eps^{0.1}$,
\[
\|\Pi_\delta\mcal{L}_k[A]\Pi_\delta\|_{op} \leq (1 + \eps^{2.1}\tau)^k \|A\|_{op}.
\]
By duality, if $\rho$ is a trace-class operator satisfying $\rho = \Pi_\delta\rho\Pi_\delta$,
\[
\|\mcal{L}_k[\rho]\|_{tr} \leq (1 + \eps^{2.1}\tau)^k \|\rho\|_{tr},
\]
\end{lemma}
\begin{proof}
We first expand
\[
\Pi_\delta\mcal{L}_k[A]\Pi_\delta = \Expec \Pi_\delta(O_{lad,k})^* A (O_{lad,k})\Pi_\delta.
\]
By the conservation of kinetic energy condition on the paths in $O_{lad,k}$ (coming from
the localization $\Xi(\Gamma)$ in the integrand of~\eqref{eq:Oladk-def}),
we have (up to negligible error)
\[
O_k^{ladder} \Pi_{\delta}
= \Pi_{\delta'} O_k^{ladder} \Pi_\delta
\]
with $\delta' = \delta - \eps^{0.8}$ (say).
Then by the operator Cauchy-Schwarz inequality we have
\[
\|\Pi_\delta\mcal{L}_k[A]\Pi_\delta\|_{op}
\leq \|A\|_{op} \|\mcal{L}_k[\Pi_{\delta'}]\|_{op}.
\]
The operator $\Pi_{\delta'}$ is an operator with good support, as it can be written
\[
\Pi_\delta = \int \One(|p|\geq \eps^{0.2}) \ket{\eta}\bra{\xi} \braket{\eta|p}\braket{p|\xi}\diff p\diff\xi\diff\eta.
\]
Therefore we can apply Lemma~\ref{lem:ladder-semigroup} and Proposition~\ref{prp:short-ladder-compare} to estimate
\begin{align*}
\|\mcal{L}_k[\Pi_\delta]\|_{op}
&\leq \|\mcal{L}^k[\Pi_\delta]\|_{op} + \eps^{100} \\
&\leq (1 + \eps^{2.1}\tau)^k.
\end{align*}

The dual statement follows from the calculation
\[
\|\mcal{L}_k[\rho]\|_{tr}
= \sup_{\|A\|_{op}=1}
\trace[A^* \mcal{L}_k[\rho]] = \trace[(\Pi_\delta \mcal{L}_{k}^*[A^*]\Pi_\delta)\rho].
\]
\end{proof}

\begin{corollary}[$L^1$-type bound on ladder functions]
\label{cor:L1-ladder-bd}
Let $L_k$ be the ladder function defined in~\eqref{eq:ladder-fx} and let $\rho\leq N^8$.  Suppose
moreover that $|\xi_p|\geq 2\eps^{0.1}$.  Then
\[
\int
|L_k\big( \mstack{\xi, &\eta}{\xi, &\eta} \big)|
\diff \eta
\leq C (1+\eps^{2.1}\tau)^k.
\]
\end{corollary}
\begin{proof}
Fix $\xi,\xi'\in\PhaseSpace$, and define the phase $\theta(\eta,\eta')$ so that
\[
e^{i\theta(\eta)}
L_k\big( \mstack{\xi, &\eta}{\xi, &\eta} \big)
= |L_k\big( \mstack{\xi, &\eta}{\xi, &\eta} \big)|.
\]
Let $B$ be the operator
\[
B = \int
e^{-i\theta(\eta)} \ket{\eta}\bra{\eta}\diff\eta
\]
The operator $B$ satisfies
\[
\trace[B^* \mcal{L}_k[\ket{\xi}\bra{\xi}]]
= \int |L_k\big( \mstack{\xi, &\eta}{\xi, &\eta} \big)|\diff\eta.
\]
Moreover, the operator $\ket{\xi}\bra{\xi'}$ satisfies
$\Pi_\delta \ket{\xi}\bra{\xi'}\Pi_{\delta'} = \ket{\xi}\bra{\xi'}$ up to
negligible error.  Therefore applying Lemma~\ref{lem:trace-growth-bd} for
the trace norm of $\mcal{L}_k[\ket{\xi}\bra{\xi'}]$ and
the Schur test to bound $\|B\|_{op}\leq C$, we have
\[
\trace[B^* \mcal{L}_k[\ket{\xi}\bra{\xi'}]]
\leq \|B\|_{op} \|\mcal{L}_k[\ket{\xi}\bra{\xi'}]\|_{tr}
\leq C (1+\eps^{2.1}\tau)^k,
\]
as desired.
\end{proof}

\subsection{The maximal ladder function}
Now we approach the maximal ladder function $M_k$ directly.  Recall
that the maximal ladder function is defined by
\[
M_k (\xi,\eta)
:=
\sup_{\xi',\eta'\in T_{N^6,\tau}(\xi,\eta)}
|L_k\big( \mstack{\xi, &\eta}{\xi', &\eta'} \big)|.
\]
We can simplify our task by applyin the Cauchy-Schwarz inequality
to the definition of $L_k$ to obtain
\begin{equation}
\begin{split}
|L_k\big( \mstack{\xi, &\eta}{\xi', &\eta'} \big)|
&\leq |L_k\big( \mstack{\xi, &\eta}{\xi, &\eta} \big)|^{1/2}
|L_k\big( \mstack{\xi', &\eta'}{\xi', &\eta'} \big)|^{1/2} \\
&\leq \max\{|L_k\big( \mstack{\xi, &\eta}{\xi, &\eta} \big)|,
|L_k\big( \mstack{\xi', &\eta'}{\xi', &\eta'} \big)|\}.
\end{split}
\end{equation}
Applying this into the definition of the maximal ladder function $M_k$ we obtain
\[
M_k(\xi,\eta)
\leq \sup_{\xi',\eta'\in T_{N^6,\tau}(\xi,\eta)} |L_k\big( \mstack{\xi', &\eta'}{\xi', &\eta'} \big)|.
\]
One tool we have to deal with the supremum in the maximal ladder function is the following
self-bounding property of the ladder functions.
\begin{lemma}
\label{lem:self-bounding}
\begin{equation}
|L_k\big( \mstack{\xi, &\eta}{\xi', &\eta'} \big)|
\leq
\int
|\Braket{\xi''|\xi'}|
|\Braket{\eta''|\eta'}|
|L_k\big( \mstack{\xi, &\eta}{\xi'', &\eta''} \big)|\diff\xi''\diff\eta''
\end{equation}
\end{lemma}
\begin{proof}
This follows from applying the identity
\[
\ket{\xi'} = \int \ket{\xi''} \braket{\xi''|\xi'}\diff\xi''
\]
into the definition of $L_k$ and then using the triangle inequality.
\end{proof}

We will use the fact that $e^{i\tau \Delta/2}\ket{\xi} \approx \ket{\xi_s}$
to derive a more convenient bound for $M_k$.  Let $\mcal{U}_s[A] := e^{is\Delta/2}Ae^{-is\Delta/2}$
be the free evolution channel.  Then we define the time-shifted ladder function
\[
L_{k,s}(\xi,\eta) := \braket{\eta | (\mcal{U}_s\mcal{L}^k\mcal{U}_{-s})[\ket{\xi}\bra{\xi}]|\eta}.
\]
\begin{lemma}
\label{lem:time-shift-lem}
The maximal ladder function $M_k$ defined in~\eqref{eq:Mk-def} satisfies
\[
M_k(\xi,\eta)
\leq C \sup_{|s|\leq \tau}\int  L_{k,s}(\xi',\eta')
\exp(-cd_r(\xi,\xi')^{0.5}) \exp(-cd_r(\eta,\eta')^{0.5}) \diff\xi'\diff\eta'.
\]
\end{lemma}
\begin{proof}
First we use the definition of the twinning set to rewrite the maximal ladder function as follows:
\begin{align*}
M_k(\xi,\eta) \leq
\sup_{|s|\leq \tau}
\sup_{d_r(\xi',\xi)\leq N^6}
\sup_{d_r(\eta',\eta)\leq N^6} |L_k\big(\mstack{\xi'_s,&\eta'_s}{\xi'_s,&\eta'_s}\big)|.
\end{align*}
Now we use the identity
$\ket{\xi'_s} = \int e^{-is\Delta/2}\ket{\xi''}\braket{\xi''|e^{is\Delta/2}|\xi'_s}\diff\alpha$,
the definition of $\mcal{L}_{k,s}$ and the triangle inequality to bound
\[
|L_k\big(\mstack{\xi'_s,&\eta'_s}{\xi'_s,&\eta'_s}\big)|.
\leq
\int
|\mcal{L}_{k,s}\big(\mstack{\xi'',&\eta''}{\xi'',&\eta''}\big)|
|\Braket{\xi''|e^{is\Delta/2}|\xi_s}|
|\Braket{\eta''|e^{is\Delta/2}|\eta_s}|\diff\xi''\diff\eta''.
\]
Then the lemma follows from the bound
\[
|\Braket{\xi''|e^{is\Delta/2}|\xi_s}|
\leq \exp(-cd_r(\xi'',\xi)^{0.5})
\]
and the self-bounding property of the ladder function.
\end{proof}

The time-shifted ladder functions are somewhat difficult to study directly because of the cutoff
functions $\chi(\omega;\xi,\eta)$ present in the integral definition of $\mcal{L}_k$.

We will work instead with the non-localized operators
\[
O_{\mbf{s},\mbf{p}} = \ket{p_k}\bra{p_0} e^{i\sum s_j|p_j|^2/2} \prod_{j=1}^k \Ft{V}(p_j-p_{j-1}),
\]
and define
\[
\overline{\mcal{L}}[A] := \int \rho(s_0)\rho(s_0')\rho(s_k)\rho(s_k')
\Expec_{ladder} O_{\mbf{s},\mbf{p}}^* A O_{\mbf{s}',\mbf{p}'}
\diff\mbf{s}\diff\mbf{s}' \diff\mbf{p}\diff\mbf{p}',
\]
where $\rho(s)$ is the cutoff function enforcing $s\geq S$ that is present in $\Xi$.  In other words,
the difference between $\mcal{L}$ and $\overline{\mcal{L}}$ is simply the absence of the
spatial localizations of $\chi_\alpha(\omega;\xi,\eta)$.
Then we can define the simplified ladder function to be
\[
\ovln{L}(\xi,\eta)
 := | \Braket{\eta | \ovln{\mcal{L}}[\ket{\xi}\bra{\xi}]|\eta}|.
\]
The same integration by parts argument
that proves Lemma~\ref{lem:int-by-pts} also proves that $\overline{\mcal{L}}\approx\mcal{L}$,
so we have the following bound.
\begin{lemma}
\label{lem:to-ovln}
For $k\leq N$,
\[
\sup_{\xi,\xi',\eta,\eta'}|
|L_k\big( \mstack{\xi_s, &\eta_s}{\xi_s, &\eta_s} \big)
- \overline{L}_k\big( \mstack{\xi_s, &\eta_s}{\xi_s, &\eta_s} \big)|
\leq \eps^{100}.
\]
\end{lemma}

Let $\overline{L}$ be the ladder function derived from the superoperator $\overline{\mcal{L}}$,
\[
\ovln{L} (\xi,\eta)
:= \braket{\eta | \ovln{\mcal{L}}[\ket{\xi}\bra{\xi}]|\eta}.
\]
We use $\ovln{\mcal{L}}$ to also define $\ovln{L}_{k,s}$ by
\[
\ovln{L}_{k,s}(\xi,\eta)
:= \braket{\eta | \mcal{U}_s\ovln{\mcal{L}}^k\mcal{U}_{-s}[\ket{\xi}\bra{\xi}]|\eta}
=  \braket{\eta | \ovln{\mcal{L}_s}^k[\ket{\xi}\bra{\xi}]|\eta},
\]
where $\ovln{\mcal{L}_s} = \mcal{U}_s\ovln{\mcal{L}}\mcal{U}_{-s}$.

Because the integration over $\xi'$ and $\eta'$ is only over a volume on the order $N^C\tau\leq \eps^{-10}$,
we can combine Lemma~\ref{lem:to-ovln} with Lemma~\ref{lem:self-bounding} and Lemma~\ref{lem:time-shift-lem}
to obtain
\begin{equation}
\label{eq:Mk-bd}
M_k(\xi,\eta) \leq C\sup_{|s|\leq \tau}
 \int |\overline{L_{k,s}}(\xi',\eta')|
\exp(-c d_r(\xi',\xi)^{0.5})
\exp(-c d_r(\eta',\eta)^{0.5})\diff\xi'\diff\eta'
+ O(\eps^{50}).
\end{equation}
Now we use the fundamental theorem of calculus to bound
\begin{equation}
\label{eq:Lks-ftoc}
|\overline{L_{k,s}}(\xi',\eta')|
\leq
|\overline{L_k}(\xi',\eta')|
+ \int_{-\tau}^\tau |\partial_s \overline{L_{k,s}}(\xi',\eta')|\diff s.
\end{equation}
Next we compute the derivative of $\partial_s \overline{L_{k,s}}$
\begin{equation}
\label{eq:Lks-derivative}
\partial_s \ovln{\mcal{L}}_{k,s}
=
\partial_s (\ovln{\mcal{L}}_s)^k
= \sum_{j=0}^{k-1} \ovln{\mcal{L}}_s^{k-j-1} \partial_s (\ovln{\mcal{L}}_s)  \ovln{\mcal{L}}_s^j,
\end{equation}
and we observe that
\[
\partial_s \ovln{\mcal{L}}_s[A] =
\partial_s \int \rho(s_0)\rho(s_k)
\rho(s_0')\rho(s_k')
\Expec e^{-is\Delta/2} O_{\mbf{s},\mbf{p}}^* e^{is\Delta/2}
A e^{-is\Delta/2} O_{\mbf{s}',\mbf{p}'} e^{is\Delta/2}
\diff\mbf{s}\diff\mbf{p}
\diff\mbf{s'}\diff\mbf{p'}.
\]
Given a time sequence $\mbf{s}=(s_0,\cdots,s_k)$, define the shifted sequence
$I_s\mbf{s} = (s_0-s,s_1,\cdots, s_{k-1},s_k+s)$.  Then
\[
e^{-is\Delta/2} O_{\mbf{s},\mbf{p}}^* e^{is\Delta/2}
= O_{I_s\mbf{s}, \mbf{p}},
\]
so by then applying a change of variables we have
\begin{equation}
\label{eq:dLs-formula}
\partial_s \ovln{\mcal{L}}_s[A] =
\int f_s(s_0,s_k,s_0',s_k')
\Expec e^{-is\Delta/2} O_{\mbf{s},\mbf{p}}^* e^{is\Delta/2}
A e^{-is\Delta/2} O_{\mbf{s}',\mbf{p}'} e^{is\Delta/2}
\diff\mbf{s}\diff\mbf{p}
\diff\mbf{s'}\diff\mbf{p'},
\end{equation}
where
\[
f_s(s_0,s_k,s_0',s_k') := \frac{d}{ds} \rho(s_0-s)\rho(s_k+s) \rho(s_0'-s)\rho(s_k'+s).
\]
The key lemma is the following estimate on the superoperator $\partial_s\ovln{\mcal{L}_s}$.

\begin{lemma}
\label{lem:dLs-bd}
If $A$ is an operator with good support, then
\[
\|\partial_s \ovln{\mcal{L}}_s[A]\|_{op} \leq C\eps^2 |\log\eps|^C \|A\|_{op}.
\]
\end{lemma}
\begin{proof}
Using the operator Cauchy-Schwarz inequality, it suffices to prove
\[
\|\partial_s \ovln{\mcal{L}}_s[P_\delta]\|_{op} \leq C\eps^2,
\]
where $\delta=\eps^{0.2}$ and $P_\delta$ is the wavepacket cutoff
\[
P_\delta = \int_{|\xi_p|\geq \delta}\ket{\xi}\bra{\xi}\diff \xi.
\]
Then we use the same integration by parts from Lemma~\ref{lem:int-by-pts} to
approximate, for $g$ smooth to scale $\eps^{-1.1}$,
\[
\int g(\mbf{s})O_{\mbf{s},\mbf{p}} \diff\mbf{s}\diff\mbf{p}
=
\int g(\mbf{s}) \chi_\alpha(\omega;\xi,\eta)
\ket{\eta}\bra{\xi} \braket{\xi|O_\omega|\eta}\diff\omega\diff\xi\diff\eta
\]
up to an error that is negligible in operator norm.  Then we estimate
$\|\partial_s\ovln{\mcal{L}}_s[A]\|_{op}$ using the Schur test and the same
calculations as in the proof of Proposition~\ref{prp:short-ladder-compare}.
In particular, without any improvements from the presence of clusters in the partition or from
the presence of recollision or tube events, the argument of Proposition~\ref{prp:short-ladder-compare}
yields
\[
\|\partial_s\ovln{\mcal{L}}_s[A]\|_{op} \leq C\|A\|_{op}\sum_{k_+,k_-=0}^{k_{max}}  (C\eps^2 \tau |\log\eps|^C)^{(k_++k_-)/2}.
\]
In summary, the bound above is obtained from the following considerations:
each pair $(a,b)$ in the generalized ladder contributes $\eps^2$ from the two factors of $\eps$,
$\tau$ from the choice of the time of the collision, and a factor of
$r^{-1}$ from the integration over the momentum annulus that is cancelled by a factor of $r$
for the time variable of the second collision in the pair.

Now we obtain an improvement because of the presence of the function $f_s(s_0,s_k,s_0',s_k')$ in the integrand.
This function is bounded by $\sigma^{-1}$ and is supported on the set
\[
\{s_0\leq \sigma\}\cup\{s_k \leq \sigma\} \cup \{s_0'\leq \sigma\} \cup \{s_k'\leq \sigma\}.
\]
The effect is that the integration over the $s_0$ variable (or $s_0'$, or $s_k$, or $s_k'$ variable)
includes an additional factor of $\sigma^{-1}$ and is integrated over a segment of length $\sigma$
rather than $\tau$.  The effect of this is to produce an additional improvement of $\sigma^{-1}\sigma/\tau=\tau^{-1}$
over the bound given above.  Since moreoever there must be at least one (and therefore at least $2$) collisions
for the integrand to be nonzero, we conclude that
\[
\|\partial_s\ovln{\mcal{L}}_s[A]\|_{op} \leq C\|A\|_{op}\sum_{2\leq |k_++k_-|\leq 2k_{max}}^{k_{max}}
C \tau^{-1}(C\eps^2 \tau|\log\eps|^C)^{(k_++k_-)/2}
\leq C\eps^2|\log\eps|^{C},
\]
as desired.
\end{proof}

We now have the ingredients needed to prove Lemma~\ref{lem:maximal-ladder-bd}.
\begin{proof}[Proof of Lemma~\ref{lem:maximal-ladder-bd}]
We integrate~\eqref{eq:Mk-bd} and apply~\eqref{eq:Lks-ftoc} to obtain
\begin{align*}
\int M_k(\xi,\eta)\diff\eta
&\leq C \int |\ovln{L_k}(\xi',\eta')|\exp(-cd_r(\xi',\xi)^{0.5}) \diff\xi'\diff\eta' \\
&\qquad +
C \int_{-\tau}^\tau\int |\partial_s\ovln{L_{k,s}}(\xi,\eta)|
\exp(-cd_r(\xi,\xi')^{0.5})\diff\xi'\diff\eta\diff s
\end{align*}
The main term comes from $\xi'$ satisfying $d_r(\xi',\xi)\leq N^4$, so we have
\begin{align*}
\int M_k(\xi,\eta)\diff\eta
&\leq C N^{8d} \sup_{d_r(\xi',\xi)\leq N^4}\int |\ovln{L_k}(\xi',\eta')|\diff\eta'\\
&\qquad +  CN^{8d} \sup_{d_r(\xi',\xi)\leq N^4}
\int_{-\tau}^\tau\int |\partial_s\ovln{L_{k,s}}(\xi,\eta)|
\diff\eta\diff s
+ C\eps^{100}.
\end{align*}
The first term can be directly bounded using Corollary~\ref{cor:L1-ladder-bd}.
The second term is bounded in the same way, by first deriving a trace-norm bound
for the operator $\partial_s\ovln{\mcal{L}}_{k,s}[\ket{\xi}\bra{\xi}]$ using
the duality argument of Lemma~\ref{lem:trace-growth-bd} and then
applying this to bound the integral over $\eta$ exactly as in the proof of Corollary~\ref{cor:L1-ladder-bd}.
\end{proof}

\subsection{Superoperators related to the ladder}
In this section we transfer the bounds we obtained on the ladder superoperator $\mcal{L}$
to some other superoperators.  The first of these is the stick superoperator
$\mcal{S}$, defined by
\[
\mcal{S}[A] := \int
\ket{\xi_0^-}\bra{\xi_0^+}
\braket{\xi_1^-|A|\xi_1^+}
\Xi(\Gamma) \Expec_{stick}
\mstack{\braket{\xi_1^+|O_{\omega^+}|\xi_0^+}}
{\braket{\xi_1^-|O_{\omega^-}|\xi_0^-}^*}\diff\Gamma,
\]
where now the expectation is over stick partitions (that is, ladders without rungs).  This can also
be expressed in terms of colored operators,
\[
\mcal{S}[A] = \Expec (O_{stick})^* A O_{stick},
\]
where $O_{stick}$ includes a sum over all stick colorings $\Psi_{rl}^{+,0}$
instead of all ladders $\Psi_{rl}^+$.  Since stick colorings are precisely ladder colorings
with no rungs, $O_{stick} = \Expec O_{lad}$, so $\mcal{S}[A]$ can also be written
\[
\mcal{S}[A] = (\Expec O_{lad})^*  A (\Expec O^{lad}).
\]
Therefore all bounds on the operator norm of $\mcal{L}[A]$ are inherited by $\mcal{S}[A]$,
since $(\Expec O^{ladder})^*(\Expec O^{ladder}) \leq \Expec (O^{ladder})^*(O^{ladder})$.

This discussion can be generalized to handle ladder superoperators with a specified number of rungs.
Define the extended ladder operator with a specified number of rungs
\[
O_{lad,k,=h} =
\int_{\Omega_{\alpha,\sigma}^k}
\Xi(\Gamma) O_{\Gamma}^{\Psi_{rl}^{+,h}(K(\Gamma))}\diff\Gamma.
\]
Then we define
\[
\mcal{L}_{k,=h}[A] := \Expec O_{lad,k,h}^* A O_{lad,k,h}.
\]
Note that $\mcal{L}_{1,0}$ is the stick superoperator considered above.
Let $L_{k,=h}$ be the ladder functions associated to $\mcal{L}_{k,=h}$,
\begin{equation}
L_{k,=h} \big(\mstack{\xi,&\eta}{\xi',&\eta'}\big)
:=
|\Braket{\eta|\mcal{L}_{k,=h}[\ket{\xi}\bra{\xi'}]|\eta'}|.
\end{equation}
An application of Cauchy-Schwarz shows that
\begin{equation}
L_{k,=h} \big(\mstack{\xi,&\eta}{\xi',&\eta'}\big)
\leq
L_{k,=h}\big(\mstack{\xi,&\eta}{\xi,&\eta}\big)^{1/2}
L_{k,=h}\big(\mstack{\xi',&\eta'}{\xi',&\eta'}\big)^{1/2}
= L_{k,=h}(\xi,\eta)^{1/2} L_{k,=h}(\xi',\eta')^{1/2}.
\end{equation}

Now we will compare the functions $L_{k,=h}$ to $L_k$.  The key observation is that
$\mcal{Q}(\Psi^{+,h}_{rl}(A), \Psi^{+}_{rl}(B))$ is the set of ladder partitions with exactly
$h$ rungs (as all other partitions have singleton sets).  Therefore  $\mcal{L}_{k,=h}$ has the
remarkable representation
\[
\mcal{L}_{k,=h}[A]
= \Expec O_{lad,k,=h}^* A O_{lad,k}.
\]
In particular, $L_{k,=h}(\xi,\eta)$ can be written
\[
L_{k,=h} =
|\Expec \braket{\eta | O_{lad,k,=h}^*|\xi} \braket{\xi|O_{lad,k}|\eta}|,
\]
so that an application of the Cauchy-Schwarz inequality shows
\[
L_{k,=h}(\xi,\eta) \leq L_{k,=h}(\xi,\eta)^{1/2} L_k(\xi,\eta)^{1/2},
\]
which upon rearranging becomes
\[
L_{k,=h}(\xi,\eta) \leq L_k(\xi,\eta).
\]
Therefore all bounds we obtained for the ladder function $L_k$ apply to the ladder functions $L_{k,=h}$
(including the case $h=0)$.  In particular, for $k\leq N$,
\begin{equation}
\label{eq:maximal-rung-bd}
\sup_{\xi} \int M_{k,=h}(\xi,\eta)\diff \eta \leq CN^C.
\end{equation}

\section{The path integral}
\label{sec:duhamel}
In this section we provide a more detailed derivation of Lemma~\ref{lem:int-by-pts}.

We start with the Duhamel identity
\[
    e^{-i\tau H} = e^{i\tau\Delta/2} +
    i\eps \int_0^\tau e^{-i(\tau-s) H} V e^{is\Delta/2} \diff s.
\]
Iterating this identity $k_{max}$ times (later we will set
$k_{max}= 10^{100}$), we arrive at the expression
\[
    e^{-i\tau H} = e^{i\tau\Delta/2} + \sum_{k=1}^{k_{max}}
    T_k(\tau) + R_{k_{max}}(\tau)
\]
where
\[
    T_k(\tau) := (i\eps)^k
    \int_{\triangle_k(\tau)}
    e^{-is_k\Delta/2} V e^{-is_{k-1}\Delta/2}
    V\cdots V e^{-is_0\Delta/2} \diff s_0\dots \diff s_k
\]
and
\[
    R_{k_{max}}(\tau) :=
    \int_0^\tau e^{-i(\tau-s) H} T_{k_{max}}(s)\diff s,
\]
where $\triangle_k(\tau)$ is the simplex
\[
    \triangle_k(\tau) :=
    \{(s_0,\dots,s_k)\in (\Real_+)^{k+1} \mid
    \sum_{j=0}^k s_j = \tau\}.
\]
An interpretation of this expansion is that the
 first term accounts for the possibility that
the particle moves freely (in a straight line), and the term $T_k$
accounts for the trajectories a particle might take which involve $k$
scattering events.  This intuitive picture is only valid for times
$\tau\ll \eps^{-2}$ for which the first term dominates.  For longer times,
although the identity still holds, there is significant cancellation
between the terms which precludes any simple interpretation of the terms
on their own.

Our next step is to further decompose the operators $T_k$ as an integral
over possible phase-space paths of the particle, taking into account both
the spatial localization of the particle and the changes in momentum of
the particle due to scattering.  To keep track of the latter we use the
Fourier inversion formula, as expressed in the
identity
\[
\Id = \int \ket{p}\bra{p}\diff p.
\]

In combination with the following formula for the free evolution of
momentum eigenstates $\braket{x|p} = e^{ip\cdot x}$
\[
    e^{is\Delta/2}\ket{p} = e^{-is|p|^2/2} \ket{p},
\]
we obtain the identity
\[
    T_k(\tau) = (i\eps)^k \int_{\triangle_k(\tau)}
    \int_{(\Real^d)^{k+1}}
    \prod_{j=0}^k e^{-is_j|p_j|^2/2}
    \prod_{j=1}^k \braket{p_j | V | p_{j-1}}
    \ket{p_k}\bra{p_0}
    \diff \mbf{s}\diff\mbf{p}.
\]
Here we have used the boldfaced symbols
$\mbf{s}=(s_0,\dots,s_k)$ and $\mbf{p}=(p_0,\dots,p_k)$ as a shorthand
for the list of all variables.  The next step is to decompose the potential
$V$ spatially.  To do this let $b\in C_c^\infty(\Real^d)$ be a smooth
positive bump function supported in the unit ball and satisfying $\int b=1$,
and let $b_r(x) = r^{-d} b(x/r)$ be the scaled version of $b$ which preserves
the integral.  Then we have the identity
\[
    V(x) = \int V(x) b_r(y-x)\diff y.
\]
Now define the localized potential $V_y$ by
\[
    V_y(x) := V(x+y) b_r(y),
\]
where $b_r = r^{-d}b(x/r)$ is a scaled version of a bump function
$b\in C_c^\infty(\Real^d)$ that we choose to be supported in $B_1$,
be normalized so that $\int b =1$, and satisfy the Gevrey-type Fourier
decay
\[
    |\Ft{b}(p)| \leq C\exp(-c|p|^{0.999}).
\]
Applying this decomposition into the term $\braket{p|V|q}$ yields
\[
    \braket{p|V|q} = \Ft{V}(p-q) =
    \int e^{-iy\cdot (p-q)} \Ft{V_y}(p-q)\diff y,
\]
and so we write
\[
    T_k(\tau) = (i\eps)^k \int_{\triangle_k(\tau)}
    \int_{(\Real^d)^{k+1}} \int_{(\Real^d)^k}
    \prod_{j=0}^k e^{-is_j|p_j|^2/2}
    \prod_{j=1}^k e^{-iy_j (p_j-p_{j-1})}
    \Ft{V_{y_j}}(p_j-p_{j-1})
\ket{p_k}\bra{p_0}
    \diff \mbf{s}\diff\mbf{p} \diff\mbf{y}.
\]
For convenience we introduce the operator
\[
    O_\omega := \ket{p_k}\bra{p_0}
    \prod_{j=0}^k e^{-is_j|p_j|^2/2}
    \prod_{j=1}^k e^{-iy_j (p_j-p_{j-1})}
    \Ft{V_{y_j}}(p_j-p_{j-1}).
\]
There is an oscillatory phase present in the integral, given by
\begin{align*}
    \varphi_k(\mbf{s},\mbf{p},\mbf{y}) &:=
    \sum_{j=0}^k s_j |p_j|^2/2
    + \sum_{j=1}^k y_j \cdot (p_j-p_{j-1}).
\end{align*}
Aside from this phase, the integrand is smooth to scale $r^{-1}$
in the $\mbf{p}$ variable, as $V_y$ is defined to be localized to the
ball of radius $r$.
Thus it makes
sense to compute the stationary points of $\varphi_k$ with respect
to perturbations in $\mbf{p}$.  We observe that, for
$0 < j <k$,
\begin{align*}
    \partial_{p_j} \varphi_k = y_j + s_jp_j - y_{j+1}.
\end{align*}
We can therefore use stationary phase to localize the integral to sequences
of collision locations $\mbf{y}$ which are approximately connected by
the straight line segments $s_jp_j$.  This confirms the picture that
$T_k$ is an integral over paths that the particle can take
in which scattering events are interspersed with free evolution.

We will also need to perform stationary phase in the $\mbf{s}$ variable,
but this is delicate due to the fact that we are integrating over the
simplex $\triangle_k$, which has sharp corners.  Ultimately, integration
by parts in the $\mbf{s}$ variables will lead to an approximate conservation
of kinetic energy.  There is a slight complication, which is that, due to
the singularity coming from the constraint $s_j\geq0$, the particle may
violate conservation of kinetic energy for short times.

To prove Lemma~\ref{lem:int-by-pts} we will first prove a general
oscillatory integral estimate that is adapted for analytic-type phases
and functions.  The analyticity is not necessary but simply makes for
convenient estimates.

\begin{lemma}[General oscillatory integral estimate]
    \label{lem:osc-int}
    Let $\varphi:\Real^N\to\Real$ be an analytic phase satisfying the
    estimate
    \begin{equation}
        \label{phase-est}
        |\partial^\beta\varphi(\mbf{x})| \leq
        \prod_{a=1}^N (C\ell_a (1+\beta_a))^{\beta_a}.
    \end{equation}
    for any multi-index $\beta$ with $|\beta|\geq 2$ and some
    scales $\{\ell_a\}_{a=1}^N$.
    Also let $w,W:\Real^N\to\Real$ be functions satisfying
    \begin{equation}
        \label{wW-bd}
        |\partial^\beta w(\mbf{x})| \leq
        \prod_{a=1}^N \ell_a^{\beta_a} (C(1+\beta_a))^{1.001\beta_a}
        W(\mbf{x}).
    \end{equation}
    Then
    \begin{equation}
        \label{exp-decay-bd}
    \Big|\int e^{i\varphi(\mbf{x})} w(\mbf{x})\diff\mbf{x}\Big|
    \leq C(N) \int W(\mbf{x})
    \prod_{j=1}^N \exp(-c_N|\ell_j^{-1}\partial_j\varphi|^{0.99})
    \diff \mbf{x}.
    \end{equation}
\end{lemma}
\begin{proof}
By a rescaling, we can assume that $\ell_a=1$ for each $a\in[N]$.

Fix a partition of unity
\[
    1 = b_0(x) + \sum_{j=1}^\infty b_j(x)
\]
where $b_j\in C_c^\infty(\Real)$, $b_0$ is supported on $[-2,2]$, and
$b_j$ is supported on the set $\{2^{j-1} \leq |x|\leq 2^{j+1}\}$, with
\begin{equation}
    \label{bGev-bd}
    \sup_x |\partial^k b_j(x)| \leq 2^{-jk} (Ck)^{1.001k}.
\end{equation}
For each $a\in[N]$ we therefore have the identity
\[
    1 = b_0(\partial_a\varphi) +
    \sum_{j=1}^\infty b_j(\partial_a\varphi).
\]
We arrive at the decomposition
\begin{equation}
    \label{int-decomp}
\int e^{i\varphi(\mbf{x})} w(\mbf{x})\diff\mbf{x}
= \sum_{j_a}
\int e^{i\varphi(\mbf{x})} w(\mbf{x})
\prod_{a=1}^N b_{j_a}(\partial_a \varphi)\diff \mbf{x}.
\end{equation}
The indices $j_a$ dictate how many times we integrate by parts in
each variable.  Specifically, when $j_a>0$ we integrate by parts
$\lfloor j_a^\delta\rfloor$ times, and we let $\alpha(\mbf{j})$
denote the multi-index corresponding to the indices $\mbf{j}$.
If $j_a>0$ we integrate by parts in the $a$ variable
using the differential operators
\begin{equation}
    \begin{split}
    L_a^*f &= -i \frac{1}{\partial_a\varphi} \partial_a \\
    L_af &= i \partial_a (\frac{f}{\partial_a\varphi}),
    \end{split}
\end{equation}
which satisfy
\[
    L_a^*e^{i\varphi(\mbf{x})} = e^{i\varphi(\mbf{x})}.
\]
We write $L^{\alpha}$ to mean
\[
    L_N^{\alpha_N} \cdots L_1^{\alpha_1}.
\]
The term corresponding to $\mbf{j}=(j_1,\dots,j_N)$, after integration
by parts, is
\begin{equation}
\begin{split}
\int e^{i\varphi(\mbf{x})} w(\mbf{x})
\prod_{a=1}^N b_{j_a}(\partial_a \varphi)\diff \mbf{x}
=
\int e^{i\varphi(\mbf{x})} L^{\alpha(\mbf{j})}
\Big(w(\mbf{x})
\prod_{a=1}^N b_{j_a}( \partial_a \varphi)\Big)\diff \mbf{x}
\end{split}
\end{equation}
Only $C^N$ terms in the sum~\eqref{int-decomp} contribute at any
point $\mbf{x}$, so it suffices to prove the bound
\begin{equation}
    \label{LSK-bd}
    L^{\alpha(\mbf{j})}
\Big(w(\mbf{x})
\prod_{a=1}^N b_{j_a}( \partial_a \varphi)\Big)
\leq C^N W(x) \prod_{a=1}^N\exp(-c|\partial_a\varphi|).
\end{equation}
The left hand side here decomposes as a sum of terms of the form
\begin{equation}
    \label{expanded-eq}
    \Big(\prod_{a=1}^N
    (\partial_a\varphi)^{-\alpha_a-M_a}
    \prod_{k=1}^{m_a} \partial^{\beta_{k,a}}\partial_a\varphi\Big)
    \partial^\gamma w
    \prod_{a=1}^N
    \partial^{\nu_a}b_{j_a}( \partial_a \varphi)
\end{equation}
where the indices $M_a$, $m_a$, $\beta_{k,a}$, $\gamma$, and $\nu$
satisfy
\begin{align*}
    \sum_{k=1}^{m_a} |\beta_{k,a}| &= M_a\geq 0 \\
    \sum_{a} \sum_{k=1}^{m_a} \beta_{k,a} +
    \gamma + \sum_{j=1}^N \nu_j
    &= \alpha(\mbf{j}).
\end{align*}
Applying~\eqref{phase-est}, \eqref{wW-bd}, and~\eqref{bGev-bd}, we can
bound~\eqref{expanded-eq} by
\begin{equation}
\begin{split}
\Big(\prod_{a=1}^N
&(\partial_a\varphi)^{-\alpha_a-M_a}
\prod_{k=1}^{m_a} \partial^{\beta_{k,a}}\partial_a\varphi\Big)
\partial^\gamma w
\prod_{a\in[N]}
\partial^{\nu_a}b_{j_a}( \partial_a \varphi) \\
&\leq
\Big( \prod_{a=1}^N (2^{j_a})^{-\alpha_a-M_a} C(1+M_a)^{M_a} \Big)
\prod_{a=1}^N (C(1+\gamma_a))^{1.001\gamma_a} W(\mbf{x})
\prod_{a=1}^N
2^{-j_a}\prod_{k=1}^N(C(1+\nu_{a,k}))^{\nu_{a,k}} \\
&\leq
\Big(\prod_{a=1}^N 2^{-\alpha_a j_a} (C(1+\alpha_a))^{1.001\alpha_a}\Big)
W(\mbf{x}).
\end{split}
\end{equation}

Because $\partial_j\varphi$ depends on only $O(1)$ variables, each of the
multi-indices $\nu_a$ and $\beta_{k,a}$ are constrained to have only
$O(1)$ nonzero values (or else the term does not contribute).  Thus
there are only
\[
    \prod_{a=1}^N (C|\alpha_a|)^{|\alpha_a|}.
\]
possible choices for these indices.

To conclude we show that each such term is bounded by the
right hand side of~\eqref{LSK-bd}.  Note that only second-order derivatives
of $\varphi$ appear in~\eqref{expanded-eq}, which is what allows us
to bound $\varphi^\beta_{k,j}\partial_j\varphi$ as well as the derivatives
of $b(\partial_j\varphi)$.  To see how the $(1+|\partial_j\varphi|)^{-K}$
term comes about, there are two cases to consider.  When $j\in S$
this decay comes directly from the $(\partial_j\varphi)^{-K-M_j}$ term
as well as the fact that $|\partial_j\varphi|>1$ due to the
$(1-b(\partial_j\varphi))$ localization.

On the other hand when $j\not\in S$ we have $|\partial_j\varphi|\leq 1$,
so $(1+|\partial_j\varphi|)^{-K} \geq 2^{-K}$, and so by multiplying
$2^{KN}$ we can remove such terms from the right hand side.
\end{proof}
\begin{lemma}
    Let $\omega = (\mbf{s},\mbf{p},\mbf{y})\in\Omega_k(\tau)$ be a path,
and let $J$ be the index with maximal $s_J$.  Suppose that $\omega$ satisfies
\[
||p_j|^2/2 - |p_J|^2/2| \leq C\alpha \max\{s_j^{-1}, |p_J| r^{-1}\}.
\]
for each $j\in[0,k]$.  Then $\omega$ also satisfies
\[
||p_j|^2/2 - |p_{j'}|^2/2| \leq C'\alpha \max\{s_j^{-1}, s_{j'}^{-1}, |p_{j}|r^{-1}, \alpha r^{-2}\}
\]
for any $j,j'\in[0,k]$.
\end{lemma}
\begin{proof}
Using the triangle inequality we have
\[
||p_j|^2/2 - |p_{j'}|^2/2| \leq 2C\alpha
\max\{s_j^{-1}, s_{j'}^{-1}, |p_J|r^{-1}\}.
\]
In the case $\max\{s_j^{-1},s_{j'}^{-1},|p_J|r^{-1}\} \geq |p_J|r^{-1}$, the conclusion is
trivial.  Otherwise, we have in particular $s_{j}^{-1} \leq |p_J|r^{-1}$ and so
\[
||p_j|^2/2 - |p_J|^2/2| \leq C\alpha |p_J|r^{-1}.
\]
There are two cases to consider now.  If $|p_J| \gg C\alpha r^{-1}$, then it follows that
$||p_j| - |p_J|| \lsim r^{-1}$, and so in particular $|p_j| \gtrsim |p_J|$. Otherwise
$|p_J|r^{-1} \leq \alpha r^{-2}$, and so
\[
\max\{s_j^{-1}, s_{j'}^{-1}, |p_J|r^{-1}\}
\lsim \max\{s_j^{-1}, s_{j'}^{-1}, \alpha r^{-2}\}.
\]
\end{proof}

\begin{proof}[Proof of Lemma~\ref{lem:int-by-pts} using Lemma~\ref{lem:osc-int}]
Using the identity $\Id = \int \ket{\xi}\bra{\xi}\diff\xi$ we have
\begin{equation}
\begin{split}
T_k &= \int_{\PhaseSpace}\int_{\PhaseSpace}
\ket{\eta}\bra{\xi}\braket{\eta| T_k|\xi}\diff\xi\diff\eta. \\
&= \int_{\PhaseSpace}\int_{\PhaseSpace}\int_{\Omega_k(\tau)}
\ket{\eta}\bra{\xi}
\braket{\eta | O_\omega | \xi} \diff\omega \diff\xi\diff\eta
\end{split}
\end{equation}
The idea is to use Lemma~\ref{lem:osc-int} in the integration over
the $\omega$ variables.
As a first step we expand the definition
of the integrand.
\begin{equation}
    \langle \phi_\eta, O_\omega\phi_\xi\rangle
=
\Ft{\phi_\eta}(p_k) (\Ft{\phi_\xi}(p_0))^*
\prod_{j=0}^k e^{-is_j|p_j|^2/2}\prod_{j=1}^k e^{-iy_j(p_j-p_{j-1})}
\Ft{V_{y_j}}(p_j-p_{j-1}).
\end{equation}
Letting $\xi=(x,p)$ and $\eta=(y,q)$, we have
\begin{equation}
    \begin{split}
    \Ft{\phi_\xi}(p_0)
    &= r^{d/2} e^{-ip_0\cdot x} \Ft{\chi_{env}}(r(p_0-p)) \\
    \Ft{\phi_\eta}(p_k)
    &= r^{d/2} e^{-ip_k\cdot y} \Ft{\chi_{env}}(r(p_k-q)),
    \end{split}
\end{equation}
where $\chi_{env}$ is the envelope function defining the family of
wavepackets, and we can choose this family so that
$\Ft{\chi_{env}}$ has support in the ball $B_{1/r}$.

The integration over $\omega$ therefore has an oscillation described
by the phase function
\begin{equation}
\varphi_{\xi,\eta}(\omega)
= p_0\cdot x - p_k\cdot y - \sum_{j=0}^k s_j|p_j|^2/2
- \sum_{j=1}^k y_j\cdot(p_j-p_{j-1}).
\end{equation}
Before we can apply Lemma~\ref{lem:osc-int} we need to deal with the
fact that we integrate over $\Omega_k(\tau)$, which is singular due to
the integration over the simplex $\triangle_k(\tau)$.  To deal
with this problem we decompose the indicator function of the simplex
as a sum of smooth functions.   The integration over $\triangle_k$
has the form
\[
    \int_{\triangle_k(\tau)} f(\mbf{s}) \diff\mbf{s}
    = \int_{\Real^{k+1}} \delta(\tau - \sum s_j) f(\mbf{s})
    \prod_{j=0}^k \One_+(s_j) \diff \mbf{s},
\]
where $\One_+(s) = \One_{\{s>0\}}$ is the indicator function for the
set of positive numbers.

First we decompose $\One_+$ as a sum of smooth functions,
\[
    \One_{\{s>0\}} = \sum_{j=-\infty}^\infty \chi_j(s),
\]
where $\rho_j$ has support
on the set $[2^{j-1}, 2^{j+1}]$, and each $\rho_j$ satisfies
\[
    \sup |D^k \rho_j| \leq (C 2^{-j} k^{1.001})^k.
\]
Given a multi-index $\mbf{a}=(a_0,a_1,\dots,a_k)$, we choose a special index
$J=J(\mbf{a})$ so that $a_J$ is minimized (breaking ties by, say, choosing the
minimal such index).  We take care of the delta function on the $\mbf{s}$
variables by writing
\[
s_J = \tau - \sum_{j\not=J(\mbf{a})} s_j.
\]
We will also decompose the integral according to the magnitude of the kinetic
energy $|p_J|^2/2$.  In this case we write
\[
1 = \sum_{j=b_0}^\infty \chi_j^{ke}(|p_J|^2/2),
\]
where $2^{b_0}\sim \alpha^2 r^{-2}$ and $\chi_{b_0}^e$ is supported on $[-10r^{-2},10r^{-2}]$
and is smooth to scale $r^{-2}$, and each of the $\chi_j^e$ is the same as above
(the letter $ke$ simply denotes that these cutoffs are used for kinetic energy).

Finally we can write down the decomposition of the integral that we will use:
\begin{equation}
    \label{a-Tk-decomp}
\begin{split}
\langle \phi_\eta, T_k \phi_\xi\rangle
= \sum_{\mbf{a}}\sum_{b=b_0}^\infty
&\int_{\Real^{k+1}}\diff\mbf{s}\int_{(\Real^d)^{k+1}}\diff\mbf{p}
\int_{(\Real^d)^k} \diff\mbf{y}
\chi_{a_J}(\tau -\sum_{j\not= J} s_j)
\Big(\prod_{j\not=J(\mbf{a})}^k \chi_{a_j}(s_j)\Big)
\chi_b^{ke}(|p_J|^2/2) \\
&\qquad e^{i\varphi(\mbf{s},\mbf{p},\mbf{y})}
r^d\Ft{\chi_{env}}(r(p_0-p)) \Ft{\chi_{env}}(r(p_k-q))
\prod_{j=1}^k \Ft{V_{y_j}}(p_j-p_{j-1}).
\end{split}
\end{equation}
When discussing a particular term in the sum, we will use $E=E(b)=2^b$ to
denote the `dominant' kinetic energy scale of the path.
Now we compute the differential of $\varphi$ with respect to the $\omega$
variables (except $s_J$, which is determined by the rest of the $s_j$)
\begin{equation}
\begin{split}
    \partial_{s_j} \varphi &= |p_j|^2/2 - |p_J|^2/2 \\
    \partial_{p_j} \varphi &= y_{j+1} - y_j - s_jp_j.
\end{split}
\end{equation}
For each $j$ let $\ell_j := \alpha \max\{2^{-a_j}, E^{1/2} r^{-1}\}$, and set
\[
    \chi_{main,\mbf{a}}(\omega)
    = \prod_{j\not= J(\mbf{a})} \rho(\ell_j^{-1}(|p_j|^2/2-|p_J|^2/2))
\prod_{j=0}^k \rho( \alpha^{-1} r^{-1} (y_{j+1}-y_j-s_jp_j)).
\]
The idea is to decompose
\begin{equation*}
\begin{split}
\braket{\phi_\eta, T_k,\phi_\xi}
= \sum_{\mbf{a}} &\int_{\Real^{k+1}}\diff\mbf{s} \int_{(\Real^d)^{k+1}} \diff\mbf{p}
\int_{(\Real^d)^k}\diff\mbf{y} (\prod_{j=0; j\not= J(\mbf{a})}^k \chi_{a_j}(s_j))
\chi_0(\tau - \sum_{j\not= J(\mbf{a})} s_j) \\
&\qquad \times e^{i\varphi(\mbf{s},\mbf{p},\mbf{y})}
 r^d \Ft{\chi_{env}}(r(p_0-p))\Ft{\chi_{env}}(r(p_k-q)) \prod_{j=1}^k \Ft{V_{y_j}}(p_j-p_{j-1}) \\
&\qquad \qquad \times ( \chi_{main,\mbf{a}}(\omega) + (1 - \chi_{main,\mbf{a}})(\omega)).
\end{split}
\end{equation*}
We will keep the first term (with $\chi_{main,\mbf{a}}$ as the main term, and apply
Lemma~\ref{lem:osc-int} to the second term (with $(1-\chi_{main,\mbf{a}})$) to obtain a bound that is
negligible in the operator norm (after applying Schur's lemma).

To get the error to be negligible we have to verify that the cutoff function
$\chi_{main,\mbf{a}}$ is sufficiently smooth.  First we investigate the smoothness
of $\chi_{main,\mbf{a}}$ in the $p_j$ variable.  The term
\[
\rho(\alpha^{-1}r^{-1}(y_{j+1}-y_j -s_jp_j))
\]
is smooth to scale $\alpha r s_j^{-1}$, which is larger than $r$ because $r^2 \geq\tau$.  The
other term is
\[
\rho(\ell_j^{-1}(|p_j|^2/2-|p_J|^2/2)),
\]
which is smooth in the $p_j$ variable to scale $\ell_j |p_j|^{-1}$.  We claim that
$\ell_j|p_j|^{-1}\gtrsim r$.  Rearranging and using our choice of $\ell_j$, this is equivalent
to showing
\begin{equation}
\label{eq:pj-smoothness-bd}
|p_j| \lsim \alpha \max\{2^{-a_j}r, E^{1/2}\}
\end{equation}
To see this we use the cutoff itself to estimate
\[
||p_j|^2/2 - |p_J|^2/2|\leq \ell_j
\]
so that
\[
||p_j|^2/2| \leq 2E + \ell_j
\]
There are two cases to consider.  The first is that
$2^{-a_j} < E^{1/2}r^{-1}$.  In this case, $\ell_j = \alpha E^{1/2}r^{-1}$,
and so in particular $|p_j|^2 \leq 2E$, since $E\geq \alpha^2r^{-2}$.
Then~\eqref{eq:pj-smoothness-bd} immediately follows.

The second case is that $2^{-a_j} > E^{1/2}r^{-1}$.  In this case,
$\ell_j = \alpha 2^{-a_j}$, so
\[
|p_j| \leq \sqrt{2E} + \sqrt{\alpha 2^{-a_j}}.
\]
The bound $\sqrt{2E} \leq \frac{\alpha}{2} E^{1/2}$ is obvious,
while $\sqrt{\alpha 2^{-a_j}} \leq \alpha 2^{-a_j} r$ follows from
$\tau \leq \alpha r^2$.

It remains to understand the smoothness of $\chi_{main,\mbf{a},b}$ in the
$s_j$ variable.  We only need to look at the term
\[
\rho(\alpha^{-1}r^{-1}(y_{j+1}-y_j-s_jp_j)),
\]
which is smooth in the $s_j$ variable to scale $\alpha r|p_j|^{-1}$.
We need to show that $\alpha r|p_j|^{-1} \geq \ell_j^{-1}$ in order for the
localization in $\partial_{s_j}\varphi$ to be effective, and fortunately this is
precisely the inequality we just proved.

Now, when we apply Lemma~\ref{lem:osc-int} to the error term (with $(1-\chi_{main,\mbf{a},b}(\omega))$),
the integrand is supported on a set with negligible magnitude, and there is no trouble in applying
Schur's lemma.
\end{proof}

\section{A first operator moment estimate}
\label{sec:first-operator-bd}
Lemma~\ref{lem:int-by-pts} will allow us to
estimate moments of the collision operator $T_k$, and in particular
in this section we aim to show that
\[
    \Expec \|e^{i\sigma H} - e^{-i\sigma\Delta/2}\|_{op}^m
\]
with $m\sigma\ll \eps^{-2}$.
We will use this with $\sigma=\eps^{-1.1}$
and $m=\eps^{-0.5}$ to show that we can occasionally `turn off' the
potential for a duration $\sigma$.  In particular, this allows us to enforce
a constraint such as $s_0>\sigma$.  Combined with the localization
provided by Lemma~\ref{lem:int-by-pts}, this gives a more useful constraint
on the kinetic energy of the path.  To see this, consider a path
$\omega\in\Omega_{\alpha,k}(t;E)$, and recall the kinetic energy
constraint
\[
    ||p_j|^2/2-E|\leq \alpha\max\{E^{1/2}r^{-1}, s_j^{-1}\}.
\]
If $s_j>E^{-1/2}r$ (so that the time between the $j$-th and $j+1$-th
collisions is somewhat bigger than $\eps^{-1}$), then this implies that
\[
    ||p_j|^2/2-E| \leq \alpha E^{1/2}r^{-1},
\]
which in turn is approximately equivalent to
\[
    ||p_j|-(2E)^{1/2}| \leq \alpha r^{-1}.
\]
Notice that this also implies, that for these segments
so that in particular $|x_{j}-x_{j+1}|\gtrsim r$, which is the condition
that the $j$-th and $(j+1)$-th collisions are spatially separated by
the wavepacket scale.
The benefit then of enforcing that $s_0\geq \sigma$ is that one can replace
$E$ by $|p_0|^2/2$.

It is also worth pointing out that the most interesting energy scale
for us is $E\sim 1$.  Wavepackets with high energy do not interact
very much with the weak potential, and in $d\geq 3$ there are not enough
momenta at low energies to scatter into, and so scattering is also
suppressed in the lower energy regime.  So while we keep around the
parameter $E$ throughout our discussion of paths, it is safe to imagine
that $E\sim 1$.  The following lemma justifies this more formally.

The following lemma states a bound in terms of operator norm.  This is
impossible for most stationary potentials because there must exist,
with high probability, \emph{some} region in $\Real^d$ on which the
potential behaves badly.  Therefore for the next lemma we assume that the
potential is supported in a ball of radius $R=\eps^{-1000}$ (say), so
we work with the potential multiplied by a cutoff, $V\chi_{R}$.

\begin{lemma}
    \label{lem:collision-strength}
There exists a constant $C=C(d)$ such that for all
admissible potentials and for $k,m\in\bbN$ satisfying $km\leq \eps^{-1}$,
the following bound holds:
\begin{equation}
        (\Expec \|T_k(s) \|_{op}^{2m})^{1/{(km)}}
        \leq R^{C/m} \eps^2 s \,\,(C(km)^8 |\log\eps|^{10}).
\end{equation}
\end{lemma}

\begin{proof}
To estimate the operator norm we will use the bound
\[
    \|A\|_{op}^{2m} \leq \trace((A^*A)^m).
\]
We will compute the trace of the operator
$B^m := (T_k^*T_k)^m$ in the wavepacket basis.
We have for any trace-class operator $B$ the identity
\[
    \trace(B) =
    \int_{\PhaseSpace} \braket{\xi|B|\xi} \diff \xi,
\]
which follows from $\trace(B) = \int K_B(x,x)\diff x$ (where $K_B$ is
the Schwartz kernel of $B$).  Therefore
\begin{equation}
    \label{trace-bd}
    \Expec \|T_k\|_{op}^{2m} \leq
    \int_{\PhaseSpace} \Expec \braket{\xi | ((T_k^\chi)^*T_k^\chi)^m | \xi} \diff \xi.
\end{equation}
We expand each of the operators $T_k^\chi$ as an integral over
paths $\omega_\ell$
\begin{equation}
\label{eq:Gamma-ordering}
    \Gamma =
    (\xi_0,\omega_1,\xi_1,\xi_2,\omega_2,\xi_3,\xi_4,
    \cdots, \xi_{4m-2},\omega_{2m}, \xi_{4m-1})
\end{equation}
be a sequence of paths with intermediate points in phase space, and write

\begin{equation}
\label{eq:trace-integral}
\begin{split}
\Expec \braket{\xi|((T_k^\chi)^*T_k^\chi)^m|\xi}
=
\int \braket{\xi|\xi_{4m-1}}\braket{\xi_0|\xi}
\prod_{\ell=1}^{2m-1} \braket{\xi_{2\ell-1} | \xi_{2\ell}}
\prod_{\ell=1}^{2m} \chi_\alpha(\omega_\ell;\xi_{2\ell-2},\xi_{2\ell-1})
\Expec \prod_{\ell=1}^{2m} \braket{\xi_{2\ell-1} | O_{\omega_\ell}|\xi_{2\ell}}
\diff\bm{\omega}\diff\bm{\xi}.
\end{split}
\end{equation}

We expand the integrand above using the identity
\[
\braket{\xi_{2\ell-2}|O_{\omega_\ell}|\xi_{2\ell-1}}
=
\eps^k
\braket{\xi_{2\ell-2}|p_{\ell,0}}\braket{p_{\ell,k}|\xi_{2\ell-1}}
e^{i\varphi(\omega_\ell)} \prod_{j=1}^k \Ft{V_{y_{\ell,k}}}(q_{\ell,k}).
\]
Let $K = [1,2m] \times [1,k]$ be the indexing set for the collisions, and
let $P(\Gamma)\in\mcal{P}(K)$ be the partition of $K$ induced by the path
$\Gamma$ ($P(\Gamma)$ is the finest partition of $K$ such that $a\sim_{P} b$
when $|a-b|\leq 2r$).  Then using $|\braket{\xi|p}|\leq Cr^{d/2}$ and
Lemma~\ref{lem:admissible-V} we can bound the expectation in the integrand
as follows:
\[
\big|
\Expec \prod_{\ell=1}^{2m} \braket{\xi_{2\ell-1}|O_{\omega_\ell}|\xi_{2\ell}}
\big|
\leq (Cr)^{dmk} \eps^{2mk}
\prod_{a\in K} (1 + |q_a|)^{-20d} \One(|y_a|\leq R)
\sum_{Q\leq P(\Gamma)}
\prod_{S\in Q} |S|^{2|S|}
r^{-d} \exp(- c |\frac{r}{|S|} \sum_{a\in S} q_a|^{0.99}).
\]
The term $\One(|y_a|\leq R)$ comes from the fact that we are actually
working with the localized potential $V\chi_R$ on a ball of radius $R=\eps^{-100}$.

The combinatorial factor $\prod_{S\in Q} |S|^{2|S|}$ is at most $(2km)^{4km}$,
and the sum over $Q\leq P(\Gamma)$ involves at most $|\mcal{P}(K)|\leq (Ckm)^{2km}$ terms.

To prove Lemma~\ref{lem:collision-strength} it therefore suffices to
prove that for any $Q,P\in\mcal{P}(K)$ with $Q\leq P$ the bound
\begin{equation}
\label{eq:trivial-path-bd}
\begin{split}
\int
&
r^{dmk}
|\Braket{\xi|\xi_{4m-1}}| |\Braket{\xi_0|\xi}|
 \prod_{\ell=1}^{2m-1} |\Braket{\xi_{2\ell-1}|\xi_{2\ell}}
\prod_{\ell=1}^{2m} \chi_\alpha(\omega_\ell;\xi_{2\ell-2},\xi_{2\ell-1})
\One(P(\Gamma)=P) \\
&
\qquad \qquad \prod_{a\in K} (1+|q_a|)^{-20d} \One(|y_a|\leq R)
\prod_{S\in Q} r^{-d}\exp(-c | \frac{r}{|S|} \sum_{a\in S} q_a|^{0.99}) \diff\Gamma \\
& \qquad \qquad \qquad \qquad \qquad
\qquad \qquad \qquad \qquad \qquad
\leq
(Cmk)^{2mk} (\eps^{2}t)^{km}
(1 + R^{-C} d_r(\xi,(0,0)))^{-2d}.
\end{split}
\end{equation}

First we will prove the bound without the localization $(1+R^{-C}d_r(\xi,\bm{0}))^{-2d}$,
and then explain how to modify the argument to obtain this (very mild) localization,
which is simply required to have integrability in $\xi$.

The proof of~\eqref{eq:trivial-path-bd} follows very closely the argument in the proof of
Proposition~\ref{prp:short-ladder-compare}.  That is, we order the variables
in the path $\Gamma$ and define constraints on each
variable $p_a$, $y_a$, and $\xi_\ell$ in terms of variables that are earlier in the ordering.

The variable labels are given by the set
\[
\Lambda =
\{(\qt{s},a), (\qt{p},a), (\qt{y},a)\}_{a\in K} \cup \{(\qt{\xi},\ell)\}_{\ell=0}^{4m-1}.
\]
We order the variables so that $(\qt{y},q) \leq \qt{p},a)\leq (\qt{s},a)$ and
$(\qt{\xi},2\ell-2) \leq (\qt{y},\ell,j)\leq (\qt{\xi},2\ell-1)$, and otherwise
so that the ordering is consistent with the usual ordering on $K$ and $[0,4m-1]$.
We use thise ordering to define the partial paths $\Gamma_{\leq \lambda}$ and
$\Gamma_{<\lambda}$ for $\lambda\in\Lambda$.

In the path integral we are considering there is no global conservation of kinetic energy because
we do not have any lower bounds on $s_{\ell,0}$ or $s_{\ell,k}$.  Given a path $\Gamma$,
we define a local kinetic energy $E_\ell$ to be the value of $|p_{\ell,j}|^2/2$ with
$j=\argmax s_{\ell,j}$.  Now we can define the local variable constraints:

\begin{equation}
f_{\qt{p},a}(\bm{\Gamma}_{\leq\qt{p},a})
:=
\begin{cases}
r^d \exp(-c | \frac{r}{km}\sum_{a\in S}q_a|^{0.99}),
&a = \max S \text{ for some } S\in Q\\
r^d \One(|p_{\ell,0} - (\xi_{2\ell-2})_p|\leq \alpha r^{-1}),
&a = (\ell,0) \\
(1+|p_a-p_{a-1}|)^{-20d} \One(||p_a| - E_\ell||\leq \alpha \max\{|E_\ell|^{-1/2}s_a^{-1},r^{-1}\}),
&\text{ else}.
\end{cases}
\end{equation}
The standard $y$ constraints are given by
\begin{equation}
f_{\qt{y},a}(\bm{\Gamma}_{\leq\qt{y},a}) :=
\begin{cases}
r^{-d}\One(|y_{\ell,1} - ((\xi_{2\ell-2})_x + s_{\ell,0}p_{\ell,0})| \leq \alpha r),
&a = (\ell,1) \\
r^{-d}\One(|y_{\ell,j} - (y_{\ell,j-1}+s_{\ell,j-1}p_{\ell,j-1})|\leq \alpha r), &a = (\ell,j)
\text{ for } j>1.
\end{cases}
\end{equation}
Finally, the standard $\xi$ constraints are
\begin{equation}
f_{\qt{\xi},\ell}(\bm{\Gamma}_{\leq\qt{\xi},\ell})
:=
\begin{cases}
|\braket{\xi|\xi_0}| & \ell = 0 \\
|\braket{\xi_{2n-1}|\xi_{2n}}|, &\ell = 2n, n\in [1,2m-1] \\
\One(d_r(\xi_{2n-1}, (y_{n,k} + s_{n,k}p_{n,k},p_{n,k})) \leq \alpha), &\ell=2n-1, n\in [1,2m].
\end{cases}
\end{equation}
A simple accounting of terms (unpacking the constraints
present in $\chi_\alpha(\omega_\ell;\xi_{2\ell-2},\xi_{2\ell-1})$ confirms the estimate
\begin{equation}
\begin{split}
&r^{dmk}
|\Braket{\xi|\xi_{4m-1}}| |\Braket{\xi_0|\xi}|
 \prod_{\ell=1}^{2m-1} |\Braket{\xi_{2\ell-1}|\xi_{2\ell}}
\prod_{\ell=1}^{2m} \chi_\alpha(\omega_\ell;\xi_{2\ell-2},\xi_{2\ell-1}) \\
& \qquad \qquad \prod_{a\in K} (1+|q_a|)^{-20d} \One(|y_a|\leq R)
\prod_{S\in Q} r^{-d}\exp(-c | \frac{r}{|S|} \sum_{a\in S} q_a|^{0.99}) \\
&
\qquad\qquad
\qquad\qquad
\qquad\qquad
\leq
\prod_{a\in K}
f_{\qt{p},a}(\Gamma_{\leq \qt{p},a})
f_{\qt{y},a}(\Gamma_{\leq \qt{p},a})
\prod_{\ell=0}^{4m-1}
f_{\qt{\xi},\ell}(\Gamma_{\leq \qt{\xi},\ell}).
\end{split}
\end{equation}

Now for the time variable constraints $f_{\qt{s},a}$ we use
the indicator function $\One(P(\Gamma)=P)$ as follows:
\begin{equation}
f_{\qt{s},a}(\Gamma_{\leq \qt{s},a})
:=
\begin{cases}
\One_{[0,t]}(s_a), & a = \min S \text{ for some } S\in Q \\
\One(|y_{\min S} - (y_{a-1} + s_{a-1}p_{a-1})| \leq 4\alpha r)
, & a \in S\setminus \{\min S\} \text{ for some } S\in Q.
\end{cases}
\end{equation}

For $a=(\ell,j)\in S\setminus\{\min S\}$, we have the estimate
\[
\sup_{\Gamma_{< \qt{s},a}}
\int f_{\qt{s},a}(\Gamma_{\leq \qt{s},a}) \diff s_a
\leq 4E_\ell^{-1/2}\alpha r.
\]
For such $a$, the integral over $p_{a-1}$ produces a factor
of $\alpha r^{-1}\min\{E_\ell^{(d-1)/2}, 1\}$, so multiplying these
factors yields at a contribution of at most $C \alpha^2$,
since $\min\{E^{(d-2)/2}, E^{-1}\} \leq 1$ for any $E$.

Therefore, multiplying the factors from all of the variables together,
and using Lemma~\ref{lem:sp-integral} to deal with immediate
recollisions, we arrive at the estimate
\begin{equation*}
\begin{split}
\int
&
r^{dmk}
|\Braket{\xi|\xi_{4m-1}}| |\Braket{\xi_0|\xi}|
 \prod_{\ell=1}^{2m-1} |\Braket{\xi_{2\ell-1}|\xi_{2\ell}}
\prod_{\ell=1}^{2m} \chi_\alpha(\omega_\ell;\xi_{2\ell-2},\xi_{2\ell-1})
\One(P(\Gamma)=P) \\
&
\qquad \qquad \prod_{a\in K} (1+|q_a|)^{-20d} \One(|y_a|\leq R)
\prod_{S\in Q} r^{-d}\exp(-c | \frac{r}{|S|} \sum_{a\in S} q_a|^{0.99}) \diff\Gamma \\
& \qquad \qquad \qquad \qquad \qquad
\qquad \qquad \qquad \qquad \qquad
\leq
(C\alpha^{10} mk)^{2mk} \eps^{2km} t^{|Q|},
\end{split}
\end{equation*}
and since $Q$ has at most $mk$ clusters, $\eps^{2km}t^{|Q|}\leq (\eps^2t)^{km}$.

The localization term $(1+R^{-C}d_r(\xi,\bm{0}))^{-2d}$ comes from the following additional
considerations.  First, if $|\xi_x| \geq 10 \alpha R|\xi_p| + 10R$, then any choice of
$s_{1,0}$ and $p_{1,0}$ forces $|y_{1,1}|>R$, and therefore the integrand vanishes.
If $|\xi_p|<R^2$ then this enforces an upper bound on $d_r(\xi,\bm{0})\leq R^C$.
On the other hand if $|\xi_p|>R^2$ then either $|p_a|\geq \eps^{-10}R$ for all $a\in K$,
in which case each time variable is constrained to be at most $\eps^{10}$ so that the collision
locations do not escape the ball $B_R$.  Otherwise if $|p_a|\leq \eps^{-10}R$ for some $a$,
then $\prod (1+|q_a|)^{-20d} \leq R^{-10d}|\xi_p|^{-10d}$ and we still obtain decay in terms
of $d_r(\xi,\bm{0})$.
\end{proof}

\section{Interspersing the free evolution}
\label{sec:free-intersperse}
The primary application of Lemma~\ref{lem:collision-strength}
is to show that, up to a small error, we may intersperse the evolution
under the Hamiltonian $H$ with periods of evolution under the free
Hamiltonian $-\Delta/2$.  This corresponds to evolution of a system in
which the potential is periodically turned off and on again in a smooth way.

This modification is convenient for ensuring conservation of kinetic
energy across path segments.  One way to see how this is helpful is the
following physical interpretation.  The channel $\Evol_\tau^N$ corresponds
to an evolution in which the potential discontinuously changes at
every time $j\tau$, $j\in[N]$.  The abrupt time-dependence of the potential
disturbs conservation of energy and in particular conservation of kinetic
energy.  We restore an approximate conservation of kinetic energy by smoothly
ramping up and down the strength of the potential.
\begin{lemma}
    \label{lem:free-replace}
    Define the unitary operator $U(s)$ by
    \[
        U(s) := e^{i(\tau-s) H} e^{-is\Delta/2}.
    \]
    Then, for $\sigma_1,\dots,\sigma_m<\eps^{-1.5}$,
and sufficiently small $\eps$, we have
    \[
        \Expec
        \|e^{im\tau H} -
U(\sigma_m)U(\sigma_{m-1})\cdots U(\sigma_1)\|_{op}^2
        \leq  m^2 \eps^{2/13}.
    \]
\end{lemma}

To prove Lemma~\ref{lem:free-replace} we first need a bound on
$\|e^{-i\tau H} - e^{i\tau\Delta/2}\|$.  Here we use the
Duhamel expansion
\[
    e^{-i\sigma H} - e^{i\sigma \Delta/2} = \sum_{k=1}^{k_{max}}
    T_k(\sigma) + R_{k_{max}}(\sigma ).
\]
To estimate the remainder $R_k$ we will use the bound
\[
    \|R_k(\sigma )\|_{op} \leq \sigma  \sup_{s<\sigma } \|T_k(s)\|_{op},
\]
which follows from the triangle inequality and the unitarity of
$e^{-i(\sigma -s)H}$.

Using these ingredients, we can compare the random evolution to the
free evolution.
\begin{lemma}
    \label{lem:approx-free}
    For $\sigma<\eps^{-1.5}$ and $m<\eps^{-1/30}$,
    and sufficiently small $\eps>0$,
    we have
    \[
        \Expec \|e^{i\sigma H} - e^{-i\sigma\Delta/2}\|^m
        \leq \eps^{m/8}.
    \]
\end{lemma}
\begin{proof}
The Duhamel formula gives the exact decomposition
\begin{equation}
e^{is H} = e^{-is\Delta/2} + \sum_{1\leq k\leq k_{max}} T_k(s) + R_k(s).
\end{equation}
Applying the triangle inequality and using unitarity in the definition
of $R_k$, we have the bound
\begin{equation}
\|e^{is H} - e^{-is\Delta/2}\|
\leq \sum_{1\leq k\leq k_{max}} \|T_k(s)\|
+ \int_0^s \|T_{k_{max}}(s')\|\diff s'.
\end{equation}
Applying the inequality $(\sum_{i=1}^k a_i)^p \leq k^p \sum a_i^p$
for $a_i\geq 0$ and $p>0$, we have
\[
\|e^{isH} - e^{-is\Delta/2}\|^m
\leq
(k_{max}+1)^m \sum_{1\leq k\leq k_{max}} \|T_k(s)\|^m
+ (k_{max}+1)^m \Big(\int_0^s \|T_{k_{max}}(s')\|\diff s'\Big)^m.
\]
Using Holder's inequality on the integral and then taking expectations
with Lemma~\ref{lem:collision-strength}, we have
\begin{equation}
\begin{split}
\Expec \|e^{isH} - e^{-is\Delta/2}\|^m
&\leq
C^m \sum_{1\leq k\leq k_{max}}
(\Expec \|T_k(s)\|^m)
+ C^m s^{m-1}
\int_0^s \Expec \|T_{k_{max}}(s')\|^m\diff s' \\
&\lsim
C^m \sum_{1\leq k\leq k_{max}}
(\Expec \|T_k(s)\|^{2Mm})^{1/2M}
+
R^{C}
(\eps^2 s^{1+1/k_{max}}
((k_{max}m)^{8}|\log\eps|^{10})^{km})^{k_{max}m}
\\
&\lsim
C^m \sum_{1\leq k\leq k_{max}}
R^{C/M} (\eps^2 s (kMm)^{8} |\log\eps|^{10})^{mk/2} \\
&\qquad\qquad +
R^{C}
(\eps^2 s^{1+1/k_{max}}
((k_{max}m)^{8}|\log\eps|^{10})^{km})^{k_{max}m}
\end{split}
\end{equation}
By choosing $M$ to be a sufficiently large constant, we can have
the bound $R^{C/M} \leq \eps^{-10^{-5}}$ (say).  Then, since
$m\leq \eps^{-1/30}$ and $s \leq \eps^{-1.5}$,
\[
\eps^2 s m^8 \leq \eps^2 s \eps^{-1/10} \leq \eps^{1/4}.
\]
Therefore the sum over $k$ converges geometrically so
that the main contribution comes from $k=1$, which
gives
$R^{C/M} (\eps^2sm^{8}|\log\eps|^{10})^{m/2} \lsim \eps^{m/8}$.
\end{proof}

We are now ready to prove Lemma~\ref{lem:free-replace}.
\begin{proof}[Proof of Lemma~\ref{lem:free-replace}]
With $\bm{\sigma} = (\sigma_1,\dots,\sigma_m)$, let
\[
U(\bm{\sigma}) :=
U(\sigma_m) U(\sigma_{m-1})\cdots U(\sigma_1),
\]
and let $\bm{\sigma}^{(k)}$ be the sequence obtained by replacing
the first $k$ elements of $\bm{\sigma}$ by zeros.  That is,
\[
\bm{\sigma}^{(k)} :=
(0,\cdots, 0, \sigma_{k+1},\cdots,\sigma_m).
\]
Then we use the triangle inequality,
\[
\| U(\bm{\sigma}^{(m)}) - U(\bm{\sigma}^{(0)})\|_{op}
\leq
\sum_{k=1}^m \|U(\bm{\sigma}^{(k)}) - U(\bm{\sigma})^{(k-1)})\|_{op}.
\]
We square both sides and use Cauchy-Schwartz to find
\[
\| U(\bm{\sigma}^{(m)}) - U(\bm{\sigma}^{(0)})\|_{op}^2
\leq
m
\sum_{k=1}^m \|U(\bm{\sigma}^{(k)}) - U(\bm{\sigma})^{(k-1)})\|_{op}^2.
\]
By unitarity, we can simplify the terms on the right hand side to
\[
\Expec \| U(\bm{\sigma}^{(m)}) - U(\bm{\sigma}^{(0)})\|_{op}^2
\leq
m
\sum_{k=1}^m \Expec \|e^{i\sigma_k H} - e^{-i\sigma_k\Delta/2}\|_{op}^2
\leq m^2 \eps^{2/13}.
\]

\end{proof}
The reason we are doing this is that we want to replace the short time
evolution $e^{i\tau H}$ by an evolution operator that allows for a small
amount of free evolution at the beginning -- this is convenient both
for using conservation of kinetic energy as discussed at the top of this
section, and to guarantee some amount of decorrelation between adjacent
path segments.

Let $S=\eps^{-1.5}$ and let
$\rho\in C_c^\infty([\sigma,2\sigma])$ be a smooth Gevrey function satisfying
$\int \rho=1$.\footnote{The Gevrey regularity is for convenience, but
$C^m$ for some large enough $m$ depending on $k_{max}$ is fine.}
Let $U_\rho$ be the operator
\[
    U_\rho(\tau) :=
    \int e^{i(\tau-s)H} e^{-is\Delta/2} \rho(s)\diff s.
\]
By averaging
Lemma~\ref{lem:free-replace}, we have
\[
    \Expec \|e^{iN\tau} - U_\rho(\tau)^N\|_{op} \leq N\eps^{1/13}.
\]
Moreover, $U_\rho(\tau)$ has a Duhamel expansion of the form
\[
    U_\rho(\tau) =
e^{-i\tau\Delta/2} + \sum_{k=1}^{k_{max}} T_{k,\rho}(\tau)
    + R_{k_{max},\rho}(\tau),
\]
where
\[
    T_{k,\rho}(\tau) :=
    \int_{\Omega_k(\tau)} P(s_0) O_\omega \diff\omega,
\]
and
\begin{equation}
\label{eq:P-def}
P(s) = \int_0^s \rho(s')\diff s'.
\end{equation}

The benefit of doing this is that
we can therefore enforce $s_0>\eps^{-1.5}>r$ in every Duhamel expansion,
and this more easily guarantees that the kinetic energy level will be
maintained in every path segment.  To notate this, we define
\[
    \Omega_{\alpha, k}(t,E,S) :=
    \{ (\mbf{s},\mbf{p},\mbf{y}) \in \Omega_{\alpha,k}(t,E) \mid
    s_0 > S\}.
\]
Note that, for such paths, if $E^{1/2}r^{-1} > S$,
$||p_0|^2/2-E| \leq \alpha E^{1/2} r^{-1}$, so we have for each $j\in[k]$ the
bound
\[
    ||p_j|^2/2- |p_0|^2/2| \leq 2\alpha \max\{ E^{1/2}r^{-1}, s_j^{-1}\}.
\]
Thus, except at very low frequencies ($E \lsim \eps^{0.5}$, so
$|p_0|\lsim \eps^{0.25}$), we can replace $E$ by $|p_0|^2/2$.

\appendix
\section{The Boltzmann limit for short times}
\label{sec:boltzmann}
In this section we compare the
random Schrodinger evolution on short time scales $s\ll \eps^{-2}$
to the Boltzmann evolution given by the equation
\begin{equation}
\label{eq:linear-boltz-eq}
\partial_t a(x,p) + p\cdot \nabla_x a(x,p) =
\eps^2
\int \delta(|p|^2-|q|^2)  R(p-q) (a(x,q)-a(x,p))\diff q.
\end{equation}

To compare the quantum evolution channel $\Evol_s$ defined in~\eqref{eq:evol-def}
and the linear Boltzmann equation~\eqref{eq:dual-boltz}, we use a perturbative
expansion for each equation and compare each term.  Because we are working
with such short time scales (compared to the kinetic time $\eps^{-2}$),
we only need to use a few terms in the expansion.

\subsection{The finite Duhamel expansion}
Indeed we have as a consequence of Lemma~\ref{lem:collision-strength} that just two
terms of the Duhamel expansion suffices to approximate the Schrodinger evolution.
\begin{lemma}
\label{lem:second-order-approx}
For times $s <\eps^{-1.5}$ and any $\delta>0$ there exists $C_\delta$ such that
\[
\Expec \|e^{-isH} - (T_0(s)+T_1(s)+T_2(s))\|_{op}^2
\leq C_\delta \eps^{-\delta}(\eps^2s)^3.
\]
\end{lemma}
\begin{proof}
We follow the proof of Lemma~\ref{lem:approx-free}.
We beign with the exact Duhamel expansion
\begin{equation}
e^{is H} = \sum_{k\leq 2}T_k(s) + \sum_{2< k\leq k_{max}} T_k(s) + R_{k_{max}}(s).
\end{equation}
Applying the triangle inequality and using unitarity in the definition
of $R_k$, we have the bound
\begin{equation}
\|e^{is H} - \sum_{k\leq 2}T_k(s)\|_{op}
\leq \sum_{2< k\leq k_{max}} \|T_k(s)\|_{op}
+ \int_0^s \|T_{k_{max}}(s')\|\diff s'.
\end{equation}
Squaring both sides and using the Cauchy-Schwartz inequality, we have
\[
\|e^{isH} - e^{-is\Delta/2}\|^2
\leq
(k_{max}+1)^2 \sum_{1\leq k\leq k_{max}} \|T_k(s)\|^2
+ (k_{max}+1)^2 s \int_0^s \|T_{k_{max}}(s')\|^2\diff s'.
\]
Then taking expectations using Lemma~\ref{lem:collision-strength}
we have
\begin{equation}
\begin{split}
\Expec \|e^{isH} - e^{-is\Delta/2}\|^2
&\lesssim
\sum_{1\leq k\leq k_{max}}
(\Expec \|T_k(s)\|^{4M})^{1/2M}
+ s
\int_0^s \Expec \|T_{k_{max}}(s')\|^2\diff s' \\
&\lsim
C^m \sum_{1\leq k\leq k_{max}}
R^{C/M} (C^M \eps^2 s |\log\eps|^{10})^{k} \\
&\qquad\qquad +
R^{C}
(C\eps^2 s^{1+1/k_{max}} |\log\eps|^{10})^{k_{max}}
\end{split}
\end{equation}
By choosing $M$ to be a sufficiently large constant, we can have
the bound $R^{C/M} \leq \eps^{-\delta}$.  The main term
is the term coming from $T_3(s)$ which gives the bound above.
\end{proof}

Using Lemma~\ref{lem:second-order-approx} and the operator Cauchy-Schwartz
inequality we have for $s\leq \eps^{-1.5}$ that
\begin{equation}
\label{eq:second-order-expansion}
\|\Evol_s[A] - \Expec (\sum_{k=0}^2T_k(s))^*A (\sum_{k=0}^2T_k(s))\|_{op}
\leq C( \eps^{1.99}s)^3\|A\|_{op}.
\end{equation}
Defining the superoperator
\[
\mcal{T}_{k,k'}(s)[A] = \Expec T_k(s)^*AT_{k'}(s),
\]
the approximation~\eqref{eq:second-order-expansion} can be written
\[
\Evol_s \approx \mcal{T}_{0,0} + \mcal{T}_{1,0}
+ \mcal{T}_{0,1} + \mcal{T}_{1,1} + \mcal{T}_{2,0}
+\mcal{T}_{0,2} + \mcal{T}_{2,1}  + \mcal{T}_{1,2}
+ \mcal{T}_{2,2}.
\]
The terms $\mcal{T}_{0,1}$ and $\mcal{T}_{1,0}$ vanish identically
because $\Expec V=0$.  Moreover the higher order terms satisfy
\[
\|\mcal{T}_{2,1}[A] + \mcal{T}_{1,2}[A] + \mcal{T}_{2,2}[A]\|_{op}
\leq C (\eps^3 s + \eps^4 s^2) \|\Wigner{A}\|_{C^{2d+1}}.
\]
The factor $\eps^3 s$ comes from $\mcal{T}_{2,1}[A]+\mcal{T}_{1,2}[A]$,
in which there are three collisions (and therefore $\eps^3$) which must
form a single cluster (hence the single factor of $s$).  The term
$\mcal{T}_{2,2}[A]$ contributes $\eps^4s^2$ because there may be
two clusters of two collisions each that may happen at any pair of times.

We therefore have the approximation
\begin{equation}
\Evol_s \approx \mcal{T}_0(\sigma) +
\mcal{T}_{1,1} + \mcal{T}_{0,2} + \mcal{T}_{2,0}
\end{equation}
where, expanding the definitions we can write
\begin{equation}
\begin{split}
\mcal{T}_0[A] &= e^{-i\sigma\Delta/2}Ae^{i\sigma\Delta/2} \\
\mcal{T}_{1,1}[A]
&:= \Expec \int_0^\sigma \diff s \int_0^\sigma \diff s'\,\,
e^{-i(\sigma-s')\Delta/2} V e^{-is'\Delta/2} A
e^{is\Delta/2} V e^{i(\sigma-s)\Delta/2} \\
\mcal{T}_{0,2}[A] &:=
\Expec \int_0^\sigma \diff s \int_0^{\sigma-s}\diff \tau\,\,
e^{-i\sigma\Delta/2} A e^{is\Delta/2}Ve^{i\tau\Delta/2}V
e^{i(\sigma-s-\tau)\Delta/2} \\
\mcal{T}_{2,0}[A] &:=
\Expec \int_0^\sigma \diff s \int_0^{\sigma-s}\diff \tau\,\,
e^{-i(\sigma-s-\tau)\Delta/2}Ve^{-i\tau\Delta/2}Ve^{-is\Delta/2}
 A e^{i\sigma \Delta/2}.
\end{split}
\end{equation}
The following lemma makes this approximation precise.
\begin{lemma}
\label{lem:three-term-evol}
Let $A\in\mcal{H}(L^2(\Real^d))$ be a bounded operator and let $s\leq \eps^{-1.5}$.  Then
\begin{equation}
    \|\Evol_s[A] -
(\mcal{T}_0(s) + \mcal{T}_{1,1}(s) + \mcal{T}_{0,2}(s)+\mcal{T}_{2,0}(s))[A]\|_{op}
\leq \eps^{2.5} s \|\Wigner{A}\|_{C^{2d+1}}.
\end{equation}
\end{lemma}

\subsection{Comparing to the linear Boltzmann equation}

We recall below the statement of Proposition~\ref{prp:short-time-compare} that
\shorttime*
%\begin{proposition}
%Let $a_0\in C^{2d+1}$ be an observable
%supported on the set $\{(x,p) | |p|>\eps^{0.2}\}$, and suppose that $a_s$
%solves the linear Boltzmann equation~\eqref{eq:linear-boltz-eq}.
%Then for $\sigma\leq \eps^{-1.5}$,
%\begin{equation}
%\|\Op(a_\sigma) - \Evol_\sigma[\Op(a_0)]\|_{op}
%\leq \eps^{2.01}\sigma
%\|a\|_{C_r^{2d+10}}.
%\end{equation}
%\end{proposition}
%
The first step is to write out the action of the superoperators $\mcal{T}$ on the Wigner
transforms of the operator.
An elementary computation using the formula
\[
\Expec \Ft{V}(q) \Ft{V}^*(q') = \delta(q-q') \Ft{R}(q),
\]
where $R(x) = \Expec V(0)V(x)$ yields the following identities:
\begin{align}
\label{eq:wigner-zero}
\Wigner{\mcal{T}_0(\sigma) A}(x,p) &= \Wigner{A}(x+\sigma p,p) \\
\label{eq:wigner-one}
\Wigner{\mcal{T}_{1,1}A}(x,p) &=\eps^2
\int_0^\sigma\int_0^\sigma
\int_{\Real^d} e^{i(s'-s)(|p'|^2-|p|^2)/2} R(p-p')
\\ \nonumber
&\qquad\qquad\qquad
\Wigner{A}(x + (\sigma-\frac{s+s'}{2})p + \frac{s+s'}{2}p', p')
\diff p' \diff s\diff s' \\
\label{eq:wigner-two}
\Wigner{\mcal{T}_{0,2}A+ \mcal{T}_{2,0}A}(x,p) &=2\eps^2
\Rept \int_0^\sigma
\int_0^{\sigma-S}
\int_{\Real^d} e^{-i u(|p'|^2-|p|^2)/2}
R(p'-p) \\
\nonumber
&\qquad\qquad\qquad
\Wigner{A}(x-(\sigma-u/2) p + u p'/2,p) \diff p' \diff S\diff u.
\end{align}

The classical solution has an analogous expansion in terms of the Duhamel formula,
\[
a_s \approx T_0(s) a_0 + T_1(s) a_0 + T_2(s) a_0 =: T_{\leq 2}(s) a_0.
\]
where
\begin{align}
\label{eq:classical-Tzero}
T_0(t)a(x,p) &= a(x+tp,p) \\
\label{eq:classical-Tone}
T_1(t)a(x,p) &= \eps^2\int_0^t\diff s \int\diff p' \delta(|p'|^2-|p|^2)
R(p-p') a(x+(t-s)p + sp', p') \\
\nonumber
&= \eps^2\int_0^t\diff s \int\diff p' \delta(|p'|^2-|p|^2)
R(p-p') a(x+sp + (t-s)p', p') \\
\label{eq:classical-Ttwo}
T_2(t)a(x,p) &= -\eps^2t  K(p)a(x+tp,p)
\end{align}
where the scattering cross section $K(p)$ is given by
\[
K(p) = \int \delta(|p'|^2-|p|^2) R(p-p')\diff p'.
\]
We will prove the following approximation result for the classical evolution.
\begin{lemma}
\label{lem:boltzmann-expansion}
Let $a_0\in C^{2d+1}$, and suppose that $a_s$ solves the linear Boltzmann
equation~\eqref{eq:linear-boltz-eq}, and suppose that $s\leq \eps^{-1.5}$.  Then
\[
\|a_s - T_{\leq 2}(s) a_0\|_{C^{2d+1}} \leq C \eps^3 s \|a_0\|_{C^{2d+1}}.
\]
\end{lemma}

We will show that $\Wigner{\mcal{T}_{\leq 2}[A]}$ is close to $T_{\leq 2}\Wigner{A}$ by comparing
the three terms to each other one at a time.
The first terms agree identically,
\begin{equation}
\label{eq:free-exact}
\Wigner{\mcal{T}_0(\sigma)A} = T_0(\Wigner{A}),
\end{equation}
so there are two terms remaining, which we separate into two lemmas.
\begin{lemma}
\label{lem:Tone-compare}
Let $a\in C^{2d+1}(\PhaseSpace)$ be an observable supported
on $\{(x,p)||p|>\eps^{0.2}\}$ and let $\eps^{-1.2} \leq s\leq \eps^{-1.5}$.
Then
\begin{equation}
\label{eq:Tone-compare}
\|\Wigner{\mcal{T}_{1,1}(s)[\Op(a)]} - T_1(s) a\|_{C^{2d+1}_r}
\leq \eps^2s\delta \|a\|_{C^{4d+2}_{r,\delta}}.
\end{equation}
\end{lemma}
\begin{lemma}
\label{lem:Ttwo-compare}
Let $a\in C^{2d+1}(\PhaseSpace)$ be an observable supported
on $\{(x,p)||p|>\eps^{0.2}\}$ and let $\eps^{-1.2} \leq s\leq \eps^{-1.5}$.
Then
\begin{equation}
\label{eq:Ttwo-compare}
\|\Wigner{(\mcal{T}_{0,2}+\mcal{T}_{2,0})(s)[\Op(a)]} - T_2(s) a\|_{C^{2d+1}}
\leq \eps^2\delta s \|a\|_{C^{4d+2}_{r,\delta}}.
\end{equation}
\end{lemma}

Before we prove these lemmas we stop to verify that the lemmas imply the main result.
\begin{proof}[Proof of Proposition~\ref{prp:short-time-compare}]
Combining Lemma~\ref{lem:three-term-evol}, Lemma~\ref{lem:boltzmann-expansion}, and
the bound $\|\Op(a)\|_{op}\leq \|a\|_{C^{2d+1}_r}$, it suffices to prove that
\[
\|\Wigner{\mcal{T}_{\leq 2}(s)[\Op(a)]}
- T_{\leq 2}(s) a\|_{C^{2d+1}_r} \leq \eps^? \|a\|_{C^{2d+1}_r}.
\]
This bound then follows from the identity~\eqref{eq:free-exact} and the
bounds~\eqref{eq:Tone-compare} and~\eqref{eq:Ttwo-compare}.
\end{proof}

\subsection{The proof of Lemma~\ref{lem:Tone-compare}}

\begin{proof}[Proof of Lemma~\ref{lem:Tone-compare}]
We recall the formula~\eqref{eq:wigner-one} for $\mcal{T}_{1,1}[\Op(a)]$
and introduce the convenient change of variables $S=(s+s')/2$ and $u=s-s'$ to write
\begin{align*}
\Wigner{\mcal{T}_{1,1}[\Op(a)]}(x,p)
&=
\eps^2 \int_0^\sigma\int_0^\sigma\diff s\diff s'\int_{\Real^d}\diff p\,\,
\,\,  e^{i(s'-s)(|p'|^2-|p|^2)/2} R(p-p')
\Wigner{A}(x + (\sigma-\frac{s+s'}{2})p + \frac{s+s'}{2}p', p')\\
&=
\eps^2 \iint\int_{\Real^d}
\One_{[0,\sigma]}(S-u/2) \One_{[0,\sigma]}(S+u/2)
e^{iu(|p'|^2-|p|^2)/2}  \\
&\qquad\qquad \qquad\qquad
R(p-p') \Wigner{A}(x + (\sigma-S)p + Sp', p') \diff p' \diff S \diff u.
\end{align*}

One key difference between the integrand~\eqref{eq:wigner-one} defining
$\Wigner{\mcal{T}_{1,1}[\Op(a)]}(x,p)$ and
the integrand in~\eqref{eq:classical-Tone} defining $T_1 a$
is that~\eqref{eq:classical-Tone}
contains a delta function $\delta(|p|^2-|p'|^2)$,
whereas~\eqref{eq:wigner-one}
 involves a phase $e^{i(s'-s)(|p'|^2-|p|^2)}$
and an additional time variable.  If we could integrate
out the $s'-s$ variable we would reproduce this delta function exactly, but there
is a hard cutoff from $0\leq s,s'\leq \sigma$ which ruins the sharp decay off
of the annulus $|p|^2=|p'|^2$.

We therefore decompose the indicator function of the interval as the sum of
a smooth piece containing the bulk of the interval and a sharp piece at either end,
\begin{equation}
\label{eq:smooth-indicator-decomp}
\One_{[0,\sigma]}(s;p) = \chi_{[0,\sigma]}(s,p) + h(s,p),
\end{equation}
where we take $\chi_{[0,\sigma]}(s,p)$ to be smooth to scale $r|p|^{-1}$ in the variable $s$,
and $h(s,p)$ to be supported on the set
\[
\supp h(\cdot,p) \subset
\begin{cases}
[0,r|p|^{-1}] \cup [\sigma-r|p|^{-1},\sigma], &|p|\geq \frac{1}{2}\eps^{0.2} \\
[0,r\eps^{0.2}] \cup [\sigma-r\eps^{0.2}, \sigma], &|p| < \frac{1}{2}\eps^{0.2}.
\end{cases}
\]
The function $h(s,p)$ can be chosen to be smooth to scale $\eps^{0.2}$ in the momentum variable $p$
(for example by constructing it using a partition of unity on dyadic annuli, the thinnest of which has
scale $\eps^{0.1}$).
We use the decomposition~\eqref{eq:smooth-indicator-decomp} in both the $s$ and $s'$ variables
to write
\[
\Wigner{\mcal{T}_{1,1}[\Op(a)]} = Ma + Ea,
\]
where $Ma$ is the main term
\begin{equation}
\label{eq:main-Tone-term}
\begin{split}
Ma(x,p) &:= \eps^2 \iint \int_{\Real^d}
\chi_{[0,\sigma]}(S+u/2,p)
\chi_{[0,\sigma]}(S-u/2,p)
e^{iu(|p'|^2-|p|^2)/2}  R(p-p') \\
&\qquad \qquad\qquad a(x + (\sigma-S)p + Sp', p')\diff u \diff S \diff p'.
\end{split}
\end{equation}
and $Ea$ is the error term
\begin{align*}
Ea(x,p)
:= \eps^2 \iint \int_{\Real^d}
&(h(S-u/2;p) \One_{[0,\sigma]}(S+u/2) + \chi(S-u/2;p) h(S+u/2,p)) \\
&\qquad
e^{iu(|p'|^2-|p|^2)/2}  R(p-p')
a(x + (\sigma-S)p + Sp', p')\diff u \diff S \diff p'.
\end{align*}
Our next step is to demonstrate the bound
\begin{equation}
\label{eq:error-Ck-bd}
\|Ea\|_{C^{2d+1}_r} \leq \|a\|_{C^{4d+2}_r}.
\end{equation}
To exploit the cancellation in the phase, we will integrate
by parts in the $u$ variable.  To do this we first split the integration in $p'$ into a part for which the kinetic
energy is approximately conserved and a part for which the conservation of kinetic energy is significantly violated,
\[
1 = \wtild{\One}(|p|^2-|p'|^2 \lsim r^{-1}|p|) + \wtild{\One}(|p|^2-|p'|^2\gtrsim r^{-1}|p|),
\]
where the function $\wtild{\One}$ is meant to indicate a smoothed indicator function.
This allows us to decompose the error term into a term with $p'$ localized to the annulus and second more oscillatory
term,
\[
Ea = E^{KE}a + E^{osc}a,
\]
with
\begin{align*}
E^{KE}a(x,p)
:= \eps^2 \iint \int_{\Real^d}
&(h(S-u/2;p) \One_{[0,\sigma]}(S+u/2) + \chi(S-u/2;p) h(S+u/2,p))
\wtild{\One}(|p|^2-|p'|^2\lsim r^{-1}|p|)
\\
&\qquad
e^{iu(|p'|^2-|p|^2)/2}  R(p-p')
a(x + (\sigma-S)p + Sp', p')\diff u \diff S \diff p'.
\end{align*}
To bound $E^{KE}a$ in the $C^{2d+1}_r$ norm we have to first observe that all terms in
the integrand are bounded in $C^{2d+1}_r$ except for the oscillation $e^{iu(|p'|^2-|p|^2)}$
which has gradient
\[
\nabla_p e^{iu(|p'|^2-|p|^2)/2} = pu e^{iu(|p'|^2-|p|^2)/2}.
\]
To handle this term, we can integrate by parts in the $p'$ variable, and since the rest of the integrand
is smoooth in $p'$ to scale $r^{-1}$ we can obtain decay in the $u$ variable for $|u|\gg r|p|^{-1}$.
We integrate by parts $2d+1$ times (again splitting the integral into a piece with
$|u|\lsim r|p|^{-1}$ and a piece with $|u|\gg r|p|^{-1}$, and only integrating by parts in the
second piece).  We choose $2d+1$ because we need to establish differentiability to order $2d+1$ in the $p$
variable, and each derivative brings down a factor of $u$.  This comes at a cost of an additional $2d+1$
derivatives of the function $a$.
Now we can use the triangle inequality to bound the norm:
the integration over $s,s',p$ is over a set of
volume $(r^2 |p|^{-1})^2\times (r^{-1}|p|^{d-1})$ (the first factor coming from the integration in time, the
second from the integration over the annulus in $p$).  The volume of integration is therefore at most
$r|p|^{d-2}$, which is further localized in the case $|p|\gg 1$ by the presence of the term $\Ft{R}(p-p')$,
so that at the end we obtain a bound of the form
\[
\|E^{KE}a\|_{C^{2d+1}_r} \lsim \eps^2 r|p|^{-1} \min\{|p|^{d-2},1\} \|a\|_{C^{2d+1}_r}
\lsim \eps^{0.8}\|a\|_{C^{4d+2}_r}.
\]

For the second term $E^{osc}a$ supported away from this annulus,
we integrate by parts in the $u$ variable.  This produces several boundary terms $E^{bd}_ja$ and several
bulk terms $E^{bulk}_ja$.
One of the boundary terms (which we arbitrarily number $E^{bd}_1a$
comes from the delta function $\delta(S-u/2)$ in the derivative of
$h(S-u/2)$ with respect to $u$,
\[
E^{bd}_1a(x,p) := \eps^2 \int\int_{\Real^d} \frac{e^{iS/2(|p'|^2-|p|^2)}}{|p|'^2-|p|^2}
\One(|p|^2-|p'|^2\gtrsim r^{-1}|p|)
\Ft{R}(p-p') a(x+(\sigma-u/2)p + up'/2,p') \diff u\diff p'.
\]
This time, as $S\lsim r|p|^{-1}$ even the oscillatory term is sufficiently smooth in the $p$ variable
(and this holds for all of the boundary terms).  We can therefore use the triangle inequality once
more, this time being careful to account for a logarithm coming from te weak localization to the annulus
(but converging due to the $\Ft{R}$ localization) to obtain
\[
\|E^{bd} a\|_{C^{2d+1}_r} \leq \eps^{0.8}|\log\eps| \|a\|_{C^{2d+1}_r}.
\]
The remaining terms to consider are the bulk terms.  These come from the smooth
part of the derivative of $h(s,p)$.  One of these bulk terms is given by
\begin{align*}
E^{blk}_1(x,p) &:= \eps^2 \iint \int_{\Real^d} h'(S-u/2,p) \One_{[0,\sigma]}(S+u/2)
\frac{e^{iu(|p'|^2-|p|^2)/2}}{|p|^2-|p'|^2}  R(p-p') \\
&\qquad \qquad\qquad a(x + (\sigma-S)p + Sp', p')\diff u \diff S \diff p',
\end{align*}
To bound this term we integrate by parts once more in the $u$ variable.
This produces more boundary terms $E^{blk,bd}_j$ that are controlled using the argument above, in addition
to one more bulk term having improved decay off of the annulus, so for example one of the terms is
\begin{align*}
E^{blk,blk}_1 a (x,p) &:= \eps^2 \iint \int_{\Real^d} h''(S-u/2,p) \One_{[0,\sigma]}(S+u/2)
\frac{e^{iu(|p'|^2-|p|^2)/2}}{(|p|^2-|p'|^2)^2}  R(p-p') \\
&\qquad \qquad\qquad a(x + (\sigma-S)p + Sp', p')\diff u \diff S \diff p'.
\end{align*}
This term is bounded using the same argument as $E^{KE}a$
since now the localization to the annulus is strong enough
that the triangle inequality can again be used.  This concludes our proof of~\eqref{eq:error-Ck-bd}.

Now proceed with the calculation of the main term $Ma(x,p)$, which we recall below
\begin{align*}
Ma(x,p) =
\eps^2 \int
&\chi_{[0,\sigma]}(S-u/2,p)\chi_{[0,\sigma}(S+u/2,p)
e^{iu(|p'|^2-|p|^2)/2} \\
&\qquad \Ft{R}(p-p') a(x+(\sigma-S)p + S'p, p')\diff p'\diff S \diff u.
\end{align*}
A main character in this calculation is the function
\begin{equation}
g_{p,\sigma}(S,\beta)
= \int
\chi_{[0,\sigma]}(S-u/2,p)\chi_{[0,\sigma}(S+u/2,p)
e^{iu \beta} \diff u.
\end{equation}
By integrating by parts in the $u$ variable, we can see that
this function decays rapidly when
$|\beta|\gtrsim |p|r^{-1}$.
Using the coarea formula in the form
\begin{equation}
\label{eq:coarea}
\int_{\Real^d} f(p')  \diff p'
= \int_\Real \int_{\{|p|^2-|p'|^2=\beta\}} \frac{1}{2|p'|}
f(p')\diff \Hausdorff^{n-1}(p')
\diff\beta,
\end{equation}
we have the identity
\[
Ma(x,p) = \eps^2 \int_{[0,\sigma]}\int_{\Real} \diff\beta
g_{p,\sigma}(S,\beta)
\int_{\{(|p|^2-|p'|^2)/2=\beta\}}\frac{1}{|p'|}
R(p-p') a(x+(\sigma-S)p+Sp',p') \diff\Hausdorff^{n-1}(p')
\]
By the Fourier inversion formula,
\[
\int g_{p,\sigma}(S,\beta)\diff \beta =
\chi_{[0,\sigma]}(S,p)^2.
\]
The idea is that the integration in $\beta$ will therefore recover the term
$\chi_{p,\sigma}(S)^2$, so that $Ma(x,p)$ can be compared to the function
\[
T^M_1a(x,p)
:= \eps^2\int_0^t\diff s \int\diff p' \delta(|p'|^2-|p|^2)
\chi_{[0,\sigma]}(S,p)^2
R(p-p') a(x+sp + (t-s)p', p'),
\]
which we will later show is close to $T_1a$.
We will change variables to move the integration over
the sphere $|p'|^2/2 = |p|^2/2 +\beta$ to integration over the sphere $|p'|=|p|$
in order to better compare the terms.  To this end we choose $\kappa(p,\beta)$
satisfying
\[
|(1+\kappa(p,\beta)) p|^2 = |p|^2 + \beta.
\]
The function $\kappa(p,\beta)$ has the explicit formula
\[
\kappa(p,\beta) = \sqrt{1 + \frac{\beta}{|p|^2}} - 1.
\]
For $|p| > 100 r^{-1}$ (which holds except when $\beta \gg \eps^{0.2}$, which is
suppressed due to the decay of $g$), we have the bound
\[
\kappa(\beta,p) \leq C \frac{\beta}{|p|^2}.
\]

Then applying the change of variables $p' = (1+\kappa)q$, we can write
\begin{equation}
\begin{split}
Ma(x,p)
= \eps^2\int_0^\sigma &\int \diff\beta g_{p,\sigma}(S,\beta)
\int_{\{|q|=|p|\}}\frac{(1+\kappa)^{d-2}}{|p|} \\
&R(p-(1+\kappa)q) a(x+(\sigma-S)p+S(1+\kappa)q,(1+\kappa)q)
\diff\Hausdorff^{n-1}(q).
\end{split}
\end{equation}
The classical counterpart $T^M_1a$ can be written
\[
T^M_1 a(x,p) =
\eps^2\int_0^\sigma \int \diff\beta g_{p,\sigma}(S,\beta)
\int_{\{|q|=|p|\}}\frac{1}{|p|}
R(p-q) a(x+(\sigma-S)p+Sq,q) \diff\Hausdorff^{n-1}(q).
\]
To compare these integrals we define the function
\[
H(x,p;\kappa,S,q)
:= |p|^{-1} (1+\kappa)^{d-2}
R(p-(1+\kappa)q) a(x+(\sigma-S)p+S(1+\kappa)q,(1+\kappa)q).
\]
Then
\begin{equation}
\label{eq:kappa-comparison}
Ma(x,p) - T^M_1a(x,p)
= \eps^2\int_0^\sigma \int g(S,\beta)
\big(
\int_{\{|q|=|p|\}}
\int_0^{\kappa(p,\beta)}
\partial_\kappa H(x,p;\kappa,S,q)\diff \kappa\diff\mcal{H}^{n-1}(q) \big)\diff \beta\diff S
\end{equation}
And
\begin{equation}
\begin{split}
\partial_\kappa H(x,p;\kappa,S,q) &=
\frac{(d-2)(1+\kappa)^{d-3}}{|p|}
R(p-(1+\kappa)q) a(x+(\sigma-S)p+S(1+\kappa)q,(1+\kappa)q) \\
&+ \frac{(1+\kappa)^{d-2}}{|p|}
q\cdot\nabla R(p-(1+\kappa)q)
a(x+(\sigma-S)p+S(1+\kappa)q,(1+\kappa)q) \\
& +\frac{(1+\kappa)^{d-2}}{|p|}
R(p-(1+\kappa)q)
q\cdot \nabla_2a(x+(\sigma-S)p+S(1+\kappa)q,(1+\kappa)q) \\
& +\frac{(1+\kappa)^{d-2}}{|p|}
R(p-(1+\kappa)q)
Sq\cdot \nabla_1a(x+(\sigma-S)p+S(1+\kappa)q,(1+\kappa)q).
\end{split}
\end{equation}
Of these terms the most substantial are the latter two terms in which the gradient hits $a$.
Using the scaling of the derivatives in $C^k_r$ and the fact that $S\ll r^2$ we obtain the bound
\[
\|\partial_\kappa H\|_{C^k_r} \leq  r(1+|p|^{-1})
\|a\|_{C^{k+1}_r}.
\]
However to prove Lemma~\ref{lem:Tone-compare} we need to assume more than microscopic smoothness on
$a$.  Assuming more smoothness on $a$, we have the bound
\[
\|\partial_\kappa H\|_{C^k_r} \leq  r\delta (1+|p|^{-1})
\|a\|_{C^{k+1}_{r,\delta}}.
\]
Then applying the triangle inequality to~\eqref{eq:kappa-comparison} we obtain the bound
\begin{equation}
\begin{split}
\| Ma - T^M_1a\|_{C^{2d+1}_r}
&\leq  C \eps^2 \sigma  \delta \|a\|_{C^{2d+2}_{r,\delta}}.
\end{split}
\end{equation}

To conclude the proof we need to estimate $\|T^M_1a - T_1a\|_{C^{2d+1}_r}$.  This function is given
by the integral
\[
T^M_1 a - T_1a
= \int_0^\sigma (1-\chi_{[0,\sigma]}(S,p)^2)
\int \delta(|p'|^2-|p|^2) R(p-p') a(x+sp +(\sigma-s)p',p')\diff s\diff p'
\]
To obtain bounds on the momentum derivatives we treat the radial derivative $p\cdot \nabla_p$
separately from tangential derivatives $\vec{v}\cdot \nabla_p$ where $\vec{v}(p)\cdot p=0$ is a tangential
vector field.  For radial derivatives we integrate by parts in the $p'$ variable using the identity
\[
p\cdot\nabla p \delta(|p'|^2 - |p|^2) = -p'\cdot\nabla_{p'} \delta(|p'|^2-|p|^2).
\]
The tangential derivatives on the other hand vanish on the delta function.  Combining this argument with a
triangle inequality in the integral over $s$ and the fact that $1-\chi_{[0,\sigma]}^2$ is supported
on $[0,r|p|^{-1}]\cup [\sigma-r|p|^{-1},\sigma]$ and $|p|\geq \eps^{0.2}$, we conclude
\[
\|T^M_1a - T_1a\|_{C^{2d+1}_r} \lsim \eps^{0.8} \|a\|_{C^{2d+1}_r}.
\]
\end{proof}

\begin{proof}[Proof of Lemma~\ref{lem:Ttwo-compare}]
The proof of Lemma~\ref{lem:Ttwo-compare} is very similar to the proof of
Lemma~\ref{lem:Tone-compare}, and for this reason we abbreviate the argument.
Using
\[
2\Rept e^{iu(|p'|^2-|p|^2)} = e^{iu(|p'|^2 - |p|^2)} + e^{-iu(|p'|^2 - |p|^2)}
\]
to symmetrize the domain of integration over $u$, we can rewrite~\eqref{eq:wigner-two}
as
\[
\Wigner{(\mcal{T}_{0,2}+\mcal{T}_{2,0})[\Op(a)]}(x,p)
= \eps^2 \int_0^\sigma \int_{-(\sigma-S)}^{\sigma-S} \int_{\Real^d}
e^{iu(|p'|^2-|p|^2)/2}R(p'-p)
a(x-(\sigma-u/2)p+up'/2,p)\diff p'\diff u\diff S.
\]
Before directly comparing this function to $T_2 a$, we compare it to the intermediate
quantity
\[
\wtild{T}_2a(x,p)
:= \eps^2 \int_0^\sigma \int_{-(\sigma-S)}^{\sigma-S} \int_{\Real^d}
e^{iu(|p'|^2-|p|^2)/2}R(p'-p)
a(x-\sigma p,p)\diff p'\diff u\diff S.
\]
The function
\[
e(x,p;u,p') :=
a(x-\sigma p,p) - a(x-\sigma p + u (p-p')/2, p)
\]
satisfies
\[
\|e\|_{C^{2d+1}_r} \lsim \delta (r^{-1} |up'| +u|p|) \|a\|_{C^{2d+2}_{r,\delta}}.
\]
We use the smoothness of $a$ to scale $r^{-1}$ in $p$ to reduce the domain of integration
to $|up'|\lsim r$.

Again we would like to integrate by parts in the $u$ variable to obtain localization of $p'$ to
the annulus.  Following the proof of Lemma~\ref{lem:Tone-compare}, we decompose
the indicator function $\One_{[-(\sigma-S),\sigma-S]}$ as the sum of a piece that is smooth
to scale $r|p|^{-1}$ and a piece that is supported on two intervals of that width.  The bounds for
the error pieces follow exactly as in the previous proof.  What remains is a main term of the form
\[
m(x,p) =
\eps^2 \int_0^\sigma \chi(S;p)
\int_{-(\sigma-S)}^{(\sigma-S)} \chi_S(u;p)\int_{\Real^d}
e^{iu(|p'|^2-|p|^2)}R(p'-p) a(x-(\sigma-u/2)p+up'/2,p)\diff p'\diff u\diff S.
\]
We define the function
\[
g(\beta; S,\sigma,p)
:= \int \chi_S(u;p) e^{iu\beta}\diff u,
\]
and observe that the smoothness of $\chi_S$ enforces that $g(\beta;S,\sigma,p)$ decays rapidly
for $\beta \geq r^{-1}|p|$.  Moreover, by the Fourier inversion formula,
\[
\int g(\beta;S,\sigma,p)\diff \beta = 1.
\]
Then using the coarea formula as in~\eqref{eq:coarea} we have
\[
m(x,p) = \eps^2 \int_0^\sigma
\chi(S;p)\int \diff \beta g(\beta;S,\sigma,p)
\int_{(|p|^2-|p'|^2)/2=\beta}
|p'|^{-1}
R(p-p') a(x + (\sigma-u/2)p + Sp', p')
\]

\end{proof}

\section{Using graphs, forests, and partitions to compute moments}
\label{sec:forest-appendix}
In this section we introduce an abstract lemma about graphs that is
useful in a variety of contexts within this paper.  The simplest
context is in the computation of the moments
\[
\Expec \prod_{j=1}^k \Ft{V_{y_j}}(q_j)
= \int e^{-iq_j\cdot x_j}
\big(\Expec \prod_{j=1}^k V(x_j)\big)
b_r(x_j-y_j)\diff \mbf{x}.
\]
In the integral above, we can split the expectation into smaller
pieces depending on the configuration of $\mbf{x}$.  More precisely,
let $G(\mbf{x})$ be the graph on index set $[1,k]$ with edges $(i,j)$
when $|x_i-x_j|\leq 1$.  Letting $P(\mbf{x})\in\mcal{P}([k])$ be the
partition of $[k]$ into the connected components of $G(\mbf{x})$,
we have
\[
\Expec \prod_{j=1}^k V(x_j)
= \prod_{S\in P(\mbf{x})}
\Expec \prod_{j\in S} V(x_j).
\]
One way to perform the desired splitting is to write, for each pair
$i<j$,
\[
1 = (1-\One(|x_i-x_j|\leq 1)) + \One(|x_i-x_j|\leq 1).
\]
Taking a product over all $i$ and $j$ produces the identity
\begin{equation}
\label{eq:split-one-graphs}
\begin{split}
1 &= \prod_{i<j} ((1-\One(|x_i-x_j|\leq 1)) + \One(|x_i-x_j|\leq 1)) \\
&= \sum_G \prod_{(i,j)\in G} \One(|x_i-x_j|\leq 1)
\prod_{(i,j)\not\in G} (1-\One(|x_i-x_j|\leq 1)),
\end{split}
\end{equation}
where the sum is over all graphs on vertex set $[k]$.
The problem with this approach is that the number of graphs scales like
$2^{\Theta(n^2)}$, whereas we can typically only afford factors on
the order $2^{\Theta(n\log n)}$.   Since we only need enough information
about $G(\mbf{x})$ to determine $P(\mbf{x})$ (of which there are at most
$n^n$ possibilities), this seems like a wasteful sum.

The solution is to work with a compressed representation of $G(\mbf{x})$
that only records its minimal spanning forest.  For $n\in\bbN$,
let $\binom{[n]}{2}$ be the set of possible edges $\{(i,j)\mid 1\leq i<j\leq n\}$
on a graph with vertex set $[n]$.  Let $w:\binom{[n]}{2}\to\Real_+$ be
any function that assigns each edge a unique positive weight.
Given a graph $G$, let $P(G)\in\mcal{P}([n])$ be the partition of $[n]$
into the connected components of $G$.
Then, given a graph $G\subset\binom{[n]}{2}$, the minimal spanning
forest $F_G$ of $G$ is defined to be the unique graph having the same connected
components $(P(G)=P(F_G))$ and minimizing the sum $\sum_{e \in F} w(e)$.

The graph $F_G$ is acyclic, as otherwise edges can be trimmed to reduce
the total weight of the graph.  A formula due to Borchardt but attributed
to Cayley states that the number of trees on a graph of $n$ vertices is
exactly $n^{n-2}$.  Any acylic graph on $n$ vertices can be obtained by
forming a tree on $n+1$ vertices and deleting a vertex.  Therefore
\[
N_F(n) \leq (n+1)^{n-1} \leq Cn^n.
\]
This is closer to the number of partitions on $n$ vertices, so there is
not as much wasted information.

We will start with the identity
\begin{equation}
\label{eq:forest-decomposition}
1 = \sum_F \One(F_G = F),
\end{equation}
which trivially holds for any graph $G$, where the sum is over all acyclic graphs $F$.
The good thing about this sum is that it is a sum over only $O(n^n)$ terms, and
$F_G$ determines $P(G)$, so the decomposition gives enough information to split the
expectation.

We set $\mcal{G}(F)$ to be the collection of graphs
\[
\mcal{G}(F) = \{G\subset \binom{[n]}{2} \mid F_G = F\}.
\]
We will express the set $\mcal{G}(F)$ using sets that depend on more local information
about the graph.  In particular, given a set $S\subset[n]$ of vertices and a graph $G$,
we write $G|_S$ for the induced graph on $S$.  A ``local approximation'' to $\mcal{G}(F)$
is the collection $\mcal{L}(F)$,
\[
\mcal{L}(F_0) = \{G\subset \binom{[n]}{2} \mid F_{G|_S} = F_0|_S \text{ for all } S\in P(F_0)\}.
\]
We note that $\mcal{L}(F)$ has an alternative definition.  Given a forest $F$ and a weight $w$,
define the weight
\[
w_F(e) :=
\begin{cases}
w(e), & P(F \cup\{e\}) = P(F) \\
w(e) + \max_e w, & P(F \cup\{e\}) > P(F).
\end{cases}
\]
That is, $w_F$ adds extra weight to each edge between distinct connected components of $F$.
Let $F(G;w_F)$ be the minimal forest of the graph $G$ with the modified weight $w_F$.  Then
\[
\mcal{L}(F_0) = \{G\subset \binom{[n]}{2} \mid F(G;w_{F_0}) \supset F_0\}.
\]
This motivates the definition of the collection
\[
\mcal{G}(F';F_0) := \{G\subset \binom{[n]}{2} \mid F(G;w_{F_0}) = F'\}.
\]
With this definition, it is clear that we have the identity
\[
\One_{\mcal{L}(F_0)} = \sum_{F'\supset F_0} \One_{\mcal{G}(F';F_0)}.
\]
We can now apply Mobius inversion in the partially ordered set $\mcal{F}$ of forests, ordered
by inclusion.  Note that $\mcal{F}$ is an ideal in the poset of graphs, so the Mobius function
is the same.  The Mobius inversion formula, along with the  identity
$\mcal{G}(F_0;F_0) = \mcal{G}(F_0)$, yields
\begin{equation}
\label{eq:forest-indicator-decomposition}
\One_{\mcal{G}(F_0)} = \sum_{F'\supset F_0} (-1)^{|F'\setminus F_0|} \One_{\mcal{L}(F')}.
\end{equation}

We conclude this discussion by deriving polynomial identities from the
decompositions~\eqref{eq:forest-decomposition} and~\eqref{eq:forest-indicator-decomposition}.
Given a forest $F_0$, let $D(F_0)$ be the set of edges ``dominated'' by $F$,
\[
D(F_0) := \{e \in\binom{[n]}{2} \mid F_{F_0\cup\{e\}} = F_0\}.
\]
That is, $D(F_0)$ consists of those edges that are already connected by $F_0$ along
a path consisting of edges $e'$ with $w(e')<w(e)$.  Another important collection of edges
are those ``connected'' by $F_0$,
\[
C(F_0) := \{e \in\binom{[n]}{2} \mid P(F_0\cup\{e\}) = P(F_0)\}.
\]
Then we can express the indicator functions $\One_{\mcal{G}(F)}$ and $\One_{\mcal{L}(F)}$
as follows:
\[
\One_{\mcal{G}(F)} = \prod_{e\in F} \One(e\in G) \prod_{e\not\in D(F)}(1-\One(e\in G))
\]
and
\[
\One_{\mcal{L}(F)} = \prod_{e\in F} \One(e\in G) \prod_{e\in C(F)\setminus D(F)} (1-\One(e\in G)).
\]
Therefore, by~\eqref{eq:forest-decomposition}, we have that for any $\{0,1\}$-valued variables
$Y_e$,
\begin{equation}
\label{eq:poly-forest-decomp}
1 = \sum_{F} \prod_{e\in F} Y_e \prod_{e\not\in D(F)} (1-Y_e).
\end{equation}
The right hand side is a multilinear polynomial that agrees with $1$ on the hypercube $\{0,1\}^N$, so it follows
that in fact~\eqref{eq:poly-forest-decomp} holds for any variables $Y_e$.  Similarly, we have
that
\begin{equation}
\label{eq:poly-indicator-decomp}
\prod_{e\in F} Y_e \prod_{e\not\in D(F)} (1-Y_e)
=
\sum_{F'\supset F} (-1)^{|F'\setminus F|}
\prod_{e\in F'} Y_e \prod_{e\in C(F)\setminus D(F)} (1-Y_e).
\end{equation}
holds for any values $Y_e$ and any forest $F$.

\subsection{Application to computing moments}

As a useful example, we can now estimate the moments of the Fourier transform of a localized potential,
\[
\Expec \prod_{j=1}^k \Ft{V_{y_j}}(q_j).
\]

\begin{definition}[Localized potentials]
    For $y\in\Real^d$, define $V_y$ to be the function
    \[
        V_y(x) = b_r(x) V(x-y),
    \]
    where $b_r$ is a function satisfying $\int b_r(x)\diff x=1$,
    $\supp b_r\subset B_{10r}$, and
    \[
        |\Ft{b_r}(p)| \leq C \exp(-c(r|p|)^{0.999})
    \]
\end{definition}

\begin{definition}[Admissible potentials]
    \label{def:admissible-V}
    A random potential $V$ is \emph{admissible} if it is stationary,
    mean-zero, has range-$1$ dependence, and satisfies the following moment bound
    \begin{equation}
        \label{V-point-est}
        \sup_{\mbf{x}} \Big|\Expec \prod_{j=1}^k \partial^{\alpha_j}
        V(x_j)\Big| \leq (C_V k)^{k}
    \end{equation}
    for any sequence of multi-indices $\alpha_j$ with
    $0\leq |\alpha_j|\leq 20d$.

    Alternatively, $V$ is also admissible if it is Gaussian
    and stationary with a two-point correlation function
    $\Expec V(x)V(y)=R(x-y)$ satisfying $\supp R\subset B_{2r}$ and
    \[
        |\partial^\beta R(x)|
        \leq C (1+|x|)^{-d-1-|\beta|},
    \]
    for multi-indices of order $0\leq|\beta|\leq 20d$.
\end{definition}

Aside from Gaussian potentials, examples of admissible potentials include
potentials constructed from a point process, such as
\[
    V(y) = \sum_{x\in X} V_0(y-x)
\]
with a fixed smooth compactly supported potential
$V_0$ satisfying $\int V_0=0$, where $X$ is (for example) a Poisson
process with unit intensity.

Another example of an admissible potential is the lattice construction
\[
    V(y) = \sum_{x\in\bbZ^d} c_x V_0(y-x-h),
\]
where $V_0$ is a smooth and compactly supported function and $c_x$
are independent Bernoulli $\{\pm 1\}$ coinflips, and $h\in[0,1]^d$ is
a uniformly random shift (which makes $V$ stationary).

\begin{lemma}
    \label{lem:admissible-V}
    If $V$ is an admissible potential according
    to Definition~\ref{def:admissible-V}, then
    \begin{equation}
        \label{eq:Vhat-moment}
\begin{split}
    \Big|\Expec \prod_{j=1}^k \partial_q^{\alpha_j} \Ft{V_{y_j}}(q_j)\Big|
    \leq
r^{-dk} &r^{\sum_{j\in[k]}|\alpha_j|} \sum_{P\in\mcal{P}([k])}
\prod_{S\in P} (C C_V|S|)^{2|S|}
r^d \exp\Big(-c_k \big|r \sum_{j\in S} q_j\big|^{0.99}\Big) \\
&\qquad\qquad\qquad\qquad \times \prod_{j\in [k]} (1+|q_j|)^{-20d}.
\end{split}
    \end{equation}
    for any sequence $y_j\in\Real^d$.
\end{lemma}

\emph{Remark:}
The estimate~\eqref{eq:Vhat-moment}, along with stationarity, are the two
conditions we actually need for the proof.  In particular, this implies
that independent sums of admissible potentials are also valid potentials for
the proof.

\begin{proof}
We first expand the Fourier transform as follows:
\begin{align*}
    \Expec \prod_{j=1}^k
    \partial^{\alpha_j}\Ft{V_{y_j}}(q_j)
    &= \int
    \Expec \prod_{j\in [k]} V(x_j)
    e^{-i\sum q_j\cdot x_j} \prod_{j=1}^k x_j^{\alpha_j} b_r(x_j-y_j)
    \diff x_1\cdots\diff x_k.
\end{align*}
To split the expectation, let $P(\mbf{x})\in\mcal{P}([k])$ be the partition of $[k]$
into connected components of the graph $G(\mbf{x}) = \{(i,j)\mid |x_i-x_j|\leq 1\}$.
Then for any $P\geq P(\mbf{x})$,
\[
\prod_{j=1}^k V(x_j)
=
\prod_{S\in P} \Expec \prod_{j\in S} V(x_j).
\]

We will use~\eqref{eq:poly-forest-decomp} and
Let $b\in C_c^\infty(\Real^d)$ be a smooth function
that is identical to $1$ on the unit ball $B_1$, has support in $B_2$, and satisfies
\[
|\Ft{b}(p)| \leq C\exp(-c|p|^{0.999}).
\]
Then by~\eqref{eq:poly-forest-decomp}, we have for any $\bm{x}\in(\Real^d)^k$
\[
1 = \sum_F \prod_{(i,j)\in F} b(x_i-x_j) \prod_{(i,j)\not\in D(F)} (1-b(x_i-x_j)).
= \sum_F \phi_F(\bm{x}).
\]
Note that, for $\bm{x}\in\supp\phi_F$, $P(\mbf{x})\leq P(F)$.
We use this decomposition in the integral,
\begin{align*}
    \Expec \prod_{j=1}^k
    \partial^{\alpha_j}\Ft{V_{y_j}}(q_j)
    &= \sum_{F} \int
\phi_F(\mbf{x})
\prod_{S\in P(F)}
    \Expec \prod_{j\in S} V(x_j)
    e^{-i\sum q_j\cdot x_j} \prod_{j=1}^k x_j^{\alpha_j} b_r(x_j -y_j)
    \diff \mbf{x}.
\end{align*}
Next we use the identity~\eqref{eq:poly-indicator-decomp} to write
\[
\phi_G(\mbf{x}) = \sum_{F'\supset F} (-1)^{|F'\setminus F|} \psi_{F'}(\bm{x}),
\]
where
\[
\psi_F(\mbf{x}) = \prod_{e\in F} b(x_i-x_j) \prod_{e\in C(F)\setminus D(F)} (1-b(x_i-x_j)).
\]
Note that $\psi_F$ can be expressed in the form
\[
\psi_F(\mbf{x}) = \prod_{S\in P(F)} (\prod_{e\in F|_S} b(x_i-x_j) \prod_{e\in C(F|_S)\setminus D(F)} (1-b(x_i-x_j)))
=: \prod_{S\in P(F)} \psi_{F,S}(\mbf{x}).
\]
We have therefore arrived at the expression
\begin{align*}
    \Expec \prod_{j=1}^k
    \partial^{\alpha_j}\Ft{V_{y_j}}(q_j)
    &= \sum_{F\subset  F'} (-1)^{|F'\setminus F|} \int
\prod_{S\in P(F)}
    \psi_{F,S}(\mbf{x})
    \Expec \prod_{j\in S} V(x_j)
    e^{-i\sum q_j\cdot x_j} \prod_{j=1}^k x_j^{\alpha_j} b_r(x_j -y_j)
    \diff \mbf{x}.
\end{align*}
We also need to introduce the following change of variables.  Given the forest $F'\geq F$,
which induces a partition
$P'=P(F')\in\mcal{P}([k])$, we define the variables $x_S = x_{\min S}$ for each $S\in P'$.
Also for each $j\in S$ we set $z_j = x_j - x_S$.  Let $S^+ := S\setminus \min\{S\}$, and define
$\mbf{z}_{S^+} = (z_j)_{j\in S^+}$.  By the stationarity of the potential $V$, we have
\[
\prod_{S'\subset S}\Expec \prod_{j\in S'} V(x_j) =
\prod_{S'\subset S}
\Expec\prod_{j\in S'} V(z_j),
\]
where the product over $S'$ ranges over sets $S'\in P(F)$.
Moreover, we have
\[
\sum_{j\in S} x_j\cdot q_j = \big(\sum_{j\in S} q_j\big) \cdot x_S + \sum_{j\in S^+}
q_j\cdot z_j.
\]
Then applying the change of variables $(x_j)_{j\in [k]}\mapsto ((x_S)_{S\in P(F')}, ((z_j)_{j\in S})_{S\in P(F')})$,
we can write
\begin{align*}
    \Expec \prod_{j=1}^k
    \partial^{\alpha_j}\Ft{V_{y_j}}(q_j)
    &= \sum_{F\leq F'} (-1)^{|F'\setminus F|}
%\sum_{\bm{\beta}}
\int \prod_{S\in P(F')} \psi_{F',S}(\mbf{z}_S)
e^{i\sum q_j\cdot z_j}
\prod_{S\subset P(F)} \Expec \prod_{j\in S} V(z_j)
%P_{\bm{\beta}}(\mbf{z})
\\
&\qquad\qquad\qquad\qquad
\prod_{S\in P'(F)} \Big(\int
e^{i(\sum_{j\in S} q_j)\cdot x_S}
\prod_{j\in S} (x_S+z_j)^{\alpha_j} b_r(x_S+z_j-y_j)
\diff x_S\Big)
\diff\mbf{z}_{P(F')}.
\end{align*}

The inner integral satisfies is the Fourier transform of the function
\[
\prod_{j\in S} (x_S+z_j)^{\alpha_j} b(x_S+x_j-y_j),
\]
which is itself a product of $|S|$ functions whose Fourier transform has decay
of the form $\exp(-c(r|p|)^{0.999})$.  Thus
\[
\big|
\int e^{i(\sum_{j\in S} q_j)\cdot x_S}
\prod_{j\in S} (x_S+z_j)^{\alpha_j} b_r(x_S+z_j-y_j)
\diff x_S\big|
\leq \big(\prod_{j\in S} r^{\alpha_j}\big) r^{(1-|S|)d}
\exp(-c (|S|^{-1} |\sum_{j\in S}q_j|)^{0.999}).
\]
The factor $r^{(1-|S|)d}$ comes from a factor of $r^d$ which is the volume of the domain of integration
and a factor of $r^{-d}$ from each function $b_r$ (which are normalized so that $\int b_r = 1$).
Taking a derivative of the integral with respect to the $z_j$ variable at worst produces a factor of $\alpha_j$
and also a factor of $r^{-1}$.  Thus for $\beta_j \leq 20d$,
\[
\big|
\prod_{j\in S}
\partial_{z_j}^{\beta_j}
\int e^{i(\sum_{j\in S} q_j)\cdot x_S}
\prod_{j\in S} (x_S+z_j)^{\alpha_j} b_r(x_S+z_j-y_j)
\diff x_S\big|
\leq C \big(\prod_{j\in S} r^{\alpha_j}\big) r^d
\exp(-c (|S|^{-1} |\sum_{j\in S}q_j|)^{0.999}).
\]
To obtain the desired decay in the $q_j$ variables we integrate by parts $20d$ times in each $q_j$
variable satisfying $|q_j|\geq 1$.

We therefore obtain a bound of the form
\begin{align*}
\big|
&\int \prod_{S\in P(F')} \psi_{F',S}(\mbf{z}_S)
e^{i\sum q_j\cdot z_j}
\prod_{S\subset P(F)} \Expec \prod_{j\in S} V(z_j)
%P_{\bm{\beta}}(\mbf{z})
\\
&\qquad\qquad\qquad\qquad
\prod_{S\in P'(F)} \Big(\int
e^{i(\sum_{j\in S} q_j)\cdot x_S}
\prod_{j\in S} (x_S+z_j)^{\alpha_j} b_r(x_S+z_j-y_j)
\diff x_S\Big)
\diff\mbf{z}_{P(F')}\big|\\
&\qquad\qquad\qquad\qquad\qquad
\lsim
r^{\sum |\alpha_j|} r^{-kd}
\prod_{S\in P(F)} (C_V k)^k
\prod_{j\in[k]}(1+|q_j|)^{-20d}
\prod_{S\in P(F')}
r^d \exp(-c(|S|^{-1}|\sum_{j\in S}q_j|)^{0.999}).
\end{align*}
We sum over all $F$ and $F'$ with fixed partitions, observing that the number of forests having partition $P(F)=P$
is given by $\prod_{S\in P} |S|^{|S|-2}$, by Cayley's formula.
\end{proof}

\section{Wavepackets and quantization}
\label{sec:wp-quantization}
Throughout the paper we will use wavepackets to decompose functions, operators, and
quantum channels.  Our starting point is the definition of a wavepacket.  We will fix
a length scale $r=\eps^{-1}$ and an envelope $\chi_{env}\in\Schwartz(\Real^d)$ which
is compactly supported in $B(0,1)$ and satisfies the Fourier decay estimate
\[
|\Ft{\chi_{env}}(p)|\leq C \exp(-c |p|^{0.9}).
\]
We normalize $\chi_{env}$ so that $\|\chi_{env}\|_{L^2}=1$.
Then given $(x_0,p_0)\in\PhaseSpace$, we define the functions
\[
\phi_{x_0,p_0}(y) = r^{-d/2} e^{ip_0\cdot y} \chi_{env}(r^{-1}(y-x_0)).
\]
We will use greek letters like $\xi$ and $\eta$ to denote points in phase
space $\PhaseSpace$ and will write $\phi_\xi$ for the wavepacket localized
to $\xi$.  We also write $\ket{\xi}$ to denote the wavefunction $\phi_\xi$,
and use Dirac notation.  This means that for example we may write
$\braket{\xi|\psi}$ to mean $\langle \phi_\xi,\psi\rangle$.  For
$p\in\Real^d$ we also use $\ket{p}$ and $\bra{p}$ to denote the momentum
eigenfunctions $e^{ip\cdot x}$.  Thus for example
$\braket{p|\xi} = \Ft{\phi_\xi}(p)$.

For points $\xi=(\xi_x,\xi_p)$ and
$\eta=(\eta_x,\eta_p)$ in phase space we use the metric
\[
d_r(\xi,\eta) := r^{-1}|\xi_x - \eta_x|
+ r|\xi_p-\eta_p|.
\]
This metric is the natural metric for comparing wavepackets at
scale $r$.  The following bound holds for the overlap of wavepackets
\begin{equation}
\label{eq:wp-overlap-bd}
|\braket{\xi|\eta}|
\leq \exp(-c d_r(\xi,\eta)^{0.9}).
\end{equation}
To prove~\eqref{eq:wp-overlap-bd} observe that it suffices to consider
$|\xi_x-\eta_x|\leq 2r$, as otherwise $\braket{\xi|\eta}=0$.  The
bound then follows from the Fourier decay estimate on $\chi_{env}$.

We will occasionally use the
one-parameter family of maps $U_s:\PhaseSpace\to\PhaseSpace$
\[
U_s(\xi) = (\xi_x + s\xi_p, \xi_p).
\]

In Dirac notation, the Fourier inversion formula is written
\[
\Id = \int_{\Real^d} \ket{p}\bra{p}\diff p.
\]
An elementary calculation involving the Fourier inversion formula also
shows that we can decompose the identity operator into wavepackets as
follows
\begin{equation}
\label{eq:id-decomp}
\Id = \int_{\PhaseSpace} \ket{\xi}\bra{\xi}\diff \xi.
\end{equation}
This is equivalent to the identity
\[
\langle f, g\rangle
= \int_{\PhaseSpace} \langle f, \phi_\xi\rangle
\langle \phi_\xi, g\rangle \diff \xi.
\]
As a consequence of~\eqref{eq:id-decomp} we have the following
expression for the $L^2$ norm of a function $\psi\in L^2(\Real^d)$
\begin{equation}
\label{eq:ltwo-wp}
\|\psi\|_{L^2}^2
= \braket{\psi|\psi}
= \int \braket{\psi|\xi}\braket{\xi|\psi}\diff\xi
= \int |\Braket{\psi|\xi}|^2 \diff \xi.
\end{equation}

The identity~\eqref{eq:id-decomp} can be used to decompose arbitrary
bounded operators $A\in\mcal{B}(L^2(\Real^d))$.
\begin{equation}
\label{eq:wp-A-decomp}
\begin{split}
A = \Id A\Id
&= \int \ket{\xi}\braket{\xi|A|\eta}\bra{\eta}\diff\xi\diff\eta \\
&= \int \ket{\xi}\bra{\eta} \braket{\xi|A|\eta}\diff\xi\diff\eta.
\end{split}
\end{equation}
In the last line we simply rearranged $\braket{\xi|A|\eta}$, which
is the scalar $\braket{\phi_\xi, A\phi_\eta}$ and therefore commutes
with multiplication by $\bra{\eta}$.

In this paper we will be working with operators more so than
wavefunctions.  We will primarily bound operators using the operator
norm, and we use the following version of the Schur test to estimate
the operator norm.
\begin{lemma}[Schur test in the wavepacket frame]
\label{lem:schur-test}
Let $A\in\mcal{B}(L^2(\Real^d))$ be an operator of the form
\[
A = \int \ket{\xi}\bra{\eta} a(\xi,\eta)\diff \xi\diff\eta.
\]
Then
\[
\|A\|_{op}^2 \leq C
\big(\sup_\xi\int |a(\xi,\eta)|\diff \eta\big)
\big(\sup_\eta\int |a(\xi,\eta)|\diff \xi\big)
\]
\end{lemma}
\begin{proof}
Let $\psi\in L^2(\Real^d)$.  The first step is to expand,
using~\eqref{eq:ltwo-wp},
\begin{align*}
\|A\psi\|_{L^2}^2
&= \int |\Braket{\xi|A\psi}|^2 \diff \xi \\
&= \int \big|\int \Braket{\xi|\xi'}\braket{\eta|\psi}a(\xi',\eta)\diff\xi'\diff\eta\big|^2 \diff \xi.
\end{align*}
For simplicity now define
$b(\xi,\eta) = \int \braket{\xi|\xi'}a(\xi',\eta)\diff\xi'$.  Then,
using the Cauchy-Schwarz inequality we write
\begin{align*}
\|A\psi\|_{L^2}^2
&\leq
\int
\big(\int |\Braket{\eta|\psi}|^2 |b(\xi,\eta)|\diff \eta\big)
\big(\int |b(\xi,\eta)|\diff \eta\big) \diff \xi \\
&\leq
\sup_\xi \int |b(\xi,\eta)|\diff \eta
\int
\int |\Braket{\eta|\psi}|^2 |b(\xi,\eta)|\diff \eta  \\
&\leq
(\sup_\xi \int |b(\xi,\eta)|\diff \eta)
(\sup_\eta \int |b(\xi,\eta)|\diff \xi)
\|\psi\|_{L^2}^2.
\end{align*}
It remains to show that
\[
\int |b(\xi,\eta)|\diff \eta \leq  C
\sup_\xi
\int |a(\xi,\eta)|\diff \eta.
\]
and
\[
\int |b(\xi,\eta)|\diff \xi \leq  C
\sup_\eta
\int |a(\xi,\eta)|\diff \xi.
\]
The first follows from
\begin{align*}
\int |b(\xi,\eta)|\diff \eta
&= \int |\Braket{\xi|\xi'}||a(\xi',\eta)|\diff \xi'\diff\eta \\
&\leq \int |\Braket{\xi|\xi'}|\diff \xi'
\sup_\xi' \int |a(\xi',\eta)|\diff \eta,
\end{align*}
and the second from an analogous calculation.
\end{proof}

The Weyl quantization map $\Op^w$ takes as input a classical
observable $a\in C^\infty(\PhaseSpace)$ and produces an operator
$\Op^w(a)$ defined by
\[
\Op^w(a)f(x)
:= \int e^{i(x-y)\cdot p} a(\frac{x+y}{2}, p) f(y)\diff y\diff p.
\]
We define rescaled $C^k$ norms for phase space as follows
\[
\|a\|_{C_r^k(\PhaseSpace)}
:=
\sum_{|\alpha|\leq k}
\sup_{\xi\in\PhaseSpace}
|(r\partial_x)^{\alpha_x} (r^{-1}\partial_p)^{\alpha_p} a(\xi)|.
\]

Using this norm we derive the following bounds for the Weyl
quantization map.
\begin{lemma}
\label{lem:weyl-wp-bd}
If $a\in C^k(\PhaseSpace)$ then
\begin{equation}
|\Braket{\xi|\Op^w(a)|\eta}|
\leq C (1+d_r(\xi,\eta))^{-k} \|a\|_{C^{k+d+1}_r}.
\end{equation}
\end{lemma}
\begin{proof}
First we write out explicitly
\begin{align*}
\braket{\xi| \Op^w(a)|\eta}
=
r^{-d}
\int e^{i[x\cdot (p-\xi_p) + y\cdot (\eta_p-p)]} a(\frac{x+y}{2},p)
\chi_{env}(r^{-1}(x-\xi_x))
\chi_{env}(r^{-1}(y-\eta_x))\diff y\diff p\diff x.
\end{align*}
There are several cases to consider.  The first case is that
$d_r(\xi,\eta)\leq K$.  In this case we simply aim to show that
the integral is bounded by a constant.  To do this we will split up
the integral over the $p$ variable as follows.
Let $\rho \in C_c^\infty$ be a smooth cutoff function that
is supported in $B(0,2r^{-1})$ and is identically $1$ on the ball
$B(0,r^{-1})$.  We decompose the integration over $p$ using
the identity
\[
1 = (1-\rho(p-\xi_p))
+ \rho(p-\xi_p)(1-\rho(p-\eta_p)) + \rho(p-\xi_p)\rho(p-\eta_p).
\]
In the second and third terms, we can use localization to a ball
of radius $r^{-1}$ in $p$ and simply apply the triangle inequality
to bound the integral by a constant $C$.
For the first term we integrate by parts $d+1$ times in the variable $x$
to obtain integrable decay in the $p$ variable outside the ball
$B(\xi_p,r^{-1})$.

The second case is that $d_r(\xi,\eta)\geq K$ and
$r^{-1}|\xi_x-\eta_x| \geq \frac{1}{2}d_r(\xi,\eta)$.
In this case we integrate by parts $k$ times in the $p$ variable to obtain
decay in $|x-y|$, using $|x-y|\geq \frac{1}{4}|\xi_x-\eta_x|$
on the support of the envelopes.  An additional integration by parts of
order $d+1$ is required in the $x$ variable to obtain the required
localization in the $p$ variable as described above.

The final case is that $d_r(\xi,\eta)\geq K$ and
$r|\xi_p-\eta_p| \geq \frac{1}{2}d_r(\xi,\eta)$.  In this case
we perform the change of
variables $(x,y)\mapsto (\overline{x},\Delta x) :=
(\frac{x+y}{2}, x-y)$ and integrate by parts in the $\overline{x}$
variable $k$ times.  A final integration by parts $d+1$ times
in the $\Delta x$ variable is required to obtain integrability in the
$p$ variable.
\end{proof}

We note that combining Lemma~\ref{lem:weyl-wp-bd} and
Lemma~\ref{lem:schur-test} provides an alternative proof of the
Calder{\'o}n-Vaillancourt theorem on the $L^2$ boundedness
of pseudodifferential operators.

We conclude this section by comparing the Weyl quantization to the wavepacket
quantization $\Op(a)$ which we define as
\[
\Op(a) := \int \ket{\xi}\bra{\xi} a(\xi)\diff \xi.
\]
The key relationship between these quantizations is that the Wigner
function of the wavpeacket envelope, $W_{env}(x,p) := \Wigner{\chi_{env}}(x,p)$,
satisfies
\[
\Op^w(W_{env}) = \ket{(0,0)}\bra{(0,0)}.
\]
Then applying translation and modulation symmetries of the Wigner function
and of the wavepacket frame, this implies
\[
\Op^w(W_{env}(\cdot-x_0,\cdot-p_0)) = \ket{(x_0,p_0)}\bra{(x_0,p_0)}.
\]
Therefore
\[
\Op(a) = \Op^w(a\ast W_{env}),
\]
so
\[
\|\Op(a) - \Op^w(a)\|_{op} \leq \|a-a\ast W_{env}\|_{C^{2d+1}_{r}}.
\]
Then using the fundamental theorem of calculus on each partial derivative we
obtain the following bound for $\delta<1$:
\[
\|a - a(\cdot-x_0,\cdot-p_0)\|_{C^{2d+1}_r}
\leq (\delta (r^{-1}|x_0|+r|p_0|) + \delta^{2d+1}) \|a\|_{C^{2d+1}_{r,\delta^{-1}}}.
\]
Then, since $W_{env}(x,p)$ is localized to the region
$|x|\lsim r$ and $|p|\lsim r^{-1}$, we conclude
\[
\|a - a\ast W_{env}\|_{C^{2d+1}_r} \leq C \delta \|a\|_{C^{2d+1}_{r,\delta^{-1}}}.
\]

Therefore we conclude
\begin{equation}
\label{eq:quantization-compare}
\|\Op(a) - \Op^w(a)\|_{op} \leq C\delta\|a\|_{C^{2d+1}_{r,\delta^{-1}}}.
\end{equation}

\section{Elementary estimates for the linear Boltzmann equation}
\label{sec:regularity}
In this section we prove some simple estimates for solutions to the linear
Boltzmann equation
\begin{equation}
\label{eq:linear-boltz}
\partial_t f + p\cdot\nabla_x f = Lf,
\end{equation}
where $L$ is the scattering operator
\[
Lf(x,p) = \eps^2 \int \delta(|p|^2 - |q|^2) G(p-q)[f(x,q)-f(x,p)] \diff q.
\]
We are interested in rescaled $C^k$ estimates for the solution.  Recall
the $C^k_{r,L}$ norm is defined by
\[
\|f\|_{C^k_{r,L}} := \sum_{|\alpha_x|+|\alpha_p|\leq k}
\sup_{x,p} |(r L \partial_x)^{\alpha_x} (r^{-1} L \partial_p)^{\alpha_p} f(x,p)|.
\]
This norm measures  the smoothness of the function $f$ to scale $Lr$ in space and $Lr^{-1}$ in momentum.

To solve~\eqref{eq:linear-boltz} we use a series expansion
\begin{equation}
\label{eq:f-series}
f_t = \sum_{j=0}^\infty g_{t,j}
\end{equation}
where
\[
g_{t,0}(x,p) = f_0(x-tp,p)
\]
and
\begin{equation}
\label{eq:g-iteration}
g_{t,j+1}(x,p) = \int_0^t (Lg_{t-s,j})(x-sp,p)\diff s.
\end{equation}
The partial sums
\[
f_{t,N} := \sum_{j=0}^N g_{t,j}
\]
solve
\[
\partial_t f_{t,N} + p\cdot\nabla_x f_{t,N} = L f_{t,N-1}.
\]
Our task is to show that the terms $g_{t,j}$ are summable in $C^k_{r,L}$ for sufficiently small $t$.
We will prove the following result.
\begin{lemma}
\label{lem:boltz-regularity}
For every $k\in\bbN$, and dimension $d\geq 2$, there exists a constant $C=C(k,d)$ such that the following holds:
If $\supp f\subset \{(x,p)|\mid |p|\geq \theta\}$
and $T\leq C^{-1}\delta^{-1} \min\{r^{-2}, \eps^2 \theta^{k-1}\}$ then the
series~\eqref{eq:f-series} converges and the terms $g_{j,t}$ satisfy
\[
\|g_{j,t}\|_{C^k_{r,\delta L}} \leq
(C\eps^2 t \theta^{1-k})^j \|f\|_{C^k_{r,L}}.
\]
In particular, it follows that, for $t\leq T$,
\[
\|f_t\|_{C^k_{r,\delta L}} \leq C \|f\|_{C^k_{r,L}}.
\]
Moreover
\[
\|f_t - f_{t,1}\|_{C^k_{r,L}} \leq C (\eps^2 t\delta \theta^{1-k})^2 \|f\|_{C^k_{r,L}}.
\]
\end{lemma}

This reduces to demonstrating bounds for the scattering operator $L$ in the $C^k_{r,L}$ norm.
To do this we will treat the radial and tangential derivatives separately, which is why we require an
assumption that $f$ is supported away from $|p|=0$.

\subsection{The $C^k$ norm in polar coordinates}

Given a vector field $V\in C^\infty(\Real^d)$ we define the derivative operator
\[
D_V f = V(p)\cdot\nabla_p f.
\]
We say that a vector field $V$ is tangential if $p\cdot V(p)=0$ and we write $R(p) = p/|p|$ for the radial
unit vector field.

\begin{lemma}
\label{lem:polar-Ck}
Let $f\in C^\infty(\Real^d\setminus B_\theta)$.  Then for any multi-index $|\alpha|=k$,
\[
|\partial^\alpha f(x,p)|
\leq C_{k,d} \sup_{W_1,\cdots,W_k} |D_{W_1}D_{W_2}\cdots D_{W_k} f(x,p)|
+ \sum_{j<k} \theta^{j-k} \|f\|_{C^j}.
\]
where each vector field $W_j$ is either tangential or is the radial vector field $R$,
and the tangential vector fields $W_j$ are smooth on the unit sphere and are $0$-homogeneous.
\end{lemma}
\begin{proof}
The proof is by induction, with the case $k=1$ following from the fact that for any $p$ one can construct an
orthonormal basis from the radial vector $p/|p|$ and a collection of $d-1$ tangential unit vectors, and the bound
\[
\|v\| \leq C_d \max_{1\leq j\leq d} |v\cdot e_j|
\]
which holds for any orthonormal basis $e_j$.

To induct we observe that
\[
D_{W_1}D_{W_2}\cdots D_{W_k} f
= \sum_{j_1,\cdots,j_k=1}^d  \prod_{i=1}^k W_{i,j_i}  (\partial_{j_1}\partial_{j_2}\cdots\partial_{j_k} f)
+ \sum_{j=0}^{k-1} E_j[W_1,\cdots,W_k]f(p),
\]
where the remainder term $E_j[W_1,\cdots,W_k]f$ involves a sum of lower order terms in which $f$ is differentiated
$j$ times and $k-j$ derivatives of the vector fields $W$ appear.  On the set $\Real^d\setminus B_\theta$,
the vector fields satisfy
\[
|\partial^\alpha W| \leq C \theta^{-|\alpha|}
\]
by $0$-homogeneity.  In particular we conclude that
\[
\big|\sum_{j_1,\cdots,j_k=1}^d  \prod_{i=1}^k W_{i,j_i}  (\partial_{j_1}\partial_{j_2}\cdots\partial_{j_k} f)\big|
\leq
|D_{W_1}D_{W_2}\cdots D_{W_k} f|
+ C_k \sum_{j=0}^{k-1} \theta^{j-k}\|f\|_{C^j}.
\]
The result follows from summing over all choices of tangential and radial derivatives (fixing a basis of tangential
derivatives).
\end{proof}

To control higher order derivatives of the form $D_{W_1}\cdots D_{W_k}$ applied to the scattering operator $L$
we need to first work out some exact computations.  Define the generalized scattering operator with kernel $H$
to be
\[
L_H f(p) = \int \delta(|p|^2-|q|^2) H(p,q) [f(q)-f(p)]\diff q.
\]
For convenience, we split this into an ``incoming'' part $L_H^+$ and an ``outgoing'' part
$L_H^-$, with
\[
L_H^+(p) = \int \delta(|p|^2-|q|^2) H(p,q) f(q)\diff q
\]
and
\[
L_H^-(p) = -\int \delta(|p|^2-|q|^2) H(p,q) f(p)\diff q.
\]

\begin{lemma}
\label{lem:derivative-calculations}
If $H$ is a $C^\infty$ scattering kernel and $V$ is a tangential vector field, we have the identities
\begin{align}
\label{eq:radial-derivative}
D_R L_Hf &= L_H D_R f + L_{H_R} f \\
\label{eq:tangential-derivative}
D_V L_Hf &= L_H^- D_V f + L_{D_V[p] H} f
\end{align}
where
\[
H_R(p,q) = (\frac{p}{|p|}\cdot\nabla_p + \frac{q}{|q|}\cdot \nabla_q + d) H(p,q).
\]
and
\[
D_V[p] H(p,q) = p\cdot\nabla_p H(p,q).
\]
\end{lemma}
\begin{proof}
For the radial derivative, we start with the observation
\[
p\cdot \nabla_p \delta(|p|^2-|q|^2)
= -q\cdot\nabla_q \delta(|p|^2-|q|^2).
\]
On integrating by parts we have
\begin{align*}
D_R Lf
&= D_R[p](\int \delta(|p|^2 -|q|^2) H(p-q) [f(q)-f(p)]\diff q) \\
&= \int \delta(|p|^2-|q|^2) (D_R[p] H)(p-q) [f(q)-f(p)] \diff q \\
&\qquad + \int \frac{1}{|p|}\delta(|p|^2-|q|^2) \nabla_q \cdot (q (H(p-q)[f(q)-f(p)])) \diff q \\
&\qquad - \int \delta(|p|^2-|q|^2) H(p-q) D_Rf(p)\diff q \\
&= L_{D_R[p] H} f + L_H^-(D_R f)
+ \int \frac{1}{|p|}\delta(|p|^2-|q|^2) \nabla_q \cdot (q (H(p-q)[f(q)-f(p)])) \diff q.
\end{align*}
The calculation then follows from applying the product rule.

For the tangential derivative, we use the fact that
\[
V(p) \cdot \nabla_p \delta(|p|^2-|q|^2) = 0
\]
for tangential vector fields $V$.  Then~\eqref{eq:tangential-derivative} follows
from a direct calculation.
\end{proof}

We are now ready to prove our main bound.
\begin{lemma}
\label{lem:L-Ck-bd}
Let $H$ be a $C^\infty$ scattering kernel satisfying
\begin{equation}
\label{eq:H-symb-est}
|\partial^\alpha H(p,q)| \leq C_\alpha A \theta^{-|\alpha|},
\end{equation}
and suppose that $f$ is a function supported in $\{|p|\geq \theta\}$.  Then for $k\geq 1$,
\[
\|L_H f\|_{C^k} \leq  A C_k\theta^{1-m} \|f\|_{C^k}.
\]
\end{lemma}
\begin{proof}
The proof is by induction on $k$.  When $k=1$ this follows from Lemma~\ref{lem:polar-Ck} and
Lemma~\ref{lem:derivative-calculations}.
For larger $k$ it suffices to show that
\[
\|D_{W_1}\cdots D_{W_k} L_Hf\|_{C^k}
\leq \theta^{1-k} \|f\|_{C^k}.
\]
for any mixture of tangential and radial vector fields $W_1,\cdots,W_k$.
To induct we use Lemma~\ref{lem:derivative-calculations}.
If $W_k$ is radial then
\[
D_{W_1}\cdot D_{W_{k-1}} D_{W_k} L_Hf
= D_{W_1}\cdot D_{W_{k-1}} L_H (D_{W_k}f)
+ D_{W_1}\cdot D_{W_{k-1}} L_{H_R} f
\]
whereas if $W_k$ is tangential
\[
D_{W_1}\cdot D_{W_{k-1}} D_{W_k} L_Hf
= D_{W_1}\cdot D_{W_{k-1}} L^-_H (D_{W_k}f)
+ D_{W_1}\cdot D_{W_{k-1}} L_{D_{W_k}[p]H} f.
\]
Both terms are bounded using the inductive hypothesis, observing
that if $H$ satisfies~\eqref{eq:H-symb-est} for some $A$ then
$D_{W_k} H$ and $H_R$ satisfy~\eqref{eq:H-symb-est} with constant $A\theta^{-1}$.
\end{proof}

\subsection{Convergence of the series}
Before we address the convergence of the series~\eqref{eq:f-series} we quickly observe that the free
transport operator
\[
T_s f(x,p) := f(x-sp,p)
\]
satisfies the bound
\[
\|T_s f\|_{C^k_{r,L}} \leq (1 + Csr^{-2})\|f\|_{C^k_{r,L}}.
\]
Using $T_s$ we can rewrite the iteration~\eqref{eq:g-iteration} defining $g_j$ as follows:
\[
g_{j+1,t} = \int_0^t T_s L g_{j,t-s}\diff s.
\]
Then applying the triangle inequality and Lemma~\ref{lem:L-Ck-bd} we have
\[
\|g_{j+1,t}\|_{C^k_r} \leq C t \eps^2 (1+Ctr^{-2}) \theta^{1-k} \sup_{0\leq s\leq t}\|g_{j,t}\|_{C^k_r}.
\]
This concludes the proof of Lemma~\ref{lem:boltz-regularity}.

%
%\printbibliography
\bibliographystyle{plain}
\bibliography{refs}
\end{document}